%% file: main.tex
\documentclass[letter,11pt]{article}
\newcommand{\cO}{\mathcal{O}}
\newcommand{\term}[1]{#1}
\usepackage{xspace}
\usepackage{nicefrac}

\usepackage{thm-restate}
\newcommand{\pname}{\textsc}

\newcommand{\ProblemFormat}[1]{\pname{#1}}
\newcommand{\ProblemIndex}[1]{\index{problem!\ProblemFormat{#1}}}
\newcommand{\ProblemName}[1]{\ProblemFormat{#1}\ProblemIndex{#1}{}\xspace}

 \newcommand{\probVDHp}[1]{\ProblemName{Vertex Deletion to #1}}
 
 \newcommand{\probVDH}{\ProblemName{Vertex Deletion to $\cH$}}
 
 \newcommand{\probEDH}{\ProblemName{Elimination Distance to $\cH$}}

 \newcommand{\probTDH}{\ProblemName{Treewidth Decomposition to $\cH$}}

\newcommand{\probPVD}{\ProblemName{Planar Vertex Deletion}}
\newcommand{\probCVD}{\ProblemName{Chordal Vertex Deletion}}
\newcommand{\probIVD}{\ProblemName{Interval Vertex Deletion}}
\newcommand{\probOCT}{\ProblemName{Odd Cycle Transversal}}
\newcommand{\probFVS}{\ProblemName{Feedback Vertex Set}}
 \newcommand{\probVC}{\ProblemName{Vertex Cover}}
\newcommand{\rw}{\mathbf{rwd}}
\newcommand{\tw}{\mathbf{tw}}
\newcommand{\td}{\mathbf{td}}
\newcommand{\ed}{\mathbf{ed}}
\newcommand{\md}{\mathbf{mod}}

\newcommand{\Oh}{\mathcal{O}}
\newcommand{\hh}{\ensuremath{\mathcal{H}}}
\newcommand{\SSS}{\ensuremath{\mathcal{S}}}

\newcommand{\mso}{{\sf CMSO}\xspace}
\newcommand{\fii}{{\sf FII}\xspace}
\newcommand{\msotwo}{{\sf CMSO}\xspace}
\usepackage{graphicx,color}
\usepackage{empheq}
\usepackage{tcolorbox}
\definecolor{blueish}{rgb}{0.122, 0.435, 0.698}
\definecolor{dagstuhlyellow}{rgb}{0.99,0.78,0.07}
\definecolor{lightgray}{rgb}{0.9,0.9,0.9}
\newtcbox{\colbox}{
size=title,
  nobeforeafter,
  colframe=white,
  colback=blue!5!white,
  arc=10pt,
  tcbox raise base}

\newcommand{\hhtw}[1][\hh]{\tw_{#1}}

\newcommand{\hhdepth}[1][\hh]{\ed_{#1}}
\newcommand{\hhmd}[1][\hh]{\md_{#1}}
\newcommand{\hhtwfull}[1][\hh]{{#1}-treewidth}

\input{Preamblenew.tex}
\input{preamble.tex}

\begin{document}
%\title{Deleting, Eliminating and Decomposing Are Synonymous} 

\title{\vspace{-0cm} \bf Deleting, Eliminating and Decomposing to Hereditary Classes Are All FPT-Equivalent}

%\author{
%Anonymous Submission
%}
\date{}
%\title{\LARGE\sffamily\bfseries\mathversion{bold} Deleting, Eliminating and Decomposing to Hereditary Classes Are All \FPT-Equivalent}

%\title{{\sc Elimination Distance} is \FPT\ whenever {\sc Vertex Deletion} is!}

%\title{Elimination Distance to $\cal F$ is \FPT\ for Every Class $\cal F$\\
%(Terms and Conditions Apply!) }
%\title{Deleting and Eliminating Vertices to Hereditary Classes is FPT-Equivalent } 

%\author{Akanksha Agrawal}{Indian Institute of Technology Madras, Chennai, India
%}{agrawal@post.bgu.ac.il}{}{}%TODO mandatory, please use full name; only 1 author per \author macro; first two parameters are mandatory, other parameters can be empty. Please provide at least the name of the affiliation and the country. The full address is optional
%
%\author{Lawqueen Kanesh}{National University of Singapore, Singapore, Singapore}{lawqueen@comp.nus.edu.sg}{}{}
%
%\author{Daniel Lokshtanov}{University of California Santa Barbara, Santa Barbara, USA}{daniello@ucsb.edu}{}{}
%\author{Fahad Panolan}{Indian Institute of Technology, Hyderabad, India}{panolan@iith.ac.in}{}{}
%\author{M. S. Ramanujan}{University of Warwick, Coventry, UK}{R.Maadapuzhi-Sridharan@warwick.ac.uk}{}{}
%\author{Saket Saurabh}{University of Bergen, Bergen, Norway \and Institute of Mathematical Sciences, Chennai, HBNI, India}{saket@imsc.res.in}{}{}
%%%\authorrunning{A. Agrawal et al.} 
%

\author{
Akanksha Agrawal\thanks{Indian Institute of Technology Madras, Chennai, India. \texttt{akanksha@cse.iitm.ac.in}}
\and 
Lawqueen Kanesh\thanks{National University of Singapore, Singapore, Singapore, \texttt{lawqueen@comp.nus.edu.sg}}
\and 
Daniel Lokshtanov\thanks{University of California Santa Barbara, USA. \texttt{daniello@ucsb.edu}}
  \and 
 Fahad Panolan\thanks{Indian Institute of Technology, Hyderabad, India, \texttt{fahad@cse.iith.ac.in}}
 \and 
 M. S. Ramanujan\thanks{University of Warwick, Coventry, UK, \texttt{R.Maadapuzhi-Sridharan@warwick.ac.uk}}
 \and  Saket Saurabh\thanks{The Institute of Mathematical Sciences, HBNI, Chennai, India, and University of Bergen, Bergen, Norway. \texttt{saket@imsc.res.in}}
  \and Meirav Zehavi\thanks{Ben-Gurion University of the Negev, Beersheba, Israel. \texttt{meiravze@bgu.ac.il}}
}

\maketitle

\thispagestyle{empty}

\maketitle

\begin{abstract}

  Vertex-deletion problems have been at the heart of parameterized complexity throughout its history. Here, the aim is to determine the minimum size (denoted by $\hhmd$) of a modulator to a graph class $\cH$, i.e., a set of vertices whose deletion results in a graph in $\cH$. 
 Recent years have seen the development of a research programme where the complexity of modulators is measured in ways other than size. 
  For instance, for a graph class $\cH$, the graph parameters elimination distance to $\cH$ (denoted by $\hhdepth$) [Bulian and Dawar, Algorithmica, 2016], and $\cH$-treewidth (denoted by $\hhtw$) [Eiben et al. JCSS, 2021]  aim to minimize the treedepth and treewidth,  respectively, of the ``torso" of the graph induced on a modulator to the graph class $\cH$. Here, the torso of a vertex set $S$ in a graph $G$ is the graph with vertex set $S$ and an edge between two vertices $u, v \in S$ if there is a path between $u$ and $v$ in $G$ whose internal vertices all lie outside $S$. 

  In this paper, we show that from the perspective of (non-uniform) fixed-parameter tractability (FPT), the three parameters described above give equally powerful parameterizations 
  for every hereditary graph class $\cH$ that satisfies mild additional conditions. 
  In fact, 
  we show that for every hereditary graph class $\cH$ satisfying mild additional conditions, with the exception of $\hhtw$ parameterized by $\hhdepth$, for every pair of these parameters, computing one parameterized by itself or any of the others is {\FPT}-equivalent to the standard vertex-deletion (to $\cH$) problem. As an example, we prove that an {\FPT} algorithm for the vertex-deletion problem implies a non-uniform {\FPT} algorithm for computing $\hhdepth$ and $\hhtw$. 
  
  The conclusions of non-uniform {\FPT} algorithms being somewhat unsatisfactory,  we essentially prove that if $\cH$ is hereditary, union-closed, CMSO-definable, and (a) the canonical equivalence relation (or any refinement thereof) for membership in the class can be efficiently computed, or (b) the class admits a ``strong irrelevant vertex rule", then there exists a uniform {\FPT} algorithm for $\hhdepth$. Using these sufficient conditions, we obtain {\em uniform} {\FPT} algorithms for computing $\hhdepth$, when $\cH$ is defined by excluding a finite number of connected (a)~minors, or   (b)~topological minors, or (c)~induced subgraphs, or when  $\hh$ is any of bipartite, chordal or interval graphs. 
 For most of these problems, the existence of a uniform {\FPT} algorithm has remained open in the literature. In fact, for some of them, even a non-uniform {\FPT} algorithm was not known. For example, Jansen et al.~[STOC 2021] ask for such an algorithm when $\cH$ is defined by excluding a finite number of connected topological minors. We resolve their question in the affirmative.

%  Our main contribution is the unification of 
%  
% 
% 
% In this line of research, Bulian and Dawar [Algorithmica, 2016] introduced the notion of elimination distance and showed that deciding whether a given graph has elimination distance at most $k$ to any minor-closed class of graphs is fixed-parameter tractable parameterized by $k$ [Algorithmica, 2017]. There has been a subsequent series of results on the fixed-parameter tractability of elimination distance to various graph classes. However, one class of graph classes to which the computation of elimination distance has remained open is the class of graphs that are characterized by the exclusion of a family $\cF$ of finite graphs as topological minors. 
%
%In this paper, we settle this question by showing that the problem of determining elimination distance to such graphs is also fixed-parameter tractable. 
\end{abstract}
\newpage
\pagenumbering{gobble}
\tableofcontents

%\begin{itemize}
%\item directed graphs 
%\item scattered families
%\item W[hard] but approximation implies the same for us?
%\item define problems in appendix
%\end{itemize}
\newpage

\pagestyle{plain}
\setcounter{page}{1}
\pagenumbering{arabic}
\section{Introduction}

%\begin{mytheo}{}{theoexample}
%  This is the text of the theorem. The counter is automatically assigned and,
%  in this example, prefixed with the section number. This theorem is numbered with
%  \ref{th:theoexample} and is given on page \pageref{th:theoexample}.
%\end{mytheo}

\input{intro-elimDist-top-minor-free.tex}

\section{Preliminaries}
\label{sec:prelims}
\input{prelim}

%%%%%Proof of Thm 1%%%%%%%
\section{Structural Results} \label{sec:struct}
		\subsection{Bounded Modulators on Unbreakable Graphs} \label{sec:proof-of-sep-lemma}
		\input{bounded-mod}

		\input{proof-of-sep-lemma}

		\subsection{Computing \hED\ Using its Decision Oracle}
		\input{compute-ED}

		\subsection{Computing \hTD\ Using its Decision Oracle} 
		\input{compute-decomposition}

\section{Equivalences Among Deletion, Decomposition and Elimination}\label{sec:equivalencesAmongDeletionDecompositionAndElimination}
\input{schema}

	\subsection{Proof of Lemma~\ref{lem:one-implies-5-9}}\label{sec:one-to-five-nine} 
	\input{one-implies-five}

	\subsection{Prove of Lemma~\ref{lem:three-implies-one}}\label{sec:one-implies-three}
	\input{one-implies-three}

	\subsection{Proof of Lemma~\ref{lem:4/7-to-1}}\label{sec:4/7-to-1}
	\input{4,7-implies-1}
	
	%\subsection{Generalization to Structures}
	%\input{gen-struct}
%%%%%%Proof of Thm ends%%%%%

%\section{Equivalences Among Deletion, Decomposition and Elimination}
%\input{equiv.tex}
\section{Applications of Theorem~\ref{thm:mainEquiv} -- Beyond Graphs}\label{sec:cutProblems}
	\input{applications}

\section{Cross Parameterizations}\label{sec:crossParam}
\input{cross-param}

\input{uniform}

\section{Applications of the Uniform Algorithm}\label{sec:applicationsOfUniformAlgorithm}
\input{apps}

\section{Conclusion}\label{sec:conclusion}

\input{elimination-conclusion}

\bibliographystyle{plainurl}
\bibliography{elimination,references,metakernels}
%\bibliography{elimination,referencesTwo,references-old.bib,referencesMZ,referencesFPTapprox.bib,references}
%\bibliography{references}

\appendix
\input{appendix}

\end{document}

%% file: intro-elimDist-top-minor-free.tex
%%!TEX root = main.tex

Delineating borders of fixed-parameter tractability (FPT) is arguably the biggest quest of parameterized complexity. Central to this goal is the class of vertex-deletion problems (to a graph class $\hh$). In the {\probVDH} problem, the goal is to compute a minimum set of vertices to delete from the input graph, in order to obtain a graph contained in $\hh$.  For a graph $G$, $\hhmd(G)$ denotes the size of a smallest vertex set $S$ such that $G-S \in {\cal H}$. If $G-S\in \cH$, then $S$ is called a {\em modulator to $\cal H$}.  The parameterized complexity of vertex-deletion problems has been extensively studied for numerous choices of $\cH$, e.g.,  when ${\cal H} $ is planar, bipartite, chordal, interval, acyclic or edgeless, respectively, we get  the classical \probPVD,  \probOCT,  \probCVD,  \probIVD,  \probFVS, and \probVC problems (see, for example, \cite{CyganFKLMPPS15}).
 
 In the study of the parameterized complexity of vertex-deletion problems, 
 solution size and various graph-width measures have been the most frequently considered parameterizations, and over the past three decades their power has been well understood. In particular, numerous vertex-deletion problems are known to be fixed-parameter tractable (i.e., can be solved in time $f(k)n^{\bigoh(1)}$, where $k$ is the parameter and $n$ is the input size) under these parameterizations. In light of this state of the art, recent efforts have shifted to the goal of identifying and exploring ``hybrid" parameters such that:
\vspace{-3 pt}
 \begin{enumerate}[{\bf (a)}]\setlength{\itemsep}{-3pt}
\item   they are upper bounded by the solution size as well as certain graph-width measures, \item  they can be arbitrarily (and simultaneously) smaller than both the solution size and graph-width measures, and \item  a significant number of the problems that are in the class FPT when parameterized by solution size or graph-width measures, can be shown to be in the class FPT also when parameterized by these new parameters. \end{enumerate}

%\subsection{Hybrid Parameterizations}
Two recently introduced parameters in this line of research are: 
(a) $\hh$-elimination distance and (b) 
\hhtwfull{} of~$G$. The $\hh$-elimination distance of  a graph $G$ to $\cH$ (denoted $\hhdepth(G)$) was introduced by Bulian and Dawar~\cite{BulianD16} and roughly speaking, it expresses the number of rounds needed to obtain a graph in $\cH$ by removing one vertex from every connected component in each round. We refer the reader to Section~\ref{sec:prelims} for a more formal definition. The reader familiar with the notion of treedepth~\cite{NesetrilM06} will be able to see that this closely follows the recursive definition of treedepth. That is, if $\cH$ is the class of empty graphs, then the $\cH$-elimination distance of $G$ is nothing but the treedepth of $G$.  In fact, if $\cH$ is union-closed ({\em as will be the case for all graph classes we consider in this paper}), then one gets the following equivalent perspective on this notion. The $\hh$-elimination distance of $G$ is defined as the minimum possible {\em treedepth} of the torso of a modulator of $G$ to $\hh$. Here, the torso of a vertex set $S$ in a graph $G$ is the graph with vertex set $S$ and an edge between two vertices $u, v \in S$ if there is a path between $u$ and $v$ in $G$ whose internal vertices all lie outside $S$. 
Consequently, it is easy to see that $\hh$-elimination distance of $G$ is always upper bounded by both the size of the smallest modulator of $G$ to $\hh$ (i.e., $\hhmd(G)$) as well as the treedepth of $G$\footnote{For this, we always assume that $\hh$ contains the empty graph and so $V(G)$ is a trivial modulator to $\hh$.}. 

 The second parameter of interest for us is \hhtwfull{} which, roughly speaking, aims to ``generalize" treewidth and solution size, the same way that elimination distance aims to generalize treedepth and solution size. This notion was recently introduced by Eiben et al.~\cite{EibenGHK19} and builds on a similar hybrid parameterization which was first developed in the context of solving CSPs~\cite{GanianRS17d} and found applications also in   algorithms for SAT~\cite{GanianRS17a} and Mixed ILPs~\cite{GanianOR17}.   
Specifically, 
% 
% 
% 
% 
%  the value of \hhtwfull of $G$ (denoted $\hhtw(G)$) is the minimum possible width of 
  a tree $\hh$-decomposition of a graph~$G$ of width $\ell$ is a tree decomposition of $G$ along with a set $L\subseteq V(G)$ (called base vertices), such that  
% 
% 
% 
% Intuitively 
% a tree $\hh$-decomposition of a graph~$G$ is a tree decomposition of $G$ where  arbitrarily large bags are permitted, as long as 
% 
% 
%  that induce subgraphs belonging to the target class of the deletion problem as long as these subgraphs have small neighborhoods.
% 
% ~tree decomposition of~$G$, together with a set~$L \subseteq V(G)$ of base vertices. Base 
(i) each  base vertex appears in exactly one bag, (ii) the base vertices in a bag induce a subgraph belonging to~$\hh$, and (iii) the number of non-base vertices in any bag is at most $\ell+1$.  The value of \hhtwfull\ of $G$ (denoted $\hhtw(G)$) is the minimum width taken over all tree $\hh$-decompositions of $G$. The utility of this definition arises from the fact that \hhtwfull\ of $G$ is always upper bounded by the treewidth of $G$ (indeed, one could simply take a tree-decomposition of $G$ attaining the treewidth of $G$ and set $L=\emptyset$) and moreover, one can design fixed-parameter algorithms for several  problems parameterized by the $\hhtw$ of the graph~\cite{EibenGHK19,JansenKW21}. On the other hand, for union-closed graph classes $\cH$, the $\cH$-treewidth of $G$ is nothing but the minimum possible {\em treewidth} of the torso of a modulator of $G$ to $\hh$. This immediately implies that the $\cH$-treewidth of $G$ is always upper bounded by both the treewidth of $G$ as well as the size of the smallest modulator of $G$ to $\hh$ (i.e., $\hhmd(G)$). 
%execute a dynamic programming algorithm on a tree $\hh$-decomposition of graph~$G$

It is fairly easy to see that $\hhdepth(G)$ (respectively, $\hhtw(G)$) can be arbitrarily smaller than both $\hhmd(G)$ and the treedepth of $G$ (respectively, the treewidth of $G$). Similarly, $\hhtw(G)$ itself can be arbitrarily smaller than $\hhdepth(G)$. 
We refer the reader to \cite{GanianRS17,JansenK021} for some illustrative examples of this fact.  As  both these parameters satisfy Properties {\bf (a)} and {\bf (b)} listed above, there has been a sustained effort in the last few years to investigate the extent to which Property {\bf (c)} is satisfied by these two parameters. In this effort, one naturally encounters the following two fundamental algorithmic questions:

\smallskip

\begin{tcolorbox}[colback=blue!5!white,colframe=white!100!black] 
% \begin{tcolorbox}[colback=gray!5!white,colframe=gray!75!black]
\begin{description}
\item[Question 1:]  For which families $\hh$ of graphs  is  \probEDH (\probTDH)  \FPT\ when parameterized by $\hhdepth$ (respectively, $\hhtw$)? 

\label{abc}
\item[Question 2:]  For which families $\hh$ of graphs is \probVDH parameterized by  $\hhdepth(G)$ ($\hhtw(G)$) \FPT? 
\end{description}
 \end{tcolorbox}
 
% All the work so far in this line of research has focused on
In {\probEDH} ({\probTDH}), the input is a graph $G$ and integer $k$ and the goal is to decide whether $\hhdepth(G)$ (respectively, $\hhtw(G)$) is at most $k$. The parameter in both problems is $k$.  
Both the questions listed above are extremely wide-ranging and challenging. Indeed, for Question 1, notice that not even an \XP\ algorithm (running in time $n^{\cO(g(k))}$, where $k$ is $\hhdepth$ or $\hhtw$) is obvious even for well-understood graph classes $\cH$, such as bipartite graphs. On the other hand, $\probVDH$ in this case has a trivial  $n^{\cO(k)}$-time algorithm where one simply guesses the minimum modulator to $\cH$ and checks whether the graph induced by the rest of the vertices is bipartite, in linear time. 
In the absence of a resolution to Question 1, Question 2 then brings with it the challenge of solving \probVDH\ (or indeed, any problem) without necessarily being able to efficiently compute $\hhdepth$ or $\hhtw$. 

%obtaining an appropriate decomposition of the input such as a tree $\cH$-decomposition, which can then be exploited to solve \probVDH\ or even other problems. 

%Bulian and Dawar~\cite{BulianD16},  in the paper which introduced the parameter $\hhdepth(G)$, studied {\sc Graph Isomorphism} parameterized by $\hhdepth(G)$, where $\hh$ is a family of graphs of degree at most $d$, and showed that the problem is \FPT\ parameterized by $\hhdepth(G)$ for every fixed value of $d$.
% Let us remark that  in the running time, the polynomial dependence on the input graph depends on $d$ and further the algorithm for {\sc Graph Isomorphism} does not compute an $\hh$-elimination forest of depth $\hhdepth(G)$. 
% In a follow-up work, 
\paragraph{State of the art for Question 1.} 
 In their work, Bulian and Dawar~\cite{BulianD17} showed that the $\probEDH$ 
%deciding whether a given graph $G$ has 
%$\hhdepth(G) \leq k$ 
problem
is \FPT, when $\hh$ is a minor-closed class of graphs and asked whether  
it is \FPT, when $\hh$ is the family of graphs of degree at most  $d$. 
%However, Bulian and Dawar~\cite{BulianD16} showed that Graph Isomorphism is FPT parameterized by $\hhdepth(G)$ {\em without} being able to compute an exact bound on $\hhdepth(G)$, where $\hh$ is the family of graphs of degree at most  $d$.. 
 %$\hhdepth(G)\leq k$ is \FPT, when $\hh$ is a family of graphs of degree $d$. 
 In a partial resolution to this question, Lindermayr et al.~\cite{LindermayrSV20} showed that \probEDH is 
 \FPT\ when we restrict the input  graph to be planar. Finally,  Agrawal et al.~\cite{AgrawalKFR21}, resolved this question completely by showing that the problem is  (non-uniformly) {\FPT} (we refer the reader to Section~\ref{sec:towardsUniform} for a definition of non-uniform {\FPT}). In fact, they obtained their result for all 
 $\hh$ that are characterized by a finite family of induced subgraphs. 
  Recently,  Jansen and de Kroon~\cite{corr/abs-2106-04191} extended the aforementioned result of Agrawal et al.~\cite{AgrawalKFR21} further, and showed that \probTDH is also (non-uniformly) \FPT\ for $\hh$ that are characterized by  a finite family of induced subgraphs. In the same paper they also showed that 
  \probTDH  (\probEDH) is non-uniformly \FPT\  for $\hh$ being the family of bipartite graphs.  Even more recently, Fomin et al.~\cite{corr/abs-2104-02998} showed that for every graph family $\hh$  expressible by a first order-logic formula 
%  $\varphi\in \Sigma_3$, that is,  of the form $\varphi=\exists x_1\cdots \exists x_r\ \ \forall y_{1}\cdots \forall y_{s}\ \ \exists z_1\cdots \exists z_t~~ \psi,$ \todo{Do we need the "that is, ..." part? --msr}
%where $\psi$ is a quantifier-free first-order formula, 
\probEDH is (non-uniformly) \FPT. Since  a family of graphs characterized by a finite set of forbidden induced subgraphs is expressible in this fragment of logic,  this result also generalizes the result of Agrawal et al.~\cite{AgrawalKFR21}. Until this result of Fomin et al., the research on Question 1 has essentially proceeded on a case-by-case basis, where each paper considers a specific choice of $\cH$. 
\paragraph{State of the art for Question 2.} 
In a recent paper, Jansen et al.~\cite{JansenK021} 
%not resolve the parameterized complexity of \probEDH or \probTDH for any family \hh, they 
provide a general framework to  design {\em {\FPT}-approximation} algorithms 
for $\hhdepth$ and $\hhtw$ for various choices of $\cH$. For instance, when  $\cH$ is bipartite or characterized by a finite set of forbidden (topological) minors, they give \FPT\ algorithms (parameterized by $\hhtw$) that compute a tree $\hh$-decomposition of $G$ whose width is {\em  not necessarily optimal}, but polynomially bounded in the $\cH$-treewidth of the input, i.e., an approximation. 
%
%
% Their result can be summarized as follows.  Let $\hh$ be a hereditary class of graphs that is defined by a finite number of forbidden (a)~minors, or (b)~topological minors, or (c)~induced subgraphs, or (d)~$\hh \in \{\mathsf{bipartite},\mathsf{chordal},\mathsf{interval}\}$. There is an algorithm that, given an $n$-vertex graph $G$, computes in $f(\hhtw(G)) \cdot n^{\Oh(1)}$ time a tree $\hh$-decomposition of $G$ of width~$\hhtw(G)^{\Oh(1)}$. Analogously, there is an algorithm that computes in~$f(\hhdepth(G)) \cdot n^{\Oh(1)}$ time an $\hh$-elimination forest of $G$ of depth~$\hhdepth(G)^{\Oh(1)}$. 
These approximation algorithms enable them to address Question 2 for various classes $\cH$ without having to exactly compute $\hhdepth(G)$ or $\hhtw(G)$ (i.e., without resolving Question 1 for these classes). 
 Towards answering Question~$2$, they give the following {\FPT} algorithms for \probVDH parameterized by $\hhtw$. 
  Let $\hh$ be a hereditary class of graphs that is defined by a finite number of forbidden {connected} (a)~minors, or (b) induced subgraphs, or (c)~$\hh \in \{\mathsf{bipartite~graphs},\mathsf{chordal~graphs}\}$. There is an algorithm that, given an $n$-vertex graph $G$, computes a minimum vertex set~$X$ such that~$G - X \in \hh$ in time~$f(\hhtw(G)) \cdot n^{\Oh(1)}$. Note that all of these \FPT\  algorithms are uniform.
%   compared to non-uniform \FPT\ algorithms mentioned......
%   \red{in the last paragraph}\todo{The last paragraph is a long way away and people may have forgotten by this point -- msr}. 

\subsection{Our Motivation}

%For question 1, nothing more general than forbidden subgraphs or forbidden (topological) minors is known.
The starting point of our work lies in the aforementioned recent advances made by Agrawal et al.~\cite{AgrawalKFR21}, Fomin et al.~\cite{corr/abs-2104-02998} and Jansen et al.~\cite{JansenK021}. A closer look at these algorithms shows an interesting property of these 
algorithms: known algorithms for \probEDH and \probTDH utilize the corresponding (known) algorithms for \probVDH in a non-trivial manner. This fact, plus the recent successes in designing (typically, non-uniform) {\FPT} algorithms for  \probEDH and \probTDH naturally raises the following questions. 

\medskip

\begin{tcolorbox}[colback=blue!5!white,colframe=white!100!black] 
% \begin{tcolorbox}[colback=gray!5!white,colframe=gray!75!black]
\begin{description}
\item[Question 3:]  When (if at all) is the parameterized complexity of \probEDH\ or \probTDH\ different from that of \probVDH?  
%Are there natural families of graphs where the parameterized complexity of 
%\probEDH or \probTDH is different from that of \probVDH? 

\label{abc}
\item[Question 4:] When (if at all) is the parameterized complexity of
%Are there natural families of graphs where the parameterized complexity of 
\probVDH is different from that of \probVDH parameterized by  $\hhdepth(G)$, or that of  \probVDH parameterized by  $\hhtw(G)$? 
 
 \item[Question 5:] Could one obtain uniform \FPT\ algorithms for \probEDH and \probTDH for 
natural families of graphs?
% where the parameterized complexity of 
%\probVDH, \probVDH parameterized by  $\hhdepth(G)$, and  \probVDH parameterized by  $\hhtw(G)$ are different? 

\end{description}
 \end{tcolorbox}
 
 The main objective of the paper is to provide satisfactory answers to Questions $3$, $4$ and $5$. 
 
 \subsection{Answering Questions 3 and 4: {\FPT}-equivalence of deletion, elimination and decomposition}
 
 Roughly speaking, we show: 
 
 \smallskip
\begin{tcolorbox}[colback=green!5!white,colframe=white!100!black] 
% \begin{tcolorbox}[colback=gray!5!white,colframe=gray!75!black]
\begin{description}
 	\item[For Question 3:] If $\cH$ has certain properties which are also possessed by numerous well-studied graph classes, then there is {\em no difference} in the parameterized complexity of {\probEDH},   {\probTDH}, and \probVDH.   This explains, unifies and extends almost all known results aimed at Question 1 in the literature.
 	\item[For Question 4:] If $\cH$ has the properties referred to above, then there is {\em no difference} in the parameterized complexity of \probVDH, \probVDH parameterized by  $\hhdepth(G)$, and  \probVDH parameterized by  $\hhtw(G)$. This explains, unifies and extends almost all known results aimed at Question 2 in the literature. 
 	\item[For Question 5:] There are several fundamental graph classes $\cH$ for which we are able to give sufficiency conditions that imply the {\em first uniform} {\FPT} algorithms for \probEDH. In some of these cases, not even a non-uniform {\FPT} algorithm was previously known. 

 \end{description}
  \end{tcolorbox}

%Our answer to Question 3 enables us to make a major advance towards answering Question 1 in its entirety. 

 In what follows, we formally state our results and go into more detail about their implications and significance.  Towards that we first define a notion of \FPT-equivalence.  We say that two parameterized problems are (non-uniformly) 
 {\em \FPT-equivalent} if, given an \FPT\ algorithm for any one of the two problems we can obtain a (non-uniform) \FPT\ algorithm for the other problem. Let $\hh$ be a family of graphs. Then, by parameterizing   \probVDH, \probEDH, and \probTDH, by  any of $\hhmd(G)$, $\hhdepth(G)$,  and $\hhtw(G)$ we get nine different problems. 
 Our first result shows \FPT-equivalences among eight of these.
 
% \medskip 

%  \begin{tcolorbox}[colback=blue!5!white,colframe=blue!40!black] 
 %[colback=blue!5!white,colframe=blue!75!black]
 %[colback=blue!2,colframe=blue!25]
\begin{theorem}
\label{thm:mainEquiv}
% \begin{restatable}{theorem}{main:Euiv}\label{thm:mainEquiv}
$\cH$ be a hereditary family of graphs that is \mso\footnote{In this paper, when we say  {\mso}, we refer to the fragment that is sometimes referred to as ${\mso}_2$ in the literature. We refer the reader to Section~\ref{sec:prelimsCMSO} for the formal description of this fragment.} definable and closed under disjoint union. Then the following problems are (non-uniformly) {\FPT}-equivalent.  
%\todo{do we need hereditary}
\begin{enumerate}
\setlength{\itemsep}{-2pt}
\item  \probVDH parameterized by  $\hhmd(G)$
\item   \probVDH parameterized by  $\hhdepth(G)$
\item   \probVDH parameterized by  $\hhtw(G)$
\item \probEDH parameterized by  $\hhmd(G)$
\item   \probEDH  parameterized by  $\hhdepth(G)$
%\item \red{  \probEDH  parameterized by  $\hhtw(G)$ does not know how to prove}
\item \probTDH parameterized by  $\hhmd(G)$
\item \probTDH parameterized by  $\hhdepth(G)$
\item \probTDH parameterized by  $\hhtw(G)$
\end{enumerate}
\end{theorem}
%\end{restatable}
%\end{tcolorbox}

%The only problem for which we are not able to show  \FPT-equivalence with the above eight is 
%\probEDH  parameterized by  $\hhtw(G)$. 

%\medskip
%\noindent  

Notice that because $\hhtw(G)\leq \hhdepth(G)\leq \hhmd(G)$, an {\FPT} algorithm for one problem parameterized by a smaller parameter also implies an {\FPT} algorithm for the problem parameterized by the larger parameter. However, the implications in the other direction are surprising and insightful. 

%We begin by showing how this result allows us to unify and significantly extend known results in the literature. 

\subsubsection{Implications of Theorem~\ref{thm:mainEquiv}}
%{\bf Implications of Theorem~\ref{thm:mainEquiv}.} 

We now describe the various applications and consequences of our first main theorem (Theorem~\ref{thm:mainEquiv}). 

\paragraph{As a classification tool that unifies and extends many known results. } Theorem~\ref{thm:mainEquiv} is a powerful classification tool which states that  
as far as the (non-uniform) fixed-parameter tractability of computing any of the parameters $\hhmd(G)$,  $ \hhdepth(G)$ and $\hhtw(G)$ is concerned,  
%and when considering classification of problems in the class FPT,
%  $\hhmd(G)$,  $ \hhdepth(G)$ and $\hhtw(G)$ 
they  are essentially the ``same parameter'' for many frequently considered graph classes $\cH$. In other words, to obtain an \FPT\  algorithm for any of the problems mentioned in Theorem~\ref{thm:mainEquiv}, it is sufficient to design an \FPT\ algorithm for the standard vertex-deletion problem, namely, \probVDH.
%Furthermore, we can use known algorithms for  \probVDH, developed over the last two decades, to design an \FPT\ algorithm for the above problems.
This implication unifies several known results in the literature.

%  For example, to  show whether \probEDH  (\probTDH) is \FPT\  for  $\hh$ being 
%family of graphs of degree at most $d$, or to design an \FPT\ algorithm for  \probVDH  parameterized by  $\hhdepth(G)$ ($\hhtw(G)$), all we need is the $d^{k}n^{\cO(1)}$-time simple branching algorithm for \probVDH. 
For example, let $\hh$ be the family of graphs of degree at most $d$ and recall that it was only recently that  Agrawal et al.~\cite{Agrawal020} and Jansen et al.~\cite{corr/abs-2106-04191} showed that \probEDH\  and \probTDH\ are (non-uniformly) \FPT, respectively. However, using  Theorom~\ref{thm:mainEquiv},  the fixed-parameter tractability of these two problems and in fact,  even 
  the fixed-parameter tractability of \probVDH parameterized by  $\hhdepth(G)$ (or $\hhtw(G)$), is implied by the straightforward $d^{k}n^{\cO(1)}$-time branching algorithm for \probVDH (i.e., the problem of deleting at most $k$ vertices to get a graph of degree at most $d$).

Moreover, for various well-studied families of $\hh$, 
we immediately derive \FPT\ algorithms for all combinations of \probVDH, \probEDH, \probTDH parameterized by  
any of $\hhmd(G)$ $\hhdepth(G)$ and $\hhtw(G)$, which are  covered in Theorem~\ref{thm:mainEquiv}. For instance, we can invoke this theorem using well-known \FPT\ algorithms for   \probVDH for 
several families of graphs that are \mso\ definable and closed under disjoint union, such as  
families defined by a finite number of  forbidden connected (a)~minors, or   (b)~topological minors, or (c)~induced subgraphs, or (d)  $\hh$ being bipartite, chordal, proper-interval, interval,  and 
distance-hereditary; to name a few~\cite{BevernKMN10,CaoM15,CaoM16,Cai96,EibenGK18,FominLPSZ20,FominLMS12,JansenLS14,LokshtanovNRRS14,Marx10,ReedSV04,RobertsonS95b,SauST20,DBLP:journals/talg/0002LPRRSS16}.  Thus, Theorem~\ref{thm:mainEquiv}  provides a unified understanding of many recent results and resolves the parameterized complexity of several questions left open in the literature. 

Of particular significance among the new results is the case where,  invoking Theorem~\ref{thm:mainEquiv} and taking $\hh$ to be a class  defined by a finite number of  forbidden connected topological minors gives the {\em first} \FPT\ algorithms for computing $\hhdepth$ and $\hhtw$,  resolving an open problem posed by Jansen et al.~\cite{JansenK021}.
%\todo[inline]{Ganian et al. "ask" about interval graphs. Should we mention  this here? -- msr}  
%Finally, we remark that we are not able to get an \FPT\ algorithm for   \probEDH  parameterized by  $\hhtw(G)$. 
%get the following corollaries. 

%  \todo[inline]{Explain Directed graph thing, mixed hgraphs, Multiway Cut}

\paragraph{Deletion to Families of Bounded Rankwidth.} 
%A graph $G$ is a \emph{distance hereditary graph} 
% (also called a completely separable graph~\cite{hammer1990completely}) if the distances between vertices in every connected induced subgraph of $G$ are the same as in the graph $G$. Distance hereditary graphs were named and first studied by Hworka~\cite{howorka1977characterization}. However,  an equivalent family of graphs was earlier studied by Olaru and Sachs~\cite{sachs1970berge} and shown to be perfect. 
% It was later discovered that these graphs are precisely the graphs of rankwidth 1~\cite{Oum05}. 
 %That is, distance hereditary graphs are precisely those graphs which has rankwidth one. 
 %Let us discuss the notion of rankwidth of a graph.
% 
We observe that Theorem~\ref{thm:mainEquiv} can be invoked by taking $\cH$ to be the class of graphs of bounded {\em rankwidth}, extending a result of Eiben et al.~\cite{EibenGHK19}. 

 Rankwidth is a graph parameter introduced by Oum and Seymour~\cite{OumS06} to approximate yet another graph parameter called Cliquewidth. 
 The notion of cliquewidth was defined by  Courcelle and Olariu~\cite{CourcelleO00} as a measure of how ``clique-like''  the input graph is. 
% This is similar to the notion of treewidth, which measures how ``tree-like'' the input graph is. 
 One of the main motivations was that several {\sf NP}-complete problems
 become tractable on the family of cliques (complete graphs), the assumption was that these algorithmic properties extend to ``clique-like'' graphs~\cite{CourcelleMR00}. 
%Indeed, it has been shown that graphs of bounded cliquewidth have good algorithmic properties.  That is,  many {\sf NP}-hard graph problems can be solved in polynomial time, if the input graphs have bounded cliquewidth~\cite{CourcelleMR00}. In fact, every graph that has treewdith at most $\eta$, its cliquewidth is upper bounded by $3\cdot 2^{\eta -1}$~\cite{CorneilR05}. However, the converse does not hold. That is,  if a graph has bounded cliquewidth it does not imply that its treewidth is bounded. For example, the family of cliques has unbounded treewidth, but the cliquewidth of this family is at most $2$.  
However, computing cliquewidth and the corresponding cliquewidth decomposition seems to be computationally intractable. 
This then motivated the notion of rankwidth, which is a graph parameter that approximates cliquewidth well while also being algorithmically tractable~\cite{OumS06,Oum08}. For more information on cliquewidth and rankwidth, we refer to the surveys by Hlinen{\'{y}} et al.~\cite{HlinenyOSG08} and Oum~\cite{Oum17}. 

For a graph $G$, we will use $\rw(G)$ to denote the rankwidth of $G$. Let $\eta \geq 1$ be a fixed integer and  let $\hh_\eta$ denote the class of graphs of rankwidth at most $\eta$.  It is known that \probVDHp{$\hh_\eta$} is \FPT~\cite{CourcelleO07}.  The algorithm is based on the fact that for every integer $\eta$, there is a finite set ${\cal C}_\eta$ of graphs such that for every graph $G$, $\rw(G) \leq \eta$ if and only if no vertex-minors of 
$G$ are isomorphic to a graph in ${\cal C}_\eta$~\cite{Oum05,Oum08a}. Further, it is known that vertex-minors can be expressed in \mso, this together with the fact that we can test whether a graph $H$ is a 
vertex-minor of $G$ or not in $f(|H|)n^{\cO(1)}$ time on graphs of bounded rankwidth leads to the desired algorithm~\cite[Theorem 6.11]{CourcelleO07}. It is also important to mention that for \probVDHp{$\hh_1$}, also known as the {\sc Distance-Hereditary Vertex-Deletion} problem, 
 there is a dedicated algorithm running in time $2^{\cO(k)}n^{\cO(1)}$~\cite{EibenGK18}. 
For us,  two properties of $\hh_\eta$ are important: (a) expressibility in \mso and (b) being closed under disjoint union.  These two properties, together with the result in~\cite{CourcelleO07} imply that Theorem~\ref{thm:mainEquiv} is also applicable to  $\hh_\eta$. Thus, we are able to generalize and extend the result of Eiben et al.~\cite{EibenGHK19}, who showed that for every $\eta$,
% if $\cH=\hh_\eta$, then
  computing $\tw_{\hh_\eta}$ is FPT.  
%  \todo{Say somewhere that whenever we talk about a problem being FPT without mentioning the parameter, the parameter is the solution value? --msr}.  
  Since we do not need the notion of rankwidth and vertex-minors in this paper beyond this application, we refer the reader for further details to~\cite{HlinenyOSG08,Oum17}. 
%scattered modulator and scattered graph families. 

\paragraph{Beyond graphs: Cut problems. }
Notice that in the same spirit as we have seen so far, one could also consider the parameterized complexity of other classical problems such as cut problems (e.g.,  {\sc Multiway Cut}) as long as the parameter is smaller than the standard parameter studied so far. With this view, we obtain the first such results for several cut problems such as {\sc Multiway Cut}, {\sc Subset FVS} and {\sc Subset OCT}. That is, we obtain {\FPT} algorithms for these problems that is parameterized by a parameter whose value is upper bounded by the standard parameter (i.e., solution size) and which can be arbitrarily smaller. For instance, consider the {\sc Multiway Cut} problem, where one is given a graph $G$ and a set of vertices $S$ (called terminals) and an integer $\ell$ and the goal is to decide whether there is a set of at most $\ell$ vertices whose deletion separates every pair of these terminals. The standard parameterization for this problem is the solution size $\ell$. Jansen et al.~\cite{JansenK021} propose to consider annotated graphs (i.e., undirected graphs with a distinguished set of terminal vertices) and study the parameterized complexity of {\sc Multiway Cut} parameterized by the elimination distance to a graph where each component has at most one terminal. Notice that this new parameter is always upper bounded by the size of a minimum solution. 

Thus, an {\FPT} algorithm for {\sc Multiway Cut} with such a new parameter would naturally extend the boundaries of tractability for the problem. We are able to obtain such an algorithm by using Theorem~\ref{thm:mainEquiv}. We then proceed to obtain similar {\FPT}  algorithms for the other cut problems mentioned in this paragraph. Recall that in the {\sc Subset FVS} problem, one is given a graph $G$, a set $S$ of terminals and an integer $\ell$ and the goal is to decide whether there is a set of at most $\ell$ vertices that hits every cycle in $G$ that contains a terminal. Similarly, in the {\sc Subset OÇT} problem, one is given a graph $G$, a set $S$ of terminals and an integer $\ell$ and the goal is to decide whether there is a set of at most $\ell$ vertices that hits every {\em odd} cycle in $G$ that contains a terminal.

Building on our new result for {\sc Multiway Cut},  we obtain an {\FPT} algorithm for  {\sc Subset FVS} parameterized by the elimination distance to a graph where no terminal is part of a cycle, and an {\FPT} algorithm for {\sc Subset OCT} parameterized by the elimination distance to a graph where no terminal is part of an odd cycle. 
We summarize all three results as follows:

\begin{theorem}[Informal version of Theorem~\ref{thm:app}]
The following problems are {\rm \FPT}:
\begin{enumerate}
	\item {\sc Multiway Cut} parameterized by elimination distance to an annotated graph where each component has at most one terminal.
	\item {\sc Subset Feedback Vertex Set} parameterized by elimination distance to an annotated graph where no terminal occurs in a cycle. 
	\item {\sc Subset Odd Cycle Transversal} parameterized by elimination distance to an annotated graph where no terminal occurs in an odd cycle. 
\end{enumerate}
	
\end{theorem}

In fact, we also strengthen the parameterization to an analogue of $\cH$-treewidth in the natural way and obtain corresponding results. The details can be found in Section~\ref{sec:cutProblems}. 
To achieve these results, we use Theorem~\ref{thm:mainEquiv}. However, note that that Theorem~\ref{thm:mainEquiv} is defined only when $\hh$ is a family of graphs. In order to capture problems such as {\sc Multiway Cut}, {\sc Subset FVS} and {Subset OCT}, we express our problems in terms of appropriate notions of structures and then give a reduction to a pure graph problem on which Theorem~\ref{thm:mainEquiv} can be invoked.  

  These results make concrete advances in the direction proposed by Jansen et al.~\cite{JansenK021} to develop {\FPT} algorithms for \textsc{Multiway cut} parameterized by the elimination distance to a graph where each component has at most one terminal. 

\subsubsection{Modulators to Scattered Families}
Recent years have seen another new direction of research on {\probVDH} -- instead of  studying the computation of a modulator to a single family of graphs $\cH$, one can aim to compute small vertex sets whose deletion leaves a graph where each connected component comes from a particular pre-specified graph class~\cite{GanianRS17b,JacobM020,corr/abs-2105-04660}. As an example of problems in this line of research, consider the following. Given a graph $G$, and a number $k$, find a modulator $S$ of size at most $k$ (or decide whether one exists) such that in $G-S$,  each connected component is either chordal, or bipartite or planar. Let us call such an $S$, {\em a  scattered modulator}. 
  Such scattered modulators (if small) can be used to design new {\FPT} algorithms for certain problems by taking separate {\FPT} algorithms for the problems on each of the pre-specified graph classes and then combining them in a non-trivial way ``through" the scattered modulator. However, the quality of the modulators considered in this line of research so far has been measured in the traditional way, i.e., in terms of the size. 
 In this paper, by using Theorem~\ref{thm:mainEquiv}, we initiate a new line of research where, again, it is not the modulator size that is the measure of structure, but in some sense, the treedepth or treewidth of the torso of the graph induced by the modulator.   That is, we introduce the first extensions of ``scattered modulators" to ``scattered  elimination distance" and  ``scattered  $\hh$-tree decompositions" and obtain results regarding the computation of the corresponding modulators as well as their role in the design of {\FPT} algorithms for other problems (i.e., cross parameterization).

% generalized to modulators whose deletion yields a graph that can be described in terms of  multiple families of graphs.

The first study of scattered modulators was undertaken by  Ganian et al.~\cite{GanianRS17b}, who introduced this notion in their work on constraint satisfaction problems. Recently, Jacob et al.~\cite{corr/abs-2105-04660,JacobM020} initiated the study of scattered modulators explicitly for ``scattered" families of graphs.  In particular, let $\hh_{1},\ldots,\hh_{d}$ be families of graphs. Then, the scattered family of graphs $\otimes(\hh_{1},\ldots,\hh_{d})$ is defined as the set of all graphs $G$ such that every connected component of $G$ belongs to $\bigcup_{i=1}^d \hh_i$. That is, each connected component of $G$ belongs to some $\hh_i$.  As their main result, Jacob et al.~\cite{corr/abs-2105-04660} showed that \probVDH is \FPT\ whenever 
% FPT when there are a finite number of graph classes, the deletion problem corresponding to each of the finite classes is known to be FPT and the properties that a graph belongs to each of the classes is expressible in CMSO logic.
\probVDHp{$\hh_i$}, $i\in\{1,\ldots,d\}$, is \FPT, and each of $\hh_i$ is \mso expressible. Here, $\hh$ is the scattered family $\otimes(\hh_{1},\ldots,\hh_{d})$.  Notice that 
if each of  $\hh_i$ is \mso expressible then so is $\hh$. Further, it is easy to observe that if each of $\hh_i$ is closed under disjoint union then so is $\hh$. 
The last two properties together with the result of Jacob et al.~\cite{corr/abs-2105-04660} enable us to  
invoke Theorem~\ref{thm:mainEquiv} even when $\hh$ is a scattered graph family. The effect can be formalized as follows. 

%In particular, notice that if $\cH'=\bigcup_{i=1}^d\cH_i$, then $\ed_{\cH'}(G)=\ed_{\cH}(G)$
%where $\cH=\otimes(\hh_1,\dots, \hh_d)$ and moreover, $\ed_{\cH}(G)$ is at most the size of a smallest scattered modulator of $G$ to $\cH$. 

\begin{theorem}
	Let $\hh_{1},\ldots,\hh_{d}$ be hereditary, union-closed, {\mso} expressible families of graphs such that \ProblemName{Vertex Deletion to $\cH_i$} is {\rm \FPT} for every $i\in [d]$. 
	Let  $\cH=\otimes(\hh_1,\dots, \hh_d)$ and  $\cH'=\bigcup_{i=1}^d\cH_i$. Then, \ProblemName{Elimination Distance to $\cH$} and \ProblemName{Treewidth Decomposition to $\cH$} are also {\rm \FPT}. 
\end{theorem}

Notice in the above statement that if we take $\cH'=\bigcup_{i=1}^d\cH_i$, then the size of a  modulator to $\cH$ is different from that of a smallest modulator to $\cH'$, whereas the elimination distance to $\cH$ and $\cH'$ are the same.

\subsection{New Results on Cross Parameterizations}
Another popular direction of research in Parameterized Complexity is cross parameterizations: that is parameterization of one problem with respect to alternate parameters. For an illustration, consider {\sc Odd Cycle Transversal} (OCT) on chordal graphs. Let $\hh$ denote the family of chordal graphs. 
It is well known that OCT is polynomial-time solvable on chordal graphs.  Further, given a graph $G$ and a modulator to  chordal graphs  of size  $\hhmd(G)$, OCT admits an algorithm with running time $2^{\OO(\hhmd(G)))}n^{\OO(1)}$. It is therefore natural to ask whether OCT admits an algorithm with running time $f(\hhdepth(G))n^{\OO(1)} $ or $f(\hhtw(G))n^{\OO(1)} $, given an $\hh$-elimination forest of $G$ of depth~$\hhdepth(G)$ and an $\hh$-decomposition of $G$ of width~$\hhtw(G)$, respectively. The question  is also relevant, in fact more challenging, when an $\hh$-elimination forest of $G$ of depth~$\hhdepth(G)$ or an $\hh$-decomposition of $G$ of width~$\hhtw(G)$ is {\em not} given. Jansen et al.~\cite{JansenK021} specifically asked to consider this research direction in their paper.   Quoting them: 

\medskip

 \begin{tcolorbox}[colback=gray!5!white,colframe=white!80!black]
``\ldots the elimination distance can also be used as a parameterization away from triviality for solving other parameterized problems~$\Pi$, when using classes~$\hh$ in which~$\Pi$ is polynomial-time solvable. This can lead to interesting challenges of exploiting graph structure. For problems which are FPT parameterized by deletion distance to~$\hh$, does the tractability extend to elimination distance to~$\hh$? For example, is \textsc{Undirected Feedback Vertex Set} \FPT\ when parameterized by the elimination distance to a subcubic graph or to a chordal graph? The problem is known to be FPT parameterized by the deletion distance to a chordal graph~\cite{JansenRV14} or the edge-deletion distance to a subcubic graph~\cite{MajumdarR18}.''
\end{tcolorbox}

A step in this direction can be seen in the work of  Eiben et al.~\cite[{Thm.~4}]{EibenGHK19}. They
 %Eiben et al.~\cite[{Thm.~4}]{EibenGHK19} 
 present a meta-theorem that yields non-uniform \FPT\ algorithms when~$\Pi$ satisfies several conditions, which require a technical generalization of an \FPT\ algorithm for $\Pi$ parameterized by deletion distance to $\hh$.  
 Here, we avoid resorting to such requirements and instead,   provide sufficient conditions {\em on the problem itself}, which are usually very easily checked and which enables us to obtain fixed-parameter algorithms for vertex-deletion problems  (or edge-deletion problems) parameterized by $\hhdepth(G)$ ($\hhtw(G)$) when given an $\hh$-elimination forest of $G$ of depth~$\hhdepth(G)$ (respectively, an $\hh$-decomposition of $G$ of width~$\hhtw(G)$).  As a consequence, we resolve the aforementioned open problem of Jansen et al.~\cite{JansenK021} on the parameterized complexity of {\sc Undirected Feedback Vertex Set} parameterized by the elimination distance to a chordal graph.  
 
Let us now state an informal version of our result so that we can convey the main message and highlight some consequences without delving into excessive detail at this point. We refer the reader to Section~\ref{sec:crossParam} for the formal statement and full proof.  

\begin{theorem}[Informal version of Theorem~\ref{thm:fullCrossParameterization}]\label{thm:intro:crossParam}
Let $\cal H$ be a hereditary family of graphs and $\Pi$ be a parameterized graph problem satisfying the following properties. 
%having the property of finite integer index ({\rm FII}).  
%and $(G,k)$ be an instance of $\Pi$. 
%Further assume that we have following. 
\begin{enumerate}
\setlength{\itemsep}{-1pt}
%\item $\Pi$ is \mso\ definable property \red{(I do not think we need this). }
%\item $\Pi$ has {\rm FII}.
\item  $\Pi$ has the property of finite integer index ({\rm \fii}).
\item $\Pi $ is {\rm \FPT}  parameterized by $\hhmd(G)$. 
\item  Either, a $\hh$-decomposition of $G$ of width~$\hhtw(G)$  (or a $\hh$-elimination forest of $G$ of depth~$\hhdepth(G)$) is given; or $\hh$ is \mso definable, closed under disjoint union and \probVDH is {\rm \FPT}. 
\end{enumerate}
 %Further assume that there is an algorithm for decid
%Let $\Pi$ be a 
%Then, there is an algorithm that, given an $n$-vertex graph $G$, decides whether $(G,k)\in \Pi$ 
% in time~$f(\hhtw(G)) \cdot n^{\Oh(1)}$ ($f(\hhdepth(G)) \cdot n^{\Oh(1)}$). That is,  
 Then, $\Pi$ is {\rm \FPT}  parameterized by $\hhdepth(G)$ or $\hhtw(G)$.
\end{theorem}

\fii is a technical property satisfied by numerous graph problems and often easily verified. The term \fii first appeared in the works of~\cite{BodlaenderF01,de1997algorithms}  and is similar to the notion of finite state~\cite{abrahamson1993finite,BoriePT92,Courcelle90}. 
An intuitive way (though, formally, not correct)  to understand \fii is as follows.  Let $\Pi$  be a graph problem. Further, for a graph $G$, let ${\sf opt}_{\Pi}(G)$ denote the optimum (minimum or maximum) value of solution to $\Pi$. For example, if $\Pi$ is {\sc Dominating Set} then ${\sf opt}_{\Pi}(G)$ denotes the size of a minimum dominating set; and if $\Pi$ is 
{\sc Cycle Packing} then ${\sf opt}_{\Pi}(G)$ denotes the cardinality of a set containing maximum number of pairwise vertex disjoint cycles.  In a simplistic way we can say that a graph problem has \fii, if for every graph $G$, and a separation  $(L, R)$ ($L\cup R=V(G)$) we have that:  
$${\sf opt}_{\Pi}(G)={\sf opt}_{\Pi}(G[L\setminus R])+{\sf opt}_{\Pi}(G[R\setminus L])\pm h(|L\cap R|).$$
Here, $h$ is a function of the order of separation  ($|L\cap R|$) only.  This immediately implies that problems such as  {\sc Dominating Set} and {\sc Cycle Packing} have \fii. In the context of this work, it is sufficient for the reader to know that \probVDH has \fii, whenever $\hh$ is hereditary,  \mso definable, and closed under disjoint union. 
%We also point out that there are problems such as {\sc Longest Path} or {\sc Independent Dominating Set} that do not have \fii.
  We refer the reader to \cite{BodlaenderF01,de1997algorithms,BodlaenderFLPST16} for more details on \fii. 

Now as a corollary to Theorem~\ref{thm:intro:crossParam} we get that \textsc{Undirected Feedback Vertex Set} is \FPT\  parameterized by $\hhdepth(G)$ or $\hhtw(G)$, where $\hh$ is a family of chordal graphs, answering the problem posed by Jansen et al.~\cite{JansenK021}.   Similarly, we can show that {\sc Dominating Set} is \FPT\  parameterized by $\hhdepth(G)$ or $\hhtw(G)$, where $\hh$ is a family of interval graphs.  
 
 %\begin{proof}
%Let $(T, \chi, L)$  be a $\hh$-decomposition of $G$ of width~$\hhtw(G)$. 
%We start the algorithm by computing the 
%
%\begin{enumerate} 
%\item Compute the decomposition.
%\item Replace the leaf bags using Lemma~\ref{lem:red2finiteindex}
%\item Treewidth is small
%\item Apply Courcelle's Theorem
%\end{enumerate}
%\end{proof}

 %!TEX root = main.tex

\subsection{Answering Question 5: Towards Uniform \FPT\ Algorithms}\label{sec:towardsUniform}
The \FPT\ algorithms obtained via Theorems~\ref{thm:mainEquiv} and \ref{thm:intro:crossParam} and the extension to families of structures are non-uniform. In fact, to the best of our knowledge, all of the current known 
\FPT\ algorithms for \probEDH or \probTDH, are non-uniform; except for \probEDH, when $\hh$ is the family of empty graphs (which, as discussed earlier, is simply the problem of computing treedepth). 
%At this point it is important to mention
However, we note that the \FPT-approximation algorithms in Jansen et al.~\cite{JansenK021} (in fact, all the algorithms obtained in~\cite{JansenK021}) are uniform. 

For the sake of clarity, we formally define the notion of  uniform and non-uniform \FPT\ algorithms~\cite[Definition $2.2.1$]{DowneyF13}. 
\begin{definition}[Uniform and non-uniform \FPT] 
Let $\Pi$ be a parameterized problem.
\begin{enumerate}[(i)]
\item  We say that $\Pi$ is {\em uniformly} \FPT\ if there is an algorithm ${\cal A}$, a constant $c$, and an arbitrary function $f :\mathbb{N} \to \mathbb{N} $ such that: the running time of ${\cal A}(\langle x,k\rangle)$ 
is at most $f(k)|x|^c$ and $\langle x,k\rangle \in \Pi$ if and only if  ${\cal A}(\langle x,k\rangle)=1$.

\item We say that $\Pi$ is {\em non-uniformly} \FPT\ if there is  collection of algorithms  $\{{\cal A}_k : k \in \mathbb{N}\}$, a constant $c$, and an arbitrary function 
$f :\mathbb{N} \to \mathbb{N} $,  such that :
for each $k\in \mathbb{N}$,   the running time of ${\cal A}_k(\langle x,k\rangle)$  is $f(k)|x|^c$ and $\langle x,k\rangle \in \Pi$ if and only if  ${\cal A}_k(\langle x,k\rangle)=1$.
\end{enumerate}
\end{definition}

Towards unification, we first present a general set of demands that, if satisfied, shows that \probEDH\ parameterized by $\hhdepth(G)$ is uniformly \FPT. Like before, we consider a hereditary family $\cal H$ of graphs such that $\hh$ is \mso definable, closed under disjoint union and \probVDH is \FPT.  However, we strengthen the last two demands. To explain the new requirements, we briefly (and informally) discuss a few notions concerning boundaried graphs and equivalence classes. Essentially, a {\em boundaried graph} ({\em a $t$-boundaried graph}) is a graph with an injective labelling of some of its vertices by positive integers (upper bounded by~$t$). When {\em gluing} two boundaried graphs $G$ and $H$, denoted by $G \oplus H$, we just take their disjoint union, and unify vertices having the same label. Concerning some $\cal H$, two boundaried graphs $G_1$ and $G_2$ are equivalent (under the {\em canonical equivalence relation}) if, for any boundaried graph $H$, $G_1\oplus H\in{\cal H}$ if and only if $G_2\oplus H\in{\cal H}$. A {\em refinement of the canonical equivalence relation} (or just a refinement, for short) is an equivalence relation where any two boundaried graphs considered equivalent, are equivalent according to the canonical equivalence relation (but not vice versa).

Now, the new requirements are to define a refinement $\mathsf{R}$ whose number of equivalence classes for $t$-boundaried graphs is a function of $t$ only, here called a {\em finite (per $t$) refinement}, such that:
\begin{itemize}
\item {\bf $\mathsf{R}$ is closed under disjoint union:} That is, if we have a boundaried graph that belongs to some equivalence class $R\in \mathsf{R}$, then the disjoint union of that boundaried graph and a non-boundaried graph from ${\cal H}$ also belongs to $R$. [Strengthens the closeness under disjoint union property of $\cal H$.]
\item {\bf {\sc Vertex Deletion to $\mathsf{R}$} is \FPT:} We have a uniform \FPT-algorithm (with parameters $t$ and $k$) that, given a $t$-boundaried graph $G$, an equivalence class $R\in \mathsf{R}$, and $k\in\mathbb{N}$, decides whether there exists $S\subseteq V(G)$ of size at most $k$ such that $G-S$ belongs to an equivalence class ``at least as good as'' $R$. [Strengthens that \probVDH~is~\FPT.]
\end{itemize}

We prove the following.

\begin{theorem}[Informal]\label{thm:towardsUniform}
Let $\cal H$ be a hereditary family of graphs and a finite (per $t$) refinement $\mathsf{R}$ satisfying the following properties. 
\begin{enumerate}
\setlength{\itemsep}{-1pt}
\item  $\hh$ is \mso definable.
\item $\mathsf{R}$ is closed under disjoint union.
\item  {\sc Vertex Deletion to $\mathsf{R}$} is \FPT\ (parameterized by boundary and solution sizes).
\end{enumerate}
 Then, \probEDH\ is uniformly {\rm \FPT}  parameterized by $\hhdepth(G)$.
\end{theorem}

We also give two conditions that seem easier to implement. Together with $\hh$ being \mso definable and closed under disjoint union, the satisfaction of the two new condition yields  the conditions of Theorem \ref{thm:towardsUniform}. In particular, they allow the user to not deal with equivalence classes at all, and to deal with boundaried graphs only with respect to a condition posed on obstructions. Roughly speaking, the simpler conditions are as follows.
\begin{itemize}
\item  {\bf Finite Boundaried Partial-Obstruction Witness Set:} $\hh$ admits a characterization by a (possibly infinite) set $\mathbb{O}$ of obstructions as induced subgraphs.\footnote{Even if the more natural characterization is by forbidden minors/topological minors/subgraphs, we can translate this characterization to one by induced subgraphs (which can make a finite obstruction set become~infinite).} Let ${\cal O}_t$ be the set of $t$-boundaried graphs that are subgraphs of obstructions from $\mathbb{O}$.
Then, there exists ${\cal O}'_t\subseteq {\cal O}_t$ of finite size (depending only on $t$) such that: for any  boundaried graph $G$, if there exists $O\in{\cal O}_t$ such that $G\oplus O\notin{\cal H}$, then there exists $O'\in{\cal O}'_t$ such that $G\oplus O'\notin{\cal H}$.
\item {\bf A ``Strong'' Irrelevant Vertex Rule:} There exist families of graphs $\cal X$ and $\cal Y$ (possibly ${\cal X}={\cal Y}$),  such that: {\em (i)} Each ``large'' graph in $\cal X$ contains a ``large'' graph from $\cal Y$ as an induced (or not) subgraph.
{\em (ii)} Given a (not boundaried) graph of ``large'' treewidth, it contains a ``large'' graph in $\cal X$ as an induced subgraph. {\em (iii)} Given $k\in\mathbb{N}$ and a (not boundaried) graph that contains a ``large'' graph in $\cal X$ as an induced subgraph, which, in turn, (by condition (i)), contains a ``large'' graph $Y$ in $\cal Y$ as an induced subgraph, ``almost all'' (depending on $k$) of the vertices in $Y$ are {\em $k$-irrelevant}.\footnote{A vertex $v$ in $G$ is $k$-irrelevant if the answers to $(G,k)$ and $(G-v,k)$ are the same (as instances of \probVDH).}
\end{itemize}
The reason why we claim that these conditions are ``simple'' is that known algorithms (for the problems considered here) already implicitly yield them as part of their analysis. So, the satisfaction of these conditions do not seem (in various cases) to require much ``extra'' work compared to the design of an \FPT\ algorithm (or a kernel) to the problem at hand. Using these sufficient conditions, we obtain {\em uniform} {\FPT} algorithms for computing $\hhdepth$, when $\cH$ is defined by excluding a finite number of connected (a)~minors, or   (b)~topological minors, or (c)~induced subgraphs, or when  $\hh$ is any of bipartite, chordal or interval graphs. 
 For most of these problems, the existence of a uniform {\FPT} algorithm has remained open in the literature. In fact, for some of them, even a non-uniform {\FPT} algorithm was not known.

%% file: prelim.tex
% !TEX root = main.tex
%\todo[inline]{Post paper finish remove the definitions we do not need. }
\subsection{Generic Notations}
We begin by giving all the basic notations we use in this paper. For a graph $G$, we use $V(G)$ and $E(G)$ to denote its vertex set and edge set, respectively. 
For a graph $G$, whenever the context is clear, we use $n$ and $m$ to denote $|V(G)|$ and 
$|E(G)|$, respectively. Consider a graph $G$. 
%For two vertex subsets $X,Y\subseteq V(G)$ such that $X\cap Y=\emptyset$, $E(X,Y)$ denote the set of edges with one endpoint in $X$ and the other in $Y$. 
%$\{\{x,y\}\in E(G)~\vert~x\in X \mbox{ and }y\in Y\}$. 
For $X \subseteq V(G)$, $G[X]$ denotes the graph with vertex set $X$ and the edge set $\{\{x,y\}\in E(G) \mid x,y\in X\}$. By $G-X$,  we denote the graph $G[V(G) \setminus X]$. For a vertex 
$v \in V(G)$,  $N_G(v)$ and $N_G[v]$ denote the set of open neighbors and closed neighbors of $v$ in $G$. That is, $N_G(v)=\{u\in V(G) \mid \{u,v\}\in E(G)\}$ and $N_G[v]=N_G(v) \cup \{v\}$. %The degree of  $v$ in $G$, denoted by $deg_G(v)$, is the number of edges incident on $v$ in $G$. 
For a set $U\subseteq V(G)$, $N_G(U)=\bigcup_{u\in U}N_G(u)\setminus U$ and $N_G[U]=N_G(U)\cup U$.  For a vertex set $C\subseteq V(G)$, by slightly abusing the notation, 
we use  $N_G(C)$, and $N_G[C]$ to denote  $N(V(C))$, and $N[V(C)]$, respectively. In all the notations above, we drop the subscript $G$ whenever the context is clear.

   A \emph{path} $P=(v_1,v_2,\cdots,v_\ell)$ in $G$ is a subgraph of $G$ where $V(P) = \{ v_1, v_2, \cdots, v_\ell \} \subseteq V(G)$ is a set of distinct vertices and $E(P) = \{ \{v_i,v_{i+1}\} \mid i \in \ns{\ell-1}\}\subseteq E(G)$, where $|V(P)| = \ell$ for some $\ell \in \ns{|V(G)|}$. The above defined path $P$ is called as $v_1-v_{\ell}$ path.  We say that the graph $G$ is \emph{connected} if for every $u,v\in V(G)$, there exists a $u-v$ path in $G$. A \emph{connected component} of $G$ is an inclusion-wise maximal connected induced subgraph of $G$.
%The set of connected components of $G$ is denoted by $\C{C}(G)$. For a tree $T$, and vertices $u,v \in V(T)$, we use ${\sf Pth}_T(u,v)$ to denote the unique path between $u$ and $v$. 

A \emph{rooted tree} is a tree with a special vertex  designated to be the \emph{root}. Let $T$ be a rooted tree with root $r\in V(T)$.  We say that a vertex $v \in V(T)\setminus \{r\}$ is a \emph{leaf} of $T$ if $v$ has exactly one neighbor in $T$. Moreover, if $V(T) = \{r\}$, then $r$ is the leaf (as well as the root) of $T$.\footnote{A root is not a leaf in a tree, if the tree has at least two vertices.} A vertex which is not a leaf, is called a \emph{non-leaf} vertex. 
Let $t,t'\in V(T)$ such that $\{t,t'\} \in E(T)$ and $t'$ is not contained in the $t-r$ path in $T$, then we say that $t$ is the \emph{parent} of $t'$ and $t'$ is a \emph{child} of $t$. 
A vertex $t' \in V(T)$  is a \emph{descendant} of $t$ ($t'$ can possibly be the same as $t$), 
if there is a $t-t'$ path in $T-\{{\sf par}_T(t)\}$, where ${\sf par}_T(t)$ is the parent of $t$ in $T$. 
Note that when $t=r$, then $T-\{{\sf par}_T(t)\}=T$, as the parent of $r$ does not exist. That is, every vertex in $T$ is a descendant of $r$. %The set of all descendants of $t$ in $T$ is denoted by $\textsf{desc}_T(t)$. We drop the subscript $T$ from ${\sf par}_T(\cdot)$ and ${\sf desc}_T(\cdot)$, when the context is clear. 

A {\em rooted forest} is a forest where each of its connected component is a  rooted tree. For a rooted forest $F$, a vertex $v\in V(F)$ that is not a root of any of its rooted trees is a \emph{leaf} if it is of degree exactly one in $F$. %We denote the set of leaves in a rooted forest by $\Lf{F}$. 
The \emph{depth} of a rooted tree $T$ is the maximum number of edges in a root to leaf path in $T$. The \emph{depth} of a rooted forest is the maximum depth among its rooted trees. 

\subsection{Graph Classes and Decompositions}
We always assume that $\hh$ is a hereditary class of graphs, that is, closed under taking induced subgraphs.
A set $X \subseteq V(G)$ is called an~$\hh$-deletion set if $G - X \in \hh$.
The task of finding a smallest $\hh$-deletion set is called the $\hh$-\textsc{deletion} problem (also referred to as $\hh$-\textsc{vertex deletion} but we abbreviate it since we do not consider edge deletion problems).  
%Note that if $\hh$ is not closed under taking disjoint unions of graph, it might be that $X$ is an~$\hh$-deletion set but $G - X \not\in \hh$, so we explicitly make this distinction.
Next, we give the notion of \emph{elimination-distance} introduced by Bulian and Dawar~\cite{BulianD16}. We rephrase their definition, and our definition is (almost) in-line with the equivalent definition given by Jansen et al.~\cite{JansenK021}.%\footnote{We note that the definition given here and the ones given in, say,~\cite{JansenK021,BulianD16} and \cite{DBLP:conf/iwpec/Agrawal020,MSRstacs2020} are all equivalent.}

\iffalse
\begin{definition}
An $\hh$-decomposition of graph $G$ is a tree $\mathcal{T}$ with functions $f \colon V(\mathcal{T}) \to V(G) \cup \{\bot\}$ and $T$ mapping the nodes of $\mathcal{T}$ to subgraphs of $G$, such that: 
\begin{enumerate}
    \item for the root node $r$ of $\mathcal{T}$ we have $T(r) = G$,
    \item if $f(t) \ne \bot$, then $f(t) \in T(t)$ and for each connected component $C$ of $T(t) \setminus f(t)$, there is a child $t'$ of $t$ such that $T(t') = C$, and $t$ has no other children,
    \item if $f(t) = \bot$, then $T(t)$ belongs to $\hh$ and $t$ has no children.
\end{enumerate}
The corpus of the decomposition is defined as $V(G) \cap f(\mathcal{T})$.
A subgraph $T(t)$ for $f(t) = \bot$
is called a canopy component.
The $\hh$-depth of $G$ is the smallest depth of an $\hh$-decomposition of $G$.
\end{definition}
\fi

\begin{figure}[h]
    \begin{center}
    \includegraphics[scale=0.11]{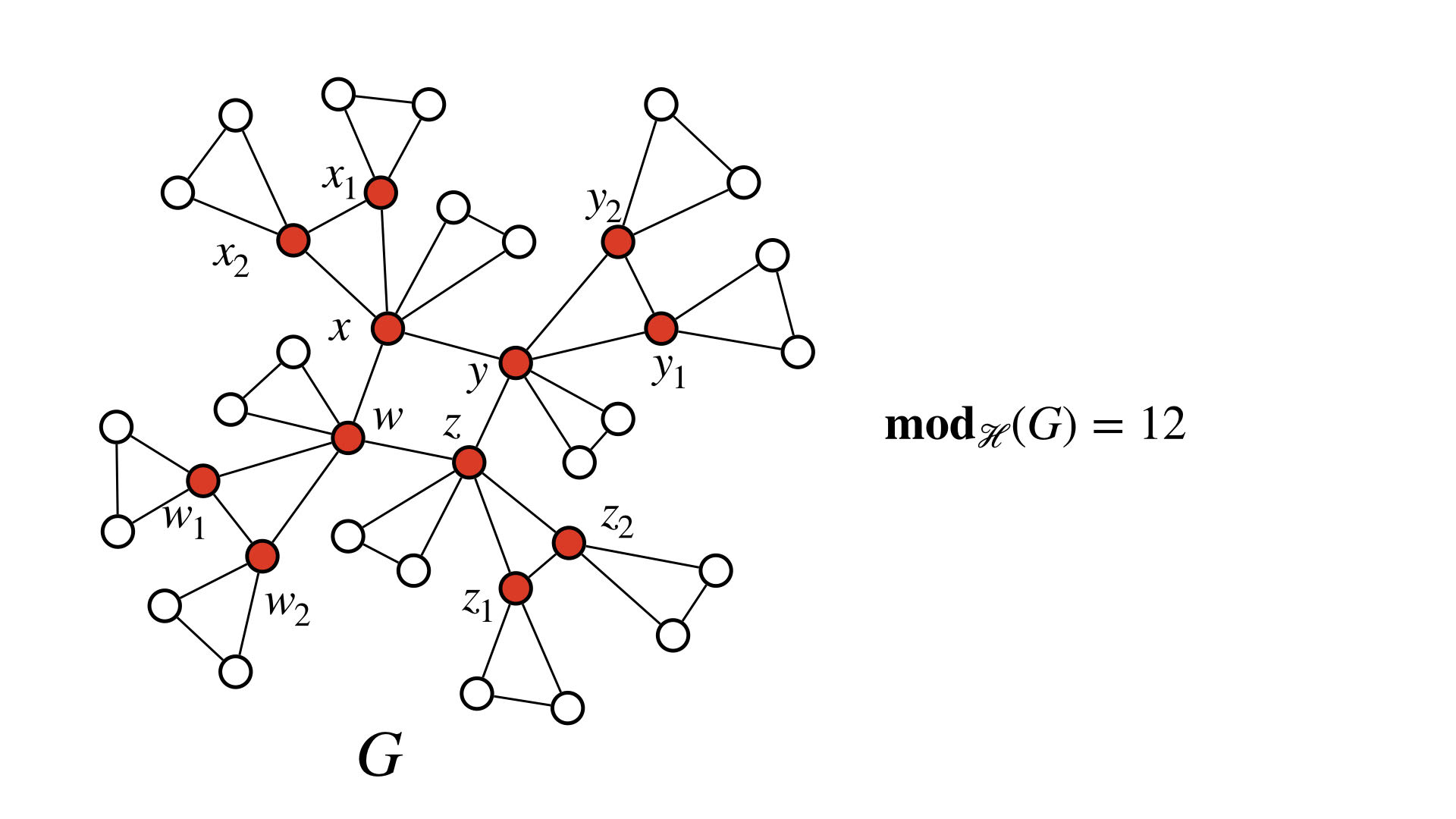}
        \includegraphics[scale=0.11]{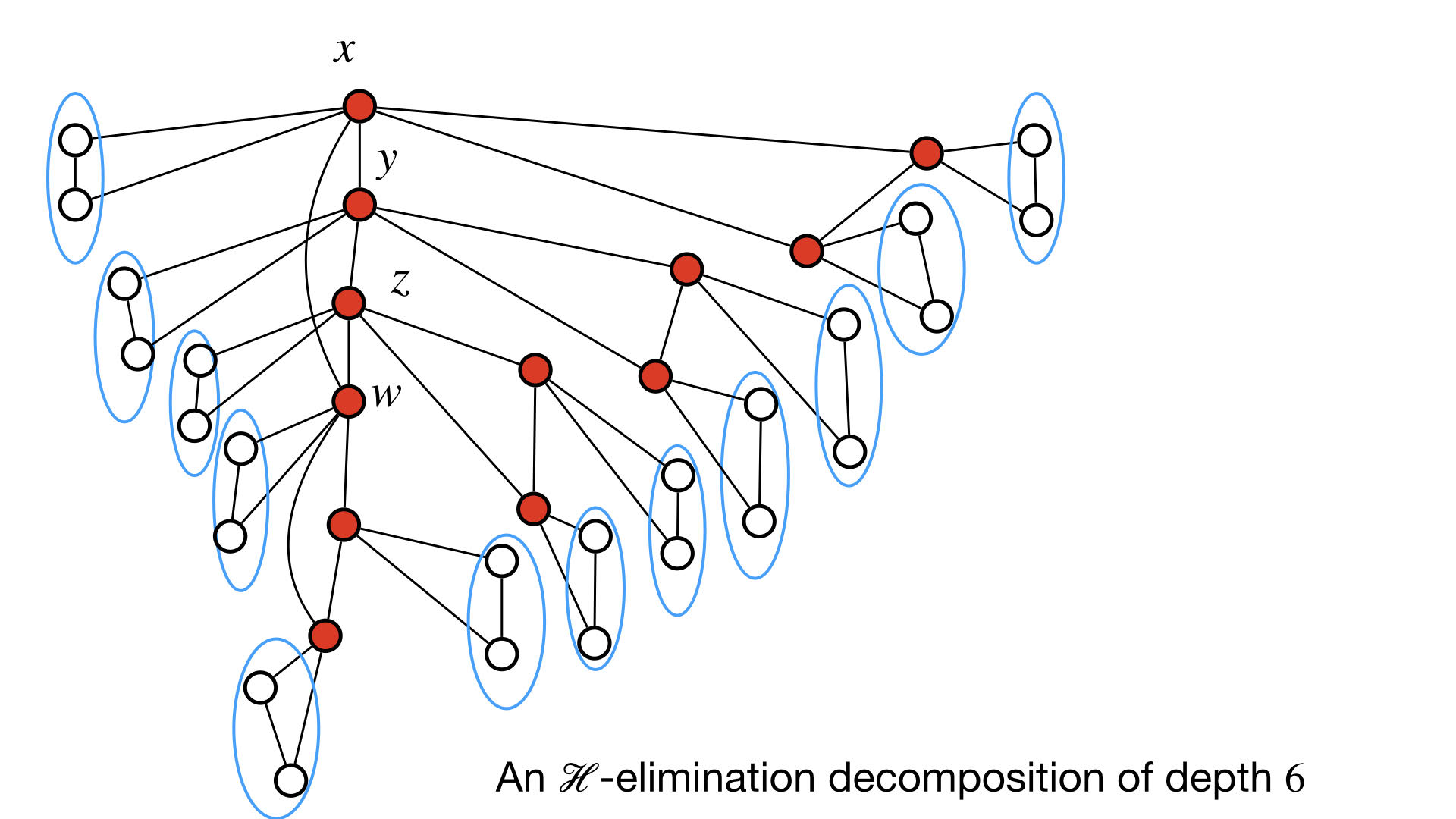}
           \includegraphics[scale=0.11]{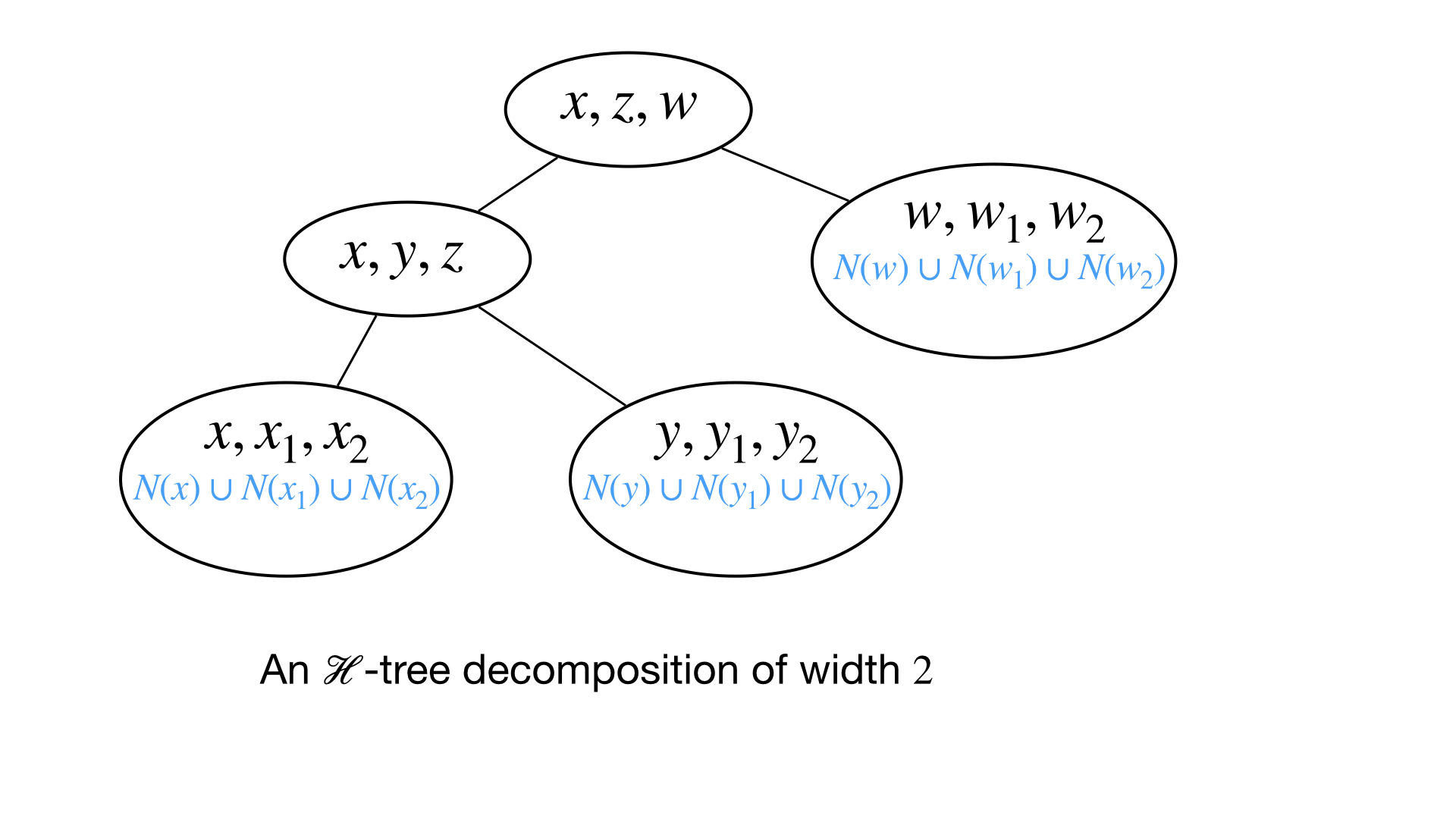}
        \end{center}
    \caption{Let $\cal H$ be the family of triangle-free graphs. The figure shows a graph $G$ with a modulator to $\cal H$ of size $12$ (in red), an $\cal H$-elimination decomposition of depth $6$, an $\cal H$-tree decomposition of width $2$.}
    \label{fig:fig1}
\end{figure}

\begin{definition}\label{def:H-elim}
\label{def:eldeco}
{\rm
%\micr{I would like to also have a triple $(T, \chi, L)$ here, where $L$ is the sum of base components. I need to refer to this set in some proofs.}
For a graph class~$\hh$, an {\em \hed} of graph $G$ is pair~$(T, \chi, L)$ where~$T$ is a rooted forest and~$\chi \colon V(T) \to 2^{V(G)}$ and $L \subseteq V(G)$, such that:
\begin{enumerate}
\setlength{\itemsep}{-1pt}
    \item For each internal node~$t$ of~$T$ we have~$|\chi(t)| \leq 1$ and $\chi(t) \subseteq V(G) \setminus L$. 
    \item The sets~$(\chi(t))_{t \in V(T)}$ form a partition of~$V(G)$.
    \item For each edge~$uv \in E(G)$, if~$u \in \chi(t_1)$ and~$v \in \chi(t_2)$ then~$t_1, t_2$ are in ancestor-descendant relation in~$T$.
    \item For each leaf $t$ of~$T$, we have $\chi(t) \subseteq L$ and the graph~$G[\chi(t)]$, {called a base component}, belongs to~$\hh$. 
\end{enumerate}

The \emph{depth} of~$T$ is the maximum number of {\bf edges} on a root-to-leaf path (see Figure~\ref{fig:fig1}).
{We refer to the union of base components as the set of base vertices.} 
The $\hh$-elimination distance of~$G$, denoted $\hhdepth(G)$, is the minimum depth of an $\hh$-elimination forest for~$G$.
{A pair $(T, \chi)$ is a (standard) elimination forest if \hh{} is the class of empty graphs, i.e., the base components are empty. The treedepth of~$G$, denoted $\td(G)$, is the minimum depth of a standard elimination forest.}
%\micr{Elimination distance to $\hh$? Unification needed}
%\bmpr{\cite{BulianD16} defines a graph from $\hh$ to have distance $0$, so to be compatible we need to subtract one from our value, or take the number of edges instead of the number of vertices on the path.}
}
\end{definition}

{It is straight-forward to verify that for any~$G$ and~$\hh$, the minimum depth of an $\hh$-elimination forest of~$G$ is equal to the $\hh$-elimination distance as defined recursively in the introduction. (This is the reason we have defined the depth of an $\hh$-elimination forest in terms of the number of edges, while the traditional definition of treedepth counts vertices on root-to-leaf paths.)}
%{%We remark that the notion of \hhdepthfull{} coincides with the notion of elimination distance to \hh{},
%which has appeared in the literature before.
% {In the remainder of the paper}, we deliberately refer to this measure as \hhdepthfull{} because it is more natural for studying the properties expressible in the language of elimination forests.}
{The following definition captures the relaxed notion of tree decomposition.}

\begin{definition}[\cite{JansenK021}]
\label{def:tree:h:decomp}{\rm
For a graph class $\hh$, an {\em \htd} of graph $G$ is a triple $(T, \chi, L)$ where~$L \subseteq V(G)$,~$T$ is a rooted tree, and~$\chi \colon V(T) \to 2^{V(G)}$, such that:
\begin{enumerate}
\setlength{\itemsep}{-1pt}
    \item For each~$v \in V(G)$ the nodes~$\{t \mid v \in \chi(t)\}$ form a {non-empty} connected subtree of~$T$. \label{item:tree:h:decomp:connected}
    \item For each edge~$\{u,v\} \in E(G)$ there is a node~$t \in V(G)$ with~$\{u,v\} \subseteq \chi(t)$.
    \item For each vertex~$v \in L$, there is a unique~$t \in V(T)$ for which~$v \in \chi(t)$,  with~$t$ being a leaf of~$T$. \label{item:tree:h:decomp:unique}
    \item For each node~$t \in V(T)$, the graph~$G[\chi(t) \cap L]$ belongs to~$\hh$. \label{item:tree:h:decomp:base}
\end{enumerate}
%\micr{We need a condition that $v \not\in L$ also belongs to some bag, even if it is isolated.}
{The \emph{width} of a tree $\hh$-decomposition is defined as~$\max(0, \max_{t \in V(T)} |\chi(t) \setminus L| - 1)$.} The $\hh$-treewidth of a graph~$G$, denoted~$\hhtw(G)$, is the minimum width of a tree $\hh$-decomposition of~$G$.
The connected components of $G[L]$ are called base components
{and the vertices in $L$ are called base vertices.} 

{A pair~$(T, \chi)$ is a (standard) \emph{tree decomposition} if~$(T, \chi, \emptyset)$ satisfies all conditions of an $\hh$-decomposition; the choice of~$\hh$ is irrelevant.}
}
\end{definition}
{In the definition of width, we subtract one from the size of a largest bag to mimic treewidth. The maximum with zero is taken to prevent graphs~$G \in \hh$ from having ~$\hhtw(G) = -1$.}

\medskip
\noindent\textbf{Topological minors:} Roughly speaking, a graph $H$ is a \emph{topological
minor} of $G$ if $G$ contains a subgraph $G'$ on $|V(H)|$ vertices, where the edges in $H$ correspond to (vertex disjoint) paths in $G'$. We now formally define the above. Let ${\sf Paths}(G)$ be the set of all paths in $G$. We say that a graph $H$ is a {\em topological minor} of $G$ if there are injective functions $\phi: V(H)\rightarrow V(G)$ and $\varphi: E(H)\rightarrow {\sf Paths}(G)$ such that $(i)$ for all $e=\{h,h'\}\in E(H)$, $\phi(h)$ and $\phi(h')$ are the endpoints of $\varphi(e)$, $(ii)$ for distinct $e,e'\in E(H)$, the paths $\varphi(e)$ and $\varphi(e')$ are internally vertex-disjoint, and $(iii)$ there does not exist a vertex $v$ such that $v$ is in the image of $\phi$ and there is an edge $e\in E(H)$ where $v$ is an internal vertex in the path $\varphi(e)$.

\medskip
\noindent
{\bf Unbreakable Graphs.} To formally introduce the notion of unbreakability, we rely on the definition of a separation:

\begin{definition}{\rm [\bf Separation]} A pair $(X, Y)$ where $X \cup Y = V (G)$ is a {\em separation} if $E(X \setminus
Y, Y \setminus X) = \emptyset$. The order of $(X, Y)$ is $\vert X\cap Y\vert$.
\end{definition}

Roughly speaking, a graph is breakable if it is possible to ``break'' it into two large parts by removing only a small number of vertices. Formally,

\begin{definition}{\rm [\bf $(s,c)$-Unbreakable graph]} Let $G$ be a graph. If there exists a separation $(X,Y)$ of order at most $c$ such that $\vert X\setminus Y\vert \geq s$ and $\vert Y\setminus X\vert \geq  s$, called an {\em $(s,c)$-witnessing separation}, then $G$ is {\em $(s,c)$-breakable}. Otherwise, $G$ is {\em $(s,c)$-unbreakable}.
\end{definition}

We next state an observation which immediately follows from the above definition.

\begin{observation}\label{obs:boundedsep}
Consider a graph $G$, an integer $k$ and a set $S \subseteq V(G)$ of size at most $k$, where $G$ is $(\ak,k)$-unbreakable graph with $|V(G)| > 2\ak+k$. Then, there is exactly one connected component $C^*$ in $G-S$ that has at least $\ak$ vertices and $|V(G)\setminus V(C^*)| < \ak + k$. 
\end{observation}

\subsection{Counting Monadic Second Order Logic}\label{sec:prelimsCMSO}
The syntax of Monadic Second Order Logic (MSO) of graphs includes the logical connectives $\vee,$ $\land,$ $\neg,$ 
$\Leftrightarrow,$ $\Rightarrow,$ variables for vertices, edges, sets of vertices and sets of edges, the quantifiers $\forall$ and $\exists$, which can be applied to these variables, and five binary relations: 
\begin{enumerate}
\setlength{\itemsep}{-2pt}
\item $u\in U$, where $u$ is a vertex variable and $U$ is a vertex set variable; 
\item $d \in D$, where $d$ is an edge variable and $D$ is an edge set variable;
\item $\mathbf{inc}(d,u),$ where $d$ is an edge variable, $u$ is a vertex variable, and the interpretation  is that the edge $d$ is incident to  $u$; 
\item $\mathbf{adj}(u,v),$ where $u$ and $v$ are vertex variables, and the interpretation is that $u$ and $v$ are adjacent;
\item  equality of variables representing vertices, edges, vertex sets and edge sets.
\end{enumerate}

Counting Monadic Second Order Logic (CMSO) extends MSO by including atomic sentences testing whether the cardinality of a set is equal to $q$ modulo $r,$ where $q$ and $r$ are integers such that $ 0\leq q<r $ and $r\geq 2$. That is, CMSO is MSO with the following atomic sentence: 
$\mathbf{card}_{q,r}(S) = \mathbf{true}$ if and only if $|S| \equiv q \pmod r$, where $S$ is a set.
We refer to~\cite{ArnborgLS91,Courcelle90,Courcelle97} for a detailed introduction to  CMSO.

We will crucially use the following result of Lokshtanov et al.~\cite{LokshtanovR0Z18} that allows one to obtain a (non-uniform) {\FPT} algorithm for {\msotwo}-expressible graph problems by designing an {\FPT} algorithm for the problem on unbreakable graphs. 
%Given a CMSO formula $\psi$, the {\sc CMSO}$[\psi]$ problem is defined as follows. 
%The input of {\sc CMSO}$[\psi]$ is a graph, and the objective is to output $\sigma_\psi(\alpha)$. 

%\begin{restatable}{theorem}{maintheorem}
%\label{thm:main}
%Let $\psi$ be a CMSO formula. For all $c\in\mathbb{N}$, there exists $s\in\mathbb{N}$ such that if there exists an algorithm that solves {\sc CMSO}$[\psi]$ on $(s,c)$-unbreakable structures in time $\OO(n^d)$ for some $d>4$, then there exists an algorithm that solves {\sc CMSO}$[\psi]$ on general structures in time $\OO(n^{d})$.
%\end{restatable}

\begin{proposition}[Theorem 1, \cite{LokshtanovR0Z18}]\label{prop:CMSOMetaTheorem}
Let $\psi$ be a CMSO sentence and let  $d>4$ be a positive integer. There exists a function $\alpha : \mathbb{N} \rightarrow \mathbb{N}$, such that for every $c\in \mathbb{N}$ there is an $\alpha (c)\in \mathbb{N}$, if there exists an algorithm that solves {\sc CMSO}$[\psi]$ on $(\alpha(c),c)$-unbreakable graphs in time $\bigoh(n^d)$, then there exists an algorithm that solves {\sc CMSO}$[\psi]$ on general graphs in time $\bigoh(n^{d})$.
\end{proposition}

\subsection{Parameterized graph problems}
A parameterized graph problem  $\Pi$ is usually defined as a subset of $\Sigma^{*}\times \Bbb{Z}^{+}$
where, in each instance $(x,k)$ of $\Pi,$ $x$ encodes a graph and $k$ is the parameter (we denote by $\Bbb{Z}^{+}$ the set of all non-negative integers). In this paper we use an extension of this definition (also used by Bodlaender et al.~\cite{BodlaenderFLPST16} and Fomin et al.~\cite{FominLST18}) that permits the parameter $k$ to be negative 
with the additional constraint that either all pairs with non-positive values of the parameter 
are in $\Pi$ or that no such pair is in $\Pi$. Formally, a parametrized problem $\Pi$
is a subset of $\Sigma^{*}\times \Bbb{Z}$ where for all $(x_{1},k_{1}),(x_{2},k_{2})\in\Sigma^{*}\times \Bbb{Z}$
with $k_{1},k_{2}<0$ it holds that $(x_{1},k_{1})\in\Pi$ if and only if  $(x_{2},k_{2})\in\Pi$.
This extended definition encompasses the traditional one and is needed for technical reasons  
(see Subsection~\ref{subsec:finiinteginde}).
In an instance of a parameterized problem $(x,k),$ the integer $k$ is called the parameter. 
\input{newprotrusion.tex}

%\newpage

%\newpage

%% file: newprotrusion.tex
%!TEX root = main.tex
%\newcommand{\mar}[1]{#1}
%A parameterized graph problem  $\Pi$ can be seen as a subset of $\Sigma^{*}\times \Bbb{Z}^{+}$
%where, in each instance $(x,k)$ of $\Pi,$ $x$ encodes a graph and $k$ is the parameter (we denote by 
%$\Bbb{Z}^{+}$ the set of all non-negative integers).  
%\todo[inline]{This needs heavy cleaning and citations from Dom Set kernel paper of TALG}
%In the next section we introduce the notion of a ``generalized protrusion''.  Recall that a protrusion in a graph  is a subgraph of constant treewidth which is separated from the rest of the graph by at most a constant number of vertices. In our variant of protrusions, instead of stipulating that the subgraph be of constant treewidth, we ask that it contains a  constant number of vertices from a solution. In this section we show that even if we have a generalized protrusion then we can find a replacement for it efficiently.  
%%In this section we introduce a notion of a ``generalized protrusion'' and  build a theory of replacement for these. 
%We first start with some relevant definitions and concepts. 
%%Towards this we define the notion of {\em boundaried graphs} and various operations on them.
%

\subsection{Boundaried Graphs} 
\label{subsec:boungrap}
Here we define the notion of {\em boundaried graphs} and various operations on them.
\begin{definition}{\rm [\bf Boundaried Graphs]}\label{def:boungraph}
A \term{boundaried graph} is a graph $G$ with a set $B\subseteq V(G)$ 
of  distinguished vertices and an injective labelling $\lambda_G$ 
from $B$  to the set $\Bbb{Z}^{+}$. The set $B$ is called the \term{{\em boundary}} of $G$ and  the vertices in $B$  are called  {\em boundary vertices} or \term{{\em terminals}}. 
Given a boundaried graph $G,$ we denote its boundary by ${\delta(G)},$
we denote its labelling by $\lambda_G$, 
and we define its {\em label set} by $\Lambda(G)=\{\lambda_{G}(v)\mid v\in \delta(G)\}$.
Given a finite set $I\subseteq \Bbb{Z}^{+}$, we define 
${{\cal F}_{I}}$  to denote the class of all boundaried graphs whose label set is $I$. 
%Similarly, we define ${\cal F}_{\subseteq I}=\bigcup_{I'\subseteq I}{\cal F}_{I'}$.
We also denote by ${{\cal F}}$ the class of all boundaried graphs.
Finally we say that a boundaried graph is a {\em $t$-boundaried} graph if $\Lambda(G)\subseteq \{1,\ldots,t\}$.
\end{definition}
%\todo[inline]{fix the mark in the definition}
%
%We remark that  in the labelling of the boundary of a $t$-boundaried graph, not all $t$ available labels are necessary used.

%
%For a graph $G=(V,E)$ and a vertex set $S \subseteq V,$ we sometime consider the graph $G[S]$ as the 
%$|\partial(S)|$-boundaried graph with $\partial(S)$ being the boundary.

\begin{definition}{\rm [\bf Gluing by $\oplus$]} Let $G_1$ and $G_2$ be two  boundaried graphs. We denote by $G_1 {\oplus} G_2$ the  graph 
(not boundaried) obtained by taking the disjoint union of $G_1$ and $G_2$ and identifying equally-labeled vertices of the boundaries of $G_{1}$ and $G_{2}.$ In $G_1 \oplus G_2$ there is an edge between two vertices if there is  an edge between them either in $G_1$ or in $G_2$, or both.  

\end{definition}

We remark that if $G_1$ has a label which is not present in $G_2$, or vice-versa, then in $G_1 \oplus G_2$ we just forget that label. 

\begin{definition} {\rm [\bf Gluing by $\oplus_\delta$]}
The {\em boundaried gluing operation} $\oplus_{\delta}$ is similar to the normal gluing operation, but results in a boundaried graph rather than a graph. Specifically $G_1 \oplus_\delta G_2$ results in a boundaried graph where the graph is $G = G_1 \oplus G_2$ and a vertex is in the boundary of $G$ if it was in the boundary of $G_1$ or  of $G_2$. Vertices in the boundary of $G$ keep their label from $G_1$ or $G_2$. 
%Both for gluing and boundaried gluing we will refer to $G_1 \oplus G_2$ or $G_1 \oplus_\delta G_2$ as the {\em sum} of $G_1$ and $G_2$, and %$G_1$ and $G_2$ are the {\em terms} of the sum.
%\todo[inline]{we do not need this}
\end{definition}

%\begin{definition}
%Let $G=G_{1}\oplus G_{2}$ where $G_{1}$ and $G_{2}$ are boundaried graphs.
%We define the \term{{\em glued}} set of $G_{i}$ as the set $B_{i}=\lambda_{G_{i}}^{-1}(\Lambda(G_{1})\cap \Lambda(G_{2})), i=1,2$. For a vertex $v\in V(G_{1})$ we define its \term{{\em heir}} $\mar{h(v)}$ in 
%$G$ as follows: if $v\not\in B_{1}$ then $h(v)=v$, otherwise $h(v)$ is the result of the identification 
%of $v$ with an equally labeled vertex in $G_{2}$. The {\em heir} of a vertex in $G_{2}$ is defined symmetrically. The \term{{\em common boundary}} of $G_{1}$ and $G_{2}$ in $G$ is equal 
%to $h(B_{1})=h(B_{2})$ where the evaluation of $h$ on vertex sets is defined in the obvious way.
%The {\em heir} of an edge $\{u,v\}\in E(G_{i})$ is the edge $\{h(u),h(v)\}$ in $G$.
%\end{definition}

Let ${\cal G}$ be a class of (not boundaried)  graphs.
By slightly abusing notation we say that a boundaried graph {\em belongs to a graph class ${\cal G}$} if the underlying graph belongs to ${\cal G}.$

\begin{definition}{\rm [\bf Replacement]}\label{defn:replacement}
Let $G$ be a $t$-boundaried graph containing a set $X\subseteq V(G)$
such that $\partial_{G}(X)=\delta(G).$ Let $G_1$ be a $t$-boundaried graph. The result of {\em replacing $X$ with $G_1$} is the graph $G^{\star}\oplus G_{1},$
where $G^{\star}=G\setminus (X\setminus \partial (X))$ is treated as a $t$-boundaried graph with  $\delta(G^{\star})=\delta(G).$
\end{definition}

\subsection{Finite Integer Index}
\label{subsec:finiinteginde}
\begin{definition}{\rm [\bf Canonical equivalence on boundaried graphs.]}
Let $\Pi$ be a parameterized graph problem whose instances are pairs of the form $(G,k).$
 Given two boundaried graphs $G_1,G_2~\in {\cal F},$ we say that \term{$G_1\!\equiv _{\Pi}\! G_2$} if 
$\Lambda(G_{1})=\Lambda(G_{2})$
 and there exists a \term{{\em transposition constant}}
$c\in\Bbb{Z}$ such that 
\begin{eqnarray*}
\forall(F,k)\in {\cal F}\times \Bbb{Z} &&  (G_1 \oplus F, k) \in \Pi \Leftrightarrow (G_2 \oplus F, k+c) \in \Pi.\label{eq:fiidef}
\end{eqnarray*}
Here, $c$ is a function of the two graphs $G_1$ and $G_2$. 
 %\sed{$c$ can be a positive or negative integer?}
%\end{itemize}
\end{definition}
Note that  the relation $\equiv_{\Pi}$  is
an equivalence relation. Observe that $c$ could be negative in the above definition. This is the reason we allow the parameter in parameterized problem instances to take negative values.

Next  we define a notion of ``transposition-minimality'' for the members 
of  each equivalence class of $\equiv_{\Pi}.$

\begin{definition}{\rm [\bf Progressive representatives~\cite{BodlaenderFLPST16}]}
\label{def:progrepr}
Let $\Pi$ be a parameterized graph problem whose instances are pairs of the form $(G,k)$
and let ${\cal C}$ be some equivalence class of $\equiv_{\Pi}$. We say that $J\in{\cal C}$ is a \term{{\em progressive 
representative}}
of ${\cal C}$ if for every $H\in{\cal C}$
there exists $c\in\Bbb{Z}^{-},$ such that 
\begin{eqnarray}
\forall(F,k)\in {\cal F}\times \Bbb{Z} \ \ \  (H \oplus F, k) \in \Pi \Leftrightarrow (J\oplus F, k+c) \in \Pi. \label{eq:progfii}
\end{eqnarray}
\end{definition}

The following lemma guarantees the existence of a progressive representative for each equivalence class of 
$\equiv_{\Pi}$. 
%Consider a graph $H$ in the equivalence class that 

\begin{lemma}[\cite{BodlaenderFLPST16}]
\label{lem:existprog}
Let $\Pi$ be a parameterized graph problem whose instances are pairs of the form $(G,k)$.
%and let $t\in\Bbb{Z}^{+}.$  
 Then each  equivalence class of $\equiv_{\Pi}$ has a progressive representative.
\end{lemma}

Notice that two  boundaried graphs with different label sets belong to 
different equivalence classes of $\equiv_{\Pi}.$ Hence for every equivalence 
class ${\cal C}$ of $\equiv_{\Pi}$ there exists some finite set $I\subseteq\Bbb{Z}^{+}$ such that 
${\cal C}\subseteq  {\cal F}_{I}$. We are now in position  to give the following definition.

\begin{definition}{\rm [\bf Finite Integer Index]}
\label{def:deffii}
A parameterized graph problem $\Pi$ whose instances are pairs of the form $(G,k)$
has {\em Finite Integer Index} (or  is \term{{\em FII}}), if and only if for every finite $I\subseteq \Bbb{Z}^+,$
the number of equivalence classes of  $\equiv_{\Pi}$ that are subsets of ${\cal F}_{I}$
is finite. For each $I\subseteq \Bbb{Z}^{+},$ we define ${\cal S}_I$ to be
a set containing exactly one progressive representative of each equivalence class of $\equiv_{\Pi}$
that is a subset of ${\cal F}_{ I}$. We also define ${\cal S}_{\subseteq I}=\bigcup_{I'\subseteq I} {\cal S}_{I'}$. 
\end{definition}

The proof of next lemma is identical to the one given for~\cite[Lemma $8.4$]{BodlaenderFLPST16}
\begin{lemma}[\cite{BodlaenderFLPST16}]
\label{lem:stronglymonotone}
Let $\hh$ family of graphs that is CMSO definable and union closed. 
Then, \probVDH has {\rm \fii}.
%Every CMSO definable strongly monotone \pmin{} and every strongly monotone \pmax{} problem has FII.
\end{lemma}

\begin{lemma}{\rm (\cite[Lemma $2$]{corr/abs-2106-04191}).}
\label{lem:cmsoDefinable}
Let $\hh$ family of graphs that is CMSO definable and union closed. 
Then, \probEDH and \probTDH is CMSO definable. 
%Every CMSO definable strongly monotone \pmin{} and every strongly monotone \pmax{} problem has FII.
\end{lemma}

\subsection{Replacement lemma}
%The result of this section will be applicable in replacing  the following kind of protrusions. 
%\begin{Definition}{\rm [\bf $r$-$\Pi$-protrusion]} Let $\Pi$ be a {\sc $p$-min-CMSO} vertex subset problem. 
% Given a graph $G$, we say that a set $X\subseteq V(G)$ is an {\em $r$-$\Pi$-protrusion} of $G$ if 
%   the number of vertices in $X$ with a neighbor in $V(G)\setminus X$ is at most $r$ and there exists a 
%   subset $S\subseteq X$ of size at most $r$ such that  $(G[X],S)\models \psi.$. 
%\end{Definition}
%
This subsection is verbatim taken from Fomin et al.~\cite[Section $3.3$]{FominLST18} and is provided here only for completion. We only need to make few simple modifications to suit our need.

\begin{definition}
Let $\cal G$ denote the set of all graphs. A graph {\em parameter}  is a function  $\Psi \colon \cal G \to \mathbb{Z}^{+}$.  That is, $\Psi$ associates a non-negative integer to a graph $G \in \cal G$. The parameter $\psi$ is called {\em monotone}, if for every $G \in \cal G$, and for every $V_1\subseteq V_2$, 
$\Psi(G[V_2])\geq \Psi(G[V_1])$.  
\end{definition}

We can use $\Psi$ to define several graph parameters such as {\em treewidth}, or given a family $\cal F$ of graphs, a minimum sized vertex subset $S$ of $G$, called modulator, such that $G-S\in {\cal F}$.  
Next we define a notion of monotonicity for parameterized problems. 

\begin{definition}{\rm (\cite[Definition $3.9$]{FominLST18}).}
We say that a parameterized graph problem $\Pi$ is {\em positive monotone} if for every graph $G$ 
there exists a unique $\ell \in \Bbb{N}$ such that for all $\ell'\in \mathbb{N}$ and $\ell' \geq \ell$, $(G,\ell')\in \Pi$ and for all 
$\ell'\in \mathbb{N}$ and $\ell' < \ell$, $(G,\ell')\notin \Pi$.  A parameterized graph problem $\Pi$ is {\em negative monotone} if for every graph $G$ 
there exists a unique $\ell \in \Bbb{N}$ such that for all $\ell'\in \mathbb{N}$ and $\ell' \geq \ell$, $(G,\ell')\notin \Pi$ and for all 
$\ell'\in \mathbb{N}$ and $\ell' < \ell$, $(G,\ell')\in \Pi$. $\Pi$ is monotone if it is either positive monotone or negative monotone.  
We denote the integer $\ell$ by {\sc Threshold($G,\Pi$)} (in short  {\sc Thr($G,\Pi$)}). 
\end{definition}

We first give an intuition for the next definition.  We are considering monotone functions and thus for every graph $G$ 
there is an integer $k$ where the answer flips. However, for our purpose we need a corresponding notion for 
boundaried graphs.   If we think of the representatives as some ``small perturbation'', then it is the max threshold over all small perturbations (``adding a representative = small perturbation''). This leads to the following definition. 

\begin{definition}{\rm (\cite[Definition $3.10$]{FominLST18}).}\label{def:equiv-boundaried-graph}
Let $\Pi$ be a monotone parameterized graph problem that has {\fii} and $\Psi$ be a graph parameter. Let  ${\cal S}_t$  be
a set containing exactly one progressive representative of each equivalence class of $\equiv_{\Pi}$ that is a subset of 
${\cal F}_{I}$, where $I=\{1,\ldots,t\}$.  
For a $t$-boundaried graph $G$, we define   
\begin{eqnarray*}
\iota(G) & = &  \max_{G'\in {\cal S}_t}  \mbox{{\sc Thr($G\oplus G',\Pi$)}},\\
\mu(G) & = &  \max_{G'\in {\cal S}_t}  \Psi(G\oplus G') 
\end{eqnarray*}
\end{definition}

%\todo[inline]{do we change the kappa definition here --- as this conflicts with the similar definition in treewidth section}

The next lemma says the following. Suppose we are dealing with some {\fii} problem and we are given a boundaried graph with boundary size $t$.  We know it has a representative of size $h(t)$ and we want to find this representative. In general finding a representative for a boundaried graph is more difficult than solving the corresponding problem. 
The next lemma says basically that if  we can find the  ``OPT'' of a boundaried graph efficiently 
then we can efficiently find its representative. Here by ``OPT''  we mean $\iota(G)$, which is a robust version of the threshold function (under adding a representative). 
And by efficiently we mean as efficiently as solving the problem on normal (unboundaried) graphs.

\begin{lemma}{\rm (\cite[Lemma $3.11$]{FominLST18}).}
%\todo{The running time should be independent of $k$ and updated accotdingly}
\label{lem:red2finiteindex}
Let $\Pi$ be a monotone parameterized graph problem that has {\fii} and $\Psi$ be a graph parameter.  Furthermore, let $\cal A$ be an algorithm for $\Pi$ that, given a pair $(G,k)$, decides whether it is in $\Pi$ in time $f(|V(G)|,\Psi(G))$. 
Then for every $t\in\Bbb{N},$ there exists a $ \xi_t \in\Bbb{Z}^{+}$ (depending on $\Pi$ and $t$), and 
an algorithm that, given a $t$-boundaried graph $G$  with $|V(G)|>\xi_t,$ outputs, in  $\cO(\iota(G)(f(|V(G)|+\xi_t,\mu(G)))$ steps,
a $t$-boundaried graph $G^\star$  such that $G\equiv_{\Pi}G^\star$ and  $|V(G^\star)| < \xi_t$. Moreover we can compute the translation 
constant  $c$ from $G$ to $G^\star$ in the same time.
\end{lemma}

\begin{proof}
We give prove the claim for positive monotone problems $\Pi$; the proof for negative monotone problems is identical. 
%Let $G^*$ be a progressive representative of the equivalence class $C$ of $\equiv_{\Pi}$ to which $G$ belongs. Since $%\equiv_{\Pi}$ has finite index we get that there is a finite number 
%Recall that we denote by  ${\cal S}_{t}$  a set of (progressive) representatives for $(\Pi, t)$ and let
% $\xi_t=\max_{Y\in {\cal S}_{t}}|Y|.$ 
 Let  ${\cal S}_t$  be
a set containing exactly one progressive representative of each equivalence class of $\equiv_{\Pi}$ that is a subset of 
${\cal F}_{I}$, where $I=\{1,\ldots,t\}$, and  let  $\xi_t=\max_{Y\in {\cal S}_{t}}|V(Y)|.$ The set  ${\cal S}_{t}$ is hardwired in the description of the algorithm. 
 Let $Y_1,\ldots,Y_\rho$ be the set of progressive representatives in ${\cal S}_{t}$. Let ${\cal F}_{t}={\cal F}_{I}$. Our objective is to find 
  a representative $Y_\ell$  for  $G$ such that 
 \begin{eqnarray}
\forall (F,k)\in {\cal F}_{t}
\times \Bbb{Z} & &   (G \oplus F, k) \in \Pi  \Leftrightarrow    (Y_\ell \oplus F, k-\vartheta(X,Y_\ell)) \in \Pi. 
\label{eq:progresivereplacement}
\end{eqnarray}
Here, $\vartheta(X,Y_\ell)$ is a constant  that depends on $G$ and $Y_\ell$.  Towards this   
we define the following matrix for the set of representatives. Let 
$$A[Y_i, Y_j]=  \mbox{{\sc Thr($Y_i\oplus Y_j,\Pi$)}}$$
The size of the matrix $A$ only depends on $\Pi$ and $t$ and is also hardwired in the description of the algorithm.  Now given $G$ we find its representative as follows. 
\begin{itemize}
\item Compute the following row vector ${\cal X}=[ \mbox{{\sc Thr($G\oplus Y_1,\Pi$)}}, \ldots ,  
\mbox{{\sc Thr($G\oplus Y_\rho,\Pi$)}})]$. For each $Y_i$ we decide whether $(G\oplus Y_i,k)\in \Pi$ using the assumed algorithm for deciding 
$\Pi$,  letting $k$ increase from $1$ until the first time $(G\oplus Y_i,k)\in \Pi$. Since $\Pi$ is positive monotone this will happen for some 
$k\leq \iota(G)$. Thus the total time to compute the vector ${\cal X}$ is $\cO(\iota(G)(f(|V(G)|+\xi_t,\mu(G)))$. 

\item Find a translate row in the matrix $A(\Pi)$. That is, find an integer $n_o$ and a representative 
$Y_\ell$ such that  
\begin{eqnarray*}
[ \mbox{{\sc Thr($G\oplus Y_1,\Pi$)}},  \mbox{{\sc Thr($G\oplus Y_2,\Pi$)}}, \ldots ,  
\mbox{{\sc Thr($G\oplus Y_\rho,\Pi$)}}] \\
=[ \mbox{{\sc Thr($Y_\ell\oplus Y_1,\Pi$)}}+n_0,  \mbox{{\sc Thr($Y_\ell\oplus Y_2,\Pi$)}}+n_0, \ldots ,  
\mbox{{\sc Thr($Y_\ell\oplus Y_\rho,\Pi$)}}+n_0]
\end{eqnarray*}
Such a row must exist since ${\cal S}_t$ is  a set of representatives for $\Pi$; the representative $Y_\ell$ for the equivalence class to which $G$ belongs, satisfies the condition.  
\item Set $Y_\ell$ to be $G^\star$ and the translation constant to be $-n_0$.
\end{itemize}
From here it easily follows that $G\equiv_{\Pi}G^\star$. This completes the proof.  
\end{proof}
 We remark that the algorithm whose existence is guaranteed by the Lemma~\ref{lem:red2finiteindex} assumes that the set  ${\cal S}_{t}$ of representatives  are hardwired in the algorithm.  In its full generality we currently donot known of a procedure that for problems having {\fii} outputs such a representative set. Thus, the algorithms using Lemma~\ref{lem:red2finiteindex}  are not uniform.

 Next we illustrate a situation in which one can  can apply  Lemma~\ref{lem:red2finiteindex} to reduce a portion of a graph. Let $\cal F$ be a family of interval graphs. Further, let $\Pi$ be the {\sc Dominating Set} problem and $\Psi$ denote the modulator to $\cal F$. That is, given a graph $G$, 
 $$\Psi(G)=\min_{ S\subseteq V(G), G-S \in {\cal F}} |S|. $$
It is possible to show that {\sc Dominating Set} parameterized by $\Psi(G)$ is \FPT. That is, we can design an algorithm that can decide whether an instance $(G,k)$ of   {\sc Dominating Set} is an \yes-instance in time 
  $f(\Psi(G))\cdot n^{\cO(1)}$. In fact, in time  $2^{\cO(\Psi(G))}n^{\cO(1)}$. This implies that if we have a 
  $t$-boundaried graph $G$, then we can find a representative of it with respect to {\sc Dominating Set} in time   $2^{\cO(\mu(G))}n^{\cO(1)}$. We will see its uses in this way in Section~\ref{sec:crossParam}. 
 
 %associates the minimum size of  a vertex subset $S$ 

%% file: bounded-mod.tex
% !TEX root = main.tex
In this section we show that for any $(\ak,k)$-unbreakable graph $G$ that has more than $3(\ak+k)$ vertices and its \hed\ (resp. \htd), $(T,\chi,L)$ of depth at most $k$, we have $|V(G) \setminus L| \leq \ak+k$ and there is a large connected component in $G[L]$, by proving the following two lemmas. %In fact we will prove a slightly more general statement, which allows us to capture a subset of modulator that our algorithms recursively computes as it progresses. To this end we begin by defining a notion of ``consistent'' decomposition.

%\begin{definition}\label{def:consistent}{
%Consider an $(\ak,k)$-unbreakable graph $G$ that has more than $\ak,k$, and an \hed\ (resp. \htd) $(T,\chi,L)$ of depth (resp. width) at most $k$. 
%}\end{definition} 

%We say that a decomposition $(T,\chi,L)$  is {\em consistent} with a subset $A \subseteq V(G)$ if the following holds:  for each $u, u' \in A \cup N_G(C^*)$, where $t,t' \in V(T)$, we have $u \in \chi(t)$ and $u' \in \chi(t')$. 
%where $u \in \chi(t)$ and $v \in \chi(t')$, $t$ and $t'$ are in ancestor-descendant relationship in $T$.

\begin{lemma}\label{lem:boundedsep}
Consider a graph $G$, an integer $k$, and any \hed\ (resp. \htd) $(T,\chi,L)$ of depth (resp. width) at most $k$ (resp. $k-1$), where $G$ is $(\ak,k)$-unbreakable graph. Then, one of the following holds:
\begin{enumerate}
\item $|V(G)|\leq 3(\ak+k)$, or
\item there is exactly one connected component $C^*$ in $G[L]$ that has at least $\ak$ vertices, and $|V(G) \setminus V(C^*)| \leq \ak+k$.
\end{enumerate}
\end{lemma}

%% file: proof-of-sep-lemma.tex
% !TEX root = main.tex
%\subsubsection{Proof of Lemma~\ref{lem:boundedsep}}\label{sec:proof-of-sep-lemma}

%We will prove that either the number of vertices in $G$ is bounded by $3(\ak+k)$ or the size of any $k$-ed modulator to ${\cal H}_{\cal F}$ for $G$ has at most $\ak+k$ vertices. In the former case the algorithm can brute-force and check if the graph has a $k$-ed modulator. In the later case, since the size of any $k$-ed modulator in $G$ has at most $\ak +k$ vertices, we ``guess'' the forest $F$ and use the algorithm for \auedtf\ (Lemma~\ref{lem:an_un}, Section~\ref{sec:analgo}) to resolve the instance. 

Towards proving the lemma, we begin by stating  a folklore result regarding weighted trees (an explicit proof can be found, for instance, in the full version of~\cite{Agrawal020}).  

%The following is a well known proposition, proof of which can be found in \cite{DBLP:conf/iwpec/Agrawal020}. This will be useful to streamline arguments in later proofs.

\begin{proposition}\label{pro:balancedsep}
Consider a tree $T$ and a weight function $w:V(T)\rightarrow \mathbb{N}$, $\tau=\sum_{t\in V(T)}w(t)$ such that $1\leq w(t)\leq 2\tau/5$, for each $t\in V(T)$. Then, there exists a non-leaf vertex $\what{t}$ in $T$ such that the connected components of $T-\{\what{t}\}$ can be partitioned into two sets $\C{C}_1$ and $\C{C}_2$, with $\sum_{C\in \C{C}_1}w(C)\leq 2\tau /3$ and $\sum_{C\in \C{C}_2}w(C)\leq 2\tau /3$, where $w(C)=\sum_{t\in V(C)}w(t)$. 
\end{proposition} 

%Consider a family of graphs that is hereditary. We say that an \hed\ $(T,\chi,L)$ for a graph $G$ is {\em special} if: i) for each $t \in V(T)$, $|\chi(t)| \geq 1$, and ii) for each $t \in V(T)$, where $\chi(t) \cap L \neq \emptyset$, we have $G[\chi(t) \cap L]$ is a connected graph. Note that if a graph $G$ admits an \hed\ of depth at most $k$, then it also admits a special \hed\ of depth at most $k$, as $\C{H}$ is hereditary.  

Next we prove useful lemma(s) about unbreakable graphs.% its a special \hed.  

%will be useful in proving the desired result. 
\begin{lemma}\label{lem:uniquelargecc}
Consider a family of hereditary graphs $\C{H}$. Furthermore, consider an integer $k$, an \akku graph $G$ with more than $3(\ak+k)$ vertices, and an \hed\ $(T, \chi, L)$ of depth at most $k$ for $G$. %\footnote{We note that the condition, ``for each $t \in V(T)$, $|\chi(t)| \geq 1$'' is not needed for the proof, as it can be easily achieved by contracting edges. We still choose to add the condition so as to avoid confusions (arising from its proof) while we invoke the lemma with a structured decomposition.} 
Then, $G[L]$ has a connected component with at least $\ak$ vertices. %\todo{either recall that for the whole section $V(G)\geq3(\ak+k)$ or mention it in the lemma.}
%with $|V(G)|>3(\ak+k)$. If $(G,k)$ is a \yes instance of \uedtf and $\BB{D}=(X,F,f:V(X)\rightarrow V(F))$ is a $(k,\hf)$-decomposition of $G$, then $G-X$ has a connected component of size at least $\ak$.
\end{lemma}
\begin{proof}
As $\C{H}$ is hereditary, note that $G$ must also admit an \hed\ $(T',\chi',L)$, such that: i) for each $t \in V(T')$ we have $|\chi(t)| \geq 1$, ii) for each $t \in V(T')$, if $\chi(t) \cap L \neq \emptyset$, then $G[\chi(t) \cap L]$ is connected, and iii) the depth of $T'$ is at most the depth of $T$. Hereafter, we will consider an \hed\ $(T',\chi',L)$ for $G$ that satisfies the above conditions. Towards a contradiction we suppose that the size of each connected component in $G[L]$ is strictly less than $\ak$. With the above assumption, we will exhibit an $(\ak,k)$-witnessing separation, which will contradict that $G$ is \akku.

%Let us first consider the case when $X = V(G) \setminus L = \emptyset$. Let ${C_1, C_2,\ldots , C_\ell}$ be the set of connected components in $G[L]$, where $|V(C_1)|\leq |V(C_2)|\leq \ldots \leq |V(C_\ell)| $. Consider the smallest integer $\ell^*\in [\ell]$ such that $\sum_{i\leq \ell^*}|V(C_i)| \geq \ak$. As $|V(C_1)|< \ak$, we have $\ell^* \geq 2$. As $\sum_{i\in [\ell^* -1]} |V(C_i)|< \ak$ and $|V(C_{\ell^*})|<\ak$, we have $\sum_{i\in [\ell^*]}|V(C_i)|\leq 2\ak$. Let $V_1=\{v \mid v\in V(C_i),i\in [\ell^*]\}$. Then $(V_1,V(G)\setminus V_1)$ is an $(\ak,k)$-witnessing separation as $|V(G)|>3(\ak+k)$.

%Hereafter, we assume that $X\neq \emptyset$. 

%Note that whenever $G$ admits an \hed\ of depth at most $k$, it also admits an \hed\ $(T',\chi',L')$, such that for each $t \in V(T')$, we have $|\chi(t)| \geq 1$. Thus, without loss of generality we assume that for $(T,\chi, L)$ satisfies the above property.\footnote{Note that if $(T,\chi, L)$ does not satisfy the required condition then we can contract each vertex violating the condition.} Let $T'$ be the tree obtained from $T$ by 

Let $T^*$ be the tree obtained from $T'$, by arbitrarily connecting one of the roots in $T'$ to all the other roots. Formally, let $\{r_1,r_2,\cdots, r_p\}$ be the set of roots in $T'$ (note that $p$ must be the number of connected components in $T'$). We let $V(T^*) = V(T')$ and $E(T^*) = E(T') \cup \{\{r_1,r_i\} \mid i \in [p]\setminus \{1\}\}$. Also, let ${\sf wht}: V(T^*) \rightarrow \mathbb{N}$, such that for $t \in V(T^*) = V(T')$, we have ${\sf wht}(t) = \chi(t)$.

By the construction of $T^*$ and ${\sf wht}$, we can obtain that for each $t \in V(T^*)$, $1 \leq {\sf wht}(t) < \ak$ and $\tau = \sum_{t \in V(T^*)} {\sf wht}(t) = n$, where $n = |V(G)|$. As $n \geq 3(\ak+k)$, we have $2\tau /5>(\ak+k)$, and hence for each $t\in V(T^*)$, $1 \leq {\sf wht}(t) \leq 2\tau/5$. The tree $T^*$ and the weight function ${\sf wht}$, satisfies the premises of Proposition \ref{pro:balancedsep}. Thus, using the proposition, there is a non-leaf node $t^*$ in $T^*$ such that connected components in $T^*-\{t^*\}$ can be partitioned into two sets $\C{C}_1$ and $\C{C}_2$, such that for each $j\in [2]$, ${\sf wht}(\C{C}_j)= \sum_{C \in \C{C}_j} {\sf wht}(C)\leq 2\tau /3 \leq 2n/3 \leq 2(\ak+k)$, where ${\sf wht}(C)=\sum_{t\in V(C)}{\sf wht}(t)$. 

For $i\in [2]$, let $Z_i=\bigcup_{C \in \C{C}_i}V(C)$ and $\what{V}_i=\{v \in V(G)\setminus L \mid \mbox{ for some } t \in Z_i, \chi(t) = \{v\}\}$. For $i \in [2]$, let $\what{U}_i = \cup_{t \in Z_i} (\chi(t) \cap L)$. Let $V_i = \what{V}_i \cup \what{U}_i$, for each $i \in [2]$. From construction of $T^*$ and ${\sf wht}$, we obtain that $|V_1|\leq 2(\ak+k)$. As $n>3(\ak+k)$, we have $|V_2|>\ak+k$. By similar arguments, we obtain $|V_1|>\ak+k$. Observe that $V_1 \cup V_2=V(G)\setminus \chi(t^*)$. Let $S$ be the set containing each vertex $v \in V(G) \setminus L$, such that $t_v \in V(T^*)$, where $\chi(t_v) = \{v\}$, $t_v$ is an ancestor of $t^*$ (possibly $t_v = t^*$). Note that $|S| \leq k$. 

%By the construction of $T^*$ and Proposition~\ref{pro:balancedsep}, we have $t^*$ is not a leaf in $T^*$, thus we conclude that $t^* \in V(F')\subseteq V(F)$. Let $T'_i$ the tree rooted at $r_i$ in $F'$, containing $\what{t}$. Furthermore, Let $S = \{x \in V(G) \setminus L \mid  \mbox{ is an ancestor of } \what{t} \mbox{ in } F\}$. As $(X,F,f)$ is a pseudo solution, the depth of $F$ is at most $k$ and $f$ is injective, and thus, $|S| \leq k$. %Let $V'_1=V_1\cup S$ and $V'_2=V_2\cup S$. 

Let $Y_1=V_1\cup S$ and $Y_2=V_2\cup S$. 
Note that $S=Y_1\cap Y_2$, $|Y_1 \setminus Y_2|=|Y_1 \setminus S| \geq \ak$, $|Y_2 \setminus Y_1|= |Y_2 \setminus S| \geq \ak$, $|S| \leq k$, and $Y_1 \cup Y_2 = V(G)$. Moreover, by the construction of $Y_1$ and $Y_2$, we have that there is no edge $\{u,v\} \in E(G)$, such that $u \in Y_1 \setminus S$ and $v \in Y_2 \setminus S$. This implies that $(Y_1,Y_2)$ is a separation of order $k$ in $G$ 
such that $|Y_1 \setminus Y_2| \geq \ak$, $|Y_2 \setminus Y_2| \geq \ak$
This contradicts the assumption that $G$ is $(\ak,k)$-unbreakable. 
\end{proof}

Analogous to the above, we can prove a result regarding \htd s. 
\begin{lemma}\label{lem:uniquelargecc-td}
Consider a family of hereditary graphs $\C{H}$. Furthermore, consider an integer $k \geq 2$, an \akku graph $G$ with more than $3(\ak+k)$ vertices, and an \htd\ $(T, \chi, L)$ of width at most $k-1$ for $G$. Then, $G[L]$ has a connected component with at least $\ak$ vertices.
\end{lemma}
\begin{proof}
If $L = V(G)$, then note $G$ must have a connected component that have exactly one connected component with at least $\ak$ vertices, as $G$ is \akku. Hereafter we assume that $V(G) \setminus L \neq \emptyset$. As $\C{H}$ is hereditary, $G$ must also admit an \htd\ $(T',\chi',L)$, such that: 
\begin{enumerate}
\item for every distinct vertices $t,t' \in V(T')$ with $\chi(t)\cap L = \emptyset$ and $\chi(t') \cap L = \emptyset$, we have $\chi(t) \not \subseteq \chi(t')$,
\item for each $t \in V(T')$ with $\chi(t) \cap L \neq \emptyset$, $\chi(t) \setminus L \subseteq \chi(t_{\sf par})$, where $t_{\sf par}$ is the parent of $t$ in $T'$. 
\item for each $t \in V(T')$, if $\chi(t) \cap L \neq \emptyset$, then $G[\chi(t) \cap L]$ is connected,
\item the width of $T'$ is at most the width of $T$, and
\item $T'$ is connected. 
\end{enumerate}
Hereafter, we will consider an \htd\ $(T',\chi',L)$ for $G$ that satisfies the above conditions. Towards a contradiction we suppose that the size of each connected component in $G[L]$ is strictly less than $\ak$ vertices. We define a (partition) function $\rho: V(T') \rightarrow 2^{V(G)}$ using $\chi$ and the properties of $(T',\chi',L)$ as follows. For each $t \in V(T')$, where: i) $\chi'(t) \cap L \neq \emptyset$, we set $\rho(t) = \chi'(t) \cap L$, and ii) otherwise, we set $\rho(t) = \chi'(t) \setminus \chi'(t_{\sf par})$, where $t_{\sf par}$ is the parent (if it exists) of $t$ in $T'$.\footnote{If $t$ is a root, then $\rho(t) = \chi'(t)$.} Notice that for each $t\in V(T')$, $\rho(t) \neq \emptyset$ (recall $V(G) \setminus L \neq \emptyset$) and $\{\rho(t) \mid t \in V(T')\}$ is a partition of $V(G)$. 
%For $t \in V(T')$, we let $\rho(t) = \chi(t) \setminus \chi(t_{\sf par})$, where $t_{\sf par}$ is the parent (if it exists) of $t$ in $T'$.\footnote{$\rho(t) = \chi(t)$, when $t$ is a root in $T'$.} From item 1 of Definition~\ref{def:tree:h:decomp} we can obtain that for distinct $t,t' \in V(T')$, $\rho(t) \cap \rho(t') =\emptyset$. Furthermore, as for every distinct $t,t' \in V(T')$, we have $\chi(t) \not \subseteq \chi(t')$, we can obtain that for each $t\in V(T')$, $\rho(t) \neq \emptyset$. 
We define the weight function ${\sf wht}: V(T') \rightarrow \mathbb{N}$, by setting, for each $t \in V(T')$, ${\sf wht}(t) = |\rho(t)|$. By the construction of ${\sf wht}$ we can obtain that, for each $t \in V(T)$, $1 \leq {\sf wht}(t) \leq \max\{\ak-1 + k\} \leq \ak + k$ (recall our assumption that each connected component in $G[L]$ has less than $\ak$ vertices). Furthermore, we have $\tau = \sum_{t\in V(T')} {\sf wht}(t) = n = |V(G)|$. 

As $n \geq 3(\ak+k)$, we have $2\tau /5>(\ak+k)$, and hence for each $t\in V(T')$, $1 \leq {\sf wht}(t) \leq 2\tau/5$. The tree $T'$ and the weight function ${\sf wht}$, satisfies the premises of Proposition \ref{pro:balancedsep}. Thus, there is a non-leaf node $t^*$ in $T'$ such that connected components in $T'-\{t^*\}$ can be partitioned into two sets $\C{C}_1$ and $\C{C}_2$, such that for each $j\in [2]$, ${\sf wht}(\C{C}_j)= \sum_{C \in \C{C}_j} {\sf wht}(C)\leq 2\tau /3 = 2n/3$, where ${\sf wht}(C)=\sum_{t\in V(C)}{\sf wht}(t)$. Moreover, ${\sf wht}(\C{C}_1),{\sf wht}(\C{C}_2)  \geq n/3 \geq \ak +k$. Let $S = \chi'(t^*)$, and note that as $t^*$ is a non-leaf node, we have $|S| \leq k$. For $i \in [2]$, $V_i = \cup_{t \in \C{C}_i} \rho(t)$. Notice that $(S \cup V_1, S \cup V_2)$ is a separation of order $k$, where $|V_1\setminus S|, |V_2\setminus S|\geq \ak$, which contradicts that $G$ is \akku. 
\end{proof}

Using the above two lemmas we obtain the desired result (Lemma~\ref{lem:boundedsep}). 
%\begin{lemma}%\label{lem:boundedsep}
%Consider an integer $k$, and an \akku graph $G$ with $|V(G)|>3(\ak+k)$. If $(G,k)$ is a \yes instance of \uedtf and $\BB{D}=(X,F,f:V(X)\rightarrow V(F))$ is a $(k,\hf)$-decomposition of $G$, then $G-X$ has a unique connected component of size at least $\ak$ and $|X|\leq \ak+k$.
%\end{lemma}
\begin{proof}[Proof of Lemma~\ref{lem:boundedsep}] 

Consider an integer $k$, an \akku\ 
graph $G$, and an \hed\ (resp. \htd) $(T,\chi,L)$ for $G$, of depth (resp. width) $k$ (resp. $k-1$). If $|V(G)| \leq 3(\ak +k)$, then the condition required by the lemma is trivially satisfied. Now consider the case when $|V(G)| > 3(\ak +k)$. Note that $G$ must also admit an \hed\ (resp. \htd), say, $(T',\chi',L)$ such that for each $t\in V(T')$, $G[\chi'(t) \cap L]$ is connected. From Lemma \ref{lem:uniquelargecc}, $G[L]$ has a connected component of size at least $\ak$, and let $D$ be such a connected component, and $t^* \in V(T')$ be a vertex such that $V(D) \subseteq \chi'(t^*)$. Let $S = \{v \in V(G) \setminus L \mid \mbox{ for some } t \in V(T') \setminus \{t^*\}, \chi'(t) = \{v\}, \mbox{ where $t$ is an ancestor of } t^*\}$ (resp., let $S = \chi'(t^*) \setminus L$). Note that $|S| \leq k$, and $N_G(D) \subseteq S$. The above, together with the assumption that $G$ is \akku\ implies that $V(G) \setminus (S \cup V(D)) < \ak$. From the above we can obtain that $|V(G) \setminus L| \leq \ak + k$ and there is exactly one connected component in $G[L]$ that has more than $\ak$ vertices. This concludes the proof. 
%
%As $(X,F,f)$ is a pseudo solution and $G$ is $(\ak,k)$-unbreakable, we can obtain that, $G-X$ does not have two distinct connected components with at least $\ak$ vertices. From the above discussions we can conclude that $G-X$ has exactly one connected component of size at least $\ak$. Let $D^*$ be the unique connected component in $G-X$ of size at least $\ak$. As $(X,F,f)$ is a pseudo solution, we can obtain that $|N_G(D^*)| \leq k$. The above together with our assumption that $G$ is $(\ak,k)$-unbreakable implies that $|X| <  \ak +k$. 
%Now we will show that $|X| \leq \ak +k$. Towards a contradiction, suppose that the above statement is not true. Next, we show that $|X|\leq \ak+k$. Let $C^*$ be the unique connected component of size at least $\ak$ and $t\in \Lf{F}$ be such that $N(C^*)\subseteq P^{\mathbb{D})}_t$. For a contradiction, assume that $|X|>\ak+k$. Then, $N(C^*)\subseteq P^{\BB{D}}_x$. Observe that $|P^{\BB{D}}_x|\leq k$. Then, the separation $(N_G[C^*],V(G) \setminus V(C^*))$ is an $(\ak,k)$-witnessing separation, a contradiction to $G$ being $(\ak, k)$-unbreakable graph. 
\end{proof}

%% file: compute-ED.tex
% !TEX root = main.tex
The objective of this section is to prove the following lemma. 

\begin{lemma}\label{lem:self-reduce-h-ed}
Consider an algorithm $\SC{A}$, for \probEDH, that runs in time $f(\ell')\cdot g(|V(G')|)$, for an instance $(G',\ell')$ of the problem.\footnote{We will use the standard assumption from Parameterized Complexity that the functions $f$ and $g$ are non-decreasing. For more details on this, please see Chapter 1 of the book~\cite{CyganFKLMPPS15}.} Then, for any given graph $G$ on $n$ vertices, we can compute an \hed\ for $G$ of depth $\ell = \hhdepth(G)$, in time bounded by $f(\ell) \cdot g(n) \cdot n^{\C{O}(1)}$. 
\end{lemma}

 Consider a family of graphs $\C{H}$, for which \probTDH\ admits an algorithm, say, $\SC{M}_{\sf elm}$, which given a graph $G$ on $n$ vertices and an integer $k$, runs in time $f(k) \cdot g(n)$, and output $1$ if $\hhed(G) \leq k$ and $0$, otherwise. We design a recurive algorithm that given a graph $G$ and an integer $k$, and returns an \hed\ of depth at most $k$, or returns that no such decomposition exists.

Consider a given graph $G$ and an integer $k$. We assume that the graph $G$ is connected, as otherwise, we can apply our algorithm for each of its connected components. We will explicitly ensure that, while making recursive calls, we maintain the connectivity requirement. We now state the base cases of our recursive algorithm.

%\subparagraph{Base Cases.} 
%\begin{itemize}
\vspace{1.5mm}
\noindent{\bf Base Case 1.} If $G \in \C{H}$ and $k \geq 0$, then the algorithm returns $(T = (\{r\}, \emptyset), \chi: \{r\} \rightarrow 2^{V(G)}, V(G))$, where $\chi(r) = V(G)$, as the \hed\ of $G$. 

\vspace{1.5mm}
\noindent{\bf Base Case 2.} If Base Case 1 is not applicable and $k \leq 0$, then return that $\hhed(G) > k$.  

For each $v \in V(G)$, let $\C{C}_v$ be the set of connected components in $G-\{v\}$

\vspace{1.5mm}
\noindent{\bf Base Case 3.} If Base Case 1 and 2 are not applicable, and there is no $v \in V(G)$, such that for every $C \in \C{C}_{v}$, $(C,k-1)$ is a yes-instance of \probEDH, then return that $\hhed(G) > k$. 
%\end{itemize} 

The correctnesses of Base Case 1 and 2 are immediate from their descriptions. If the first two base cases are not applicable, then $k \geq 1$ and $\hhed(G) \geq 1$ must hold. Thus, for any \hed, say, $(T,\chi,L)$ for $G$, $T$ must have at least one vertex which is not a leaf. The third base case precisely returns that $\hhed(G) > k$, when the above condition cannot be satisfied, thus its correctness follows. Using $\SC{M}_{\sf elm}$, we can test if $(G,0)$ is a yes-instance of \probEDH\ in time bounded by $f(k) \cdot g(n)$. Thus, we can test/apply Base Case 1, 2 and 3 in time bounded by $f(k) \cdot g(n) \cdot n^{\C{O}(1)}$. Hereafter we assume that the base cases are not applicable. 

%\begin{itemize}
%\item
\vspace{1.5mm} 
\noindent{\bf Recursive Step.} Find a vertex $v^* \in V(G)$, such that for every $C \in \C{C}_{v^*}$, $(C,k-1)$ is a yes-instance of \probEDH, using $\SC{M}_{\sf elm}$. Such a $v^*$ exists as Base Case 3 is not applicable. 

Recursively obtain an \hed\ $(T_C,\chi_C, L_C)$ for the instance $(C, k-1)$, for each $C \in \C{C}_{v^*}$. Let $L^* = \bigcup_{C\in \C{C}_{v^*}} L_C$, and let $T^*$ be the forest defined as follows. We have $V(T^*) = \{r^*\} \cup (\bigcup_{C\in \C{C}_{v^*}} V(T_C))$, where $r^*$ is a new vertex, and $E(T^*)$ contains all edges in $E(T_C)$, for each $C \in \C{C}_{v^*}$, and for each root $r$ in some forest in $T_C$, for some $\C{C}_{v^*}$, the edge $\{r^*,r\}$ belongs to $E(T^*)$. Finally, let $\chi^*: V(T^*) \rightarrow 2^{V(G)}$ be the function such that $\chi^*(r^*) = \{v^*\}$, and for each $C \in \C{C}_{v^*}$ and $a \in V(T_C)$, we have $\chi^*(a) = \chi_C(a)$. Return $(T^*,\chi^*,L^*)$ as the \hed\ for $G$. 

The correctness of the above recursive step follows from its description. Moreover, it can be execute in time bounded by $f(k) \cdot g(n) \cdot n^{\C{O}(1)}$, using $\SC{E}_{\sf mod}$. 

The overall correctness of the algorithm follows from the correctness of each of its base cases and recursive step. Moreover, as we can always assume that $k \leq n$ and the depth of the recursion tree can be bounded by $k+1$, we can obtain that our algorithm runs in time $f(k) \cdot g(n) \cdot n^{\C{O}(1)}$, using $\SC{E}_{\sf mod}$. By exhibiting the above algorithm, we have obtained a proof of Lemma~\ref{lem:self-reduce-h-ed}.  
%\end{itemize}

%We will assume that $\C{H}$ is not the family of all graphs, otherwise, the problem is trivial, i.e., we can return $(\Phi, \phi, V(G))$ as the $\C{H}$-tree decomposition of width $0$, where $\Phi$ is the graph with no vertices and $\phi : \emptyset \rightarrow \emptyset$ is the empty function. 

%% file: compute-decomposition.tex
% !TEX root = main.tex
The objective of this section is to prove the following lemma.

\begin{lemma}\label{lem:self-reduce-h-td}
Consider an algorithm $\SC{A}$, for \probTDH, that runs in time $f(\ell')\cdot g(|V(G')|)$, for an instance $(G',\ell')$ of the problem.\footnote{Again, we assume that the functions $f$ and $g$ are non-decreasing.} Then, for any given graph $G$ on $n$ vertices, we can compute an \htd\ for $G$ of width $\ell = \hhtw$, in time bounded by $\big (f(\ell) \cdot g(n) + \ell^{\C{O}(\ell^2)} \big ) \cdot n^{\C{O}(1)} + h(\C{F})$, where $h(\C{F})$ depends only on the family $\C{H}$. 
\end{lemma}

 Consider a family of graphs $\C{H}$, for which \probTDH\ admits an algorithm, say, $\SC{T}_{\sf tw}$, which given a graph $G$ on $n$ vertices and an integer $k$, runs in time $f(k) \cdot g(n)$, and output $1$ if $\hhtw(G) \leq k$ and $0$, otherwise. We will assume that $\C{H}$ is not the family of all graphs, otherwise, the problem is trivial, i.e., we can return $(T=(\{t\},\emptyset), \chi, V(G))$, where $\chi(t) = V(G)$, as the $\C{H}$-tree decomposition of width $0$. 

We will design an algorithm $\SC{D}_{\sf tw}$ which, for given a graph $G$ on $n$ vertices, will construct an \htd\ for $G$ of $\C{H}$-treewidth $\ell = \hhtw$, in time bounded by $\ell^{\C{O}(\ell^3)} \cdot f(\ell) \cdot |V(G)|^{\C{O}(1)} + h(\C{F})$, where $h(\C{F})$ is a number depending on the family $\C{F}$. Intuitively speaking, we will attach a flower of obstructions on each vertex and check if the resulting graph has its $\C{H}$-treewidth exactly the same as $\hhtw(G)$. If the $\C{H}$-treewidth does not increase, then we will be able to obtain that this vertex can be part of the modulator. We repeat this procedure to identify the vertices $S \subseteq V(G)$ that go the the modulator. After this, we take the torso of $G[S]$ in $G$, to obtain the graph (to be denoted by) $\wtilde{G}_S$. Then using the known algorithm of Bodlaender~\cite{Bodlaender96}, we compute a tree decomposition for $\wtilde{G}_S$, using which we construct an $\C{H}$-tree decomposition for $G$.  
 
We next state an observation that will be useful in constructing an obstruction, i.e., a graph outside $\C{H}$.  

\begin{observation}\label{obs:obstruction}
There exists a number $h = h(\C{H})$,\footnote{That is, $h$ depends on the family $\C{H}$.} such that we can find a graph $H \notin \C{H}$ in $h$ many steps, where each step can be execute in constant time.  
\end{observation}
\begin{proof}
We initialize $i=1$ and do following steps:
\begin{enumerate}
\item Construct the set, $\C{G}_i$, that contains all graph on exactly $i$ vertices.
\item For each $H \in \C{G}_i$, check if $(H',0)$ is a no-instance of \probTDH, using the algorithm $\SC{T}_{\sf tw}$, and if it is a no-instance, then return the graph $H$ (and exit). Otherwise, increment $i$ by $1$ and go to Step $1$.   
\end{enumerate}  

Let $q \geq 1$, such that $\C{H}$ does not contain some graph on (exactly) $q$ vertices. Note that $q$ is well-defined, as $\C{H}$ is not the family of all graphs by our assumption. Notice that at the iteration where $i=q$, we will be able to output a graph that is not in $\C{H}$. Also the number of steps executed by the procedure we described depends only on $q$ (which in turn depends only on $\C{H}$). This concludes the proof.   
\end{proof}

We next state an easy observation using which we can compute $\ell = \hhtw(G)$ with the help of the algorithm $\SC{T}_{\sf tw}$. 

\begin{observation}\label{obs:compute-H-tw}
For a given graph $G$ on $n$ vertices, we can compute (using $\SC{T}_{\sf tw}$) $\ell = \hhtw(G)$ in time bounded by $\C{O}(n) \cdot f(\hhtw(G)) \cdot g(n)$.  
\end{observation}
\begin{proof}
We iterate over $i \in \mathbb{N}$ (starting from $0$) and check whether $(G,i)$ is a yes-instance of \probTDH\ using $\SC{T}_{\sf tw}$, and stop at the iteration where the instance is a yes-instance. Note that, the iteration $i$ at which we stop, it must hold that $i=\hhtw(G)$. As $\hhtw(G) \leq n$ and $f$ is a non-decreasing function by our assumption, our procedure achieves the claimed running time bound. 
\end{proof}

We now move to formal description of our algorithm. We fix an arbitrary ordering of vertices in $G$, and let $V(G) = \{v_1,v_2,\cdots, v_n\}$. We compute $\ell = \hhtw(G)$, using Observation~\ref{obs:compute-H-tw}. Let $H^*$ be the graph returned by Observation~\ref{obs:obstruction}, and let $V(H^*) =\{u^*_1,u^*_2, \cdots, u^*_q\}$. We will construct a graph $G'_i$ and a set $S_i \subseteq V(G)$, for each $i \in [n]$, where we add a flower of obstruction at $v_i$ (and add $v_i$ to $S_i$) if and only if after adding such obstructions, the $\C{H}$-treewidth of the resulting graph doesn't change. Formally, we do the following. 
\begin{enumerate}
\item Set $G'_0 = G$ and $S_0 = \emptyset$. 
\item For each $i \in [n]$ (in increasing order), we do the following: 
\begin{enumerate}
\item Initialize $G'_1 = G'_{i-1}$ and $S_i = S_{i-1}$. 
\item We obtain the graph $\what{G}_i$ obtained from $G'_{i-1}$ by adding $k+2$ copies of $H$ at $v$ as follows. For $j \in [k+2]$, let $H^*_j$ be the graph such that $E(H^*_j) = \{u^*_{1,j}, u^*_{2,j},\cdots, u^*_{q,j}\}$ and $E(H^*_j) =\{\{u^*_{p,j}, u^*_{r,j}\} \mid p,r \in [q], \{u^*_p,u^*_r\} \in E(H)\}$. Furthermore, let $H'_j = H^*_j - \{u^*_{1j}\}$. We let $\what{G}_i$ be the graph with $V(\what{G}_i) = V(G'_{i-1}) \bigcup \big(\cup_{j \in [k+2]} V(H'_j)\big)$ and $E(\what{G}_i) = E(G'_{i-1}) \bigcup \big (\cup_{j \in [k+2]} E(H'_j) \big ) \bigcup \big\{\{v_i, u^*_{p,j}\} \mid j \in [k+2], p \in [q] \mbox{ and } \{u^*_1, u^*_p\} \in E(H) \big\}$.
\item Check if $(\what{G}_i, \ell)$ is a yes-instance of \probTDH\ using $\SC{T}_{\sf tw}$. If the above is true, then set $G'_i = \what{G}_i$ and $S_i =S_{i-1} \cup \{v_i\}$, and otherwise, set $G'_i = G'_{i-1}$ and $S_i =S_{i-1}$. 
\end{enumerate}
\end{enumerate}

Next we show that, there is an $\C{H}$-tree decomposition of optimal width for $G$ which puts exactly the vertices in $S_n$ in the modulator.

\begin{lemma}\label{lem:equate-modulator-htd}
There is an $\C{H}$-tree decomposition, $(T,\chi,V(G) \setminus S_n)$, for $G$ of width $\ell = \hhtw(G)$. 
\end{lemma}
\begin{proof}
For each $i \in [n]_0$, the construction of $G'_i$ implies that $G'_i$ must admit an $\C{H}$-tree decomposition $(T'_i,\chi'_i, L'_i)$ of width $\hhtw(G)$, such that $S_i = (V(G'_i)\setminus L'_i) \cap \{v_1,v_2, \cdots, v_i\}$.\footnote{Here we use the fact that $H^*$ is a graph that has smallest number of vertices, which does not belong to $\C{H}$. Thus, deletion of $v_i$ from $G'_i$ would imply that each of the newly attached obstructions at $v_i$ (if any) are intersected.} As $\C{H}$ is a hereditary family of graphs, we can obtain that for each $i \in [n]$ and $\chi_i: V(T'_i) \rightarrow 2^{V(G)}$, where for $t\in V(T'_i)$, $\chi_i(t) = \chi'_i(t) \cap V(G)$, $(T'_i,\chi_i, L'_i \cap V(G))$ is an $\C{H}$-tree decomposition for $G$ of width $\hhtw(G)$, such that $S_i = V(G) \setminus L'_i$. The above in particular implies that, $(T'_n,\chi_n, V(G) \setminus S_n)$ is an $\C{H}$-tree decomposition for $G$ of width $\hhtw(G)$. This concludes the proof. 
\end{proof}

Let $S = S_n$ and $L = V(G) \setminus S$. Let $\wtilde{G}$ be the graph obtained from $G$ with vertex set $V(G)$, by taking a torso with respect to the connected components in $G[L]$. That is, $V(\wtilde{G}) = V(G)$ and for $u,v \in V(\wtilde{G})$, $\{u,v\}$ is an edge in $\wtilde{G}$ if and only if one of the following holds: i) $\{u,v\} \in E(G)$, or ii) there is a connected component $C$ in $G-S$ ($=G[L]$), such that $u,v \in N_G(C)$ and $u \neq v$. We let $\wtilde{G}_S = \wtilde{G}[S]$

We have the following observation regarding $\wtilde{G}$, which follows from its construction and Lemma~\ref{lem:equate-modulator-htd}. 

\begin{observation}\label{obs:compute-H-td} The following properties hold: 
\begin{enumerate}
\item A tuple $(T,\chi,V(G) \setminus S)$ is either an $\C{H}$-tree decomposition for both $G$ and $\wtilde{G}$, or none. Moreover, there is at least one such $\C{H}$-tree decomposition of width $\hhtw(G)$ for $\wtilde{G}$.
\item Treewidth of $\wtilde{G}_S$ is $\hhtw(G)$.
\end{enumerate}  
\end{observation}

Due to the above observation, it is now enough to compute a tree decomposition of $\wtilde{G}$ of width $\hhtw(G)$, to obtain an $\C{H}$-tree decomposition for $G$ of width $\hhtw(G)$. We next use the following result, which immediately follows as a corollary from the result of Bodlaender et al.~\cite{Bodlaender96}. 

\begin{proposition}[see,~\cite{Bodlaender96} or Theorem 7.17~\cite{CyganFKLMPPS15}]\label{prop:chordal-tw}
There is an algorithm, which given a graph $G$ on $n$ vertices, in time bounded by ${\sf tw}(G)^{\C{O}({\sf tw}(G)^2)} \cdot n$, computes a tree decomposition of $G$ of width ${\sf tw}(G)$.   
\end{proposition}

We are now ready to prove Lemma~\ref{lem:self-reduce-h-td}. 

\begin{proof}[Proof of Lemma~\ref{lem:self-reduce-h-td}]
Consider a graph $G$ on $n$ vertices. We construct the graph $\wtilde{G}$ (and $\wtilde{G}_S$) as described previously. From Observation~\ref{obs:obstruction} and the constructions of $G'_{i}$ and $S_i$, for $i\in [n]$, implies that $\wtilde{G}$ (and $\wtilde{G}_S$) can be constructed in time bounded by $f(\ell) \cdot g(n) \cdot n^{\C{O}(1)} + h(\C{F})$, where $\ell = \hhtw(G)$. Then using Proposition~\ref{prop:chordal-tw} we can compute a tree decomposition, $(\wtilde{T}_S,\wtilde{\chi}_S: V(\wtilde{T}_S) \rightarrow 2^{V(\wtilde{G}_S)})$ of width $\ell$, for $\wtilde{G}_S$ in time bounded by $\ell^{\C{O}(\ell^2)} \cdot n^{\C{O}(1)}$ (see item 2 of Observation~\ref{obs:compute-H-td}). Recall that $S = V(\wtilde{G}_S)$. 
For each connected component $C$ in $G-S$, the construction of $\wtilde{G}_S$ implies that $N_G(C)$ ($\subseteq S$) induces a clique in $\wtilde{G}_S$. Thus there must exist $t \in V(\wtilde{T}_S)$ such that $N_G(C) \subseteq \wtilde{\chi}_S(t)$. We construct a tree from $\wtilde{T}_S$ and a function $\chi: V(T) \rightarrow 2^{V(G)}$ as follows. Initialize $T = \wtilde{T}_S$ and $\chi = \wtilde{\chi}_S$. For each connected component $C$ in $G-S$, we add a new node $t_C$ and add the edge $\{t_C,t^*_C\}$ to $E(T)$, where $t^*_C$ is an arbitrary selected node (if it exists) in $\wtilde{T}_S$, such that $N_G(C) \subseteq \wtilde{\chi}_S(t^*_C)$.\footnote{If $t^*_C$ does not exist, in particular, when $S = \emptyset$, then we just add the node $t_C$.} Furthermore, we set $\chi(t_C) = N_G(C) \cup V(C)$. The above construction together with Observation~\ref{obs:compute-H-td} implies that $(T,\chi, V(G) \setminus S)$ is an $\C{H}$-tree decomposition for $G$. Note that we can construct $(T,\chi, V(G) \setminus S)$ in time bounded by $\big (f(\ell) \cdot g(n) + \ell^{\C{O}(\ell^2)} \big ) \cdot n^{\C{O}(1)} + h(\C{F})$. This concludes the proof.
\end{proof}

%% file: schema.tex
% !TEX root = main.tex
The overall schema of proof of the theorem is presented in Figure~\ref{fig:Thm1}. Notice that ones the implications depicted in the figure are obtained, we can conclude the proof of Theorem~\ref{thm:mainEquiv}. We will next discuss the results that are used to obtain the proof, and we begin with a simple observation which directly follows from the fact that $\hhtw(G) \leq \hhdepth(G) \leq \hhmd(G)$.

\begin{figure}[t]
\center
\includegraphics[scale=0.55]{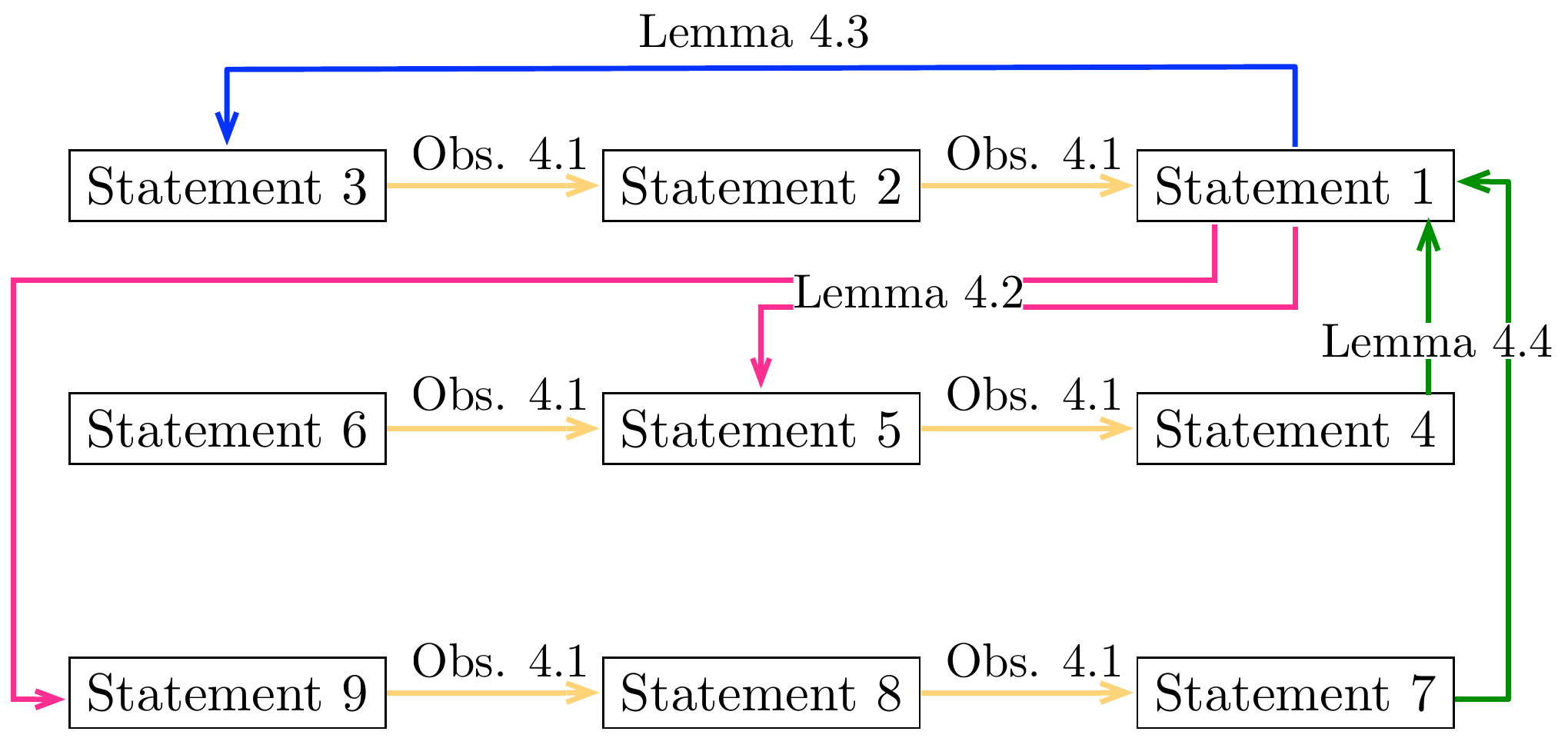}
\caption{Implications used to obtain proof of Theorem~\ref{thm:mainEquiv}.}\label{fig:Thm1}
\end{figure}

\begin{observation}\label{obs:simple-value-implications}
The following implications hold:
\begin{enumerate}
\item Statement 3 $\Rightarrow$ Statement 2 $\Rightarrow$ Statement 1.

\item Statement 6 $\Rightarrow$ Statement 5 $\Rightarrow$ Statement 4.

\item Statement 9 $\Rightarrow$ Statement 8 $\Rightarrow$ Statement 7.
\end{enumerate}
\end{observation}

We prove that Statement 1 implies Statement 5 and 9 in Section~\ref{sec:one-to-five-nine}, by proving the following lemma. 

\begin{lemma}\label{lem:one-implies-5-9}
Consider a family $\C{H}$ of graphs that is CMSO definable and is closed under disjoint union and induced subgraphs. If \probVDH\ parameterized by $\hhmd(G)$ is {\rm \FPT}, then i) \probEDH parameterized by $\hhdepth(G)$ is {\rm \FPT}, and ii) \probTDH parameterized by $\hhtw(G)$ is {\rm \FPT}. 
\end{lemma}

Intuitively speaking, we obtain the proof of the above lemma as follows. Suppose that \probVDH, parameterized by $\hhmd(G)$, admits an {\FPT} algorithm, say, $\SC{M}_{\sf mod}$.\footnote{For ease in readability, as much as possible, we will use the letters $\SC{M}, \SC{E}$ and $\SC{T}$ for algorithms for the problems \probVDH, \probEDH and \probTDH, respectively. Moreover, the subscripts ${\sf mod}$, ${\sf elm}$ and ${\sf tw}$ will denote the parameterizations $\hhmd(G)$, $\hhdepth(G)$ and $\hhtw(G)$, respectively. We note that we aren't fixing such algorithms, but whenever a need to assume/obtain such algorithms arises, we will be using the above letters/subscripts.} 
We will intuitively explain how we obtain an {\FPT} algorithm for \probEDH parameterized by  $\hhdepth(G)$, using $\SC{M}_{\sf mod}$. Consider an \hed\ $(T,\chi,L)$ of depth at most $k$ for $G$ (if it exists). From Lemma~\ref{lem:boundedsep}, either the number of vertices in $G$ is bounded by $3(\ak+k)$, in which we can resolve the instance by a brute-force procedure, or $G[L]$ has exactly one large connected component, denoted by $D^*$, and $V(G) \setminus V(D^*)$ has size bounded by $\ak+k$. Let $t$ be the parent (if it exists) of $t^*$, where $t^*$ is the leaf containing $V(D^*)$. Roughly speaking, we will try to determine the large component $D^*$ completely, and then resolve the remaining instance. To this end, we will maintain a subset, $A^* \subseteq V(G)$, which will also be the subset of vertices from $G$ that are associated with the root-to-$t$ path in $T$, and thus, we will always have $|S| \leq k$. We will look at the unique large connected component $C^*$ in $G-A^*$ (see Observation~\ref{obs:boundedsep}), and try to fix it as much as possible in the following sense. We will (roughly speaking) show that either an arbitrary solution for $(C^*,k)$ as an instance for \probVDH, obtained using the assumed algorithm $\SC{M}_{\sf mod}$, is enough for us to completely determine $D^*$, or we will be able to find a small connected set contained in $C^*$, with small neighborhood containing an obstruction to $\C{H}$. In the latter case, we will further be able to show that any such maximal connected set (with bounded neighborhood) must have a non-empty intersection with the set of vertices in $G$ associated with the root-to-$t$ path in $T$. Thus, we will either be able to ``grow'' our set $A^*$ (upto size at most $k$), or resolve the instance by brute-force. The above will give us an algorithm as required by the lemma. We note that the algorithm for the case of \probTDH parameterized by $\hhtw(G)$ can be obtained in a very similar fashion, but for this case, we will maintain that the set $S$ is the subset of vertices present in the bag of \htd, that contains all the vertices in the large component in $G[L]$. 

In Section~\ref{sec:one-implies-three} we show that Statement 1 implies Statement 3, assuming that Lemma~\ref{lem:one-implies-5-9} holds, by proving the following result.  

\begin{lemma}\label{lem:three-implies-one}
Consider a family $\C{H}$ of graphs that is CMSO definable and is closed under disjoint union and induced subgraphs. If \probVDH\ parameterized by $\hhmd(G)$ is {\rm \FPT}, then the problem is also {\rm \FPT} when parameterized by $\hhtw(G)$. 
\end{lemma}

Intuitively speaking, we obtain a proof of the above lemma as follows. Suppose that \probVDH\ parameterized by $\hhmd(G)$ admits an {\FPT} algorithm, say, $\SC{M}_{\sf mod}$. Consider an instance $(G,k)$ of the problem \probVDH. From Lemma~\ref{lem:one-implies-5-9}, we can obtain that \probTDH parameterized by $\hhtw(G)$ has an {\FPT} algorithm, say, $\SC{T}_{\sf tw}$. Using the algorithm $\SC{T}_{\sf tw}$ and Lemma~\ref{lem:self-reduce-h-td}, we compute an \htd\ for $G$. For each leaf $t$ in $T$, where the graph $G[\chi(t)\cap L]$ has large number of vertices, we replace $G[\chi(t)\cap L]$ by another graph, using Lemma~\ref{lem:red2finiteindex}, still maintaining equivalence without increasing the parameter. After this, we are obtain to bound the (standard) treewidth of the graph, and resolve the instance using Courcelle's Theorem~\cite{Courcelle90}. The above gives us an {\FPT} algorithm for \probVDH, when parameterized by $\hhtw$. 

In Section~\ref{sec:4/7-to-1} we prove that Statement 4 (resp.~Statement 7) implies Statement 1, by proving the following lemma. 

\begin{lemma}\label{lem:4/7-to-1}
Consider a family $\C{H}$ of graphs that is {\msotwo} definable and is closed under disjoint union and induced subgraphs. If 
\probEDH (resp. \probTDH) parameterized by $\hhmd(G)$ is {\rm \FPT}, then \probVDH\ parameterized by $\hhmd(G)$ is also {\rm \FPT}.
\end{lemma}

Roughly speaking, the above lemma is proved as follows. Consider an instance $(G,k)$ of \probVDH. From Proposition~\ref{prop:CMSOMetaTheorem}, it is enough for us to focus in $(\ak,k)$-unbreakable graphs, and thus we assume that $G$ is such a graph. If $G$ has at most $2\ak + k$ vertices, we resolve the instance by trying all possible subsets. Otherwise, using an assumed {\FPT} algorithm $\SC{X}_{\sf mod}$ for \probEDH (resp. \probTDH) parameterized by $\hhmd(G)$, we compute an \hed\ (resp. \htd), say, $(T,\chi,L)$. Now using Observation~\ref{obs:boundedsep} we will to able to conclude that: i) there is exactly one connected component $C^*$ in $G[L]$ which has more than $\ak$ vertices, ii) $S^* = N_G(C^*)$ has size at most $k$, and iii) $Z^* = V(G) \setminus V(C^*)$ has size at most $\ak+k$. Using the above, we are either able branch on vertices of $Z^*$, or conclude that $S^*$ is already a solution for the given \probVDH\ instance. The above gives us an algorithm for \probVDH, when parameterized by $\hhmd(G)$, using the assumed {\FPT} algorithm for \probEDH (resp. \probTDH) parameterized by $\hhmd(G)$. 

%Using the above stated results, notice that the proof of Theorem~\ref{thm:main-9-equiv} follows.    

%% file: one-implies-five.tex
% !TEX root = main.tex 
The objective of this section is to prove Lemma \ref{lem:one-implies-5-9}. We will present the result for \probEDH, and later comment how we can adapt exactly the same idea for \probTDH. Fix any family $\C{H}$ of graphs that is CMSO definable and is hereditary, such that \probVDH\ admits an {\FPT} algorithm, say, $\SC{M}_{\sf mod}$, which given an instance $(G',\ell)$, where $G'$ is an $n$ vertex graph, correctly resolves the instance in time bounded by $f_{\sf mod}(\ell)\cdot n^{\C{O}(1)}$.

Let $(G,k)$ is an instance of the problem \probEDH. From Proposition \ref{prop:CMSOMetaTheorem}, it is enough for us to design an algorithm for  $(\ak,k)$-unbreakable graphs, and thus, we assume that $G$ is  $(\ak,k)$-unbreakable. If $G$ has at most $3(ak+k)$ vertices, then we can resolve the instance in FPT time, by brute force. Thus, the interesting case is when $G$ has more than $3(ak+k)$ vertices and it is \akku. We will begin by defining an extension version of the problem for (large) unbreakable graphs, called \auedtffull\ (\auedtf, for short), which will lie at the heart of our {\FPT} algorithm for \probEDH\ (using $\SC{M}_{\sf mod}$ as a subroutine). Roughly speaking, the problem \auedtf\ will (recursively) try to compute (some of) the vertices that will be mapped to the root-to-leaf path leading the the large connected component (see Lemma~\ref{lem:boundedsep}), which will be enough for us to identify the large connected component in a decomposition. Once we have the above set, we will be able to determine the unique large connected component in the final decomposition, and then solve the remainder of the problem using brute force as the number of vertices outside the large connected component can be bounded by a function of $k$. We would like to remark that, an intuitive level the above illustrates how the deletion and the elimination problems coincide.

\smallskip
\noindent{\bf (Problem Definition)} {\auedtffull\ (\auedtf)}\\ 
\noindent{\em Input: } A graph $G$, an integer $k$, a set $A^* \subseteq V(G)$ of size at most $k$, such that $G$ is an $\scbrksize$-unbreakable graph and $|V(G)| > \grsize$.\\
\noindent{\em Question: } Test if there is an \hed\ $(T,\chi,L)$, where $D^*$ is the unique connected component in $G[L]$ of size at least $\ak$ (see Lemma~\ref{lem:boundedsep}) and $t^*\in V(T)$ is the vertex with $V(D^*)\subseteq \chi(t^*)$, such that the following holds: 
\begin{enumerate}
\item $L \cap A^* = \emptyset$.
\item For each $v \in N_G(D^*) \cup A^*$, there is (unique non-leaf vertex) $t \in V(T) \setminus \{t^*\}$, such that $\chi(t) =\{v\}$ and $t$ is an ancestor of $t^*$ in $T$.  
\end{enumerate} 

In the above, we say that $(T,\chi,L)$ is a {\em solution} to the \auedtf\ instance $(G,k,A^*)$.

The objective of the remainder of this section is to prove the following lemma. 

\begin{lemma}\label{lem:an_un}
Equipped with the algorithm $\SC{M}_{\sf mod}$, we can design an algorithm for 
\auedtf, that given an instance $(G,k,A^*)$, where $G$ is a graph on $n$ vertices, correctly decides whether or not it is a yes-instance of the problem in time bounded by $(f_{\sf mod}(k) + {(\ak+k)}^{\C{O}(\ak+k)}) \cdot n^{\C{O}(1)}$. 
\end{lemma}

We prove the above lemma by exhibiting such an algorithm for \auedtf. Let $(G,k, A^*)$ be an instance of \auedtf. As $|V(G)| \geq 3(\ak+k)$, from Observation~\ref{obs:boundedsep}, $G-A^*$ has a unique connected component of size at least $\ak$, we denote that connected component by $C^*$. Note that from the observation we also have $|V(G) \setminus V(C^*)| < \ak +k$. 

We design a branching algorithm for the problem, and we will use $k - |A^*|$ as the measure to analyse the running time of our algorithm. We start with some simple cases, when we can directly resolve the given instance. (The time required for the execution of our base cases and branching rules will be provided in the runtime analysis of the algorithm.)

 %We begin with the following structural result which will be crucially used in our algorithm. 

%Our algorithm begins with the following simple base case.
%
%\smallskip
%\noindent{\bf Base Case 1:} If $|V(G)|\leq 3(\ak+k)$, then brute force and solve the instance.
%\smallskip
%From now onwards we assume that $|V(G)|>2\ak+k$. Before we describe our remaining algorithm, we define some notations and define an extension version of the problem \probEDH, which will be useful in using a partial solution that we compute as our algorithm progresses.  
%Consider a set $A^*\subseteq V(G)$ of size at most $k$, and 

%Consider any \hed\ $(T,\chi,L)$ of depth at most $k$, of $G$ (if it exists), and let $C^*$ be the connected component in $G[L]$ that has more than $\ak$  vertices. (Note that as Base Case 1 is not applicable, ) 

%(recall that by Observation~\ref{obs:boundedsep} there is a unique connected component of size at least $\ak$ in $G[L]$). We say that $(T,\chi,L)$  is {\em consistent} with $A^*$ if for each $u, v \in A^* \cup N_G(C^*)$, where $u \in \chi(t)$ and $v \in \chi(t')$, $t$ and $t'$ are in ancestor-descendant relationship in $T$.

%We define an extension of \probEDH for $(\ak,k)$-unbreakable graphs, which we call as \aedh as follows. In \aedh together with an $(\ak,k)$-unbreakable graph $G$ and integer $k$, we are also given a set $A\subseteq V(G)$ and the question is does there exist an \hed\ of $G$ of depth at most $k$ which is consistent with $A$? We say that an \hed consistent with $A$ is a {\em solution} to $(G,A,k)$ of \aedh.

Notice that \probEDH\ is a special case of \aedh, namely, when $A^* =\emptyset$.  Thus to prove Lemma~\ref{lem:one-implies-5-9}, it is sufficient to prove the more general Lemma~\ref{lem:an_un}. 
%design an FPT algorithm for \aedh using FPT algorithm  $\SC{M}_{\sf mod}$ of \probVDH. 
Next, we give description of our branching algorithm for the \aedh problem. %We will use $k-|A|$ as the measure to analyse our algorithm. We start with some simple sanity checks.

\smallskip
\noindent{\bf Base Case 1:} If $k-|A|= 0$ and $C^* \notin \C{H}$, then return that the instance has no solution.
\smallskip

Notice that if the instance admits a solution, say, $(T,\chi,L)$, where $D^*$ is the unique large connected component in $G[L]$, then $V(D^*) \subseteq V(C^*)$ (see Observation~\ref{obs:boundedsep}). All the vertices in $N_G(D^*) \cup A^*$ must be mapped to non-leaf vertices in a single root-to-leaf path. In Base Case 1, as $C^* \notin \C{H}$, some neighbor of $D^*$ in $G$ outside $A^*$ must also be mapped to (a non-leaf in) the same root-to-leaf path. Thus the correctness of Base Case 1 follows.

Intuitively speaking, our next base case deals with the case when an arbitrary deletion set for $C^*$ of size at most $k$ can be added to $A^*$ to make the resulting $C^*$ belong to $\C{H}$. We will later prove that if our branching rules (to be described later) and Base Case 1 is not applicable, then indeed we can obtain a solution to our instance using this base case. Before formally describing our next base case, we state a simple observation and introduce some notations.

The following observation uses the well-known self-reducibility property of NP-complete problem, using which we can find a solution with the help of a decision oracle. %result that using an algorithm for a decision version of an NP-complete problem, a deletion set can be obtained by using the well known self reducibility property. 

\begin{observation}\label{obs:basecase3}
Given an instance $(G^*,\ell)$ of \probVDH, using  the {\FPT} algorithm $\SC{M}_{\sf mod}$ of \probVDH, in time $f_{\sf mod}(\ell)\cdot n^{\C{O}(1)}$ we can obtain a minimum sized set $\what{S}\subseteq V(G^*)$ of size at most $\ell$ such that $G^*-\what{S}\in \C{H}$  or return that no such set exists.
\end{observation} 

We define a boolean variable ${\sf bool}$ as follows. If $(C^*,k-|A|)$ is a no-instance of \probVDH, then we set ${\sf bool} = 0$. Otherwise, let $S^*$ be a minimum sized set computed using Observation~\ref{obs:basecase3}. If for any subset $Z \subseteq V(G) \setminus (V(C^*) \cup A^*)$, there is a solution $(T,\chi,V(G) \setminus (A^* \cup S^* \cup Z))$, for the instance $(G,k,A^*\cup S^*)$, then set ${\sf bool} = 1$, and otherwise, set ${\sf bool} = 0$. By Cayley's Theorem \cite{cayley1889theorem}, the number of distinct labelled forest on $q$ vertices is bounded by $q^{\C{O}(q)}$. Using the above together with Observation~\ref{obs:basecase3} and the fact that $|V(G) \setminus (V(C^*) \cup A^*)| < \ak + k$ (see Observation~\ref{obs:boundedsep}), allows us to construct ${\sf bool}$ in time bounded by $(f_{\sf mod}(k) + {(\ak+k)}^{\C{O}(\ak+k)}) \cdot n^{\C{O}(1)}$. We are now ready to state our next base case.

\smallskip
\noindent{\bf Base Case 2:} If ${\sf bool} = 1$, then return that $(G,k,A^*)$ is a yes-instance of the problem.  

The correctness of the above base case follows from the fact that whenever ${\sf bool}$ is set to $1$, then we do have a solution for the instance $(G,k,A^*)$. 

%Consider the case when $C$ does not contain an $(\ak,k)$-connected set $B$ in $G-A$, such that $G[B]$ contains an obstruction. If $(C,k-|A|)$ is a yes instance of \probVDH. For a minimum integer $\what{k}\leq k-|A|$, find a set $S$ such that $S$ is a minimum solution to $(C,k-|A|)$ of \probVDH. If for some $Z\subseteq V(G)\setminus (V(C)\cup A)$, every connected component in $G-(A\cup S\cup Z)$ is in $\C{H}$ and ${\sf Torso}_{A\cup S}(A\cup S\cup Z)$ has treedepth at most $k$, then return that $(G,A,k)$ is a yes instance of \aedh. 

%we construct an instance $(G,k, \what{A})$ as follows.
Roughly speaking, we will be able to say that, if we have a yes-instance, then for any solution $(T,\chi,L)$ for the instance, when our base cases are not applicable, there is a small connected subset $X^* \subseteq V(C^*)$, with bounded neighborhood containing an obstruction, that must be separated from the large connected component in $G[L]$ (see Lemma~\ref{lem:boundedsep}). Moreover, as $X^*$ is contained in the large connected component $C^*$, we will be able to guarantee that at least one vertex in the closed neighborhood of $X^*$ must be mapped to some node in the root-to-leaf path of the node to which vertices in $D^*$ are mapped. Before discussing further, we introduce some notations. 

For a graph $\what{G}$ and integers $p,q\in \mathbb{N}$, a set $B \subseteq V(\what{G})$ is a {\em $(p,q)$-connected set} in $\what{G}$, if $\what{G}[B]$ is connected, $|B| \leq p$ and $|N_{\what{G}}(B)|\leq q$. We say that a $(p,q)$-connected set $B$ in $\what{G}$ is {\em maximal} if there does not exist another $(p,q)$-connected set $B^*$ in $\what{G}$, such that $B \subset B^*$.  Below we state a result regarding computation of maximal $(p,q)$-connected sets in a graph, which directly follows from Lemma 3.1 of Fomin and Villanger~\cite{DBLP:journals/combinatorica/FominV12}.  

\begin{proposition}\label{lem:sep}
For a graph $\what{G}$ on $n$ vertices and integers $p,q \in [n]$, the number of $(p,q)$-connected sets in $\what{G}$ is bounded by $2^{p+q} \cdot n$. Moreover, the set of all maximal $(p,q)$-connected sets in $\what{G}$ can be computed in time bounded by $2^{p+q}\cdot n^{\C{O}(1)}$. 
\end{proposition}

We enumerate the set of all maximal $(\ak,k)$-connected sets in $C^*$, and let $\C{C}_{\sf conn}$ be the set containing all maximal $(\ak+k,k)$-connected sets $Q \subseteq V(C^*)$, such that $G[N_{C^*}[Q]]$ does not belong to $\C{H}$. Note that using $\SC{M}_{\sf mod}$ and Proposition~\ref{lem:sep}, we can construct $\C{C}_{\sf conn}$ in time bounded by $f_{\sf mod}(0) \cdot 2^{\C{O}(\ak + k)} \cdot n^{\C{O}(1)}$. We are now ready to state our branching rule. 

\smallskip
\noindent{\bf Branching Rule 1:} If there is some $C \in \C{C}_{\sf con}$, then for each $v \in N_G[C]$ ($=N_G[C]\setminus A^*$), solve (recursively) the instance $(G,k, A^* \cup \{v\})$. Moreover, return yes, if and only if one such instance is a yes-instance of the problem.

The next lemma lies at the crux of our algorithm, which will help us establish that one of the base cases or branching rule must be applicable. For the following lemma, we suppose that $(G,A,k)$ has a solution and Base Case 1 and 2 are not applicable, and we consider a solution $(T,\chi,L)$ for the instance, which maximizes $|V(G) \setminus L|$. As $G$ is $(\ak,k)$-unbreakable, from Lemma~\ref{lem:boundedsep} we can obtain that there is exactly one connected component, say, $D^*$, in $G[L]$, that has at least $\ak$ vertices, and we let $t^*\in V(T)$ such that $\chi(t^*)\subseteq V(D^*)$. We define let the set $S \subseteq V(G)\setminus L$ contain all the vertices mapped to the root-to-$t^*$ path in $T$, that is $S = \{v \in V(G) \mid \chi(t) = \{v\}, t\neq t^*,t \text{ is an ancestor of }t^* \mbox{ in } T\}$.

\begin{lemma}\label{lem:brorbasecase3}
We have $\C{C}_{\sf conn} \neq \emptyset$, and for each $C \in \C{C}_{\sf conn}$, we have $S \cap N_{C^*}[C] \neq \emptyset$.   
\end{lemma}
\begin{proof}
%If $(G,A,k)$ has no solution, then the claim trivially follows. We now consider the case when $(G,A,k)$ has a solution, and let $(T,\chi,L)$ be a solution for it. Recall that, as $G$ is $(\ak,k)$-unbreakable, from Lemma~\ref{lem:boundedsep}, there is exactly one connected component, say $D^*$ in $G[L]$, that has at least $\ak$ vertices. 
As $C^*$ is the (unique) connected component in $G-A^*$ that has at least $\ak$ vertices, thus, we can obtain that $V(D^*) \subseteq V(C^*)$. Let $Y$ be the set of vertices in $C^*$ that do not belong to the set $S \cup V(D^*)$, i.e., $Y = V(C^*) \setminus (V(D^*) \cup S)$. Let $\C{Y}_{\sf obs}$ be the set of connected components in $G[Y]$ that do not belong to $\C{H}$. We let $S_{\sf obs} \subseteq S \cap V(C^*)$ be the set of vertices in $S \cap V(C^*)$ that are neighbors of some connected components in $\C{Y}_{\sf obs}$, i.e., $S_{\sf obs} = \{v \in S \cap V(C^*) \mid v \in N_G(C) \mbox{ for some } C \in \C{Y}_{\sf obs}\}$. Furthermore, let $S_{\sf rem} = (S \cap V(C^*)) \setminus S_{\sf obs}$. Notice that $V(D^*), S_{\sf rem}, S_{\sf obs}, Y$ is a partition of $V(C^*)$ (possibly with empty-sets).

%We next define a partition, $S_{\sf obs}$ and $S_{\sf rem}$ of $S \cap V(C^*)$ as follows. We let $S_{\sf obs} \subseteq S \cap V(C^*)$ be the set of neighbors of connected components in $G[Y]$ containing an obstruction, i.e., $S_{\sf obs}$ contains each $v \in S \cap V(C^*)$, for which there is a connected component $C_v$ in $G[Y]$, such that i) $v \in N_G(C_v)$ and ii) $C_v \notin \C{H}$. 

%We define sets $S_{\sf rem},S_{\sf obs},Y$ with following properties: (1) $V(D)\uplus S_{\sf rem} \uplus S_{\sf obs} \uplus Y=V(C)$, (2) $Y= V(C)\setminus ( S \cup V(D))$, (3) $S_{\sf rem}\cup S_{\sf obs}= S\cap V(C)$,   $S_{\sf rem}\cap S_{\sf obs}= \emptyset$, and (4) for every vertex $v\in S\cap V(C)$ if there exists a connected component $C_v\in G[Y]$,  $v\in N(V(C_v))$, $C_v$ contains an obstruction and $N(V(C_v))\subseteq S_{\sf obs}\cup A$, then we have $v\in S_{\sf obs}$.
% Also, for each vertex $v\in S_R$, every connected component $C_v\in G[Y]$ such that $v\in N(V(C_v))$, $C_v$ does not contain an obstruction.

Firstly consider the case when $S_{\sf obs} =\emptyset$, and we will argue that Base Case 2 must be applicable, which will lead us to a contradiction. Note that for the above case, $S_{\sf rem} = S \cap V(C^*)$, where $|S| \leq k$, and each connected component in $C^* - S_{\sf rem} = C^* - S$ must belong to $\C{H}$. As $(T,\chi,L)$ is a solution for the instance, we must have $A^* \subseteq S$. From the above two statements we can obtain that $(C^*, k-|A^*|)$ is a yes-instance of \probVDH. Let $S^* \subseteq V(C^*)$ be the (same) minimum sized set computed using Observation~\ref{obs:basecase3}, while computation of ${\sf bool}$, such that each connected component in $C^*-S^*$ is in $\C{H}$. Note that $|S^*| \leq |S_{\sf rem}| = |S \cap V(C^*)|$ (recall that $S_{\sf obs} = \emptyset$). Let $Z = V(G) \setminus (L \cup A^*) \subseteq V(G) \setminus (V(C^*) \cup A^*)$. We will argue that $(G,k, A^*)$ admits a solution $(T,\chi', V(G) \setminus (A^* \cup S^* \cup Z))$, for an appropriately constructed $\chi'$, in which case ${\sf bool} =1$, contradicting that Base Case 2 is not applicable. Set $\chi(t^*) = \chi(t^*) \cup (V(C^*) \setminus S^*)$ and for each $t \in V(T)\setminus \{t^*\}$, such that $t \cap L \neq \emptyset$, we set $\chi'(t) = \chi(t)\setminus V(C^*)$. Let $S_{\sf rem} = \{u_1,u_2,\cdots, u_p\}$ and $S^* = \{w_1,w_2,\cdots, w_q\}$, where note that $q \leq p$. For each $i \in [q]$, we set $\chi'(t_i) = \{w_i\}$, where $t_i$ is the vertex in $T$, such that $\chi(t_i) = \{u_i\}$. Notice that, as $C^*-S$ is in $\C{H}$ and $\C{H}$ is closed under disjoint union, we can obtain that for each $t \in V(T)$, $G[\chi'(t)]$ is in $\C{H}$. By the construction of $\chi'$, note that $(T,\chi', V(G) \setminus (A^* \cup S^* \cup Z))$ satisfies item 1 to 3 of Definition~\ref{def:H-elim}. Thus we can conclude that $(T,\chi', V(G) \setminus (A^* \cup S^* \cup Z))$ is a solution for $(G,k,A^*)$. 

We will next consider the case when $S_{\sf obs} \neq \emptyset$. As $|V(G) \setminus V(D^*)| < \ak + k$ (see Observation~\ref{lem:boundedsep}), we can obtain that $|Y|$, and thus, the number of vertices in each connected component in $\C{Y}_{\sf obs}$ is less than $\ak +k$. Notice that for each $C \in \C{Y}_{\sf obs}$, $|N_G(C)| \leq |S|\leq k$, and thus, $C$ is an $(\ak+k,k)$-connected set in $G$. Thus, for each $C \in \C{Y}_{\sf obs}$, there must exists some $C' \in \C{C}_{\sf conn}$, such that $V(C) \subseteq V(C')$. The above together with the assumption $S_{\sf obs} \neq \emptyset$ (and thus, $\C{Y}_{\sf obs} \neq \emptyset$) we can obtain that $\C{C}_{\sf conn} \neq \emptyset$. By construction, for each $C \in \C{C}_{\sf conn}$, $V(C) \subseteq V(C^*)$ and $N_{C^*}[C]$ is not in $\C{H}$. Moreover, as $C^*$ is connected and $C$ is a maximal $(\ak+k,k)$-connected set in $C^*$, we can obtain that $N_{C^*}[C] \cap S \neq \emptyset$. This concludes the proof. 
%
%Now consider any $C \in $
%
%Note that in this case, $\C{C}_{\sf conn} \neq \emptyset$.  
%
%We will argue that $S_{\sf rem}$ is a solution to $(C,k-|A|)$ of \probVDH. As every connected component in $C-S_{\sf rem}$ is either completely contained in $G[Y]$ or is $D$, and by our definition of $S_{\sf rem}$, they belong to $\C{H}$, therefore there is no obstruction in $C-S_{\sf rem}$. This implies that Base Case 3 must be applicable, a contradiction.  
%
%Firstly consider the case when $S_{\sf obs}\neq \emptyset$, and consider a vertex $v\in S_{\sf obs}$. By the construction of $S_{\sf obs}$, there is some connected component in $C_v \in \C{Y}_{\sf obs}$, such that $v \in N_G(C_v)$. Moreover, we have $C_v \notin \C{H}$, $|V(C_v)| < \ak$, and $N_{C^*}(C_v) \subseteq S_{\sf obs} \subseteq S$, where $|S| \leq k$. Thus we can obtain that, there is some $Q \in \C{C}_{\sf conn}$, such that $V(C_v) \subseteq Q$ (and also $N_{C^*}[C_v] \subseteq N_{C^*}[Q]$). Thus we can obtain that $\C{C}_{\sf conn} \neq \emptyset$.    
%Note that the size of connected component $C_v$ in $G[Y]$ that contains an obstruction and for which $v\in N(V(C_v))$ is at most $\ak$. Also observe that $N(V(C_v))\subseteq S_{\sf obs}\cup A$, hence $|N(V(C_v))\cap V(C)|\leq k$. This implies that $V(C_v)$ is an $(\ak,k)$-connected set in $C$ in $G-A$ that contains an obstruction, a contradiction. Therefore, we must have that $S_{\sf obs}= \emptyset$.
\end{proof}

We are now ready to prove Lemma~\ref{lem:an_un}. 

\begin{proof}[Proof of Lemma \ref{lem:an_un}.]
From Proposition \ref{prop:CMSOMetaTheorem}, it is enough for us to design an algorithm for  $(\ak,k)$-unbreakable graphs, therefore we design algorithm for case when the input graph is  $(\ak,k)$-unbreakable.
Consider an instance $(G,A^*,k)$ of \auedtf, where $G$ is an $(\ak,k)$-unbreakable graph. If the instance can be resolved using Base Case 1 or 2, then the algorithm resolve it. Otherwise, from Lemma~\ref{lem:brorbasecase3} we know that the branching rule must be applicable. By our previous discussion, the we can test/apply our base cases in time bounded by $(f_{\sf mod}(k) + {(\ak+k)}^{\C{O}(\ak+k)}) \cdot n^{\C{O}(1)}$. Recall that for each set in $\C{C}_{\sf conn}$, we have $|N_{C^*}[C]|\leq \ak+ 2k$, and $\C{C}_{\sf conn}$ can be constructed in time bounded by $f_{\sf mod}(0) \cdot 2^{\C{O}(\ak + k)} \cdot n^{\C{O}(1)}$. Also, the depth of the recursion tree can be bounded by $k+1$ (see Base Case 1). Thus, the number of nodes in the recursion tree can be bounded by ${(\ak+2k)}^{\C{O}(\ak+2k)})$. Thus we can bound the running time of our algorithm for \auedtf\ by $(f_{\sf mod}(k) + f_{\sf mod}(0)) \cdot (\ak+k)^{\C{O}(\ak+k)} \cdot n^{\C{O}(1)}$. This concludes the proof.
\end{proof}

We will end this section with a remark regarding how we can obtain an {\FPT} algorithm for \probTDH parameterized by $\hhtw(G)$, using $\SC{M}_{\sf mod}$. Consider an instance $(G,k)$ of \probTDH, and an \htd, $(T,\chi,L)$ for it, if it exists. For the above case, we will define our extension version with respect to $A^*\subseteq V(G)$ of size at most $k$, where we would want the unique large connected component in $G[L]$ to have all its neighbors in $A^*$. Notice that all our argument will work exactly with the above minor modification. Thus, using exactly the same ideas as we presented for the case of \probEDH, we will be able to obtain an {\FPT} algorithm for \probTDH parameterized by $\hhtw(G)$, as required by the lemma.

%% file: one-implies-three.tex
% !TEX root = main.tex
Fix any family $\C{H}$ of graphs that is \mso\ definable and is closed under disjoint union and taking induced subgraphs, such that \probVDH\ admits an FPT algorithm, say, $\SC{M}_{\sf mod}$, when parameterized by $\hhmd(G)$, running in time $f_{\sf mod}(k) \cdot n^{\C{O}(1)}$. From Lemma~\ref{lem:one-implies-5-9} (see Section~\ref{sec:one-to-five-nine}), we can obtain that \probTDH\ admits an FPT algorithm, say, $\SC{T}_{\sf tw}$, when parameterized by $\hhtw(G)$, running in time $f_{\sf tw}(k) \cdot n^{\C{O}(1)}$. We will design an FPT algorithm, $\SC{M}_{\sf tw}$, for \probVDH, parameterized by $\hhtw(G)$, using $\SC{M}_{\sf mod}$ and $\SC{T}_{\sf tw}$ as subroutines. %Our algorithm begins with the following very simple reduction rule.

%\begin{redrule}\label{red:1-3-replace} 
%If $G$ has a connected component which belongs to $\C{H}$, then remove $D$ from $G$, i.e, solve the instance $(G-V(D),k)$ and return the same answer. 
%\end{redrule} \todo{remove this}

%The safeness of the above reduction rule immediately follows from the assumption that $\C{H}$ is closed under taking disjoint union. Hereafter, we assume that the above rule is not applicable. 

%Next we use the following lemma (Lemma~\ref{lem:self-reduce-h-td}), which gives us a self-reducibility like property for the problem \probTDH, using which we can compute an \htd, given an algorithm for the decision version of the problem. The proof of Lemma~\ref{lem:self-reduce-h-td} can be found in Section~\ref{sec:compute-decomposition}. 
%
%\begin{lemma}\label{lem:self-reduce-h-td}
%Consider an algorithm $\SC{A}$, for \probTDH, that runs in time $f(\ell')\cdot g(|V(G')|)$, for an instance $(G',\ell')$ of the problem.\footnote{We will use the standard assumption from Parameterized Complexity that the functions $f$ and $g$ are non-decreasing. For more details on this, please see Chapter 1 of the book~\cite{CyganFKLMPPS15}.} Then, for any given graph $G$ on $n$ vertices, we can compute an \htd\ for $G$ of $\C{H}$-treewidth $\ell = \hhtw$, in time bounded by $\big (f(\ell) \cdot g(n) + \ell^{\C{O}(\ell^2)} \big ) \cdot n^{\C{O}(1)} + h(\C{F})$, where $h(\C{F})$ depends only on the family $\C{H}$.
%\end{lemma}

Let $\ell = \hhtw(G)$. By instantiating Lemma~\ref{lem:self-reduce-h-td} with the algorithm $\SC{T}_{\sf tw}$, we can compute an \htd, $(T, \chi, L)$ of $\C{H}$-treewidth exactly $\ell$, for $G$ in time bounded by $\big (f_{\sf tw}(\ell) + \ell^{\C{O}(\ell^2)} \big ) n^{\C{O}(1)}+ h(\C{F})$. We let $S = V(G) \setminus L$. We will assume that the constants $\xi_i$, for each $i \in [\ell+1]$ is hardcoded in the algorithm, and due to this and Lemma~\ref{lem:red2finiteindex} our algorithm will be non-uniform. If we are able to bound the size of each connected component in $G[L]$ by $\sum_{i \in [\ell+1]} \xi_i$, where $\xi_i$ is the number from Lemma~\ref{lem:red2finiteindex}, then notice that the treewidth of $G$ can be bounded by $\ell + \sum_{i \in [\ell+1]} \xi_i$. As \probVDH\ is \mso\ expressible, for the case when $G$ has bounded treewidth, we can check whether $(G,k)$ is a yes-instance of \probVDH\ using Courcelle's Theorem~\cite{Courcelle90} in time bounded by $f(\ell)\cdot n^{\C{O}(1)}$, where $f$ is some computable function. Thus (roughly speaking) our next objective will be to bound the size of $\chi^{-1}(t)$, for each $t \in V(T)$ by replacements using Lemma~\ref{lem:red2finiteindex}. Let $\wtilde{A}$ be the set of nodes in $T$ whose bag contain at least $\sum_{i \in [\ell+1]} \xi_i$ vertices from $L$, i.e., $\wtilde{A} = \{t \in V(T) \mid |\chi(t) \cap L| \geq \sum_{i \in [\ell+1]} \xi_i\}$. Furthermore, let $\wtilde{\C{G}}$ be the set of graphs induced by vertices in $L$, in each of the bags of nodes in $\wtilde{A}$, i.e., $\wtilde{\C{G}} = \{G[\chi(t) \cap L] \mid t \in \wtilde{A}\}$. We let $\wtilde{\C{G}} = \{\wtilde{G}_1, \wtilde{G}_2, \cdots, \wtilde{G}_{q}\}$. 

We create a sequence of $G_0, G_1, \cdots, G_{q}$ graphs and a sequence of constant $c_0,c_1,\cdots, c_{q}$ as follows. Intuitively speaking, we will obtain the above sequence of graph by replacing $\wtilde{G}_i$, for $i \in [q]$, by some graph obtained using Lemma~\ref{lem:red2finiteindex}. Set $G_0 = G$ and $c_0 = 0$. We iteratively compute $G_i$, for each $i \in [q]$ (in increasing order) as follows. For the graph $\wtilde{G}_i$, let $\wtilde{t}_i$ be the unique leaf in $T$, such that $V(\wtilde{G}_i) \subseteq \chi(\wtilde{t}_i)$, and let $\wtilde{B}_i = \chi(\wtilde{t}_i) \setminus V(\wtilde{G}_i)$ and $\wtilde{b}_i = |\wtilde{B}_i|$. Note that $|\wtilde{B}_i| \leq \ell +1$, as $(T,\chi, L)$ is an $\C{H}$-tree decomposition of $G$. %Furthermore, as Reduction Rule~\ref{red:1-3-replace} is not applicable, we can obtain that $\wtilde{B}_i \neq \emptyset$, for each $i\in [q]$. 
Fix an arbitrary injective function $\lambda_{\wtilde{G}_i} : \wtilde{B} \rightarrow \{1,2,\cdots, \wtilde{b}_i\}$, and then $\wtilde{G}_i$ is the boundaried graph with boundary $\wtilde{B}_i$.\footnote{We have slightly abused the notation, and used $\wtilde{G}_i$ to denote both a graph and a boundaried graph.} Note that $V(\wtilde{G}_i) \subseteq V(G_{i-1})$. Also, let $G'_i$ be the boundaried graph $G_{i-1}-V(\wtilde{G}_i)$, with boundary $\wtilde{B}_i$. Using Lemma~\ref{lem:red2finiteindex} and the algorithm $\SC{M}_{\sf mod}$, we find the graph $\wtilde{G}^*_i$ and the translation constant $c_i$, such that $\wtilde{G}_i \equiv_{\Pi} \wtilde{G}^*_i$ and $|V(\wtilde{G}^*_i)| \leq \xi_{\wtilde{b}_i}$ in time bounded by $\cO(f_{\sf mod}(\ell) \cdot n^{\C{O}(1)} )$.\footnote{Note that for any graph $\what{G}$, $\iota(\what{G}) \leq |V(\what{G})|$ (see Definition~\ref{def:equiv-boundaried-graph}).}

We let $G_i$ be the graph $\wtilde{G}^*_i \oplus G'_i$, which can be computed in time bounded by $\cO(f_{\sf mod}(\ell) \cdot n^{\C{O}(1)} )$. Let $c^* = \sum_{i \in [q]} c_i$. With the constructions described above, we are now in a position to prove Lemma~\ref{lem:three-implies-one}.

\begin{proof}[Proof of Lemma~\ref{lem:three-implies-one}]
As $\C{H}$ is \mso\ definable to prove the lemma, it is enough to establish the following statements (together with Courcelle's Theorem~\cite{Courcelle90}).   
%\begin{lemma}\label{red:safe-1-3-replace}
%All of the following conditions are satisfied:
\begin{enumerate}
\item The instance $(G_q,k+c^*)$ can be constructed in time bounded by $\cO(f(\ell) \cdot n^{\C{O}(1)})$, for some function $f$, 
\item $(G_q, k+c^*)$ and $(G,k)$ are equivalent instances of \probVDH, and 
\item The treewidth of $G_q$ is at most $\ell + \max_{i \in [q]}{c_i} + \sum_{i \in [\ell]} \chi_i$ and the $\C{H}$-treewidth of $G_q$ is at most $\ell + \max_{i \in [q]}{c_i}$.\footnote{We remark that although $k+c^*$ can possibly be much larger than $k$, both the treewidth and the $\C{H}$-treewidth of $G_q$ are at most some additive constants (depending on $\C{H}$) away from $k$.}  
\end{enumerate} 
%\end{lemma}
%\begin{proof}
As stated perviously, we will assume that the constants $\xi_i$, for $i \in [\ell+1]$ are hardcoded in the algorithm. Thus, we can construct the set $\wtilde{\C{G}}$ in polynomial time. Also, note that for any $i \in [q]$, $\iota(\wtilde{G}^*_i \oplus G'_i) \leq \iota(G) \leq |V(G)|$ (see Definition~\ref{def:equiv-boundaried-graph}). Thus, for some function $f$, we can construct the instance $(G_q, k+c^*)$ in time bounded by $f(\ell) \cdot n^{\C{O}(1)}$. 

We will inductively argue that for each $i \in [q]_0$, $(G_i, k + \sum_{j \in [i]_0} c_j)$ and $(G,k)$ are equivalent instances of \probVDH. As $G_0 = G$ and $k+ c_ 0 =k$, the claim trivially follows for the case when $i = 0$. Next we assume that for some $q' \in [q-1]_0$, for each $i' \in [q']_0$, $(G_{i'}, k + \sum_{j \in [i']_0} c_j)$ and $(G,k)$ are equivalent instances of the problem. We will next prove the statement for $i = i'+1$. It is enough to argue that $(G_{i-1}, k + \sum_{j \in [i-1]_0} c_j)$ and $(G_{i}, k + \sum_{j \in [i]_0} c_j)$ are equivalent instances. Recall that, by construction, $G[V(\wtilde{G}_i)] = G_{i-1}[V(\wtilde{G}_i)] = \wtilde{G}_i$ and $G'_i = G_{i-1} - V(\wtilde{G}_i)$ are boundaried graphs with boundary $\wtilde{B}_i$, and $G_i = \wtilde{G}_i \oplus G'_i$. From Lemma~\ref{lem:red2finiteindex}, $\wtilde{G}^*_i \equiv_{\Pi} \wtilde{G}_i$. Thus by definition, we have that $(G_{i-1}, k + \sum_{j \in [i-1]_0} c_j)$ and $(G_{i}, k + \sum_{j \in [i]_0} c_j)$ are equivalent instances of \probVDH. 

To prove the third statement, note that it is enough to construct an $\C{H}$-tree decomposition, $(T_q, \chi_q, L_q)$, of $G_q$, where for each $t \in V(T_q)$, we have $|\chi_q^{-1}(t)| \leq \ell + \max_{i \in [q]}{c_i} + \sum_{i \in [\ell]} \chi_i$ and $|\chi_q^{-1}(t) \setminus L_q| \leq \ell + \max_{i \in [q]}{c_i}$. Let $X = \cup_{i \in [q]} V(\wtilde{C}_i)$, and $L_q = (L \setminus X ) \cup (V(G_q) \setminus V(G))$ and $T_q = T$. For each $t \in V(T) \setminus \wtilde{A}$, we set $\chi_q(t) =\chi(t)$, and for each $i \in [q]$, we set $\chi_q(\wtilde{t}_i) = (\chi(\wtilde{t}_i) \setminus V(\wtilde{G}_i)) \cup V(\wtilde{G}^*_i)$. For each $i \in [q]$, note that $|V(\wtilde{G}^*_i)| \leq \xi_{\wtilde{b}_i} \leq \sum_{j \in [\ell]} \xi_{j}$. Thus we can obtain that $(T_q, \chi_q, L_q)$ is an $\C{H}$-tree decomposition of $G_q$ that satisfies all the required properties. This concludes the proof. 
\end{proof}

%% file: 4,7-implies-1.tex
% !TEX root = main.tex
Fix any family $\C{H}$ of graphs that is CMSO definable and is closed under disjoint union and taking induced subgraphs, such that \probEDH (resp. \probTDH) admits an {\FPT} algorithm, say, $\SC{X}_{\sf mod}$ running in time $f(\ell) \cdot n^{\C{O}(1)}$, where $n$ is the number of vertices in the given graph $G$ and $\ell = \hhmd(G)$.% is any of the numbers $\hhmd(G), \hhdepth(G)$, or $\hhtw(G)$.\footnote{We wish to design an FPT algorithm for the problem \probVDH, parameterized by $\hhmd$, and we will be using the algorithm with the parameter $\hhmd$. Thus, even we are fine with any FPT algorithm for the \probEDH (resp. \probTDH) problem where the parameter is at most $\hhmd$.} 

Let $(G,k)$ be an instance of the problem \probVDH. From Proposition~\ref{prop:CMSOMetaTheorem} it is enough for us to design an algorithm for $(\alpha(k),k)$-unbreakable graphs, and thus, we assume that $G$ is $(\alpha(k),k)$-unbreakable. We begin with the following simple sanity checks.

\vspace{1.5mm}
\noindent{\bf Base Case 1.} If $G \in \C{H}$ and $k \geq 0$, then return that $(G,k)$ is a yes-instance of the problem. Moreover, if $k < 0$, then return that the instance is a no-instance.

\vspace{1.5mm}
\noindent{\bf Base Case 2.} If $|V(G)| \leq 2\ak +k$, then for each $S \subseteq V(G)$, check if $G-S \in \C{H}$, by calling $\SC{X}_{\sf mod}$ for the instance $(G-S,0)$. If for any such $S$ we obtain that $G-S \in \C{H}$, return that $(G,k)$ is a yes-instance, and otherwise return that it is a no instance.

\vspace{1.5mm}
\noindent{\bf Base Case 3.} If $G$ does not admit an \hed\ (resp. \htd) of depth (resp.~width) at most $k$, then return that $(G,k)$ is a no-instance of the problem. 

The correctness of the Base Case 1 and 2 is immediate from their descriptions. Note that $\hhtw(G) \leq \hhdepth(G) \leq \hhmd(G)$. Thus, if $(G,k)$ is a yes-instance of \probVDH, then it must admit an \hed\ (resp. \htd) of depth (resp. width) at most $k$. The above implies the correctness of Base Case 3. Note that using $\SC{X}_{\sf mod}$, we can test/apply all the base cases in time bounded by $\max\{2^{\ak +k}, f(k)\} \cdot n^{\C{O}(1)}$.

Hereafter we assume that the base cases are not applicable. We compute an \hed\ (resp. \htd), say, $(T,\chi,L)$, by instantiating Lemma~\ref{lem:self-reduce-h-ed} (resp. Lemma~\ref{lem:self-reduce-h-td}) with the algorithm $\SC{X}_{\sf mod}$, of depth (resp. width) at most $k$. Note that the above decomposition can be computed as Base Case 3 is not applicable. 

Let $C^*$ be a connected component in $G[L]$ with maximum number of vertices, and let $S^* = N_G(C^*)$ and $Z^* = V(G) \setminus N_G[C^*]$. Note that as $(T,\chi,L)$ is an \hed\ (resp. \htd) of depth (resp. width) at most $k$, we can obtain that $|S^*| \leq k$. As $G$ is $(\ak,k)$-unbreakable, the above together with Observation~\ref{obs:boundedsep} implies that $|Z^* \cup S^*| \leq \ak + k$. We have the following observation which immediately follows from the fact that $(T,\chi,L)$ is an \hed\ (resp. \htd) for $G$ of depth (resp. width) at most $k$.  

\begin{observation}\label{obs:branch-Z-star}
For any $S \subseteq V(G)$, where $|S| \leq k$, either $Z^* \cap S \neq \emptyset$, or $G-S^* \in \C{H}$.  
\end{observation}  

The above observation leads us to the following base case and our branching rule. 

\vspace{1.5mm}
\noindent{\bf Base Case 4.} $G-S^* \in \C{H}$, then return that $(G,k)$ is a yes-instance. 

As $|S^*| \leq k$, the correctness of the above base case immediately follows. Moreover, we can apply Base Case 3 in time bounded by $f(k)\cdot n^{\C{O}(1)}$ using $\SC{X}_{\sf mod}$.  

\vspace{1.5mm}
\noindent{\bf Branching Rule.} For each $z \in Z^*$, (recursively) solve the instance $(G-\{z\},k-1)$. Return that $(G,k)$ is a yes-instance if and only if for some $z \in Z$, $(G-\{z\},k-1)$ is a yes-instance.

The correctness of the branching rule follows from Observation~\ref{obs:branch-Z-star} and non-applicability of Base Case 4. Moreover, we can create instances in the branching rule in polynomial time, given the decomposition $(T,\chi,L)$. 

Note that the depth of the recursion tree is bounded by $k+1$. Also, each fo the steps can be applied in time bounded by $\max\{2^{\ak + k}, f(k)\}\cdot n^{\C{O}(1)}$. Thus we can bound the running time of our algorithm by $k^{\C{O}(k)} \cdot \max\{2^{\ak + k}, f(k)\} \cdot n^{\C{O}(1)}$. The correctness of the algorithm is immediate from the description and Observation~\ref{obs:branch-Z-star}. The above implies proof of Lemma~\ref{lem:4/7-to-1}. 

%% file: applications.tex
% !TEX root = main.tex
In this section we see applications of our main theorem (Theorem~\ref{thm:mainEquiv}) for problems that are not precisely captured by the families of graphs mentioned in the premise of Theorem~\ref{thm:mainEquiv}. This allows us to obtain the first  analogous results for problems such as {\sc Multiway Cut}, {\sc Subset Feedback Vertex Set} ({\sc Subset FVS}, for short) and {\sc Subset Odd Cycle Transversal} ({\sc Subset OCT}, for short). That is, we obtain {\FPT}  algorithms for these problems that is parameterized by a parameter whose value is upper bounded by the standard parameter (i.e., solution size) and which can be arbitrarily smaller. For instance, consider the {\sc Multiway Cut} problem, where one is given a graph $G$ and a set of vertices $S$ (called terminals) and an integer $\ell$ and the goal is to decide whether there is a set of at most $\ell$ vertices whose deletion separates every pair of these terminals. The standard parameterization for this problem is the solution size. Jansen et al.~\cite{JansenK021} propose to consider undirected graphs with a distinguished set of terminal vertices and study the parameterized complexity of {\sc Multiway Cut} parameterized by the elimination distance to a graph where each component has at most one terminal. Notice that this new parameter is always upper bounded by the size of a minimum solution. 
Thus, an {\FPT} algorithm for {\sc Multiway Cut} with this new parameter would naturally extend the boundaries of tractability for the problem and we obtain such an algorithm by using Theorem~\ref{thm:mainEquiv} in an appropriate way. We then proceed to obtain similar {\FPT} algorithms for the other cut problems mentioned in this paragraph.

That is, we obtain an {\FPT} algorithm for  {\sc Subset FVS} parameterized by the elimination distance to a graph where no terminal is part of a cycle, and an {\FPT} algorithm for {\sc Subset OCT} parameterized by the elimination distance to a graph where no terminal is part of an odd cycle.

 To this end, we begin by formally defining  structures in the context of this paper and an extension of \hed\ and \htd\ to families of structures. 

\begin{definition}{\rm [{\bf Structure}]}
{\rm A {\em structure} $\alpha$ is a tuple whose first element is a graph, say, $G$, and each of the remaining elements is a subset of $V(G)$, a subset of $E(G)$, a vertex in $V(G)$ or an edge in $E(G)$. The number of elements in the tuple is the {\em arity} of the structure. 
}\end{definition}

We will be concerned with structures of arity $2$ where we have the second element as a vertex subset, and we denote the family of such structures by $\C{S}$. For a family $\C{S}' \subseteq \C{S}$, we will next define the notion of \sedst s and \stdst s that 
is tailored to our purpose.

\begin{definition}\label{def:H-elim-struct}{\rm 
For a family of structures, $\C{S}' \subseteq \C{S}$, an {\em \hedp{$\C{S}'$}} of a structure $(G,S)$, where $S \subseteq V(G)$, is a triplet $(T, \chi, L)$, where $L \subseteq V(G)$,~$T$ is a rooted forest, and~$\chi \colon V(T) \to 2^{V(G)}$, such that:
\begin{enumerate}
\setlength{\itemsep}{-1pt}
    \item For each internal node~$t$ of~$T$ we have~$|\chi(t)| \leq 1$ and $\chi(t) \subseteq V(G) \setminus L$. 
    \item The sets~$(\chi(t))_{t \in V(T)}$ form a partition of~$V(G)$.
    \item For each edge~$\{u,v\} \in E(G)$, if~$u \in \chi(t_1)$ and~$v \in \chi(t_2)$, then~$t_1, t_2$ are in ancestor-descendant relation in~$T$.
    \item For each leaf $t$ of~$T$, we have $\chi(t) \subseteq L$ and the structure~$(G[\chi(t)], S \cap \chi(t))$ belongs to~$\C{S}'$.  
\end{enumerate}
The \emph{depth} of~$(T,\chi,L)$ is same as the depth of $T$. The {\em $\C{S}'$-elimination distance} of the structure $(G,S)$, denoted $\hhdepthp{\C{S}'}(G,S)$, is the minimum depth of an \hedp{$\C{S}'$} of $(G,S)$. 
}\end{definition}

\begin{definition}
\label{def:tree:h:decomp-struct}{\rm
For a family of structures, $\C{S}' \subseteq \C{S}$, an {\em \hedp{$\C{S}'$}} of a structure $(G,S)$, where $S \subseteq V(G)$, is a triplet $(T, \chi, L)$, where $L \subseteq V(G)$,~$T$ is a rooted tree, and~$\chi \colon V(T) \to 2^{V(G)}$, such that:
\begin{enumerate}\setlength{\itemsep}{-1pt}
    \item For each~$v \in V(G)$ the nodes~$\{t \mid v \in \chi(t)\}$ form a {non-empty} connected subtree of~$T$. %\label{item:tree:h:decomp:connected}
    \item For each edge~$\{u,v\} \in E(G)$ there is a node~$t \in V(G)$ with~$\{u,v\} \subseteq \chi(t)$.
    \item For each vertex~$v \in L$, there is a unique~$t \in V(T)$ for which~$v \in \chi(t)$. Moreover, $t$ must a leaf of~$T$. %\label{item:tree:h:decomp:unique}
    \item For each node~$t \in V(T)$, the structure~$(G[\chi(t) \cap L], \chi(t) \cap L \cap S)$ belongs to~$\C{S}'$. %\label{item:tree:h:decomp:base}
\end{enumerate}
{The \emph{width} of an \htdp{$\C{S}'$} is $\max(0, \max_{t \in V(T)} |\chi(t) \setminus L| - 1)$.} The $\C{S}'$-treewidth of a structure $(G,S)$, denoted~$\hhtwp{\C{S}'}(G,S)$, is the minimum width of an \htdp{$\C{S}'$} of $(G,S)$.
}
\end{definition}

We remark that though the above two definitions are extendable to more general notions of  structures, we choose to give it for this restricted version as it is enough for our purpose and it is more insightful for the reader. We now will be able to capture problems such as {\sc Multiway Cut}, {\sc Subset FVS} and {\sc Subset OCT}. To this end, we begin by defining the following families of structures. 

$\C{S}_{\sf mway}=\{(G,S)~|~\mbox{every connected component of $G$ has at most one vertex from $S\subseteq V(G)$} \}$

$\C{S}_{\sf fvs}=\{(G,S)~|~\mbox{$G$ has no cycle containing a vertex from $S\subseteq V(G)$}\}$

$\C{S}_{\sf oct}=\{(G,S)~|~\mbox{$G$ has no odd length cycle containing a vertex from $S\subseteq V(G)$} \}$ \vspace{1mm}

We are now ready to state the main result of this section.

\begin{theorem}\label{thm:app}
Each of the following parameterized problems admits an {\rm \FPT} algorithm:
\begin{enumerate}
\item {\sc Multiway Cut} parameterized by $\sstwst{mway}(G,S)$ (and thus, $\ssdepthst{mway}(G,S)$). 
\item {\sc Subset FVS} parameterized by $\sstwst{fvs}(G,S)$ (and thus, $\ssdepthst{fvs}(G,S)$), and 
\item {\sc Subset OCT} parameterized by $\sstwst{oct}(G,S)$ (and thus, $\ssdepthst{oct}(G,S)$). 
\end{enumerate}
\end{theorem}

We will begin by explaining the intuitive idea behind the proof of the above theorem. For simplicity, let us fix the problem {\sc Multiway Cut}. 
Our first goal is to construct a \msotwo definable family of graph $\C{H}_{\sf mway}$ that is closed under disjoint union. Then, for a given instance $(G,S,k)$ of {\sc Multiway Cut}, by appropriate ``gadgeteering", we will construct an (equivalent) instance $(G',k)$ of \probVDHp{$\C{H}_{\sf mway}$}, ensuring that $\hhtwp{\C{H}_{\sf mway}}(G')$ is at most a polynomial factor away from $\sstwst{\C{H}_{\sf mway}}(G)$. Using a known {\FPT} algorithm for {\sc Multiway Cut} parameterized by the solution, we will be able to obtain an {\FPT} algorithm for \probVDHp{$\C{H}_{\sf mway}$}, parameterized by $\hhmdp{\C{H_{\sf mway}}}(G')$. The above statement together with Theorem~\ref{thm:mainEquiv} will imply that \probVDHp{$\C{H}_{\sf mway}$} admits an {\FPT} algorithm, when parameterized by $\hhtwp{\C{H}_{\sf mway}}(G')$. The above, together with the equivalence of the instance $(G,S,k)$ of {\sc Multiway Cut} and the instance $(G',k)$ of \probVDHp{$\C{H}_{\sf mway}$}, and the property that $\hhtwp{\C{H}_{\sf mway}}(G')$ is at most a polynomial factor away from $\sstwst{\C{H}_{\sf mway}}(G)$, will imply an {\FPT} algorithm for \probVDHp{$\C{H}_{\sf mway}$} parameterized by $\sstwst{\C{H}_{\sf mway}}(G)$. 

Before moving to the formal description, we briefly discuss the construction of $\C{H}_{\sf mway}$. For an instance $(G,S,k)$ of {\sc Multiway Cut}, we will subdivide the edges of $G$ and then attach a $K_3$ at each vertex in $S$ to obtain the graph $G'$ in the \probVDHp{$\C{H}_{\sf mway}$} instance $(G',k)$. Roughly speaking, the family $\C{H}_{\sf mway}$ will trivially contain all ``ill-formed'' graphs, i.e., the graph that cannot be obtained via a reduction that we discussed above. Apart from the above graphs, the family $\C{H}_{\sf mway}$ will contain all ``solved graphs" (very roughly speaking). The above will be categorized based on cut vertices whose removal results in connected components that are $K_3$s. We will now move towards the formal discussion of our proof of Theorem~\ref{thm:app}.

%We will construct an obstruction set $\mathbb{O}_{\sf mway}$, and the family of graphs $\C{H}_{\sf mway}$ will contain all graphs that do not contain a graph from $\mathbb{O}_{\sf mway}$ as a {\bf subgraph}. Roughly speaking, $\mathbb{O}_{\sf mway}$ will contain each graph that is a path with at least two vertex, where the endpoints of the path have an attached clique on three vertices, i.e., a $K_3$. For the construction of $G'$ from $G$, we will first sub-divide edges in $G$ so as to remove already existing $K_3$s. After this, for each of the terminal vertices we will attache $k+1$ many $K_3$s. We will now move to a formal description towards the proof of Theorem~\ref{thm:app}. 

\paragraph{Some useful graphs and graph classes.} 
For a graph $G$, the {\em subdivision} of $G$, denoted by $G_{\sf sd}$ is the graph obtained by sub-dividing the edges of $G$ exactly once, i.e., we have $V(G_{\sf sd}) = V(G) \cup \{w_e \mid e \in E(G)\}$ and $E(G_{\sf sd}) = \{\{u, w_e\},\{w_e,v\} \mid e = \{u,v\}\in E(G)\}$. A {\em double subdivision} of a graph $G$ is the graph $G_{\sf dsd}$, obtained by sub-dividing each of the edges in $G$ with two vertices, i.e., we have $V(G_{\sf dsd}) = V(G) \cup \{w_e, w'_e \mid e \in E(G)\}$ and $E(G_{\sf dsd}) = \{\{u, w_e\},\{w_e,w'_e\}, \{w'_e,v\} \mid e = \{u,v\}\in E(G)\}$. 

For a graph $G$ and a set $S \subseteq V(G)$, by $G^{S, \Delta}$, we denote the graph obtained from $G$ by attaching a $K_3$ on vertices in $S$, i.e., $V(G^{S, \Delta}) = V(G) \cup \{a_s, a'_s \mid s \in S\}$ and $E(G^{S, \Delta}) = E(G) \cup \{\{a_s, a'_s\}, \{a_s,s\}, \{a'_s,s\} \mid s \in S\}$. 

%For $\ell \in \mathbb{N}$, where $\ell \geq 2$, a {\em $K_3$-path} of {\em order} $\ell$ is a graph $P$ with exactly $\ell+4$ vertices, where $V(P) = \{x_1,x_2,y_1,y_2\} \cup \{a_1,a_2, \cdots a_\ell\}$ and $E(P) = \{\{a_1,x_1\}, \{a_1,x_2\}, \{x_1,x_2\},\{a_\ell,y_1\}, \{a_\ell,y_2\}, \{y_1,y_2\}\} \cup \{\{a_i,a_{i+1}\} \mid i \in [\ell-1]\}$. For $\ell \in \mathbb{N}$, where $\ell \geq 3$, a {\em $K_3$-ring} of {\em order} $\ell$ is a graph $R$ with exactly $\ell + 2$ vertices, where $V(R) = \{x_1,x_2\} \cup \{a_1,a_2, \cdots a_\ell\}$ and $E(R) = \{\{a_1,x_1\}, \{a_1,x_2\}, \{x_1,x_2\}\} \cup \{\{a_i,a_{i+1}\} \mid i \in [\ell-1]\} \cup \{\{a_1,a_\ell\}\}$. A $K_3$-ring is a {\em $K_3$-odd ring} if its order is an odd number. Let $\C{G}$ be the family of all graphs. We now define the following (hereditary) graph classes.

%\paragraph{Relevant Cut Vertices.} 
Consider a graph $G$. A {\em cut} vertex in $G$ is a vertex $v \in V(G)$, such that the number of connected components in $G-\{v\}$ is strictly more than the number of connected components in $G$. We say that a cut vertex $v \in V(G)$ is {\em relevant} if $G-\{v\}$ has a connected component $C$ with {\bf exactly} two vertices and for each $u \in V(C)$, we have $\{u,v\} \in V(G)$. We now define the following hereditary families of graphs. 

$\C{H}_{\sf mway}=\{G \in \C{G} \mid \mbox{ each connected component of } G \mbox{ has at most one relevant cut vertex}\}$ 

$\C{H}_{\sf fvs}=\{G \in \C{G} \mid G \mbox{ has no cycle of length at least } 4 \mbox{ containing a relevant cut vertex}\}$

$\C{H}_{\sf oct}=\{G \in \C{G} \mid G \mbox{ has no odd cycle of length at least $5$ containing a relevant cut vertex}\}$ \vspace{1mm} 

We have the following observation regarding the above defined families of graphs.

\begin{observation}\label{obs:multicut-family-cmso}
Each $\C{H} \in \{\C{H}_{\sf mway},\C{H}_{\sf fvs},\C{H}_{\sf oct}\}$ is \msotwo definable and closed under disjoint union. 
\end{observation} 
\begin{proof}
From the description of the families, it clearly follows that they are closed under disjoint union. We will next show that they are also \msotwo\ definable. We note that checking whether a vertex subset $X \subseteq V(G)$ induces a connected subgraph of $G$, is \msotwo\ definable, and we let this predicate be ${\sf conn}(X)$ (see, for example, Section 7.4.1 in the book of Cygan et al.~\cite{CyganFKLMPPS15}). Next we give a predicate which can check if a vertex is a relevant cut vertex in the given graph.  

${\sf rel\mbox{-}cut}(v) = \exists u,u' \in V(G) \setminus \{v\} \big[u \neq u' \wedge {\sf adj}(u,u') = 1 \wedge {\sf adj}(u,v) = 1 \wedge {\sf adj}(u',v) = 1 \wedge (\forall w \in V(G) \setminus \{u,u',v\}~\forall x \in \{u,u'\}~{\sf adj}(w,u) = 0) \big]$

Now we define a predicate which can test if a connected component of a graph contains more that contains two relevant cut vertices.  

${\sf in}\mbox{-}{\C{H}_{\sf mway}} = \neg \exists u,v \in V(G) \big[ u \neq v~\wedge~{\sf rel\mbox{-}cut}(u) =1~\wedge~{\sf rel\mbox{-}cut}(v) =1~\wedge~\big(\exists X \subseteq V(G) ({\sf conn}(X) = 1 \wedge u,v \in X) \big]$

Notice that using ${\sf in}\mbox{-}{\C{H}_{\sf mway}}$ (and ${\sf rel\mbox{-}cut}(v)$) we can obtain that $\C{H}_{\sf mway}$ is \msotwo\ definable.

It is well known that checking if a graph contains a cycle of length at least $5$ containing a particular vertex  is \msotwo\ definable. We denote such a predicate by ${\sf cycle}_{\geq 5}(v)$. The above together with the definition of ${\sf rel\mbox{-}cut}(v)$ implies that $\C{H}_{\sf fvs}$ is \msotwo\ definable. Since we are allowed to have condition on cardinality of a set size modulo a number $q$, using ${\sf cycle}_{\geq 5}(v)$ and ${\sf rel\mbox{-}cut}(v)$ we can obtain that $\C{H}_{\sf oct}$ is \msotwo\ definable. 
\end{proof}

\paragraph{Reduction.} In the following lemma we show how we can obtain our reduction.

\begin{lemma}\label{lem:multi-cut-to-H}Each of the following holds:
\begin{enumerate}\setlength{\itemsep}{-1pt}
\item An instance $(G,S,k)$ of {\sc Multiway Cut} is a yes-instance if and only if $(G^{S,\Delta}_{\sf sd},k)$ is a yes-instance of \probVDHp{$\C{H}_{\sf mway}$}. %Moreover, $\hhdepthp{\C{H}_{\sf mway}}(G^{S,\Delta}_{\sf sd}) \leq \hhdepthp{\C{S}_{\sf mway}}(G,S)$ and 
Moreover, $\hhtwp{\C{H}_{\sf mway}}(G^{S,\Delta}_{\sf sd}) < 3 \cdot (\ell +1) + {{\ell +1} \choose{2}}$, where $\ell = \hhtwp{\C{S}_{\sf mway}}(G,S)$.   

\item An instance $(G,S,k)$ of {\sc Subset FVS} is a yes-instance if and only if $(G^{S,\Delta}_{\sf sd},k)$ is a yes-instance of \probVDHp{$\C{H}_{\sf fvs}$}. %Moreover, $\hhdepthp{\C{H}_{\sf fvs}}(G^{S,\Delta}_{\sf sd}) \leq \hhdepthp{\C{S}_{\sf fvs}}(G,S)$ and 
Moreover, $\hhtwp{\C{H}_{\sf fvs}}(G^{S,\Delta}_{\sf sd}) < 3 \cdot (\ell +1) + {{\ell +1} \choose{2}}$, where $\ell = \hhtwp{\C{S}_{\sf fvs}}(G,S)$.

\item An instance $(G,S,k)$ of {\sc Subset OCT} is a yes-instance if and only if $(G^{S,\Delta}_{\sf dsd},k)$ is a yes-instance of \probVDHp{$\C{H}_{\sf oct}$}. %Moreover, $\hhdepthp{\C{H}_{\sf oct}}(G^{S,\Delta}_{\sf dsd}) \leq \hhdepthp{\C{S}_{\sf oct}}(G,S)$ and 
Moreover, $\hhtwp{\C{H}_{\sf oct}}(G^{S,\Delta}_{\sf dsd}) < 3 \cdot (\ell +1) + 2 \cdot {{\ell +1} \choose{2}}$, where $\ell = \hhtwp{\C{S}_{\sf oct}}(G,S)$.
\end{enumerate} 
\end{lemma}
\begin{proof}
We will prove the third statement as it is the most involved. The proof of other two statements can be obtained by following similar arguments. Let $G' = G^{S,\Delta}_{\sf dsd}$. Note that a cycle $C$ in $G$ has even (resp. odd) number of vertices, if and only if the cycle $C'$ with vertex set $V(C') = \{u,w_e,w'_e,v \mid e = \{u,v\} \in E(C)\}$ and edge set $E(C') = \{\{u,w_e\}, \{w_e,w'_e\}, \{w'_e,v\} \mid e = \{u,v\} \in E(C)\}$ in $G'$, has even (resp. odd) number of vertices. Moreover, no (simple) cycle in $G'$ of length at least $4$ can contain a vertex from $\{a_s,a'_s \mid s \in S\}$. Also, for any cycle $C'$ in $G'$ of odd (resp. even) length at least $4$, we can obtain a cycle $C$ in $G$ of odd (resp. even) length at least $3$, by contracting the edges incident to vertices in $W = \{w_e,w'_e \mid e \in E(G)\}$. Notice that any minimal set $S \subseteq V(G')$, such that $G'-S$ is in $\C{H}_{\sf odd}$ does not contain a vertex from $S' = \{a_s,a'_s \mid s \in S\}$, as there is no odd cycle in $G'$ of length at least $5$ that contains a vertex from $S'$. Also, each vertex in $W$ has degree exactly $2$ in $G$, and thus any cycle in $G'$ containing them, must also contain their neighbors. From the above we can obtain that if $(G',k)$ is a yes-instance of \probVDHp{$\C{H}_{\sf oct}$}, then there is $B \subseteq V(G) \subseteq V(G')$ of size at most $k$, such that $G'-B \in \C{H}_{\sf oct}$. From the above discussions we can obtain that $(G,S,k)$ is a yes-instance of {\sc Subset FVS} if and only if $(G',k)$ is a yes-instance of \probVDHp{$\C{H}_{\sf oct}$}, and in particular, we can obtain that for $B \subseteq V(G)$, $G-B$ has no odd cycle containing a vertex from $S$ if and only if $G'-B$ has no odd length cycle with at least $5$ vertices, containing a vertex from $S$. This concludes the first part of the proof.  

%for the \probVDHp{$\C{H}_{\sf oct}$ instance $(G^{S,\Delta}_{\sf dsd},k)$, does not contain any vertex %From the above we can obtain that 

We will now argue that $\hhtwp{\C{H}_{\sf oct}}(G') \leq \hhtwp{\C{S}_{\sf oct}}(G,S)$. To this end, consider an \htdp{$\C{S}_{\sf oct}$}, $(T,\chi,L)$ of width $\ell$, for $(G,S)$. %Let $X = V(G) \setminus L$. 
Let $T'$ be the tree obtained from $T$ on $2|V(T)|$ vertices as follows. Initialize $T' = T$. For each $t \in V(T)$, we add a vertex $t_{\sf cpy}$ to $V(T')$ and the edge $\{t,t_{\sf cpy}\}$ to $E(T')$. We define the sets $X'$ and $L'$ as follows. Initialize $X' = X$. For each edge $e=\{u,v\} \in E(G)$, such that $u,v \in X$, we add $w_e,w'_e$ to $X'$. Furthermore, for each $s \in S \cap X$, we add $a_s,a'_s$ to $X'$. This completes the construction of $X'$. We set $L' = V(G') \setminus X'$. We next define a function $\chi': V(T') \rightarrow 2^{V(G')}$, where we initialize $\chi'(t') =\emptyset$, for each $t' \in V(T')$. Roughly speaking, for each non-leaf vertex $t$ in $T'$, $\chi'(t)$ will contain the vertices corresponding to the subdivision of the edges that are contained in $\chi(t) \setminus L$. Furthermore, $\chi'(t)$ will also contain the $K_3$, if any, that are attached to the vertices in $\chi(t) \setminus L$. For $t \in V(T) \subseteq V(T')$, set $\chi'(t) = (\chi(t) \setminus L) \cup \{w_e,w'_e \mid e=\{u,v\} \in E(G) \mbox{ and } u,v \in \chi(t) \setminus L\}\cup \{a_s, a'_s \mid s \in (S \cap \chi(t)) \setminus L\}$. Note that for each $t \in V(T)$, we have $\chi'(t) \subseteq X'$ and $|\chi'(t)| \leq 3 \cdot |\chi(t) \cap L| + 2 \cdot {{|\chi(t) \cap L|}\choose{2}} \leq 3 (\ell+1) + 2 \cdot {{\ell + 1}\choose {2}}$. Next (roughly speaking) for leaf-nodes $t_{\sf cpy}$, we will assign the remaining vertices in accordance with $\chi(t)$ as follows. For each $t \in V(T)$, where $\chi(t) \cap L \neq \emptyset$, we let $Z^t_{\sf cpy} = \{w_e,w'_e \mid e=\{u,v\} \in E(G), u \in \chi(t) \cap L \mbox{ and } v \in \chi(t)\} \cup \{a_s,a'_s \mid s \in S \cap L \cap \chi(t)\}$. Furthermore, we set $\chi(t_{\sf cpy}) = \chi(t) \cup Z^t_{\sf cpy}$ (note that by the construction of $L'$, we have $L \cup Z^t_{\sf cpy} \subseteq L'$.) 
  
Recall that for any $B \subseteq V(G)$, $G-B$ has no odd cycle containing a vertex from $S$ if and only if $G'-B$ has no odd length cycle with at least $5$ vertices, containing a vertex from $S$. The above together with the construction of $G'$ and $(T',\chi',L')$, and the assumption that $(T,\chi,L)$ is an \htdp{$\C{S}_{\sf oct}$} of $(G,S)$, implies that $(T',\chi',L')$ is an \htdp{$\C{H}_{\sf oct}$} of $G'$ of width less than $3 \cdot (\ell + 1) + 2 \cdot {{\ell + 1}\choose {2}}$. This concludes the proof. 
\end{proof}

We will now show that for each $\C{H} \in \{\C{H}_{\sf mway},\C{H}_{\sf fvs},\C{H}_{\sf oct}\}$, \probVDH\ admits an {\FPT} algorithm, using the known {\FPT} algorithms for {\sc Multiway Cut}, {\sc Subset FVS} and {\sc Subset OCT} {\em parameterized by the solution size}.   

\begin{lemma}\label{lem:multi-cut-fpt-algo}
For each $\C{H} \in \{\C{H}_{\sf mway},\C{H}_{\sf fvs},\C{H}_{\sf oct}\}$, \probVDH\ admits an {\rm \FPT} algorithm, when parameterized by $\hhmd(G)$. 
\end{lemma}
\begin{proof}
%We give the proof for the case when $\C{H} =\C{H}_{\sf oct}$, as the proof for the other cases can ob obtained in a very similar fashion. 
Consider $\C{H} \in \{\C{H}_{\sf mway},\C{H}_{\sf fvs},\C{H}_{\sf oct}\}$, and an instance $(G,k)$ of \probVDH, where $k \geq 0$ (as otherwise, the problem is trivial). Let $Q$ be the set of relevant cut vertices in $G$. Note that $Q$ can be computed in polynomial time. If $Q = \emptyset$, then clearly, $G \in \C{H}$, and thus we can return that $(G,k)$ is a yes-instance of the problem. We will next consider the case when $Q \neq \emptyset$. For each $v \in Q$, let $\C{D}_v$ be the set of connected components $D$ in $G-\{v\}$ which is an edge whose both endpoints are adjacent to $v$, i.e., $D$ that has exactly two vertices and for each $u \in V(D)$, we have $\{u,v\} \in E(G)$. Note that for any $v \in Q$, we have $\C{D}_v \neq \emptyset$, and for each $D \in \C{D}_v$, $V(D) \cap Q = \emptyset$. Let $Z = \cup_{v \in Q, D \in \C{D}_v} V(D)$, and $G' = G - Z$. Notice that all of the following hold:
\begin{enumerate}
\setlength{\itemsep}{-1pt}
\item $(G,k)$ is yes-instance of \probVDHp{$\C{H}_{\sf mway}$} if and only if $(G',Q,k)$ is a yes-instance of {\sc Multiway Cut}. 
\item $(G,k)$ is yes-instance of \probVDHp{$\C{H}_{\sf fvs}$} if and only if $(G',Q,k)$ is a yes-instance of {\sc Subset FVS}.
\item $(G,k)$ is yes-instance of \probVDHp{$\C{H}_{\sf oct}$} if and only if $(G',Q,k)$ is a yes-instance of {\sc Subset OCT}.    
\end{enumerate}
Thus, to resolve the instance $(G,k)$ of \probVDHp{$\C{H}$}, when $Q \neq \emptyset$, we can now invoke the known {\FPT} algorithms (see~\cite{Marx06,DBLP:journals/siamdm/CyganPPW13,DBLP:journals/siamdm/LokshtanovMRS17,DBLP:conf/soda/KakimuraKK12}) for each  problem, as discussed above. This concludes the proof.  
\end{proof}

We are now ready to prove Theorem~\ref{thm:app} by arguing the applicability of  Theorem~\ref{thm:mainEquiv}. 

\begin{proof}[Proof of Theorem~\ref{thm:app}]
For each $\C{H} \in \{\C{H}_{\sf mway},\C{H}_{\sf fvs},\C{H}_{\sf oct}\}$, we have the following properties: i) $\C{H}$ is \msotwo definable and closed under disjoint union (see Observation~\ref{obs:multicut-family-cmso}) and ii) \probVDH\ admits an {\FPT} algorithm, when parameterized by $\hhmd(G)$. Thus, from Theorem~\ref{thm:mainEquiv}, for each $\C{H} \in \{\C{H}_{\sf mway},\C{H}_{\sf fvs},\C{H}_{\sf oct}\}$, there is an {\FPT} algorithm for \probVDH, parameterized by $\hhtw(G)$. The above combined with Lemma~\ref{lem:multi-cut-fpt-algo} implies the proof of the theorem. 
\end{proof}

%% file: cross-param.tex
% !TEX root = main.tex
%\label{sec:crossParam}
%\newpage
In this section we extend our results from previous section to problems where the parameterizations is with respect to other problems. For an illustration consider {\sc Odd Cycle Transversal} on chordal graphs. Let $\hh$ denotes the family of chordal graphs. 
It is well known that  {\sc Odd Cycle Transversal} is polynomial time solvable on chordal graphs.  Further, given a graph $G$ and a modulator to  chordal graphs  of size  $\hhmd(G)$, {\sc Odd Cycle Transversal} admits an algorithm with running time $2^{\OO(\hhmd(G)))}n^{\OO(1)}$. It is natural to ask whether {\sc Odd Cycle Transversal} admits an algorithm with running time $f(\hhdepth(G))n^{\OO(1)} $ or $f(\hhtw(G))n^{\OO(1)} $, given a $\hh$-elimination forest of $G$ of depth~$\hhdepth(G)$ and $\hh$-decomposition of $G$ of width~$\hhtw(G)$, respectively. The question  is also relevant, in fact more challenging, when $\hh$-elimination forest of $G$ of depth~$\hhdepth(G)$ or $\hh$-decomposition of $G$ of width~$\hhtw(G)$ are not given. We provide sufficient conditions which allows us to have an algorithm for vertex deletion problems  (or edge deletion problems)  when given a $\hh$-elimination forest of $G$ of depth~$\hhdepth(G)$ or $\hh$-decomposition of $G$ of width~$\hhtw(G)$. Here, $\hh$ is a family of graphs.

\begin{theorem}\label{thm:fullCrossParameterization}
Let $\cal H$ be a family of graphs and $\Pi$ be a monotone parameterized graph problem and $(G,k)$ be an instance of $\Pi$. 
Further assume that we have following. 
\begin{enumerate}
\setlength{\itemsep}{-1pt}
\item $\Pi$ has {\rm \fii}.
\item $\Pi $ is {\rm \FPT}  parameterized by $\hhmd(G)$. 
\item An \htd\ (resp. \hed) of $G$ of width~$\hhtw(G)$  (resp. depth $\hhdepth(G)$) is given.
\end{enumerate}
Then, there is an algorithm that, given an $n$-vertex graph $G$ and an integer $k$, decides whether $(G,k)\in \Pi$ in time~$f(\hhtw(G)) \cdot n^{\Oh(1)}$ (resp. $f(\hhdepth(G)) \cdot n^{\Oh(1)}$). That is,  $\Pi$ is {\rm \FPT}  parameterized by $\hhtw(G)$ (resp. $\hhdepth(G)$).
\end{theorem}
\begin{proof}
The proof of this theorem is similar to the proof of Lemma~\ref{lem:three-implies-one}, for the sake of completeness we give a complete proof here. Consider an instance $(G,k)$ of $\Pi$, and let $(T, \chi, L)$ be the given \htd\ (resp. \hed) of width (resp. depth) at most $\hhtw(G)$ (resp. $\hhed(G)$) for $G$, and let $\ell = \hhtw(G)$ (resp. $\ell = \hhed(G)$). Let $\SC{P}_{\sf mod}$ be an {\FPT} algorithm for $\Pi$, running in time $f(\hhmd(G)) \cdot n^{\C{O}(1)}$. Using Lemma~\ref{lem:red2finiteindex}, we will next bound the size of $\chi(t)$, for each leaf $t$ in $T$. 

We will assume that the constants $\xi_i$ from Lemma~\ref{lem:red2finiteindex}, for each $i \in [\ell+1]$ are hardcoded in the algorithm. We will bound the number of vertices in $G[\chi(t)\cap L]$ by $\sum_{i \in [\ell+1]} \xi_i$, for each $t \in V(T)$, then the (standard) treewidth of $G$ can be bounded by $\ell + \sum_{i \in [\ell+1]} \xi_i$ (resp. $\ell + 1 + \sum_{i \in [\ell+1]} \xi_i$). Using the property that $\Pi$ has {\fii}, (roughly speaking) our next objective will be to bound the size of $\chi(t)$, for each $t \in V(T)$ by replacements using Lemma~\ref{lem:red2finiteindex}. Let $\wtilde{A}$ be the set of nodes in $T$ whose bag contains more than $\sum_{i \in [\ell+1]} \xi_i$ vertices from $L$, i.e., $\wtilde{A} = \{t \in V(T) \mid |\chi(t) \cap L| > \sum_{i \in [\ell+1]} \xi_i\}$. Furthermore, let $\wtilde{\C{G}}$ be the set of graphs induced by vertices in $L$, in each of the bags of nodes in $\wtilde{A}$, i.e., $\wtilde{\C{G}} = \{G[\chi(t) \cap L] \mid t \in \wtilde{A}\}$. We let $\wtilde{\C{G}} = \{\wtilde{G}_1, \wtilde{G}_2, \cdots, \wtilde{G}_{q}\}$. 

We create a sequence of $G_0, G_1, \cdots, G_{q}$ graphs and a sequence of constant $c_0,c_1,\cdots, c_{q}$ as follows. Intuitively speaking, we will obtain the above sequence of graph by replacing $\wtilde{G}_i$, for $i \in [q]$, by some graph obtained using Lemma~\ref{lem:red2finiteindex}. Set $G_0 = G$ and $c_0 = 0$. We iteratively compute $G_i$, for each $i \in [q]$ (in increasing order) as follows. For the graph $\wtilde{G}_i$, let $\wtilde{t}_i$ be the unique leaf in $T$, such that $V(\wtilde{G}_i) \subseteq \chi(\wtilde{t}_i)$, and let $\wtilde{B}_i = \chi(\wtilde{t}_i) \setminus V(\wtilde{G}_i)$ and $\wtilde{b}_i = |\wtilde{B}_i|$. Note that $|\wtilde{B}_i| \leq \ell$ (resp. $|\wtilde{B}_i| \leq \ell + 1$), as $(T,\chi, L)$ is an \hed\ (resp. \htd) of $G$ of depth (resp. width) $\ell$. Fix an arbitrary injective function $\lambda_{\wtilde{G}_i} : \wtilde{B} \rightarrow \{1,2,\cdots, \wtilde{b}_i\}$, and then $\wtilde{G}_i$ is the boundaried graph with boundary $\wtilde{B}_i$.\footnote{We have slightly abused the notation, and used $\wtilde{G}_i$ to denote both a graph and a boundaried graph.} Note that $V(\wtilde{G}_i]) \subseteq V(G_{i-1})$. Also, let $G'_i$ be the boundaried graph $G_{i-1}-V(\wtilde{G}_i)$, with boundary $\wtilde{B}_i$. Using Lemma~\ref{lem:red2finiteindex} and the algorithm $\SC{P}_{\sf mod}$, we find the graph $\wtilde{G}^*_i$ and the translation constant $c_i$, such that $\wtilde{G}_i \equiv_{\Pi} \wtilde{G}^*_i$ and $|V(\wtilde{G}^*_i)| \leq \xi_{\wtilde{b}_i}$ in time bounded by $\cO(f(\ell) \cdot n^{\C{O}(1)} )$.\footnote{Note that for any graph $\what{G}$, $\iota(\what{G}) \leq |V(\what{G})|$ (see Definition~\ref{def:equiv-boundaried-graph}).}

We let $G_i$ be the graph $\wtilde{G}^*_i \oplus G'_i$, which can be computed in time bounded by $\cO(f(\ell) \cdot n^{\C{O}(1)} )$. Let $c^* = \sum_{i \in [q]} c_i$. With the constructions described above, we are now in a position to prove Lemma~\ref{lem:three-implies-one}. Note that to obtain the desired result, it is enough to prove the following statements. 
%As $\C{H}$ is CMSO definable to prove the lemma, it is enough to establish the following statements (together with Courcelle's Theorem~\cite{DBLP:journals/iandc/Courcelle90}).   
\begin{enumerate}
\item The instance $(G_q,k+c^*)$ can be constructed in time bounded by $\cO(g(\ell) \cdot n^{\C{O}(1)})$, for some function $g$, 
\item $(G_q, k+c^*)$ and $(G,k)$ are equivalent instances of $\Pi$, and 
\item The treewidth of $G_q$ is at most $\ell + \max_{i \in [q]}{c_i} + \sum_{i \in [\ell]} \zeta_i$ and the $\C{H}$-treewidth (resp. $\C{H}$-elimination distance) of $G_q$ is at most $\ell + \max_{i \in [q]}{c_i}$.\footnote{We remark that although $k+c^*$ can possibly be much larger than $k$, both the treewidth and the $\C{H}$-treewidth of $G_q$ are at most some additive constants (depending on $\C{H}$) away from $k$.} %Also, the $\C{H}$-treewidth of $G_q$ is at most $\ell + \max_{i \in [q]}{c_i}$. 
\end{enumerate} 
%\end{lemma}
%\begin{proof}
As stated perviously, we will assume that the constants $\xi_i$, for $i \in [\ell+1]$ are hardcoded in the algorithm. Thus, we can construct the set $\wtilde{\C{G}}$ in polynomial time. Also, note that for any $i \in [q]$, $\iota(\wtilde{G}^*_i \oplus G'_i) \leq \iota(G) \leq |V(G)|$ (see Definition~\ref{def:equiv-boundaried-graph}). Thus, for some function $f$, we can construct the instance $(G_q, k+c^*)$ in time bounded by $g(\ell) \cdot n^{\C{O}(1)}$. 

We will inductively argue that for each $i \in [q]_0$, $(G_i, k + \sum_{j \in [i]_0} c_j)$ and $(G,k)$ are equivalent instances of $\Pi$. As $G_0 = G$ and $k+ c_ 0 =k$, the claim trivially follows for the case when $i = 0$. Next we assume that for some $q' \in [q-1]_0$, for each $i' \in [q']_0$, $(G_{i'}, k + \sum_{j \in [i']_0} c_j)$ and $(G,k)$ are equivalent instances of the problem. We will next prove the statement for $i = i'+1$. It is enough to argue that $(G_{i-1}, k + \sum_{j \in [i-1]_0} c_j)$ and $(G_{i}, k + \sum_{j \in [i]_0} c_j)$ are equivalent instances. Recall that, by construction, $G[V(\wtilde{G}_i)] = G_{i-1}[V(\wtilde{G}_i)] = \wtilde{G}_i$ and $G'_i = G_{i-1} - V(\wtilde{G}_i)$ are boundaried graphs with boundary $\wtilde{B}_i$, and $G_i = \wtilde{G}_i \oplus G'_i$. From Lemma~\ref{lem:red2finiteindex}, $\wtilde{G}^*_i \equiv_{\Pi} \wtilde{G}_i$. Thus by definition, we have that $(G_{i-1}, k + \sum_{j \in [i-1]_0} c_j)$ and $(G_{i}, k + \sum_{j \in [i]_0} c_j)$ are equivalent instances of $\Pi$. 

To prove the third statement, note that it is enough to construct an $\C{H}$-tree decomposition, $(T_q, \chi_q, L_q)$, of $G_q$, where for each $t \in V(T_q)$, we have $|\chi_q(t)| \leq \ell + \max_{i \in [q]}{c_i} + \sum_{i \in [\ell]} \chi_i$ and $|\chi_q^{-1}(t) \setminus L_q| \leq \ell + \max_{i \in [q]}{c_i}$. (Using similar arguments we can also obtain the statement regarding \hed.) Let $X = \cup_{i \in [q]} V(\wtilde{C}_i)$, and $L_q = (L \setminus X ) \cup (V(G_q) \setminus V(G))$ and $T_q = T$. For each $t \in V(T) \setminus \wtilde{A}$, we set $\chi_q(t) =\chi(t)$, and for each $i \in [q]$, we set $\chi_q(\wtilde{t}_i) = (\chi(\wtilde{t}_i) \setminus V(\wtilde{G}_i)) \cup V(\wtilde{G}^*_i)$. For each $i \in [q]$, note that $|V(\wtilde{G}^*_i)| \leq \xi_{\wtilde{b}_i} \leq \sum_{j \in [\ell]} \xi_{j}$. Thus we can obtain that $(T_q, \chi_q, L_q)$ is an $\C{H}$-tree decomposition of $G_q$ that satisfies all the required properties. This concludes the proof. 
%
%\begin{enumerate} 
%\item Replace the leaf bags using Lemma~\ref{lem:red2finiteindex}
%\item Treewidth is small
%\item Apply Courcelle's Theorem
%\end{enumerate}
\end{proof}

%% file: uniform.tex
% !TEX root = main.tex
\newcommand{\QQ}{{\cal Q}}
\newcommand{\PP}{{\cal P}}
\newcommand{\equ}{\equiv_{u}}
\newcommand{\eqh}{\equiv_{\hh}}
\newcommand{\eqf}{{\sf eq}}
\newcommand{\ugf}{{\sf ug}}
\newcommand{\sigtest}{{\sf Sig-Test}}
\newcommand{\rsig}{{\sf real\mbox{-}sig}}
\newcommand{\heds}{$\hh$-elimination decompositions}
\newcommand{\wT}{{\widehat{T}}}
\newcommand{\imp}[3]{I_{#1}(#2,#3)}
\newcommand{\used}[2]{{\sf Used}_{#1}(#2)}

\newcommand{\restricted}[2]{{\sf Restricted}_{#1}#2}

\newcommand{\present}{\spadesuit}
\newcommand{\future}{\heartsuit}

\section{Uniform {\FPT} Algorithm}

In this section we design a uniform {\FPT} algorithm for \probEDH  parameterized by $\hhdepth(G)$ assuming we have a {\em specific type} of  algorithm for solving \probVDH parameterized by $\hhmd(G)$.   Our algorithm is generic and in the next section we explain the corollaries of this result for various families \hh. Throughout this section we assume that \hh\ is a hereditary family of graphs and it is closed under disjoint union. Throughout the section, $k$ is the parameter in the input instance of \probEDH. Unless specified all the graphs mentioned in this section are $2k$-boundaried graphs and we assume that $k\geq 2$.  Our algorithm uses the recursive understanding technique~\cite{DBLP:journals/siamcomp/ChitnisCHPP16,DBLP:conf/stoc/GroheKMW11,DBLP:conf/focs/KawarabayashiT11}.

Here, we identify a $(q,k)$-unbreakable induced subgraph from an input graph and we replace it with a smaller representative graph, where $q$ is bounded by function of $k$. Eventually, when the size of the graph become bounded by a function of $k$, we use a brute-force algorithm to solve the problem.

\begin{definition}[Canonical equivalence relation]
\label{def:canonical-eq}
The canonical equivalence relation $\eqh$ for \hh\ over the set of boundaried graphs is defined as follows. Two graphs $G_1$ and $G_2$ are equivalent if for any graph $H$, $G_1 \oplus H \in \hh \Leftrightarrow G_2 \oplus  H \in \hh$. 
\end{definition}

We assume that we are given a refinement of the equivalence relation $\eqh$ which we call a {\em user defined equivalence relation}.  We use $\equ$ to denote this refinement of the equivalence relation $\eqh$. We use $\QQ(\equ)$ to denote the set of equivalence classes of $\equ$. 

%\todo[inline]{write about user's algorithm here.}

\begin{definition}[Comparison of equivalence classes]
\label{def:eq-compare}
Let ${Q}_1$ and ${Q}_2$ be two equivalence classes in $\equ$. Let $G_1\in {Q}_1$ and $G_2\in {Q}_2$. 
We say that ${Q}_1$ is at least as good as ${Q}_2$ if the following holds. For any graph $H$, if $G_2\oplus H \in \hh$, then $G_1\oplus H \in \hh$. 
\end{definition}

Throughout the section we also assume that along with $\equiv_{u}$, we are given a user defined function ${\sf ug} \colon \QQ(\equ) \times \QQ(\equ) \mapsto \{0,1\}$ such that $\ugf(Q,Q)=1$ for all $Q\in \QQ(\equ)$ and it has the following property. If $\ugf(Q_1,Q_2)=1$, then $Q_1$ is at least as good as $Q_2$. We would like to mention that the reverse may not be true. 
That is, even if $Q_1'$ is at least as good as $Q_2'$, the user may define that $\ugf(Q_1',Q_2')=0$. 
We also assume that equivalence classes in $\equ$ satisfies the following property. 

\begin{definition}[Component deletion property]
\label{def:cdpeq}
Let $G$ be a graph and $Q$ is an equivalence class in $\equ$ such that $G\in Q$. Let $C$ be a connected component of $G$ such that $V(C)\cap \delta(G)=\emptyset$. Then, $G-C$ belongs to $Q$. Moreover, if $F$ is a graph in $\hh$ and $\delta(F)=\emptyset$, then $G\uplus F \in Q$. 
\end{definition}

Moreover, we  assume that we have an access to an algorithm, called {\em user's algorithm}, denoted by ${\cal A}_u$ that, given a graph $G$, a non-negative integer $k'$ and equivalence class $Q\in \QQ(\equ)$, outputs a vertex subset $S$ of size at most $k'$ such that $G-S$ belongs to an equivalence class which is at least as good as $Q$ (if it exists). If no such set exists, then the algorithm outputs {\sf No}. We use $f_{u}(k', |\QQ(\equ)|,n)$ to denote the running time of ${\cal A}_u$ where $n=\vert V(G)\vert$. 

Next, we define the notion of state-tuple. Here two ``graphs having  same set of state-tuples'' can be treated as identical. 

\begin{definition}[State-tuple]
A state-tuple is a tuple $(D,T,\chi,L,\PP,\eqf)$, where 
$L\subseteq D\subseteq [2k]$, $\PP$ is a partition of $L$, $\eqf \colon \PP \rightarrow \QQ(\equ)$,
$T$ is a rooted forest of depth at most $k$ and $\vert V(T) \vert \leq (2k)^2$, and $\chi \colon V(T) \to 2^D \cup \{\present,\future\}$ such that the following holds. 

\begin{enumerate}
    \item For each internal node~$t$ of~$T$ we have $\chi(t)\in  \{\present,\future,\emptyset\} \cup \{\{x\}\colon x\in (D\setminus L) \}$. 
    \item For each leaf $t$ of~$T$, $\chi(t) \in \PP \cup \{\{x\}\colon x\in (D\setminus L) \}$. 
    \item The sets~$(\chi(t))_{t \in N}$ form a partition of~$D$, where $N=\{t\in V(T) \colon \chi(t)\notin \{\present,\future,\emptyset\}\}$
%    \item For each edge~$uv \in E(G)$, if~$u \in \chi(t_1)$ and~$v \in \chi(t_2)$ then~$t_1, t_2$ are in ancestor-descendant relation in~$T$.
%    \item For each leaf $t$ of~$T$, we have $\chi(t) \subseteq L$ and the graph~$G[\chi(t)]$, {called a base component}, belongs to~$\hh$. 
\end{enumerate}
\end{definition}

Next we define a notation which we use throughout the section. Let $G$ be a graph and $T$ be a rooted forest. 
Let $\chi$ be a function from $V(T)$ to a super set of $2^{V(G)}$. Let $Y$ be the nodes of $T$ that are mapped  to subsets containing at least one vertex from $\delta(G)$ by the function  ${\chi}$. 
Then, we use $\imp{T}{G}{\chi}$ to denote the set of nodes of $T$ that belong to the unique paths in $T$ from the vertices in $Y$ to the corresponding roots. 

\begin{definition}
\label{def:satisf}
Let $G$ and be a graph and $s=(D,T,\chi,L,\PP,\eqf)$ be a state-tuple. 
We say that $G$ satisfies $s$ if $\Lambda(G)=D$ and there is an $\hh$-elimination decomposition $(\wT,\widehat{\chi},\widehat{L})$ of $G$, of depth at most $k$ with the following properties. 
%Let $Y$ the nodes of $T'$ that are mapped  to subsets containing at least one vertex from $\delta(G)$ by the function  ${\chi}'$. 
\begin{itemize}
%\item[(a)]  $T$ is an induced subgraph of $T'$ such that for any $t\in V(T)$, all the vertices in the unique path in $T'$ from $t$ to the root belong to $V(T)$. Moreover there is an induced subgraph isomorphism $\phi \colon V(T)  \rightarrow V(T')$ with the properties mentioned below. 
\item[(a)] $T$ is isomorphic to $\wT[Z]$, where $Z=\imp{\wT}{G}{\widehat{\chi}}$. For any node $t\in Z$, we use the same notation to represent the node in $T$ that maps to $t$ by the isomorphism function. That is, $V(T)=Z$.   
\item[(b)] For any node~$t$ of~$T$ with $\widehat{\chi}(t)\cap \delta(G)\neq \emptyset$, 
$\chi(t)=\lambda_G(\widehat{\chi}(t)\cap \delta(G))$
%$\chi(t)\in \PP\cup \{\{x\}\colon x\in (D\setminus L) \}$  we have that $\lambda_G(\chi'(t)\cap \Lambda(G))= \chi(t)$. 
\item[(c)] For any node $t$ of~$T$ with $\widehat{\chi}(t)\neq \emptyset$ and 
$\widehat{\chi}(t)\subseteq V(G)\setminus \delta(G)$, $\chi(t)=\present$. 
\item[(d)] Let $t_1,\ldots,t_{\ell}$ be the leaf nodes of $\wT$ that belong to $Z$. Then, $\PP=(\chi(t_i))_{i\in [\ell]}$. 
\item[(e)] For each node~$t$ of~$T$ with $\chi(t)\in \PP$, $G[\widehat{\chi}(t)]$ belongs to an equivalence class $Q\in \QQ(\equ)$ which  is at least as good as $\eqf(\chi(t))$. That is, $\ugf(Q,\eqf(\chi(t))=1$. 
\end{itemize} 
%We say that $(G,H)$ realizes $s$, if the following conditions hold along with the above conditions. 
%\begin{itemize}
%\item[(e)] For each node $t$ of~$T$ with $\chi'(t)\neq \emptyset$ and $\chi'(t)\subseteq V(H)\setminus \Lambda(G)$, we have that $\chi(t)=F$.
%\item[(d)] For each node~$t$ of~$T$ with $\chi(t)\in \PP$, we have that $G \oplus [\chi'(t)]$ is in an equivalence class $Q$ of $\equ$ which is at least as good as $Q'=\eqf(\chi(t))$. That is, $\ugf((Q,Q'))=1$. 
%\end{itemize} 
Also, we say that $G$ satisfies $s$ through $(\wT,\widehat{\chi},\widehat{L})$. 
\end{definition}

\begin{definition}
\label{def:exact-satisf}
We say that a graph $G$ exactly satisfies a state-tuple $(D,T,\chi,L,\PP,\eqf)$ if $G$ satisfies $(D,T,\chi,L,\PP,\eqf)$ such that a stricter condition than condition (e) in Definition~\ref{def:satisf}  holds. 
That is, for each node~$t$ of~$T$ with $\chi(t)\in \PP$, $G[\widehat{\chi}(t)]$ belongs to the equivalence class $\eqf(\chi(t))$. 
\end{definition}

\begin{definition}
Signature of a graph $G$, denoted by $\rsig(G)$, is the set of all state-tuples $s$ such that $G$ satisfies $s$. 
\end{definition}

\begin{observation}
Signature of a graph $G$ is unique and its cardinality is bounded by a function $k$ and $\vert \QQ(\equ)\vert$.  
\end{observation}

For two graphs $G$ and $H$, and an $\hh$-elimination decomposition $(\wT,\widehat{\chi},\widehat{L})$ of $G\oplus H$, we use $\restricted{G}{(\wT,\widehat{\chi},\widehat{L})}$ to denote the $\hh$-elimination decomposition $(\wT,\psi,F)$ of $G$ obtained by restricting $(\wT,\widehat{\chi},\widehat{L})$ to $G$. Formally $F=\widehat{L}\cap V(G)$ and for any node $t \in V(\wT)$, $\psi(t)=\widehat{\chi}(t)\cap V(G)$.

\begin{definition}
\label{def:realizes}
Let $G$ and $H$ be two graphs 
and $s=(D,T,\chi,L,\PP,\eqf)$ be a state-tuple.  We say that $(G,H)$ realizes $s$, if $\Lambda(G)=D$ and there is an $\hh$-elimination decomposition $(\wT,\widehat{\chi},\widehat{L})$ of $G^{\star}=G\oplus H$, of depth at most $k$ with the following conditions. 
\begin{itemize}
\item[$(i)$] $G$ satisfies $s$ through $\restricted{G}{(\wT,{\widehat{\chi}},{\widehat{L}})}$. 
That is, conditions (a)-(e) in Definition~\ref{def:satisf} hold.  
%\begin{itemize}
%\item For any node $t \in V(\wT)$, ${\widehat{\chi}}_1(t)=\widehat{\chi}(t)\cap V(G)$.   
%\end{itemize}
\item[$(ii)$] For each node $t$ of~$T$ with $\widehat{\chi}(t)\neq \emptyset$ and $\widehat{\chi}(t)\subseteq V(H)\setminus \delta(G)$, $\chi(t)=\future$. (Recall that $T$ is isomorphic to $\wT[Z]$, where $Z=\imp{\wT}{G}{\widehat{\chi}}$).  
\end{itemize}
Then, we say that $(G,H)$ realizes $s$ through $(\wT,\widehat{\chi},\widehat{L})$. 
We say that $(G,H)$ exactly realizes $s$, if instead of conditions $(i)$, we have that $G$ exactly satisfies $s$ through $\restricted{G}{(\wT,{\widehat{\chi}},{\widehat{L}})}$. 
%$(\wT,\widehat{\chi},\widehat{L})$. 
In that case, we say that $(G,H)$ exactly realizes $s$ through 
$(\wT,\widehat{\chi},\widehat{L})$. 
%\begin{itemize}
%\item[($e'$)] For each node~$t$ of~$T$ with $\chi(t)\in \PP$, we have that $G^{\star}[\chi'(t)]$ belongs to the  equivalence class $\eqf(\chi(t))$. 
%\end{itemize} 
\end{definition}

\begin{definition}
A marked signature is a subset of $\{s \colon s \mbox{  is a state-tuple}\}\times \{0,1\}$. 
\end{definition}

\begin{observation}
\label{obs:no-of-mark-sig}
The number of marked signatures are bounded by a function of $k$ and the number of equivalence classes in $\equ$. 
\end{observation}

In a marked signature every state-tuple is marked as $0$ or $1$, where we interpret the marking with $0$ as marking ``do not care''. Later we will define {\em marked signature of a graph $G$} where for any tuple $s$ that is marked with $0$ (do not care), there is a tuple in the signature which is ``strictly better than'' $s$. Towards that we define the following.  

\begin{definition}
Let $s_1=(D_1,T_1,\chi_1,L_1,\PP_1,\eqf_1)$ and $s_2=(D_2,T_2,\chi_2,L_2,\PP_2,\eqf_2)$ be two state tuples. We say that $s_2$ is strictly better than $s_1$ if the following holds. 
\begin{itemize}
\item The set $D_2\setminus L_2$ is a strict super set of $D_1\setminus L$.
\item  For any two graphs $G$ and $H$, if $(G,H)$ exactly realizes $s_1$, then $(G,H)$ realizes $s_2$. 
\end{itemize}  
\end{definition}

\begin{definition}[Validity]
Let $sig$ be a marked signature. We say that $sig$ is valid if the following holds. For any pair $(s_1=(D_1,T_1,\chi_1,L_1,\PP_1,\eqf_1),0)$ in $sig$ (i.e., a tuple that is marked with $0$), there is a pair $(s_2=(D_2,T_2,\chi_2,L_2,\PP_2,\eqf_2),b)$ in $sig$, such that 
%$D_2\setminus L_2$ is a strict superset of $D_1\setminus L_1$ and 
$s_2$ is strictly better than $s_1$,  where $b\in \{0,1\}$.   
\end{definition}

\begin{definition}[compatibility]
Let $G$ be a graph and $sig$ be a marked signature. We say that $sig$ is compatible with $G$ if 
$sig$ is valid and the following holds. 
\begin{itemize}
\item For any $s\in \rsig(G)$,  $\{(s,0),(s,1)\}\cap sig \neq \emptyset$. 
\item For any $s \notin \rsig(G)$, if $(s,b)\in sig$, then $b=0$. 
\end{itemize} 
\end{definition}

We would like to mention that for a graph $G$, there could be many marked signatures that are compatible with $G$.

\begin{definition}[Similarity]
Let $sig_1$ and $sig_2$ be two marked signatures. We say that $sig_1$ and $sig_2$ are similar if the following holds. 
\begin{itemize}
\item $\{s \colon (s,1) \in sig_1 \}\subseteq \{s' \colon (s',0)\in sig_2 \mbox{ or } (s',1)\in sig_2\}$, and  
\item $\{s \colon (s,1) \in sig_2 \}\subseteq \{s' \colon (s',0)\in sig_1 \mbox{ or } (s',1)\in sig_1\}$. 
\end{itemize} 
\end{definition}

Next we prove that if the ``real signatures'' of two graphs $G_1$ and $G_2$ are same, then for any two marked signatures $sig_1$ and $sig_2$ that are compatible with $G_1$ and $G_2$, respectively,  the marked signatures $sig_1$ and $sig_2$ are similar. 

\begin{lemma}
\label{lem:rev:sig:sim}
Let $G_1$ and $G_2$ be two graphs such that $\rsig(G_1)=\rsig(G_2)$. Let $sig_1$ and $sig_2$ be two marked signatures such that $sig_1$ is compatible with $G_1$ and $sig_2$ is compatible with $G_2$. Then $sig_1$ and $sig_2$ are similar. 
\end{lemma}

\begin{proof}
Suppose $sig_1$ and $sig_2$ are not similar. Then, at least one of the following statements is true.  
\begin{itemize}
\item[$(i)$] There exists a state-tuple $s$ such that $(s,1)\in sig_1$ and $(s,0),(s,1)\notin sig_2$. 
\item[$(ii)$] There exists a state-tuple $s$ such that $(s,1)\in sig_2$ and $(s,0),(s,1)\notin sig_1$. 
\end{itemize}
Since the above statements are symmetric we assume that statement $(i)$ is true and derive a contradiction. 
Since $(s,1)\in sig_1$, we have that $s\in \rsig(G_1)=\rsig(G_2)$. Since $s\in \rsig(G_2)$ we have that either $(s,1) \in sig_2$ or $(s,0)\in sig_2$. This is a contradiction to the assumption that  $(s,0),(s,1)\notin sig_2$. 
\end{proof}

\

\begin{definition}
\label{def:repuni}
We say that a graph $R$ {\em represents} a graph $G$, if the following holds. 
For any graph $H$, $(G\oplus H,k)$ is yes-instance of \probEDH if and only if $(R\oplus H,k)$ is yes-instance of \probEDH.
%$\rsig(R)=\rsig(G)$. 
%if there exists marked signatures $sig_G$ and $sig_H$ such that $sig_G$ and $sig_H$ are similar, $sig_G$ is compatible with $G$, and $sig_H$ is compatible with $H$.  
\end{definition}

Observation~\ref{obs:no-of-mark-sig} implies the following lemma. 

\begin{lemma}
\label{lem:repr}
There is a function $r$ such that for any graph $G$, there is a graph $R$ of size $r(k,\vert\QQ(\equ)\vert)$, that represents $G$. 
\end{lemma}

Next we state two important lemmas and then using them prove our results. Later we prove both the lemmas.

\begin{lemma}
\label{lem:sim:sig}
Let $G_1$ and $G_2$ be two graphs such that $\Lambda(G_1)=\Lambda(G_2)$. Let $sig_1$ and $sig_2$ be two marked signatures such that $sig_1$ is compatible with $G_1$ and $sig_2$ is compatible with $G_2$. Moreover, $sig_1$ and $sig_2$ are similar. Then, for any graph $H$, $(G_1\oplus H,k)$ is yes-instance of \probEDH if and only if $(G_2\oplus H,k)$ is yes-instance of \probEDH. 
\end{lemma}

We prove Lemma~\ref{lem:sim:sig} in Section~\ref{sec:lemsimilar}. Next we prove that given a graph $(q,k)$-unbreakable graph $G$, a marked signature that is compatible with $G$ can be computed efficiently. 
Recall that we are given an algorithm ${\cal A}_u$ to find a deletion set of an input graph to a given equivalence class in $\QQ(\equ)$ and its running time is denoted by the function $f_u$. 

\begin{lemma}
\label{lem:comp:sig}
There is a function $g$ and an algorithm that given a $(q,k)$-unbreakable graph $G$ (recall that $G$ is also a $2k$-boundaried graph) as input, runs in time $g(k,q,\vert \QQ(\equ)\vert) \cdot f_{u}(k,\vert \QQ(\equ)\vert,n) \cdot n^{\OO(1)}$, and outputs a marked signature that is compatible with $G$, where $n=\vert V(G)\vert$.   
\end{lemma}

We prove Lemma~\ref{lem:comp:sig} in Section~\ref{sec:proof:lem:comp:sig}. Now, using Lemmas~\ref{lem:sim:sig} and \ref{lem:comp:sig}, we define a uniform {\FPT} algorithm for \probEDH. We know that by Lemma~\ref{lem:repr}, for any graph $G$, there is a representative of $G$ and its size bounded a function $r$ of $k$ and $|\QQ(\equ)|$. We remark that this is an existential result and we do not the function $r$. So first we design an algorithm for \probEDH\ assuming we know the value $r(k,|\QQ(\equ)|)$. Later we explain how to get rid of this assumption. In the rest of the section we use $q=2^{k+1}\cdot r(k,|\QQ(\equ)|)$.

\begin{lemma}
\label{lem:comp:rep:unbre}
There is a function $g'$ and an algorithm ${\cal A}_r$ that given a $(q,k)$-unbreakable graph $G$, runs in time $g'(k,q,|\QQ(\equ)|)\cdot f_u(k,|\QQ(\equ)|,n) \cdot n^{\OO(1)}$ and outputs a graph $R$ of size at most  $r(k,|\QQ(\equ)|)$ such that $R$ represents $G$, where $n=\vert V(G)\vert$. Here $f_u$ is the running time of the user's algorithm ${\cal A}_u$.
\end{lemma}

\begin{proof}
Our algorithm uses Lemma~\ref{lem:comp:sig}.  The pseudocode of our algorithm is given in Algorithm~\ref{alg:Ar}.

\begin{algorithm}[H]
\SetAlgoLined
\KwResult{A graph $R$ that represents the input $(q,k)$-unbreakable graph $G$}
 Using Lemma~\ref{lem:comp:sig} compute a marked signature $sig_G$ that is compatible with $G$\;
 \For{$i=1$ to $r(k,|\QQ(\equ)|)$}{
 \For{every graph $R$ on $i$ vertices}{
 Using Lemma~\ref{lem:comp:sig} compute a marked signature $sig_R$ that is compatible with $R$\;
\If{$sig_G$ and $sig_R$ are similar}{
 Output $R$\;
}
}
 }
 
 \caption{Algorithm ${\cal A}_r$}
 \label{alg:Ar}
\end{algorithm}

Next we prove the correctness of the algorithm. By Lemma~\ref{lem:repr}, we know that there is a representative of $G$, of size at most $r(k,|\QQ(\equ)|)$. Thus, there is a graph $R$ of size at most $r(k,|\QQ(\equ)|)$ such that  $\rsig(G)=\rsig(R)$. Hence, by Lemma~\ref{lem:rev:sig:sim}, for any  
marked signatures $sig_R$ and $sig_G$ that are compatible with $R$ and $G$, respectively,  
we have that $sig_G$ and $sig_R$ are similar. 
Moreover, by Lemma~\ref{lem:sim:sig} if there is a graph $R$ and a marked signature $sig_R$ compatible with $R$ such that $sig_G$ and $sig_R$ are similar, then $R$ represents $G$.  Therefore, the algorithm is correct. 

Notice that the algorithm runs the algorithm of Lemma~\ref{lem:comp:sig} at most $r(k,|\QQ(\equ)|)+1$ times. This implies that the total running time of the algorithm is upper bounded by $g'(k,q,|\QQ(\equ)|)\cdot f_u(k,|\QQ(\equ)|,n)\cdot n^{\OO(1)}$ for some function $g'$.  This completes the proof of the lemma. 
\end{proof}

There is an algorithm to determine (approximately) whether a graph is unbreakable. 
%That result is stated below.  
We use the statement from \cite{LokshtanovR0Z18}, although lemmas similar to it can be found in \cite{DBLP:journals/siamcomp/ChitnisCHPP16}.

\begin{proposition}{\rm ~\cite{LokshtanovR0Z18}}
\label{prop:unbr}
There is an algorithm {\sf Break-ALG}, that given two positive integer $s, c \in {\mathbb N}$ and a graph $G$, runs in time $2^{\OO(c \log(s+c))}\cdot n^3\log n)$ and either returns an $(\frac{s}{2^c},c)$-witnessing separation or correctly concludes that $G$ is $(s, c)$-unbreakable.
\end{proposition}

Next we prove the following theorem. The proof of the theorem has the same {\em general template} employed by algorithms based on the recursive understanding technique. So, we give a proof sketch of Theorem~\ref{thm:uni:hypo:elim}.

\begin{theorem}
\label{thm:uni:hypo:elim}
%There is an algorithm that given an instance $(G,k)$ of 
There is a function $f$ and  an algorithm for \probEDH running in time $f(k,q,|\QQ(\equ)|)\cdot f_u(k,|\QQ(\equ)|,n) \cdot n^{\OO(1)}$. 
 %and outputs a graph $R$ of size at most $r(k,|\QQ(\equ)|)$.
  Here,  we assume that we know the value $r(k,|\QQ(\equ)|)$.  
\end{theorem}

\begin{proof}[Proof sketch]
We design a recursive algorithm, denoted by ${\cal A}$,  where the input is a $2k$-boundaried graph $G'$ and the output is a graph of size at most $r(k,|\QQ(\equ)|)$ that represents the input graph. 
%recursive call will be $2k$-boundaried graphs with possibly non-empty boundary vertices.
%Recall that $q=2\cdot r(k,|\QQ(\equ)|)$. 
 The following are the steps of the recursive algorithm ${\cal A}$. 
\begin{itemize}
%\item[(1)] If $|V(G)|\leq 2 q$ we do a brute force computation. That is we enumerate all choices all graphs $R$ of size at most $r(k,|\QQ(\equ)|)$ marked signatures $sig_G$ and $sig_R$ compatible with $G$ and $R$, respectively. If for any choice of $(R,)$

\item[(1)] Apply Proposition~\ref{prop:unbr} on $(G',q,k)$ and it outputs either an $(\frac{q}{2^k},k)$-witnessing separation $(X,Y)$ or conclude that  $G'$ is $(q,k)$-unbreakable. 

\item[(2)] If $G'$ is $(q,k)$-unbreakable, then compute a representative $R$ of $G'$ using Lemma~\ref{lem:comp:rep:unbre}.

\item[(3)] Suppose we get an $(\frac{q}{2^k},k)$-witnessing separation $(X,Y)$ in Step $(1)$. Let $G_1$ be the boundaried graph $G[X]$ with $\delta(G_1)=(\delta(G)\cup S) \cap V(G_1)$ and $G_2$ be the boundaried graph $G[Y]$ with $\delta(G_2)=(\delta(G)\cup S) \cap V(G_2)$. Since $\vert \delta(G)\vert \leq 2k$, either $G_1$ or $G_2$ is a $2k$-boundaried graph. 
%Let $i\in \{0,1\}$ such that $G_i$ is a $2k$-boundaried graph. 
Without loss of generality let  $G_1$ is be a  $2k$-boundaried graph. 
Since $(X,Y)$ is a $(\frac{q}{2^k},k)$-witnessing separation, we have that $\vert V(G_1)\vert \geq \frac{q}{2^{k}}\geq 2 r(k,|\QQ(\equ)|)$. 
\item[(4)] Recursively call ${\cal A}$ and compute a representative $R$ of $G_1$ of size at most $r(k,|\QQ(\equ)|)$. 
\item[(5)] Recursively call ${\cal A}$ on $R\oplus_{\delta} G_2$ and output the result. 
\end{itemize}

Let $(G,k)$ be the given input instance of \probEDH.  We assume that $G$ is a $2k$-boundaried graph with $\delta(G)=\emptyset$. Then run ${\cal A}$ on $G$ and let $R$ be the output. Then we do a brute force computation on $R$ and output accordingly. The correctness of our algorithm follows from the correctness of  Lemma~\ref{lem:comp:rep:unbre}, Proposition~\ref{prop:unbr}, and Definition~\ref{def:repuni}. 

Now we analyse the running time. Let $T(n,k,|\QQ(\equ)|)$ be the running time of the algorithm ${\cal A}$.  
In the algorithm we make two recursive calls where the size of the input graphs are at most $|V(G_1)|$ and $|V(G_2)|+r(k,|\QQ(\equ)|)$. Here $|V(G_1)|+|V(G_2)|\leq n+2k$. This implies that the running time has the following recurrence relation. 
$$T(n,k,|\QQ(\equ)|)=T(n_1,k,|\QQ(\equ)|)+T(n_2+r(k,|\QQ(\equ)|),k,|\QQ(\equ)|)+d(k,\QQ(\equ),n)$$
where $n_1+n_2\leq n+ 2k$, and $d(k,\QQ(\equ),n)$ is the time taken for Steps $(1)$ and $(2)$. 
The base case is $T(2q-1,k,|\QQ(\equ)| )\leq d(k,\QQ(\equ),2q-1)$.  By solving the above recursive formula, we get that the total running time is at most $f_1(k,q,|\QQ(\equ)|)\cdot f_u(k,|\QQ(\equ)|,n) \cdot n^{\OO(1)}$ for some function $f_1$. 

Now to solve the problem, first we run ${\cal A}$ on the input graph and get a representative $G^{\star}$ of size at most $r(k,|\QQ(\equ)|)$ for the input graph. Then, we do a brute force computation on $G^{\star}$ and output accordingly. The running time of the algorithm follows from the running time of ${\cal A}$ and the fact that the size of $G^{\star}$  is at most $r(k,|\QQ(\equ)|)$. This completes the proof of the theorem. 
\end{proof}

Next, we explain how to get rid of the assumption that we know the value $r(k,|\QQ(\equ)|)$ in Theorem~\ref{thm:uni:hypo:elim}. First, we notice that if we choose a value $r'$ instead of $r(k,|\QQ(\equ)|)$ and if the algorithm in Lemma~\ref{lem:comp:rep:unbre} succeeds, then the output is a representative of the input graph. If the algorithm terminates without producing an output, then our choice $r'$ for $r(k,|\QQ(\equ)|)$ is wrong.
In that case we say that the algorithm fails.  Now, we run the algorithm in Theorem~\ref{thm:uni:hypo:elim} by substituting values $1,2,3,\ldots$ instead of  $r(k,|\QQ(\equ)|)$ in the order. If for a value $r'$  the algorithm in Lemma~\ref{lem:comp:rep:unbre} fails restart the whole algorithm with the next value  $r'+1$. Clearly we will succeed on or before we reach the value $r(k,|\QQ(\equ)|)$. Thus, we have the following theorem.

\begin{theorem}
\label{thm:uni:elim}
There is a function $f$and  an algorithm for \probEDH running in time $f(k,q,|\QQ(\equ)|)\cdot f_u(k,|\QQ(\equ)|,n) \cdot n^{\OO(1)}$. Here $f_u$ is the running time of the user's algorithm ${\cal A}_u$.  
\end{theorem}

\subsection{Proof of Lemma~\ref{lem:sim:sig}}
\label{sec:lemsimilar}

%\begin{proof}
To prove the lemma it is enough to prove that 
for any graph $H$ such that $(G_1\oplus H,k)$ is a yes-instance of \probEDH,   $(G_2\oplus H,k)$ is also a yes-instance of \probEDH. Towards that let $H$ be an arbitrary graph such that $(G_1\oplus H,k)$ is a yes-instance of \probEDH. That is, $\hhdepth(G_1 \oplus H) \leq k$.  First, let us fix some notations. Let $G_1^{\star}=G_1 \oplus H$, $G_2^{\star}=G_2 \oplus H$, $B_1=\delta(G_1)$, $B_2=\delta(G_2)$, and $D=\Lambda(G_1)=\Lambda(G_2)$. Among all the \heds\ of $G_1^{\star}$, of depth at most $k$ (recall that $\hhdepth(G^{\star}_1) \leq k$), let  $(\wT_1,\widehat{\chi}_1,\widehat{L}_1)$ be an \hed\ of $G_1^{\star}$ with maximum number of vertices from $B_1$ are deleted (i.e., the number of internal nodes in $\wT_1$ that are mapped to subsets of $B_1$ by $\widehat{\chi}$ is maximized).

Now we will derive a state-tuple $s=(D,T,\chi,L,\PP,\eqf)$ in $\rsig(G_1)$ from  $(\wT_1,\widehat{\chi}_1,\widehat{L}_1)$ such that $(G_1,H)$ exactly realizes $s$ through $(\wT_1,\widehat{\chi}_1,\widehat{L}_1)$. 
Recall that $D=\Lambda(G_1)=\Lambda(G_2)$. 
\begin{itemize}
\item[$(i)$]
The rooted forest $T$ is isomorphic to $\wT_1[Z_1]$, where $Z_1=\imp{\wT_1}{G_1}{\widehat{\chi}_1}$. For each node $t\in Z_1$, we use the same notation to represent the node in $T$ that is mapped to $t$ by the isomorphism function. That is, $V(T)=Z_1$. 
%
%Let $Y$ the nodes of $\wT_1$ that are mapped  to subsets containing at least one vertex from $B_1$ by the function by $\widehat{\chi}_1$. Let $Z$ be the set of nodes of $\wT_1$ that belong to the paths from $Y$ to the corresponding roots.  
%The rooted forest $T_1$ is the subgraph of $\wT_1$ induced on $Z$, where $x\in V(T_1)$ is a root vertex if and only if $x$ is a root in $\wT_1$. 
Since $\vert B_1\vert \leq 2k$ and the depth of the rooted forest $\wT_1$ is at most $k$, we have that $\vert V(T)\vert \leq (2k)^2$. 
\item[$(ii)$] Next we define $\chi$. For any node $t$ in $T$, if $\widehat{\chi}_1(t)\cap B_1\neq \emptyset$, then we set $\chi(t)=\lambda_{G_1}(\widehat{\chi}_1(t)\cap B_1)$. For any node $t$ in $T$, if $\widehat{\chi}_1(t)\subseteq V(G_1)\setminus B_1$, then we set $\chi(t)=\present$. For any node $t$ in $T$, if $\widehat{\chi}_1(t)\subseteq V(G_1^{\star})\setminus V(G_1)$, then we set $\chi(t)=\future$. 
\item[$(iii)$]
We define $L=\lambda_{G_1}(B_1\cap \widehat{L}_1)$. From the construction of $T$ and $\chi$, notice that for any leaf node $t$ of $\wT_1$, if $t\in V(T)$, then $\chi(t)\subseteq L$. Let $t_1,\ldots,t_{\ell}$ be the leaves of $\wT_1$ that belong to $V(T)$. Then, $(\chi(t_i))_{i\in [\ell]}$ forms a partition of $L$. We set $\PP=(\chi(t_i))_{i\in [\ell]}$. 
\item[$(iv)$]
Now, for each $i\in [\ell]$, we set $\eqf(\chi(t_i))$ to be the the equivalence class  of $\equ$ that contains  $G_1[\widehat{\chi}_1(t_i)]$. 
\end{itemize}

This completes the construction of the state-tuple $s=(D,T,\chi,L,\PP,\eqf)$. From the construction of $s$, we have that $s\in \rsig(G_1)$  and $(G_1,H)$ exactly realizes $s$ through $(\wT_1,\widehat{\chi}_1,\widehat{L}_1)$ (see Definitions~\ref{def:exact-satisf} and \ref{def:realizes}).

\begin{claim}
\label{clm:sig1s}
$(s,1)$ belongs to $sig_1$. 
\end{claim}

\begin{proof}
Since $s \in \rsig(G_1)$ and $sig_1$ is compatible with $G_1$, we have that at least one among $(s,0)$ and $(s,1)$ belongs to $sig_1$. To prove that $(s,1)\in sig_1$, we prove that $(s,0)$ does not belong to $sig_1$. Since $sig_1$ is a valid marked signature, if $(s,0)\in sig_1$, then there is a pair $(s_1=(D_1,T_1,\chi_1,L_1,\PP_1,\eqf_1),b)\in sig_1$ such that $D_1\setminus L_1$ is a strict super set of $D\setminus L$ and $s_1$ is at least as good as $s$. This implies that since $(G_1,H)$ exactly realizes $s$, $(G_1,H)$ realizes $s_1$ and more labels of the boundary vertices are there in $D_1\setminus L_1$ than in $D\setminus L$. 
That is, there is an \hed\ $(T^{\star},\chi^{\star},L^{\star})$ of $G_1^{\star}$ such that the number of internal nodes in $T^{\star}$ that are mapped to subsets of $B_1$ by ${\chi}^{\star}$ is strictly more than the 
number of internal nodes in $\wT$ that are mapped to subsets of $B_1$ by $\widehat{\chi}$.  This contradicts our choice of the \hed\ $(\wT_1,\widehat{\chi}_1,\widehat{L}_1)$ of $G_1^{\star}$. 
\end{proof}

Since $(s,1)\in sig_1$, and $sig_1$ and $sig_2$ are similar, we have that either $(s,1)\in sig_2$ or 
$(s,0)\in sig_2$. Next we prove that indeed $(s,1)\in sig_2$. 

\begin{claim}
\label{clm:sig2s}
$(s,1)$ belongs to $sig_2$. 
\end{claim}
\begin{proof}
%We know that either $(s,1)\in sig_2$ or $(s,0)\in sig_2$. Suppose $(s,1)\notin sig_2$. Then $(s,0)\in sig_2$. 
%Since $sig_2$ is a valid marked signature, $\ugf$ is transitive, and $(s,0)\in sig_2$, there is a pair $(s_2=(D_2,T_2,\chi_2,L_2,\PP_2,\eqf_2),1)\in sig_2$ such that $D_2\setminus L_2$ is a strict super set of $D\setminus L$ and $s_2$ is at least as good as $s$.  Since $(s_2,1)\in sig_2$, and $sig_1$ and $sig_2$ are similar, we have that either $(s_2,1)\in sig_1$ or $(s_2,0)\in sig_1$. This implies that 
We know that either $(s,1)\in sig_2$ or $(s,0)\in sig_2$. Suppose $(s,1)\notin sig_2$. Then $(s,0)\in sig_2$. 
Since $sig_2$ is a valid marked signature and $(s,0)\in sig_2$, there is a pair $(s_2=(D_2,T_2,\chi_2,L_2,\PP_2,\eqf_2),b)\in sig_2$, for some $b\in \{0,1\}$, such that $D_2\setminus L_2$ is a strict super set of $D\setminus L$ and $s_2$ is at least as good as $s$.  This implies that since $(G_1,H)$ exactly realizes $s$, $(G_1,H)$ realizes $s_2$ and more labels of the boundary vertices are there in $D_2\setminus L_2$ than in $D\setminus L$. As like in the proof of Claim~\ref{clm:sig1s}, this contradicts our choice of the \hed\ $(\wT_1,\widehat{\chi}_1,\widehat{L}_1)$ of $G_1^{\star}$. 
\end{proof}

Because of Claim~\ref{clm:sig2s}, $s\in \rsig(G_2)$. From this we will prove that $G_2^{\star}$ is a yes-instance of \probEDH. Since $s\in \rsig(G_2)$, $G_2$ satisfies $s$. This implies that there is an $\hh$-elimination decomposition $(T',\chi',L')$ of $G_2$ with the following properties.

\begin{itemize}
%\item[(a)]  $T$ is an induced subgraph of $T'$ such that for any $t\in V(T)$, all the vertices in the unique path in $T'$ from $t$ to the root belong to $V(T)$. Moreover there is an induced subgraph isomorphism $\phi \colon V(T)  \rightarrow V(T')$ with the properties mentioned below. 
\item[(a)] $T$ is isomorphic to $T'[Z_2]$, where $Z_2=\imp{T'}{G_2}{{\chi}'}$. For any node $t\in Z_2$, we use the same notation to represent the node in $T$ that maps to $t$ by the isomorphism function. That is, $V(T)=Z_2$.      
\item[(b)] For any node~$t$ of~$T$ with ${\chi}'(t)\cap B_2\neq \emptyset$, 
$\chi(t)=\lambda_G({\chi}'(t)\cap B_2)$
%$\chi(t)\in \PP\cup \{\{x\}\colon x\in (D\setminus L) \}$  we have that $\lambda_G(\chi'(t)\cap \Lambda(G))= \chi(t)$. 
\item[(c)] For any node $t$ of~$T$ with ${\chi}'(t)\neq \emptyset$ and 
${\chi}'(t)\subseteq V(G_2)\setminus B_2$, $\chi(t)=\present$. 
\item[(d)] Let $d_1,\ldots,d_{\ell'}$ be the leaf nodes of $T'$ that belong to $Z_2$. Then, $\PP=(\chi(d_i))_{i\in [\ell']}$. 
\item[(e)] For each node~$t$ of~$T$ with $\chi(t)\in \PP$, the equivalence class $Q\in \QQ(\equ)$ that contains $G_2[{\chi}'(t)]$  is at least as good as $\eqf(\chi(t))$. That is, $\ugf(Q,\eqf(\chi(t))=1$. 
\end{itemize}

%\begin{itemize}
%\item[(a)]  $T_1$ is an induced subgraph of $T'$ and  there  induced subgraph isomorphism $\phi \colon V(T_1)  \rightarrow V(T')$ with the properties mentioned below. 
%%\item For each node~$t$ of~$T$ with $\chi(t)\in $ we have that  $\lambda_G(\chi'(t))= \chi(t)$. 
%\item[(b)] For each node~$t$ of~$T_1$ with $\chi_1(t)\in \PP_1\cup \{\{x\}\colon x\in (D_1\setminus L_1) \}$ 
%we have that $\lambda_{G_2}(\chi'(t)\cap \Lambda(G_2))= \chi_1(t)$. 
%\item[(c)] For each node $t$ of~$T$ with $\chi'(t)\neq \emptyset$ and $\chi'(t)\subseteq V(G_2)\setminus B_2$, we have that $\chi_1(t)=P$. 
%\item[(d)] For each node~$t$ of~$T$ with $\chi_1(t)\in \PP_1$, we have that $G_2[\chi'(t)]$ is in an equivalence class $Q$ of $\equ$ which is at least as good as $Q'=\eqf(\chi_1(t))$. That is, $\ugf((Q,Q'))=1$. 
%\end{itemize} 

Next we construct an $\hh$-elimination decomposition $(\wT_2,\widehat{\chi}_2,\widehat{L}_2)$ of $G^{\star}_2$ from $s$, $(\wT_1,\widehat{\chi}_1,\widehat{L}_1)$ and  $(T',\chi',L')$, and prove that that indeed $(\wT_2,\widehat{\chi}_2,\widehat{L}_2)$ is an  $\hh$-elimination decomposition  of $G^{\star}_2$, of depth at most $k$.  
\begin{itemize}
\item First we will construct the rooted forest $\wT_2$. Recall that $Z_1=\imp{\wT_1}{G_1}{\widehat{\chi}_1}$,  $Z_2=\imp{T'}{G_2}{{\chi}'}$, and $V(T)=Z_1=Z_2$. Moreover, $V(T)\subseteq V(\wT_1)$ and $V(T)\subseteq V(T')$. The vertex set of $\wT_2$ is $V(T')\uplus (V(\wT_1)\setminus V(T))$. The edge set of $\wT_2$ is the union of the edge sets of $T'$ and $\wT_1$. All the roots of $T'$ and $\wT_1$ are the roots of $\wT_2$. From the construction of $\wT_2$, it is easy to prove that $\wT_2$ is a forest. For any leaf node $t$ of $\wT_2$,  the unique path in $\wT_2$ from $t$ to the root, is a leaf to root path either in $T'$ or in $\wT_1$. Therefore, the depth of the forest $\wT_2$ is at most $k$ because the depths of $T'$ and $\wT_1$ are at most $k$. 

\item  Next we define $\widehat{\chi}_2$. 
\begin{itemize}
\item For any internal node $t$ of $\wT_2$, if $t\in V(T')$ and $\chi'(t)\neq \emptyset$, then we set $\widehat{\chi}_2(t)=\chi'(t)$. 

\item For any internal node $t$ of $\wT_2$, if $t\in V(\wT_1)$ and $\widehat{\chi}_1(t)\subseteq V(G_1^{\star})\setminus V(G_1)$, then we set $\widehat{\chi}_2(t)=\widehat{\chi}_1(t)$. Notice that because of conditions $(ii)$, $(b)$, and $(c)$, for such nodes $t$, we have that $\chi'(t)=\emptyset$. 

\item 
Now let $t$ be a leaf node of $\wT_2$. If $t$ is a leaf node in $T'$, then let $U'=\chi'(t)$. If $t$ is a leaf node in $\wT_1$, then let $U_1=\widehat{\chi}_1(t)\setminus V(G_1)$. Then, we set $\widehat{\chi}_2(t)=U'\cup U_1$. 

\item For all other nodes $t$ in $\wT_2$ (which are not considered above), we set $\widehat{\chi}_2(t)=\emptyset$. 
%The rooted forests  $T'$ and $\widehat{T}_1$ are induced subgraph of $\wT_2$. 
%This completes the definition $\widehat{\chi}_2$. 

\end{itemize}

\item Next, we define $\widehat{L}_2= (V(H)\cap \widehat{L}_1)\cup (V(G_2)\cap L')$. 

\end{itemize}

Next, we prove that $(\wT_2,\widehat{\chi}_2,\widehat{L}_2)$ is indeed an $\hh$-elimination decomposition 
of  $G^{\star}_2$ of depth at most $k$. Since we have already proved that the depth of $\wT_2$ is at most $k$,  the depth of the decomposition $(\wT_2,\widehat{\chi}_2,\widehat{L}_2)$ is at most $k$. Now we prove that 
$(\wT_2,\widehat{\chi}_2,\widehat{L}_2)$ satisfies all the conditions of Definition~\ref{def:eldeco}. 

Let $t$ be an internal node in $\wT_2$ and suppose that $\widehat{\chi}_2(t)\neq \emptyset$. Then, exactly one of the following is true. 
\begin{itemize}
\item $t$ is an internal node in $T'$ and $\widehat{\chi}_2(t)=\chi'(t)$.
\item $t$ is an internal node in $\wT_1$ and $\widehat{\chi}_2(t)=\widehat{\chi}_1(t)$. 
\end{itemize} 
Thus, since $(\wT_1,\widehat{\chi}_1,\widehat{L}_1)$ and  $(T',\chi',L')$ are $\hh$-elimination decompositions of $G^{\star}_1$ and $G_2$, we have that $\vert \widehat{\chi}_2(t)\vert \leq 1$ and $\widehat{\chi}_2(t)\subseteq V(G)\setminus \widehat{L}_2$. Therefore, condition $(1)$ of Definition~\ref{def:eldeco} is satisfied. 

From the construction of  $\widehat{\chi}_2$, notice that for any $v\in V(G_2^{\star})$, there is a node  $t\in V(\wT_2)$ such that $v\in \widehat{\chi}_2(t)$. Moreover, since $(\chi'(t))_{t\in V(T')}$ is a partition of $V(G_2)$ and   $(\widehat{\chi}_1(t)\cap (V(G^{\star}_1)\setminus V(G_1)))_{t\in V(\wT_1)}$ forms a partition of  $V(G^{\star}_1)\setminus V(G_1)$, we have that $(\widehat{\chi}_2(t))_{t\in V(\wT_2)}$ forms a partition of $V(G^{\star}_2)$. Thus, condition $(2)$ of Definition~\ref{def:eldeco} is satisfied. 

Now we prove that $(\wT_2,\widehat{\chi}_2,\widehat{L}_2)$ satisfies condition $(3)$ of Definition~\ref{def:eldeco}. Let $uv\in E(G^{\star})$ be an arbitrary edge. Let $u\in \widehat{\chi}_2(t)$ and $u\in \widehat{\chi}_2(t')$.  Suppose $uv\in E(G_2)$. Then, since $(T',\chi',L')$ is an $\hh$-elimination decomposition of $G_2$ and $T'$ is a subgraph of $\wT_2$, we have that $t$ and $t'$ are in ancestor-descendant  
relation of $\wT_2$. Suppose $uv\in E(H)$. Then, $t$ and $t'$ are also belong to $\wT_1$. Since $\wT_1$ is a subgraph of $\wT_2$ and  $(\wT_1,\widehat{\chi}_1,\widehat{L}_)$ is an $\hh$-elimination decomposition of $G_1^{\star}$, we have that $t$ and $t'$ are in ancestor-descendant  relation of $\wT_2$.

Next we prove that $(\wT_2,\widehat{\chi}_2,\widehat{L}_2)$ satisfies condition $(4)$ of Definition~\ref{def:eldeco}. That is, for each leaf node $t$ of $\wT_2$, $G_2^{\star}[\widehat{\chi}_2(t)]$ belongs to $\hh$. 
If $\widehat{\chi}_2(t)\subseteq V(G_2)$, then $G_2^{\star}[\widehat{\chi}_2(t)]=G_2[\widehat{\chi}_2(t)]$ is a base component in $(T',\chi',L')$ and hence it belongs to $\hh$. 
If $\widehat{\chi}_2(t)\subseteq V(H)$, then $G_2^{\star}[\widehat{\chi}_2(t)]=H[\widehat{\chi}_2(t)]$ is an induced subgraph of a base component in $(\wT_1,\widehat{\chi}_1,\widehat{L})$ and hence it belongs to $\hh$ because $\hh$ is a hereditary family.  

Now, we are in the case when  $\widehat{\chi}_2(t)\cap (V(G_2)\setminus B_2)\neq \emptyset$ and $\widehat{\chi}_2(t)\cap (V(G_2^{\star})\setminus V(G_2))\neq \emptyset$. Let $H'=H[\widehat{\chi}_2(t)\cap V(H)]$ and $G'_2=G_2[\widehat{\chi}_2(t)\cap V(G_2)]$. Then $G_2'\oplus H'$ is the graph $G_2^{\star}[\widehat{\chi}_2(t)]$. If $B_2\cap \widehat{\chi}_2(t)=\emptyset$, then by the construction of $\wT_2$, 
either $V(G_2')=\emptyset$ or $V(H')=\emptyset$ which is a contradiction to the assumption that 
$\widehat{\chi}_2(t)\cap (V(G_2)\setminus B_2)\neq \emptyset$ and $\widehat{\chi}_2(t)\cap (V(G_2^{\star})\setminus V(G_2))\neq \emptyset$. So we assume that 
$B_2\cap \widehat{\chi}_2(t)\neq\emptyset$. This implies that $t$ is a leaf node in $T'$ as well as a leaf node in $\wT_1$. Let $G_1'=G_1[\widehat{\chi}_1(t)\cap V(G_1)]$. Notice that  $G_1'\oplus H'$ is a base component in $(\wT_1,\widehat{\chi}_1,\widehat{L}_1)$, and hence $G_1'\oplus H'$ belongs to $\hh$.  Since $(G_1,H)$ exactly realizes $s$ through  $(\wT_1,\widehat{\chi}_1,\widehat{L}_1)$, we have that $G_1'$ belongs to the equivalence class $\eqf_1(\widehat{\chi}_1(t))$. Since $G_2$ satisfies $s$, we have that $G_2'$ belongs to an equivalence class which is at least as good as $\eqf_1(\chi_1(t))$.  This implies that $G_2'\oplus H'$ belongs to $\hh$ because $G_1'$ belongs to the equivalence class $\eqf_1(\widehat{\chi}_1(t))$ and $G_1'\oplus H'$ belongs to $\hh$. 

This completes the proof of the lemma. 
%\end{proof}

\input{compute_sig}

%% file: compute_sig.tex
% !TEX root = main.tex
\subsection{Proof of Lemma~\ref{lem:comp:sig}}
\label{sec:proof:lem:comp:sig}

The proof of Lemma~\ref{lem:comp:sig} is identical to the proof of Lemma~\ref{lem:an_un}, where we use the algorithm ${\cal A}_u$ instead of $\SC{M}_{\sf mod}$. That is, we have a branching algorithm where the main computation boils down to  executions of ${\cal A}_u$ for each state tuple $s$. Here, the algorithm is notationally cumbersome compared to the proof of Lemma~\ref{lem:an_un}. Throughout the section $G$ is a $(q,k)$-unbreakable graph.

If $\vert V(G)\vert \leq 3q+k$, then we do a brute force computation. Otherwise, 
to prove the lemma we do the following. For each state-tuple $s=(D,T,\chi,L,\PP,\eqf)$, either we identify that $s\in \rsig(G)$ or we identify a state-tuple $\tilde{s}=(\tilde{D},\tilde{T},\tilde{\chi},\tilde{L},\tilde{\PP},\tilde{\eqf}) \in \rsig(G)$ such that $\tilde{s}$ is strictly better than $s$. 
   %$\tilde{D}\setminus \tilde{L}$ is a strict super set of $D\setminus L$, 
   or we may ``fail''. If we identify that  $s\in \rsig(G)$, then we include $(s,1)$ in the output signature. If we identify that $\tilde{s}\in \rsig(G)$, then we include $(s,0)$ and $(\tilde{s},1)$ in the output signature. If we fail to identify $s$ or $\tilde{s}$, then we prove that indeed $s\notin \rsig(G)$. 
This is formalized in the following lemma. 

\begin{lemma}
\label{lem:sigletest}
Suppose $\vert V(G)\vert >3q+k$. There is a function $g_1$ and an algorithm {\sigtest} that given a graph $G$ and  a state-tuple $s$, runs in time $g_1(q,k)\cdot f_{u}(k,|\QQ(\equ)|,n)\cdot n^{\OO(1)}$,
 and outputs exactly one of the following: $(i)$ $(s,1)$, $(ii)$ $(s,0)$ and $(\tilde{s}=({D},\tilde{T},\tilde{\chi},\tilde{L},\tilde{\PP},\tilde{\eqf}),1)$, or $(iii)$ {\sf Fail}. The algorithm has the following properties. 

\begin{itemize}
\item[(a)] If the output is $(s,1)$, then $s\in \rsig(G)$.
\item[(b)] If the output is $(s,0)$ and $(\tilde{s},1)$, then $\tilde{s}\in \rsig(G)$,
%$D\setminus L \subset \tilde{D}\setminus \tilde{L}$, 
and $\tilde{s}$ is strictly better than $s$.
\item[(c)] If  $s\in \rsig(G)$, then the output is either $(i)$ or $(ii)$. 
\end{itemize}  
\end{lemma}

First assuming Lemma~\ref{lem:sigletest},  we prove Lemma~\ref{lem:comp:sig}.

\begin{proof}[Proof of Lemma~\ref{lem:comp:sig}]
Initially we set $sig=\emptyset$. For each state-tuple $s$, we run {\sigtest}. If the output is $(s,1)$, then we add $(s,1)$ to $sig$. If the output is $(s,0)$ and $(\tilde{s},1)$, then we add  $(s,0)$ and $(\tilde{s},1)$ to $sig$. Finally we output $sig$. The correctness of the algorithm follows from the correctness of {\sigtest}. Since the number of state-tuples is bounded by $k$ and $\vert \QQ(\equ)\vert$, by Lemma~\ref{lem:sigletest}, the running time of the algorithm follows. 
\end{proof}

\subsubsection{Proof of Lemma~\ref{lem:sigletest}}

%The rest of the section is devoted to prove Lemma~\ref{lem:sigletest}. 
First we give an overview of our algorithm {\sigtest}. Let $s\in \rsig(G)$. To identify that $s$ is indeed belongs to $\rsig(G)$ we have to test the existence of an $\hh$-elimination decomposition $(\wT,\widehat{\chi},\widehat{L})$ of $G$, of depth at most $k$, such that $G$ satisfies $s$ through $(\wT,\widehat{\chi},\widehat{L})$. Since $G$ is $(q,k)$-unbreakable and $\vert V(G)\vert > 3q+k$, by
Lemma~\ref{lem:boundedsep} we know that there is exactly one connected component $C^{\star}$ in $G[\widehat{L}]$ that has at least $q$ vertices and $|V(G)\setminus V(C^{\star})|\leq q+k$. This implies that 
$|V(\wT)|\leq q+k$. Thus, we can guess the rooted forest $\wT$ and the leaf node $t^{\star}$ in $\wT$ such that $V(C^{\star})\subseteq \widehat{\chi}(t^{\star})$. Let $EQ^{\star}$ be the equivalence class mentioned in the state-tuple corresponding to the leaf node $t^{\star}$ (assuming $C^{\star}$ contains some boundary vertices). Then $\widehat{\chi}(t^{\star})$ belongs to an equivalence class which is at least as good as $EQ^{\star}$.   Using  branching rules we {\em almost} identify the values $\widehat{\chi}(t)$ for all the nodes $t$ except the nodes in the unique path $P^{\star}$ from $t^{\star}$ to the root.  Finally we will have a connected component $D^{\star}$ such that $N(D^{\star})\subseteq \bigcup_{t\in V(P^{\star})} \widehat{\chi}(t)$. In this step we run the user's algorithm ${\cal A}_u$ to identify the vertices $S$ of $D^{\star}$ to be placed in $\widehat{\chi}(t)$ for the internal nodes of $T$   in the path from $t^{\star}$ to the root such that 
$\chi(t)=\present$ and $D^{\star}-S$ belongs to an equivalence class which is at least as good  $EQ$. If $S$ does not contain any vertex from $\delta(G)$, then we identify a witness (i.e., the $\hh$-elimination decomposition $(\wT,\widehat{\chi},\widehat{L})$) for the fact that $s\in \rsig(G)$. In that case the algorithm {\sigtest} outputs $(s,1)$. If $S\cap \delta(G)\neq \emptyset$, then we get a witness for another state-tuple $\tilde{s}$ and we prove that $\tilde{s}$ is at least as good as $s$. In that case the algorithm {\sigtest} outputs $(s,0)$ and $(\tilde{s},1)$. If ${\cal A}_u$ outputs {\sf No}, then {\sigtest} outputs {\sf Fail}. 

%In the rest of the section 
Now we formally prove Lemma~\ref{lem:sigletest}.  We start by defining  the notion of a partial solution. 

\begin{definition}[Partial solution] Given a graph $G$ and a state-tuple $s=(D,T,\chi,L,\PP,\eqf)$, a partial solution of $(G,s)$ is a tuple $e=(D,T^{\star}, {\chi}^{\star},{L}^{\star})$ 
where~$T^{\star}$ is a rooted forest of depth at most $k$ and~${\chi}^{\star} \colon V(T^{\star}) \to 2^{V(G)}$ and $L^{\star} \subseteq V(G)$, such that the following properties hold. 
%\begin{itemize}
%\item . 
%\item 
%\end{itemize}

\begin{itemize}
%\item[(a)]  $T$ is an induced subgraph of $T'$ such that for any $t\in V(T)$, all the vertices in the unique path in $T'$ from $t$ to the root belong to $V(T)$. Moreover there is an induced subgraph isomorphism $\phi \colon V(T)  \rightarrow V(T')$ with the properties mentioned below. 
\item[(a)] $\Lambda(G)=D$ and $\lambda_G(L^{\star}\cap \delta(G))= L$. 
\item[(b)] For each internal node~$t$ of~$T^{\star}$ we have~$|{\chi}^{\star}(t)| \leq 1$ and ${\chi}^{\star}(t)\subseteq V(G)\setminus L$.  
\item[(c)] The sets~$({\chi}^{\star}(t))_{t \in V(T^{\star})}$ form a partition of a subset of $V(G)$.
\item[(d)] $T$ is isomorphic to $T^{\star}[Z]$, where $Z=\imp{T^{\star}}{G}{{\chi}^{\star}}$. For any node $t\in Z$, we use the same notation to represent the node in $T$ that maps to $t$ by the isomorphism function. That is, $V(T)=Z$.   
\item[(e)] For any node~$t$ of~$T$ with ${\chi}^{\star}(t)\cap \delta(G)\neq \emptyset$, 
$\chi(t)=\lambda_G({\chi}^{\star}(t)\cap \delta(G))$
%$\chi(t)\in \PP\cup \{\{x\}\colon x\in (D\setminus L) \}$  we have that $\lambda_G(\chi'(t)\cap \Lambda(G))= \chi(t)$. 
\item[(f)] For each node $t$ of~$T$ with ${\chi}^{\star}(t)\neq \emptyset$ and 
${\chi}^{\star}(t)\subseteq V(G)\setminus \delta(G)$, $\chi(t)=\present$. 
\item[(g)] Let $t_1,\ldots,t_{\ell}$ be the leaf nodes of $T^{\star}$ that belong to $Z$. Then, $\PP=(\chi(t_i))_{i\in [\ell]}$. 
\item[(h)] For any two nodes $t$ and $t'$ in $T^{\star}$ such that they are not in an ancestor-descendant relationship, there is no edge in $G$ between a vertex in $\chi^{\star}(t)$ and $\chi^{\star}(t')$. 
%\item[(e)] For each node~$t$ of~$T$ with $\chi(t)\in \PP$, the equivalence class $Q\in \QQ(\equ)$ that contains $G[\widehat{\chi}(t)]$,  is at least as good as $\eqf(\chi(t))$. That is, $\ugf(Q,\eqf(\chi(t))=1$. 
\end{itemize} 
\end{definition}

For a partial solution $e=(D,T^{\star}, {\chi}^{\star},{L}^{\star})$ of $(G,s)$, we use $\used{\chi^{\star}}{G}$ to denote the set $\bigcup_{t\in V(T^{\star})} \chi^{\star}(t)$. That is, in the partial solution $e$ we have identified the places of $\used{\chi^{\star}}{G}$. The objective is to  keep including vertices from $V(G)\setminus \used{\chi^{\star}}{G}$ to the sets $\{\chi^{\star}(t)\}_{t\in V(T^{\star})}$ through branching rules.

Recall that $(G,s=(D,T,\chi,L,\PP,\eqf))$  is the input of the algorithm {\sigtest}. To understand the steps of the algorithm let us assume that $s\in \rsig(G)$ and $(\wT,\widehat{\chi},\widehat{L})$ be an $\hh$-elimination decomposition of $G$, of depth at most $k$, such that $G$ satisfies $s$ through $(\wT,\widehat{\chi},\widehat{L})$.  
As mentioned earlier, since $G$ is $(q,k)$-unbreakable and $\vert V(G)\vert > 3q+k$, by
Lemma~\ref{lem:boundedsep}, there is exactly one connected component $C^{\star}$ in $G[\widehat{L}]$ that has at least $q$ vertices and $|V(G)\setminus V(C^{\star})|\leq q+k$ and hence $|V(\wT)|\leq q+k$. Therefore, as a first step we guess the tree  $\wT$, all the nodes in $t$ such that $\widehat{\chi}(t)$ contains at least one vertex from $\delta(G)$, $\widehat{\chi}(t)\cap \delta(G)$, and the node $t^{\star}$ such that $\widehat{\chi}(t^{\star})=C^{\star}$.  Since  $\vert \delta(G)\vert \leq 2k$ and  $|V(\wT)|\leq q+k$ the number of choices for this guess is bounded by $(q+k)^{\OO(q+k)}$. Thus, we assume that initially we have a partial solution $e=(D,T^{\star}, {\chi}^{\star},{L}^{\star})$, where $T^{\star}=\wT$ with the following properties. (Recall that $T$ is an induced subgraph of $T^{\star}$). 

\begin{itemize}
\item For any node $t\in V(T)$, $\chi^{\star}(t)\subseteq \delta(G)$ and  $\chi(t)=\lambda_G(\chi^{\star}(t)\cap \delta(G))$.
\item $\delta(G)=\bigcup_{t\in V(T^{\star})} \chi^{\star}(t)$.  
\end{itemize}

We remind that we also know the node $t^{\star}\in V(T^{\star})=V(T)$ such that the vertices of the ``unknown'' large component $C^{\star}$ belong to $\widehat{\chi}(t^{\star})$. Now on, we assume that our instances contain $(G,s,e,t^{\star})$ where $e$ is a partial solution and $t^{\star}$ is the special leaf node in the rooted forest $T^{\star}$ of~$e$. The objective is to construct an $\hh$-elimination decomposition $(\wT,\widehat{\chi},\widehat{L})$ of $G$ (we call it as solution) {\em obeying} the partial solution $e$ such that either $G$ satisfies $s$ through $(\wT,\widehat{\chi},\widehat{L})$ or $G$ satisfies $\tilde{s}$ (derived from the solution) through $(\wT,\widehat{\chi},\widehat{L})$ with the property that $\tilde{s}$ is strictly better than $s$. We formally define when a solution obeys a partial solution. 

\begin{definition}
Let $s=(D,T,\chi,L,\PP,\eqf)$ be a state-tuple and $e=(D,T^{\star}, {\chi}^{\star},{L}^{\star})$ be a partial solution. We say that a solution $(\wT,\widehat{\chi},\widehat{L})$ obeys the partial solution $e$ if $\wT$ is isomorphic to $T^{\star}$ (we use the same vertex to denote its image in the isomorphism function, i.e., $V(\wT)=V(T^{\star})$) and for each $t\in V(\wT)=V(T^{\star})$, $\chi^{\star}(t)\subseteq \widehat{\chi}(t)$. 
\end{definition}

Next we explain about our branching rules. For each instances created in the branching rule, at least for one internal node $t\in V(T^{\star})$ with $\chi^{\star}(t)=\emptyset$, we will have a assigned a vertex from $V(G)$ to it.  That is, in the new instance  $(G,s,e'=(D,T^{\star}, {\chi}',L'),t^{\star})$ created there is at least one new internal node $t\in V(T^{\star})$ with $\chi^{\star}(t)=\emptyset$ and $\chi'(t)\neq \emptyset$. Since $\vert V(T^{\star})\vert \leq q+k$ and for each internal node $t$, $\vert \chi'(t)\vert \leq 1$, we use the number of internal nodes $t$ with $\chi^{\star}(t)= \emptyset$ as a measure for our branching algorithm. In our branching rules, the measure decreases by at least one after the application of each branching rule. Next we explain our branching rules. A branching step creates more instances. If for any of the instance the output is $(s,1)$, then the output of our algorithm is $(s,1)$. If this is not the case and at least one of them outputs $(s,0)$ and $(\tilde{s},1)$, then the output of our algorithm is $(s,0)$ and $(\tilde{s},1)$. If for all the instances, the output is {\sf Fail}, then we output {\sf Fail}. 

Our first branching rule attempt to fix ``neighborhood'' inconsistencies, below we begin by explaining what we mean by an inconsistency that our first branching rules attempts to fix. Suppose there are (not necessarily distinct) vertices $u,u'$ in the same connected component of $G'=G-\used{\chi^{\star}}{T^{\star}}$ and two nodes $t$ and $t'$ in $T^{\star}$ such that $t$ and $t'$ are not in ancestor-descendant relationship, and
$u$ and $u'$ have neighbors in $\chi^{\star}(t)$ and $\chi^{\star}(t')$, respectively. Then consider any $u-u'$ path $P$ in $G'$ (which exists, since $u$ and $u'$ are in the same connected component of $G'$). We remark that if $u=u'$, then $P$ is a path on one vertex. For the above, we say that $(u,u',P)$ is an {\em inconsistent triplet}.

We will observe that for any solution $(\wT,\widehat{\chi},\widehat{L})$ obeying the partial solution $e$, $V(P) \cap (V(G)\setminus L) \neq \emptyset$. In fact, we can establish a stronger statement than the above, which guarantees that at least one among the first and the last $q$ many vertices in $P$, belongs to $V(G)\setminus L$. The above property is obtained because at least one of $u$ or $u'$ cannot belongs to the unique connected component of size at least $q$ in $G[L]$, for any solution $(\wT,\widehat{\chi},\widehat{L})$ (see Lemma~\ref{lem:boundedsep}). The above leads us to the following branching rule.

\medskip
\noindent{\bf Branching Rule 1:} Let $(u,u',P)$ is an inconsistent triplet in $(G,s,e=(D,T^{\star}, {\chi}^{\star},{L}^{\star}),t^{\star})$. Then for each  vertex $v$  in the first $q$ vertices or in the last $q$ vertices in the path $P$ and for each internal node $t\in V(T^{\star})$ such that $\chi^{\star}(t)=\emptyset$ and $\chi(t)=\present$, construct a new instance $(G,s,e'=(D,T^{\star}, {\chi}',{L}^{\star}),t^{\star})$ if $e'$ is a partial solution for $(G,s)$ and solve it. Here $\chi'(t)=\{v\}$ and $\chi'(t_1)=\chi^{\star}(t_1)$ for all $t_1 \neq t$.

The correctness of the branching rule follows from Lemma~\ref{lem:boundedsep}. Notice that if Branching Rule 1 is not applicable, then there is no ``crossing edge'' inconsistencies. That is, let $(G,s,e=(D,T^{\star}, {\chi}^{\star},{L}^{\star}),t^{\star})$ be an instance such that Branching Rule 1 is not applicable. Then for any connected component $C$ of $G-\used{\chi^{\star}}{G}$, there is a leaf node $t$ in $T^{\star}$ such that for any vertex $w\in N_G(V(C))$, there is an ancestor $t'$ of $t$ ($t'$ maybe equal to $t$) with $w\in \chi^{\star}(t')$.

Next, we do a branching rule similar to the Branching Rule~1 in  the proof of Lemma \ref{lem:one-implies-5-9}. 
Recall that for a graph $G$ and integers $p,q'\in \mathbb{N}$, a set $B \subseteq V({G})$ is a {\em $(p,q')$-connected set} in ${G}$, if $\what{G}[B]$ is connected, $|B| \leq p$ and $|N_{{G}}(B)|\leq q'$. A $(p,q')$-connected set $B$ in ${G}$ is {\em maximal} if there does not exist another $(p,q')$-connected set $B^*$ in ${G}$, such that $B \subset B^*$.  Using Proposition~\ref{lem:sep} we can enumerate all maximal $(p,q')$-connected sets in time $2^{p+q'}\cdot n^{\OO(1)}$.  

Since the the family $\hh$ is closed under disjoint union, if there is a graph $H\notin \hh$, then there is a connected component $H'$ of $H$ such that $H'\notin \hh$. Let $(G,s,e=(D,T^{\star}, {\chi}^{\star},{L}^{\star}),t^{\star})$ be an instance and let $(\wT,\widehat{\chi},\widehat{L})$ be a hypothetical solution obeying the partial solution $e$. Let $C$ be a maximal $(q+k,k)$-connected set in $G$ such that $G[C'] \notin \hh$ where $C'=V(C) \setminus \used{\chi^*}{G}]$. 
Then by the hereditary property of $\hh$, there exists a 
%subset $C''\subseteq C'$ such that 
connected component $F$  in $G[C']$ such that $F\notin \hh$. 
Then observe that for the solution $(\wT,\widehat{\chi},\widehat{L})$ obeying the partial solution $e$, $V(F)\cap (V(G)\setminus \widehat{L}) \neq \emptyset$. This leads to the following branching rule. Let ${\cal C}$ be the set of all maximal $(q+k,k)$-connected set in $G$.

\medskip
\noindent{\bf Branching Rule 2:}  Let $(G,s,e=(D,T^{\star}, {\chi}^{\star},{L}^{\star}),t^{\star})$ be an instance. Let $C\in {\cal C}$ such that $G[C'] \notin \hh$ where $C'=Q\setminus \used{\chi^*}{G}$. 
Let $F$ be a connected component in $G[C']$ such that  $F\notin \hh$.
Then for each $v\in V(F)$ and for each internal node $t\in V(T^{\star})$ such that $\chi^{\star}(t)=\emptyset$ and $\chi(t)=\present$, construct a new instance $(G,s,e'=(D,T^{\star}, {\chi}',{L}^{\star}),t^{\star})$ if $e'$ is a partial solution for $(G,s)$ and solve. Here $\chi'(t)=\{v\}$ and $\chi'(t_1)=\chi^{\star}(t_1)$ for all $t_1 \neq t$. 

Towards the correctness of Branching Rule 2, notice that $V(F)$ cannot be a subset of $\widehat{\chi}(t')$ for any leaf node $t'$ because $F\notin \hh$. Now suppose $V(F)\subseteq \bigcup_{t'\in R}\widehat{\chi}(t')$ where $R$ is the set of leaf nodes of $\wT$.  Then, since $F$ is connected there exists a two distinct leaf nodes  
$t_1$ and $t_2$ such that there is an edge in $G$ between a vertex in  $\widehat{\chi}(t_1)$ and a vertex in $\widehat{\chi}(t_2)$. This contradicts the fact  that $(\wT,\widehat{\chi},\widehat{L})$ is an $\hh$-elimination decomposition. 

Now, we prove the following lemma.

\begin{lemma}
\label{lem:uniqueuni}
Suppose Branching Rule 1 is not applicable for an instance $(G,s,e=(D,T^{\star}, {\chi}^{\star},{L}^{\star}),t^{\star})$. Let $Y=\bigcup_{t'\in R}{\chi}^{\star}(t')$ where $R$ is the set of leaf nodes of $T^{\star}$. 
Then there is a unique connected component $J^*$ in $G-(\used{\chi^{\star}}{G} \setminus Y)$ that has at least $q$ vertices, and $|V(G) \setminus V(J^*)| \leq q+k$.
\end{lemma}

\begin{proof}
Since Branching Rule 1 is not applicable for any connected component $C$ in $G-(\used{\chi^{\star}}{G} \setminus Y)$, there is a leaf node $t_C$ in $T^{\star}$ such that for any vertex $w\in N_G(C)$, there is a node $t'$ in the unique path in $T^{\star}$ from $t_C$ to the root such that $w\in \chi^{\star}(t_C)$. 

 Let ${\cal G}$ be the set of all graphs.  Now we construct a ${\cal G}$-elimination decomposition $(T^{\star},\chi_1,L_1)$ from $e$ as follows. 
Initially set $\chi_1(t)=\chi^{\star}(t)$ for all $t\in V(T^{\star})$. For each connected component $C$ in $G-(\used{\chi^{\star}}{G} \setminus Y)$, add $V(C)$ to $\chi_1(t_C)$. Then $(T^{\star},\chi_1,L_1)$ is a ${\cal G}$-elimination decomposition. Now the lemma follows by applying Lemma~\ref{lem:boundedsep} on  $(T^{\star},\chi_1,L_1)$. 
\end{proof}

Now let $(G,s,e=(D,T^{\star}, {\chi}^{\star},{L}^{\star}),t^{\star})$ be an instance such that Branching Rules 1 and $2$ are not applicable. Let $(\wT,\widehat{\chi},\widehat{L})$ be a hypothetical solution obeying the partial solution $e$. Let $Y=\bigcup_{t'\in R}{\chi}^{\star}(t')$ where $R$ is the set of leaf nodes of $T^{\star}$. 
Then by Lemma~\ref{lem:uniqueuni}, there is a unique connected component $J^*$ in $G-(\used{\chi^{\star}}{G} \setminus Y)$ that has at least $q$ vertices, and $|V(G) \setminus V(J^*)| \leq q+k$. Now since $|V(G) \setminus V(J^*)|\leq q+k$, for each vertex $v\in V(G) \setminus V(J^*)$, we can guess the node $t_v\in V(T^{\star})=V(\wT)$ such that $v\in \widehat{\chi}(t_v)$. Thus, according to our guess we update the function $\chi^{\star}$. That is, we update $\chi^{\star}(t_v):=\chi^{\star}(t_v)\cup \{v\}$. 
%Thus after doing this update we have the following property. For each node $t\neq t^{\star}$, $\chi^{\star}(t)\subseteq \widehat{\chi}(t)$. 
Also notice that if there is a vertex $u\in V(J^{\star})$, an edge $uw\in E(G)$ and a node $t\in V(T^{\star})$ 
such that $w\in \chi^{\star}(t)$, and $t^{\star},t$ are not in an ancestor-descendant relation, then clearly 
there is an internal node $t_u$ in $T^{\star}$ such that $\widehat{\chi}(t_u)=\{u\}$. Thus, we also guess the node $t_u$ for such vertices.  This implies that after these steps for any vertex $x\in V(G)\setminus \used{\chi^{\star}}{G}$, and any edge $xy\in E(G)$ either $y\in V(G)\setminus \used{\chi^{\star}}{G}$ or there is a node $t$ in the unique path $P^{\star}$ in $T^{\star}$ from $t^{\star}$ to the root such that $y\in \chi^{\star}(t)$.

Next, we prove that at this stage by running one execution of the user defined function ${\cal A}_u$ we get the output $(i)$ $(s,1)$ or $(i)$ $(s,0)$ and $(\tilde{s},1)$. 
  
 \begin{lemma}
 \label{lem:finaldeletion}
 Let $I=(G,s,e,t^{\star})$ be an instance such that Branching Rules $1$ and $2$ are not applicable, where $s=(D,T,\chi,L,\PP,\eqf)$ and $e=(D,T^{\star}, {\chi}^{\star},{L}^{\star})$.  Let $(\wT,\widehat{\chi},\widehat{L})$ be a hypothetical solution obeying the partial solution $e$.  Let $P^{\star}$ be the unique path in $T^{\star}$ from  $t^{\star}$ to the root. Let $k'$ be the number of nodes $t$ in  $V(P^{\star})$  such that $\chi(t)=\present$ and $\chi^{\star}(t)=\emptyset$. Notice that for any leaf node $t$ in $T^{\star}$ such that $\chi^{\star}(t)\cap \delta(G)\neq \emptyset$, we have $\lambda_G(\chi^{\star}(t)\cap \delta(G))=\lambda_G(\widehat{\chi}(t)\cap \delta(G))=\chi(t)\in \PP$.  Let $EQ^{\star}$ be the equivalence class in $\equ$ such that $\eqf(\chi(t^{\star}))=EQ^{\star}$.   Let $U=V(G)\setminus \used{\chi^{\star}}{G}$ and $G^{\star}=G[U\cup \chi^{\star}(t^{\star})]$. Suppose we have the following two conditions.

\begin{itemize}
%\item[(a)] For each node $t\neq t^{\star}$, $\chi^{\star}(t)\subseteq \widehat{\chi}(t)$
\item[(a)] $\delta(G)\subseteq \bigcup_{t\in V(T^{\star})} \chi^{\star}(t)$.  
\item[(b)] For any vertex $x\in U$ and any edge $xy\in E(G)$, either $y\in U$ or there is a node $t \in V(P^{\star})$ such that $y\in \chi^{\star}(t)$.
\end{itemize}
  
 %Let $S$ be the output of the
% If the algorithm ${\cal A}_u$ succeed to output $S$ on the input $(G^{\star},EQ,\ell)$, then the following holds. 
Then, the below conditions are true. 
\begin{itemize}
\item[$(i)$]  There is vertex subset $S$ of size at most $k'$ such that  $G^{\star}-S$ belongs to an equivalence class which is at least as good as $EQ^{\star}$.
\item[$(ii)$] Given  a vertex subset $S$ of size at most $k'$ such that $G^{\star}-S$ belongs to an equivalence class which is at least as good as $EQ^{\star}$, then we can construct  an $\hh$-elimination decomposition $(T^{\star},\psi,Z^{\star})$  of $G$ of depth at most $k$, in polynomial time such that  the following holds. 
\begin{itemize}
\item If $S\cap \delta(G)=\emptyset$, then $G$ satisfies $s$ through $(T^{\star},\psi,Z^{\star})$.  
\item If $S\cap \delta(G)\neq \emptyset$, then we can construct a state-tuple $\tilde{s}
=(D,T,\chi',L',\PP',\eqf')$  in polynomial time such that $G$ satisfies $\tilde{s}$ through $(T^{\star},\psi,Z^{\star})$, and $\tilde{s}$ is strictly better than $s$. 
 \end{itemize}
 \end{itemize}
 \end{lemma}
 
 \begin{proof}
 
First we prove property $(i)$ of the lemma. For the sake of contradiction, suppose there is no vertex subset $S$ of size at most $k'$ in $G^{\star}$ such that $G^{\star}-S$ belongs to a class which is at least as good as $EQ^{\star}$. Recall that $(\wT,\widehat{\chi},\widehat{L})$ is a solution obeying the partial solution $e$. 
Let $S^{\star}=V(G^{\star})\setminus \widehat{L}$. Then $G^{\star}-S^{\star}$ belongs to a class which is at least as good as $EQ^{\star}$. This implies that $\vert S^{\star}\vert >k'$.  Let $X=\bigcup_{t\in V(P^{\star})\setminus \{t^{\star}\}} \widehat{\chi}(t)\cap S^{\star}$. Notice that $|X| \leq k'$. We prove that $G^{\star}-X$ belongs to a class which is at least as good as $EQ^{\star}$ and that will be a contradiction to our assumption. 
Let $F=S^{\star}\setminus X$. Since $\bigcup_{t\in V(P^{\star})} \widehat{\chi}(t)\cap F =\emptyset$, 
$F$ has neighbors only in   $\bigcup_{t\in V(P^{\star})\setminus \{t^{\star}\}} \widehat{\chi}(t)$ (because $F\subseteq V(G^{\star})$ and $(\wT,\widehat{\chi},\widehat{L})$ is a solution), and Branching Rule $2$ is not applicable, we have that $G[F]\in \hh$. Moreover, $F\cap \delta(G)=\emptyset$. Thus, by the component deletion property (see Definition~\ref{def:cdpeq}), we have that 
$G^{\star}-X=(G^{\star}-S^{\star})\uplus G[F]$ belongs to an equivalence which is at least as good as $EQ^{\star}$. This is a contradiction to our assumption that there is no vertex subset $S$ of size at most $k'$ in $G^{\star}$ such that $G^{\star}-S$ belongs to a class which is at least as good as $EQ^{\star}$.
This completes the proof of property~$(i)$.

 Now we prove property $(ii)$. Suppose there is a vertex subset $S=\{v_1,\ldots,v_{\ell}\}$ such that $\ell\leq k'$ and $G^{\star}-S$ belongs to an equivalence class which is at least as good as $EQ^{\star}$. Now we construct an $\hh$-elimination decomposition $(T^{\star},\psi,Z^{\star})$ as follows.   Let $W=\{t_{1},\ldots,t_{\ell}\}$ be a subset of nodes in $P^{\star}$ such that for all $i\in [\ell]$ $\chi(t_i)=\present$ and $\chi^{\star}(t_i)=\emptyset$. Now for each $i\in [\ell]$, we set $\psi(t_i)=\{v_i\}$. For each $t\in V(T^{\star})\setminus (W\cup \{t^{\star}\})$, we set $\psi(t)=\chi^{\star}(t)$. Finally, we set $\psi(t^{\star})=V(G^{\star})\setminus S$. Let $R$ be the set of leaf nodes in $T^{\star}$. We define $Z^{\star}=\bigcup_{t\in R}\psi(t)$.
 
 \begin{claim}
 \label{clim:eqclassesaregood}
 The graph $G[\psi(t^{\star})]$ is in an equivalence class which is at least as good as $EQ^{\star}$. For each $t\in R\setminus \{t^{\star}\}$, $\psi(t)\subseteq \widehat{\chi}(t)$ and  if $\psi(t)\cap \delta(G)\neq \emptyset$, then $G[\psi(t)]$ is in an equivalence class which is at least as good as $\eqf(\chi(t))$. 
 \end{claim}

 \begin{proof}
 Since $\psi(t^{\star})=V(G^{\star})\setminus S$, by our assumption (i.e., premise of the condition $(ii)$), we have that $G[\psi(t^{\star})]$ is in an equivalence class which is at least as good as $EQ^{\star}$. 
 
 From the construction of $\psi$ we have that  for each $t\in R\setminus \{t^{\star}\}$, $\psi(t)=\chi^{\star}(t)\subseteq \widehat{\chi}(t)$. Now fix a node $t\in R\setminus \{t^{\star}\}$ such that $\psi(t)\cap \delta(G)\neq \emptyset$. Then, we know that $G[\psi(t)]$ is an induced subgraph of $G[\widehat{\chi}(t)]$. In fact we prove that $G[\psi(t)]$ is a union of some connected components of $G[\widehat{\chi}(t)]$. Suppose not. 
 Then there is an edge $xy$ in the graph $G[\widehat{\chi}(t)]$ such that $x\in \psi(t)$ and $y\notin \psi(t)$. 
 Since $y\in \widehat{\chi}(t)$, we have that $y\notin \chi^{\star}(t')$ for any $t'\neq t$ because $(\wT,\widehat{\chi},\widehat{L})$ obeys the partial solution $e$. This implies that $y\in U$. This contradicts 
 condition $(b)$ in the lemma. So $G[\psi(t)]$ is obtained by deleting some connected components from 
 $G[\widehat{\chi}(t)]$. Moreover, we know that $\psi(t)\cap \delta(G)=\widehat{\chi}(t)\cap \delta(G)$.  
 Thus, by the component deletion property (see Definition~\ref{def:cdpeq}), we have that $G[\psi(t)]$ 
 and $G[\widehat{\chi}(t)]$ are in same equivalence class which is at least as good as $\eqf(\chi(t))$. This completes the proof of the  claim. 
 \end{proof}

 \begin{claim}
 $(T^{\star},\psi,Z^{\star})$ is an $\hh$-elimination decomposition of $G$ of depth at most $k$.  
 \end{claim}
 \begin{proof}
 Since $e$ is a partial solution, we have that the depth of $T^{\star}$ is at most $k$. Now we prove that 
 $(T^{\star},\psi,Z^{\star})$  satisfies the conditions of Definition~\ref{def:eldeco}. From the construction of $\psi$, we have that $\vert \psi(t)\vert \leq 1$ and $\psi(t)\subseteq V(G)\setminus Z^{\star}$ for any $t\in W$. For any internal node $t$ of $T^{\star}$ such that $t\notin W$,  $\psi(t)=\widehat{\chi}(t)\subseteq V(G)\setminus Z^{\star}$ and $\vert \psi(t)\vert=\vert\widehat{\chi}(t)\vert \leq 1$ because $(\widehat{T},\widehat{\chi},\widehat{L})$ is 
 $\hh$-elimination decomposition. This implies that condition $(1)$ of Definition~\ref{def:eldeco} holds. 
 From the construction of $\psi$, condition $(2)$  of Definition~\ref{def:eldeco} holds. 
 
 Now we prove condition $(3)$ of Definition~\ref{def:eldeco}.  Let $uv \in E(G)$  and $u \in \psi(t)$ and~$v \in \psi(t')$ for some $t,t'\in V(T^{\star})$.  Suppose $V(P^{\star})\cap \{t,t'\}=\emptyset$. Then, $\psi(t)=\chi^{\star}(t)\subseteq \widehat{\chi}(t)$ and $\psi(t')=\chi^{\star}(t')\subseteq \widehat{\chi}(t')$. Then, since $(\widehat{T},\widehat{\chi},\widehat{L})$ is an 
 $\hh$-elimination decomposition, we have that $t$ and $t'$ are in ancestor-descendant relationship in~$T^{\star}$. Suppose $t,t'\in V(P^{\star})$. Then $t$ and $t'$ are in ancestor-descendant relationship in~$T^{\star}$. Now consider the case when $|V(P^{\star})\cap \{t,t'\}|=1$.  Without loss of generality let $t\in V(P^{\star})$ and $t'\notin V(P^{\star})$. Thus, from the construction of $\psi$, we have that $v\in \psi(t')=\chi^{\star}(t')\subseteq \widehat{\chi}(t')$. If $u\in \chi^{\star}(t)$, then $t$ and $t'$ are in ancestor-descendant relation in~$T^{\star}$. Otherwise we contradict the fact that $e$ is a partial solution. If  $u\notin \chi^{\star}(t)$, then $u\in U=V(G^{\star}) \setminus \chi^{\star}(t^{\star})$.  Then, it will contradict condition $(b)$ of the lemma. Thus, we have proved condition $(3)$ of Definition~\ref{def:eldeco}.

 Next we prove condition $(4)$ of Definition~\ref{def:eldeco}. Recall that for any leaf node $t\in R\setminus \{t^{\star}\}$, $\psi(t)=\chi^{\star}(t)\subseteq \widehat{\chi}(t)$. Also, since $(\widehat{T},\widehat{\chi},\widehat{L})$ is an  $\hh$-elimination decomposition, $G[\widehat{\chi}(t)]\in \hh$. Thus, since $G[\psi(t)]$ is an induced subgraph of $G[\widehat{\chi}(t)]$, we get that $G[\psi(t)]\in \hh$ because $\hh$ is an hereditary family. Recall that $\psi(t^{\star}) =V(G^{\star})\setminus S$. Since $G^{\star}-S$ is in an equivalence class which is at least as good as $EQ^{\star}$ (see Claim~\ref{clim:eqclassesaregood}), we have that $G[\psi(t^{\star})]\in \hh$. This completes the proof of the claim. 
 \end{proof}

  \begin{claim}
  \label{clm:satisfiessolution}
   If $S\cap \delta(G)=\emptyset$, then $G$ satisfies $s$ through $(T^{\star},\psi,Z^{\star})$.  
  \end{claim}
\begin{proof}
Since the hypothetical solution $(\widehat{T},\widehat{\chi},\widehat{L})$ obeys the partial solution $e=(D,T^{\star}, {\chi}^{\star},{L}^{\star})$, $G$ satisfies $s$ through $(\widehat{T},\widehat{\chi},\widehat{L})$, and for each $t\in V(\wT)=V(T^{\star})$, $\chi^{\star}(t)\subseteq \widehat{\chi}(t)$. 
We need to prove that $G$ satisfies $s$ through $(T^{\star},\psi,Z^{\star})$.  Since $T^{\star}$ is isomorphic to $\widehat{T}$, condition $(a)$ of Definition~\ref{def:satisf} is true. 
Since for any node $t\in V(T^{\star})$, $\psi(t)\cap \delta(G)=\widehat{\chi}(t)\cap \delta(G)$, condition $(b)$ of Definition~\ref{def:satisf} is true. Notice that for any $v_i\in S$, we set $\psi(t_i)=\{v_i\}$, where $\chi(t_i)=\present$. This implies that condition $(c)$ of Definition~\ref{def:satisf} is true. Notice that for any vertex $v\in \delta(G)$, there is a node $t$ such that $v\in \psi(t)$ and $v\in \widehat{\chi}(t)$. This implies that condition $(d)$ of Definition~\ref{def:satisf} is true. Condition $(e)$ of Definition~\ref{def:satisf}  follows from Claim~\ref{clim:eqclassesaregood}. 
\end{proof}

\begin{claim}
If $S\cap \delta(G)\neq \emptyset$, then we can construct a state-tuple $\tilde{s}
=(D,T,\chi',L',\PP',\eqf')$  in polynomial time such that $G$ satisfies $\tilde{s}$ through $(T^{\star},\psi,Z^{\star})$, $\tilde{s}$ is strictly better than $s$. 
  \end{claim}
  
  \begin{proof}
  First we define the function $\chi'$. Recall that $S=\{v_1,\ldots,v_{\ell}\}$, and $W=\{t_{1},\ldots,t_{\ell}\}$ be a subset of nodes in $P^{\star}$ such that for all $i\in [\ell]$ $\chi(t_i)=\present$ and $\chi^{\star}(t_i)=\emptyset$.
  For any $t\in V(T^{\star})\setminus \{t_1,\ldots,t_{\ell},t^{\star}\}$, we set $\chi'(t)=\chi(t)$. 
  For any $i\in [\ell]$, if $v_i\in \delta(G)$, then $\chi'(t_i)=\{\lambda_G(v_i)\}$. Otherwise  $\chi'(t_i)=\present$. 
  %For any $j\in \{\ell+1,\ldots,k'\}$, we set $\chi'(t_j)=P$. 
  Finally, $\chi'(t^{\star})=\lambda_G(\psi(t^{\star})\cap \delta(G))$. 
  The construction of $\chi'$ and the fact that $S\cap \delta(G)\neq \emptyset$ implies that $D\setminus L\subset D\setminus L'$. Let $R'\subseteq R$ be the set of leaf nodes in $T^{\star}$ such that for any $t\in R'$, $\psi(t)\cap \delta(G)\neq \emptyset$. Then $\PP'=(\chi'(t))_{t\in R'}$. Notice that for each $t\in R'\setminus \{t^{\star}\}$, $\chi(t)=\chi'(t)$ and $\chi'(t^{\star})\subset \chi(t)$.  Now we define $\eqf'$. For any $t\in R'$, we define $\eqf'(\chi'(t))=\eqf(\chi(t))$. Finally $L'=L\setminus \lambda_G(S\cap \delta(G))$. This completes the construction of $\tilde{s}=(D,T,\chi',L',\PP',\eqf')$. Using arguments similar to the one in Claim~\ref{clm:satisfiessolution}, one can prove that $G$ satisfies $\tilde{s}$ through $(T^{\star},\psi,Z^{\star})$, and hence it is omitted here.  
  
%------------------  

  Next we prove that $\tilde{s}$ is strictly better than  $s$. We have already proved that $D\setminus L\subset D\setminus L'$. 
  %Here we can assume that $G$ exactly satisfies $s$ through $(\widehat{T},\widehat{\chi},\widehat{L})$. 
   Let $H$ be a graph such that $G\oplus H$ exactly realizes $s$ through $(\widehat{T},\eta,L_1)$ 
   and $(\wT,\widehat{\chi},\widehat{L})$ is equal to $\restricted{G}{(\widehat{T},\eta,L_1)}$. 
   Thus we have the following.
   \begin{itemize}
  \item[$(i)$] $G$ exactly satisfies $s$ through $(\wT,\widehat{\chi},\widehat{L})$. 
\item[$(ii)$] For each node $t$ of~$T$ with $\eta(t)\neq \emptyset$ and $\eta(t)\subseteq V(H)\setminus \delta(G)$, $\chi(t)=\future$.
   \end{itemize}

We will prove that $G\oplus H$ realizes $\tilde{s}$. Towards that we will construct an $\hh$-elimination decomposition $(\wT,\phi, L_2)$ of $G\oplus H$, of depth at most $k$, such that $G\oplus H$ realizes  $\tilde{s}$  through $(\wT,\phi, L_2)$.  Statement $(i)$ implies that for any $t\in R'$, $G[\widehat{\chi}(t)]$ belongs to $\eqf(\chi(t))$.
 Recall that $T^{\star}$ is isomorphic to $\wT$ and $V(T^{\star})=V(\wT)$. From the construction of $\tilde{s}$, we have that for any node $t\in V(T)$, $\chi(t)=\future$ if and only if  $\chi'(t)=\future$.  Now we construct the function $\phi$. Let $V=V(G)$ and $V'=V(G\oplus H) \setminus V(G)$.  For any node $t\in V(\wT)$, we define $\phi(t)=(\psi(t)\cap V)\cup (\eta(t)\cap V')$. Let $L_2=\bigcup_{t\in R} \phi(t)$ where $R$ is the set of leaf nodes in $\wT$.  
 
 Now we prove that $G\oplus H$ realizes $\tilde{s}$ through $(\wT,\phi, L_2)$. It is easy to verify that $(T^{\star},\psi,Z^{\star})$ is equal to   $\restricted{G}{(\wT,\phi, L_2)}$. We have already mentioned that   
 $G$ satisfies $\tilde{s}$ through $(T^{\star},\psi,Z^{\star})$.   Next we prove that indeed $(\wT,\phi, L_2)$ is an $\hh$-elimination decomposition of $G\oplus H$.  It is easy to verify that conditions $(1)$ and $(2)$ of 
 Definition~\ref{def:eldeco} holds.  Now we will prove that condition $(3)$ of Definition~\ref{def:eldeco} holds. 
 Let $uv$ be an edge in $G\oplus H$ and $t,t'\in V(\wT)$ such that $u\in \phi(t)$ and $v\in \phi(t')$. Suppose $u,v\in V(G)$. Then, since $G$ satisfies $\tilde{s}$ through $(T^{\star},\psi,Z^{\star})$, which is equal to $\restricted{G}{(\wT,\phi, L_2)}$, we have that $t$ and $t'$ are in ancestor-descendant relation.  Suppose $u,v\in V(H)$. Then, since $(\widehat{T},\eta,L_1)$ is an $\hh$-elimination decomposition of $G\oplus H$, from the construction of $\phi$, we have that $t$ and $t'$ are in ancestor-descendant relation. Thus, we have proved that condition $(3)$ of Definition~\ref{def:eldeco} holds.

 Next we prove that condition $(4)$ of Definition~\ref{def:eldeco} holds. 
 Let $S_H=\{v\in \delta(H)\colon \lambda_G(v)\in \lambda_G(S\cap \delta(G))\}$.  That is $S_H$ is the set of boundary vertices in $H$ that has the same label as the vertices in $S\cap \delta(G)$. 
 Let $t$ be a leaf node. 
 Let $G_1=G[\eta(t)\cap V(G)]=G[\widehat{\chi}(t)]$ and $H_1=H[\eta(t) \cap V(H)]$. 
 Let $G_2=G[\phi(t)\cap V(G)]= G[\psi(t)]$ and $H_2=H[(\eta(t)\cap V(H)) \setminus S_H]$. 
 Notice that $G[\eta(t)]=G_1\oplus H_1$ and $G[\phi(t)]=G_2\oplus H_2$. 
 Since $G\oplus H$ exactly realizes $s$ through $(\widehat{T},\eta,L_1)$, we have that $G_1$ belongs to the equivalence class $\eqf(\chi(t))$ and $G_1\oplus H_1 \in \hh$. 
 %Since $D\setminus L \subset D \setminus L'$, we have that 
  Since $G$ satisfies $\tilde{s}$ through $(T^{\star},\psi,Z^{\star})$
 $G_2$ belongs to an equivalence class which is at least as good as $\eqf'(\chi'(t))=\eqf(\chi(t))$. 
 This implies that $G_2\oplus H_1\in \hh$. 
 Since $H_2$ is an induced subgraph of $H_1$ and $G_2\oplus H_1\in \hh$, by the hereditary property of $\hh$, we have that $G_2\oplus H_2\in \hh$. 
 %For any internal node $t\in V(\wT)$, if $\chi_1(t)\subseteq V(H)$
\end{proof}
This completes the proof of the lemma. 
 \end{proof}

Because of Lemma~\ref{lem:finaldeletion}, in the final step of our algorithm {\sigtest}, we run the algorithm ${\cal A}_u$ on the input $(G^{\star},EQ^{\star},k')$ as defined in Lemma~\ref{lem:finaldeletion}. 
If the algorithm ${\cal A}_u$ fails to output a solution, then we output {\sf Fail}. Otherwise, 
let $S$ be the output of the algorithm ${\cal A}_u$. If $S\cap \delta(G)=\emptyset$, then we output $(s,1)$. If $S\cap \delta(G)\neq \emptyset$, then we construct $\tilde{s}$ (as mentioned in Lemma~\ref{lem:finaldeletion}) and output $(s,0)$ and $(\tilde{s},1)$.

\paragraph*{Running time analysis} Each branching rules makes a function of $q+k$ many branches. 
Since the number of nodes in $T^{\star}$ is at most $q+k$ and in each branch of the branching algorithm the number of internal nodes $t$ with $\chi^{\star}(t)=\emptyset$ decreases by $1$, the depth of the branching algorithm is at most $q+k$. In the final step we run the algorithm ${\cal A}_u$. This implies that the running time of the algorithm is upper bounded $g_1(q,k)\cdot f_{u}(k,|\QQ(\equ)|,n)\cdot n^{\OO(1)}$ for some function $g_1$.

%% file: apps.tex
%!TEX root = main.tex
The objective of this section is to prove the following theorem. 

\begin{theorem}\label{thm:apps-uniform}
\probEDH\ admits a uniform {\rm \FPT} algorithm, for the case when $\C{H}$ is any one of the following:
\begin{enumerate}\setlength{\itemsep}{-1pt}
\item the family of chordal graphs,
\item the family of interval graphs,
\item the family of bipartite graphs,
\item for a fixed finite family of graphs $\mathbb{O}$, a graph is in $\C{H}$, if and only if it does not contain any graph from $\mathbb{O}$ as an induced subgraph, 
\item for a fixed finite family of graphs $\mathbb{O}$, a graph is in $\C{H}$, if and only if it does not contain any graph from $\mathbb{O}$ as a minor,\footnote{As minor closed families are characterized by a finite family of forbidden minors by Robertson-Seymour Theorem, so the condition on finiteness for this case is not necessary.} or
\item for a fixed finite family of graphs $\mathbb{O}$, a graph is in $\C{H}$, if and only if it does not contain any graph from $\mathbb{O}$ as a topological minor. 
\end{enumerate}
\end{theorem}

For any fixed finite family of graphs $\C{F}$, we can find a finite family of graphs $\C{F}'$, such that a graph does not contain any minor from $\C{F}$ if and only if it does not contain any topological minor from $\C{F}'$. Thus, item 5 in Theorem~\ref{thm:apps-uniform} is subsumed by item 6 (see for instance,~\cite{FominLPSZ20} for a discussion on this). Thus, we will not focus on proving item 5, and our rest of the section will focus on proving the theorem for all other items its statement. 

We denote the families of bipartite, chordal, interval graphs by $\C{H}_{\sf bip}$, $\C{H}_{\sf cdl}$, and $\C{H}_{\sf int}$, respectively. For a fixed finite family of graphs $\C{F}$, we denote the families of graphs that have no graph from $\C{F}$ as an induced subgraph and a topological minor by $\C{H}_{\C{F}_{\sf ind}}$ and $\C{H}_{\C{F}_{\sf top}}$, respectively.

To invoke Theorem~\ref{thm:uni:elim}, we need to define a refinement of the canonical equivalence class (see Definition~\ref{def:canonical-eq}) for the graph families of our concern. To this end, we introduce the notion of obstructions.    

\paragraph{Obstructions to a family of graphs.} For a family of graphs $\C{H}$, a family of (possibly infinite) set of graphs $\mathbb{O}$ is an {\em induced obstruction set} to $\C{H}$, if a graph $G \in \C{H}$ if and only if $G$ does not contain any graph from $\mathbb{O}$ as an induced subgraph. Notice that we have the following: i) the family $\mathbb{O}_{\sf bib}$ of all cycles of odd length is an induced obstruction set for $\C{H}_{\sf bib}$ and ii) the family $\mathbb{O}_{\sf cdl}$ of all cycles of length at least $4$ is an induced obstruction set for $\C{H}_{\sf cdl}$.

\begin{figure}[t]
	\begin{center}
		\includegraphics[scale=0.6]{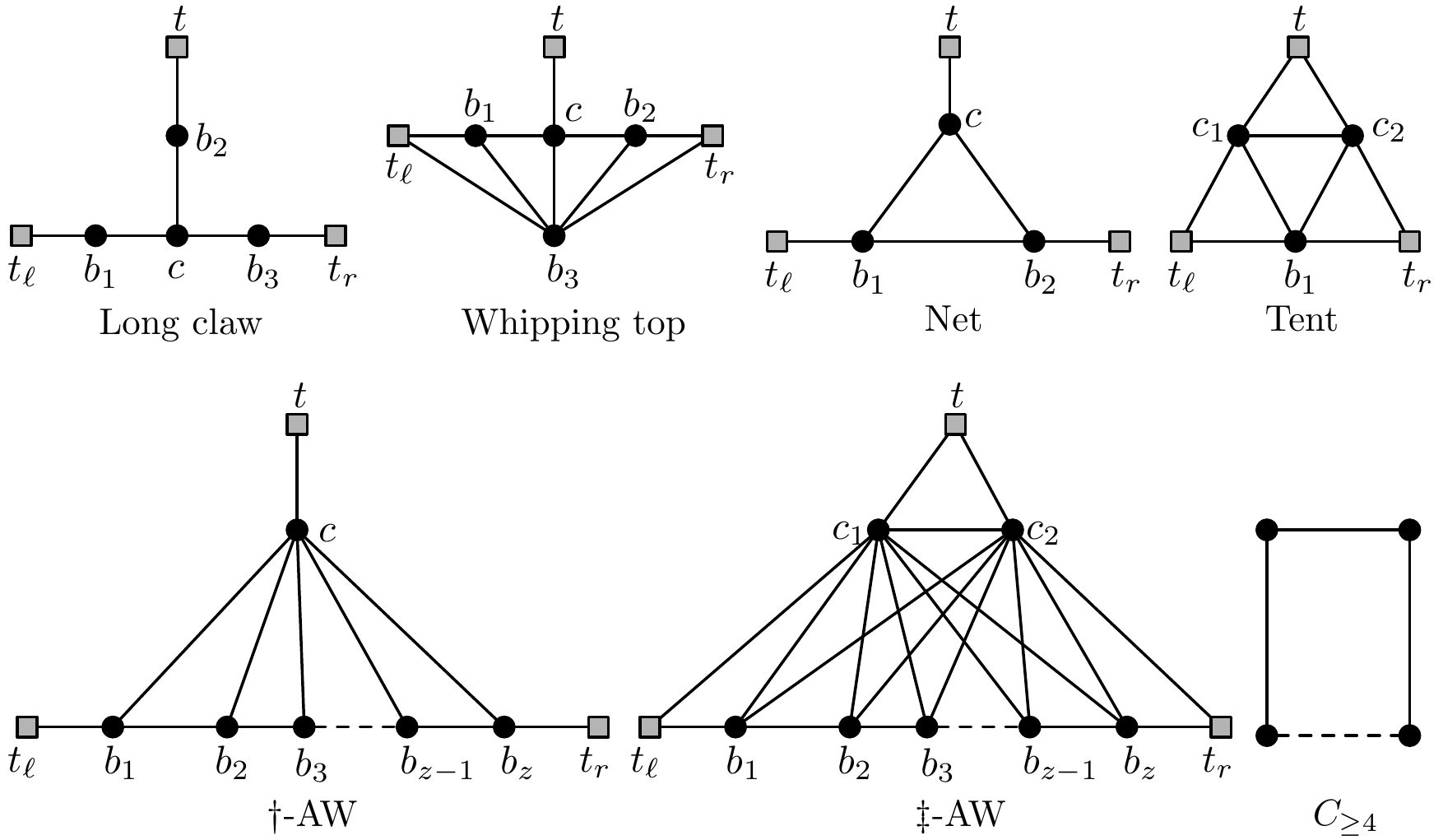}
	\end{center}
	\caption{The set of obstructions for the family of interval graph.} 
	\label{fig:list-obstruction}
\end{figure}

The set of obstructions to interval graphs have been completely characterized by Lekkerkerker and Boland,~\cite{lekkeikerker1962representation}. A graph is an interval graph if and only if it does not contain any of the following graphs as an induced subgraph (see Figure~\ref{fig:list-obstruction}).\footnote{The figure is borrowed from~\cite{AgrawalM0Z19}.}
\begin{itemize}\setlength{\itemsep}{-1pt}
\item {\bf Long Claw.} A graph $\mathbb{O}$ such that $V(\mathbb{O})=\{t_\ell,t_r,t,c,b_1,b_2,b_3\}$ and $E(\mathbb{O})=\{(t_\ell,b_1),(t_r,b_3),$ $(t,b_2),(c,b_1),(c,b_2),(c,b_3)\}$.

\item {\bf Whipping Top.} A graph $\mathbb{O}$ such that $V(\mathbb{O})=\{t_\ell,t_r,t,c,b_1,b_2,b_3\}$ and $E(\mathbb{O})=\{(t_\ell,b_1),$ $(t_r,b_2),(c,t),(c,b_1),(c,b_2),(b_3,t_\ell),(b_3,b_1), (b_3,c),(b_3,b_2), (b_3,t_r)\}$.

\item {\bf \longobs.} A graph $\mathbb{O}$ such that $V(\mathbb{O})=\{t_\ell,t_r,t,c\} \cup \{b_1,b_2,\ldots,b_z\}$, where $t_\ell=b_0$ and $t_r=b_{z+1}$, $E(\mathbb{O})=\{(t,c),(t_\ell,b_1), (t_r,b_z)\} \cup \{(c,b_i) \mid i \in [z]\}\cup \{(b_i,b_{i+1}) \mid i \in [z-1]\}$, and $z \geq 2$. A \longobs\ where $z=2$ will be called a \emph{net}.

\item {\bf \dlongobs.} A graph $\mathbb{O}$ such that $V(\mathbb{O})=\{t_\ell,t_r,t,c_1,c_2\} \cup \{b_1,b_2,\ldots,b_z\}$, where $t_\ell=b_0$ and $t_r=b_{z+1}$, $E(\mathbb{O})=\{(t,c_1),(t,c_2),(c_1,c_2),(t_\ell,b_1), (t_r,b_z), (t_\ell,c_1), (t_r,c_2)\} \cup \{(c,b_i) \mid i \in [z]\}\cup \{(b_i,b_{i+1}) \mid i \in [z-1]\}$, and $z \geq 1$. A \dlongobs\ where $z=1$ will be called a \emph{tent}.

\item {\bf Hole.} A chordless cycle on at least four vertices.
\end{itemize}

We denote the above family of graph by $\mathbb{O}_{\sf int}$, which is an induced obstruction set for $\C{H}_{\sf int}$.

Let $\C{G}$ denote the family of all graphs. Consider a finite set of graphs $\C{F}$. Notice that an induced obstruction set for $\C{H}_{\C{F}_{\sf ind}}$ is  the same as $\mathbb{O}_{\C{F}_{\sf ind}} = \C{F}$. Let $\mathbb{O}_{\C{F}_{\sf top}}$ be the family of all graphs that are not contained in $\C{H}_{\C{F}_{\sf top}}$, i.e., $\mathbb{O}_{\C{F}_{\sf top}} = \C{G} \setminus \C{H}_{\C{F}_{\sf top}}$. Notice that each graph in $\mathbb{O}_{\C{F}_{\sf top}}$ must contain some $G \in \C{F}$ as a topological minor. Thus, we can obtain that $\mathbb{O}_{\C{F}_{\sf top}}$ is an induced obstruction set for $\C{H}_{\C{F}_{\sf top}}$. 

Let $\mathbb{H} = \{\C{H}_{\sf cdl}, \C{H}_{\sf int}, \C{H}_{\sf bip}\} \cup \{\C{H}_{\C{F}_{\sf min}}, \C{H}_{\C{F}_{\sf top}} \mid \C{F} \mbox{ is a finite family of graphs}\}$. From our previous discussions we can obtain that each $\C{H}\in \mathbb{H}$ has an induced obstruction set $\mathbb{O}$. In the rest of the section, for $\C{H}\in \mathbb{H}$, we will work with the corresponding $\mathbb{O}$ that we stated previously. For the sake of simplicity, throughout this section we will assume that the label set of any boundaried graph $G$, we have $\Lambda(G) \subseteq \mathbb{N}$, where $0 \in \Lambda(G)$ and for no vertex $v\in \delta(G)$, we have $\lambda_G(v) = 0$. 

\begin{definition}[Boundaried Partial Obstruction]{\rm
Consider a family of graphs $\C{H}$ and an induced obstruction set $\mathbb{O}$ for it. A {\em partial obstruction of $\C{H}$} from $\mathbb{O}$ is any graph that is an induced subgraph of some graph in $\mathbb{O}$. A {\em boundaried partial obstruction of $\C{H}$} is boundaried graph whose unboundaried counterpart (graph after forgetting the labels of the boundary) is a partial obstruction from $\mathbb{O}$. We denote the set of all boundaried partial obstruction from $\mathbb{O}$ of $\C{H}$ by $\mathsf{BPO}_{\mathbb{O}}({\C{H}})$. (We skip the subscript in the above notation whenever the context is clear.) 
}\end{definition}

\begin{definition}[Relevant Boundaried Partial Obstruction]{\rm
Consider a family of graphs $\C{H}$ and an induced obstruction set $\mathbb{O}$ for it. We say that $O \in \mathsf{BPO}({\C{H}})$ is {\em relevant} if there does not exist another boundaried partial obstruction $O' \in \mathsf{BPO}({\C{H}})$, such that:
\begin{enumerate}\setlength{\itemsep}{-1pt}
\item $|V(O')|<|V(O)|$, and  
\item for every boundaried graph $G$, $O \oplus G \notin \C{H}$ if and only if $O' \oplus G \notin \C{H}$. 
\end{enumerate}
}\end{definition}

In the next observation we show that for each $t \in \mathbb{N}$, we can compute the set of relevant boundaried partial obstructions in time bounded by $t$ alone.   

\begin{observation}\label{obs:obst-t-boundary}
For each $\C{H} \in \mathbb{H}$, there exists a function $\zeta[\C{H}]: \mathbb{N} \rightarrow \mathbb{N}$, such that for any $t \in \mathbb{N}$, we can find an inclusion-wise maximal set of relevant $t$-boundaried partial obstructions, denoted by, $\mathsf{RBPO}({\C{H}},t)$ (with respect to $\mathbb{O}$) in time bounded by $\zeta[\C{H}](t)$, such that $|\mathsf{RBPO}({\C{H}},t)| \leq \zeta[\C{H}](t)$. 
\end{observation}
\begin{proof}
Note that for the family of chordal graphs, if we have a boundaried graph $G$ with $\Lambda(G)\subseteq \{1,\ldots,t\}$, the following state information is enough: for every pair of labels in $\Lambda(G)$, whether there no path, a path with exactly one edge, or all paths have at least two edges, between the corresponding vertices associated with the labels. Notice that for each such state information, keeping one partial obstruction is enough for us, which can be found easily based on the state information. Similarly for all other graph classes, we can argue that only bounded number of state information are required, corresponding to each of which keeping one partial obstruction is enough for our purpose. 
\end{proof}

For $\C{H} \in \mathbb{H}$ and $t \in \mathbb{N}$ let $q[\C{H},t] = |{\sf RBPO}(\C{H},t)|$, we fix an arbitrary ordering among the $t$-boundaried graphs in ${\sf RBPO}(\C{H},t)$, and we let ${\sf RBPO}(\C{H},t) = \{P[\C{H},t,1], P[\C{H},t,2], \cdots,$ $P[\C{H},t,q[\C{H},t]]\}$. We will next define the notion of ``behaviour'' of a boundaried graph, that will help us in defining our (refinement of) canonical equivalent class (see Definition~\ref{def:canonical-eq}).

\begin{definition}\label{def:stamp-eq-class-app-H}{\rm
Consider $\C{H} \in \mathbb{H}$, $t \in \mathbb{N}$ and a $t$-boundaried graphs $G$. The $[\C{H},t]$-{\em behaviour} of $G$ is the vector ${\sf bhv}[\C{H},t](G) = (z_1,z_2,\cdots, z_{q[\C{H},t]})$, where for $i \in [q[\C{H},t]]$, $z_i = 1$ if and only if $G \oplus P[\C{H},i] \in \C{H}$.
}\end{definition}

\begin{definition}\label{def:equiv-bhv}{\rm 
Consider $\C{H} \in \mathbb{H}$ and $t \in \mathbb{N}$. Let $G_1$ and $G_2$ be two $t$-boundaried graphs, such that $\Lambda(G_1) = \Lambda(G_2)$ and $t$ is the largest number in $\Lambda(G_1)$. 
 We say that $G_1$ and $G_2$ are {\em $\C{H}$-behaviour equivalent}, denoted by $G_1 \equiv_{u,\C{H}} G_2$, if ${\sf bhv}[\C{H},t](G_1) = {\sf bhv}[\C{H},t](G_2)$.\footnote{The letter $u$ in $\equiv_{u,\C{H}}$ is to denote that it is a user defined equivalence class, as we are the user in this section.}
}\end{definition}

Notice that $\equiv_{u,\C{H}}$ defines an equivalence class over the boundaried graphs. We use $\C{Q}(\equiv_{u,\C{H}})$ to denote the set of equivalence classes of $\equiv_{u,\C{H}}$. The next observation follows from Definition~\ref{def:equiv-bhv}.

\begin{observation}\label{obs:bhc-can-equiv}
For each $\C{H} \in \mathbb{H}$, and boundaried graphs $G_1$ and $G_2$, such that $G_1 \equiv_{u,\C{H}} G_2$, we have $G_1 \equiv_{\C{H}} G_2$. Also, $\equiv_{u,\C{H}}$ satisfies the requirement of Definition~\ref{def:cdpeq}. 
\end{observation}

We will next give a definition that will be useful in obtaining the user defined function ${\sf ug}_{\C{H}}$. 

\begin{definition}\label{def:app-ug}{\rm
Consider $\C{H} \in \mathbb{H}$, and two distinct equivalence classes $Q_1,Q_2 \in \C{Q}(\equiv_{u,\C{H}})$. Moreover, let $t_1$ and $t_2$ be the largest numbers in $\Lambda(G_1)$, for $G_1 \in Q_1$ and $\Lambda(G_2)$, for $G_2 \in Q_2$, respectively. We say that $Q_1$ is {\em at least as good as} $Q_2$ if each of the following holds:
\begin{enumerate}\setlength{\itemsep}{-1pt}
\item $t_1 \leq t_2$, and
\item for each $H \in {\sf RBPO}(\C{H},t_1) \subseteq {\sf RBPO}(\C{H},t_2)$, whenever for any $G_2 \in Q_2$, $G_2 \oplus H \in \C{H}$, then any $G_1 \in Q_1$ we have $G_1 \oplus H \in \C{H}$.
\end{enumerate}
}\end{definition}

Using the above definition we are now ready to define the function ${\sf ug}_{\C{H}}$. 

\begin{definition}\label{def:ug-function-apps}{\rm 
For $\C{H} \in \mathbb{H}$, we define the function ${\sf ug}_{\C{H}}: \C{Q}(\equiv_{u,\C{H}}) \times \C{Q}(\equiv_{u,\C{H}}) \rightarrow \{0,1\}$ as follows: for $Q_1,Q_2 \in \C{Q}(\equiv_{u,\C{H}})$, if $Q_1$ is at least as good as $Q_2$ as per Definition~\ref{def:app-ug}, then set ${\sf ug}_{\C{H}}(Q_1,Q_2) = 1$, otherwise set ${\sf ug}_{\C{H}}(Q_1,Q_2) = 0$.  
}\end{definition}

We have now defined a refinement of $\equiv_{\C{H}}$ and ${\sf ug}_{\C{H}}$, for each $\C{H} \in \mathbb{H}$. We will next focus on designing the ``user's algorithm''. To this end, we begin by defining the notion of an irrelevant vertex.

\begin{definition}[Irrelevant Vertex]\label{def:irrelev-one}{\rm 
Consider a family of graphs $\C{H}$, a graph $G$, and an integer $k \in \mathbb{N}$. A vertex $v \in V(G)$ is {\em $(\C{H},k)$-irrelevant} if there exists $S \subseteq V(G)$ of size at most $k$ such that $G-S \in \C{H}$ if and only if there exists $S' \subseteq V(G-\{v\})$ of size at most $k$ such that $(G-\{v\})-S' \in \C{H}$. 
}\end{definition}

We now extend the above definition to boundaried graphs. 

\begin{definition}[Boundaried Irrelevant Vertex]\label{def:irrelev-two}{\rm 
Consider a family of graphs $\C{H}$, a boundaried graph $G$, and an integer $k \in \mathbb{N}$. A vertex $v\in V(G) \setminus \delta(G)$ is {\em boundaried $(\C{H},k)$-irrelevant} if for every boundaried graph $H$, $v$ is $(\C{H},k)$-irrelevant in $G \oplus H$. 
}\end{definition}

\begin{definition}[Boundaried Obstruction Irrelevant Vertex]{\rm 
Consider $\C{H} \in \mathbb{H}$, a $t$-boundaried graph $G$, where $t$ is the largest integer in $\Lambda(G)$, and an integer $k \in \mathbb{N}$. A vertex $v \in V(G) \setminus \delta(G)$ is {\em boundaried obstruction $(\C{H},k)$-irrelevant} if for every $H \in {\sf RBPO}({\C{H}},t)$, $v$ is $(\C{H},k)$-irrelevant in $G \oplus H$.  
}\end{definition} 

\begin{lemma}\label{lemma:irrelevance}
Consider $\C{H} \in \mathbb{H}$, a $t$-boundaried graph $G$, where $t$ is the largest integer in $\Lambda(G)$, and an integer $k \in \mathbb{N}$. A vertex $v\in V(G) \setminus \delta(G)$ is boundaried $(\C{H},k)$-irrelevant if and only if it is boundaried obstruction $(\C{H},k)$-irrelevant.  
\end{lemma}
\begin{proof}
Let $v\in V(G) \setminus \delta(G)$. One direction is trivial: If $v$ is boundaried $(\C{H},k)$-irrelevant, then it is boundaried obstruction $(\C{H},k)$-irrelevant.  Indeed, suppose that $v$ is boundaried $(\C{H},k)$-irrelevant. Then, for every boundaried graph $H$, $v$ is $(\C{H},k)$-irrelevant in $G \oplus H$. In particular, for every $O \in {\sf RBPO}({\C{H}},t)$, $v$ is $(\C{H},k)$-irrelevant in $G \oplus O$. So, $v$ is boundaried obstruction $(\C{H},k)$-irrelevant.

Now, consider the reverse direction. Suppose  $v$ is boundaried obstruction $(\C{H},k)$-irrelevant. Then, for every $O \in {\sf RBPO}({\C{H}},t)$, $v$ is $(\C{H},k)$-irrelevant in $G \oplus O$. We need to prove that 
for any boundaried graph $H$, $v$ is $(\C{H},k)$-irrelevant in $G \oplus H$. Towards that, let us fix a boundaried graph $H$. Without loss of generality, we can suppose that $\delta(H)\subseteq \{1,\ldots,t\}$ (else, we can ``unlabel'' the vertices in $H$ having a larger label, and thereby we do not change $G \oplus H$). Targeting a contradiction, suppose that $v$ is not $(\C{H},k)$-irrelevant in $G \oplus H$.  So, because $\C{H}$ is hereditary, this means that there exists $S' \subseteq V(G \oplus H-\{v\})$ of size at most $k$ such that $(G \oplus H-\{v\})-S' \in \C{H}$, but there does not exist $S \subseteq V(G \oplus H)$ of size at most $k$ such that $G \oplus H-S \in \C{H}$. In particular, $G \oplus H-(S'\cup\{v\})\notin \C{H}$, so $G \oplus H$ must have an obstruction $O^\star$ in $\mathbb{O}$ (as an induced subgraph) that contains $v$ and is disjoint from $S'$. Let $O'$ be the subgraph of $O^\star[V(H)]$ whose boundary is the boundary vertices of $H$ that occur in $O^\star$. So, $O'$ is a $t$-boundaried partial obstruction of  $\C{H}$ such that $G \oplus O' - (S'\cap V(G))\notin \C{H}$ and $G \oplus O' - ((S'\cap V(G))\cup\{v\})\in \C{H}$. By the definition of ${\sf RBPO}({\C{H}},t)$, there exists $O\in {\sf RBPO}({\C{H}},t)$ such that $G \oplus O - (S'\cap V(G))\notin \C{H}$ and $G \oplus O - ((S'\cap V(G))\cup\{v\})\in \C{H}$. However, $v$ is $(\C{H},k)$-irrelevant in $G \oplus O$, and thus we have reached a contradiction.
\end{proof}

For the families of graph that we are interested in, our objective will be to find an ``irrelevant'' vertex, towards obtaining a ``user's algorithm''. We first explain {\em intuitively} how we intend to obtain such a rule. For simplicity consider the family of chordal graphs. We know that any chordal graph with large treewidth must contain a large clique. More generally, any graph from which we can delete a small set of vertices so that the resulting graph is a chordal graph, either it contains a large clique, or its treewidth can be bounded by the size of the small deletion set. We remark that whenever a graph has a large clique, using the known (FPT/kernelization) algorithms for {\sc Chordal Vertex Deletion}~\cite{Marx10}, we will be able to find an irrelevant vertex. We give generic definitions below, and then we will see how our families fit into these definitions. 

\begin{definition}{\rm
Consider families of graphs $\what{\C{G}}$ and $\wtilde{\C{G}}$. We say that $\what{\C{G}}$ has {\em treewidth property} with respect to $\wtilde{\C{G}}$, if there is a function $\tau: \mathbb{N} \rightarrow \mathbb{N}$, such that any graph $G \in \what{\C{G}}$ with $\tw(G) \geq \ell$, contains a graphs $G' \in \wtilde{\C{G}}$ with $|V(G')| \geq \tau(\ell)$ as an induced subgraph. In that above, we say that $\tau$ is a {\em tw-certificate} for the treewidth-property of $\what{\C{G}}$ with respect to $\wtilde{\C{G}}$.    
}
\end{definition}

We will now extend the above definition to the classes where we can find large subgraphs of desired type in bounded amount of time. 

\begin{definition}\label{def:app-eff-tw-prop-def}{\rm 
Consider families of graphs $\what{\C{G}}$ and $\wtilde{\C{G}}$, such that $\what{\C{G}}$ has treewidth-property with respect to $\wtilde{\C{G}}$, with tw-certificate $\tau: \mathbb{N} \rightarrow \mathbb{N}$. We say that $\what{\C{G}}$ has {\em efficient treewidth-property} with respect to $\wtilde{\C{G}}$, if there a function $g: \mathbb{N} \rightarrow \mathbb{N}$ and an algorithm, which given an $n$-vertex graph $G \in \what{\C{G}}$ and an integer $\ell \in \mathbb{N}$, in time bounded by $g(\ell) \cdot n^{\C{O}(1)}$, either correctly concludes that $\tw(G) < \ell$, or outputs a graph $G' \in \wtilde{\C{G}}$, such that: i) $|V(G')| \geq \tau(\ell)$ and ii) $G'$ is an induced subgraph of $G$.
}\end{definition}

\begin{definition}{\rm
Consider families of graphs $\what{\C{G}}$, $\wtilde{\C{G}}$, and $\C{H}$, where $\wtilde{\C{G}} \subseteq \C{H}$. We say that $\what{\C{G}}$ is $(\wtilde{\C{G}}, \C{H})$-{\em deletable}, if there is a function $\mu: \mathbb{N} \rightarrow \mathbb{N}$, such that: for any graph $G \in \what{\C{G}}$ which contains a graph $G' \in \wtilde{\C{G}}$ with at least $k$ vertices as an induced subgraph, there exists $Z \subseteq V(G') \subseteq V(G)$ of size at least $\mu(k)$, such that each vertex $u \in Z$ is $(\C{H},k)$-irrelevant in $G$. In the above, we say that $\mu$ is a {\em del-certificate} of $\what{\C{G}}$ being $(\wtilde{\C{G}}, \C{H})$-deletable.
} 
\end{definition}

\begin{definition}{\rm 
Consider families of graphs $\what{\C{G}}$, $\wtilde{\C{G}}$, and $\C{H}$, where $\wtilde{\C{G}} \subseteq \C{H}$ and $\what{\C{G}}$ is $(\wtilde{\C{G}}, \C{H})$-deletable, where $\mu$ is a del-certificate of $\what{\C{G}}$ being $(\wtilde{\C{G}}, \C{H})$-deletable. A {\em $(\what{\C{G}}, \wtilde{\C{G}}, \C{H})$-irrelevance detector} is an algorithm, which given an $n$-vertex graph $G \in \what{\C{G}}$ and an induced subgraph $G' \in \wtilde{\C{G}}$ of $G$, where $k = |V(G')|$, outputs a subset $Z \subseteq V(G') \subseteq V(G)$ of size at least $\mu(k)$, such that each vertex in $Z$ is $(\C{H},k)$-irrelevant in $G$.  
}
\end{definition}

Now we will state results which follow from the known results, that will be important in obtaining our ``user's algorithm''. The next result says that, given a graph which admits a $k$-sized deletion set to chordal graphs, along with a large clique, then we will be able to obtain large set of vertices in this clique that are irrelevant for us. We remark that each of this vertices are ``individually'' irrelevant to us, and we won't be deleting all of them together. The above discussed result simply corresponds to the marking of vertices in a maximal clique in the clique-tree of the input graphs, obtained after removing a small deletion set. (A similar result can also be obtained for interval graphs.) Explicitly we can obtain such a result in various paper, like~\cite{Marx10,JansenP18,DBLP:journals/talg/AgrawalLMSZ19}, where to the best of our knowledge, the first time such a result appeared in~\cite{Marx10}. For an explicit reference, we direct the readers to Lemma 4.1 and Section 3.3 of~\cite{DBLP:journals/talg/AgrawalLMSZ19}. We would like to remark that, whenever a reduction rule from the above result, deletes a vertex and decrements $k$ by $1$, we can instead add such a vertex to our ``relevant set''. Also, since we will assume that the input graph comes with the promise that it admits an at most $k$-sized deletion set to chordal graphs, the known algorithms cannot return no for such a graph with the integer $k$ as input. Also, the kernelization algorithm for {\sc Interval Vertex Deletion} given in, say,~\cite{DBLP:journals/talg/AgrawalLMSZ19}, never explicitly returns yes.

\begin{observation}\label{obs:chordal-app}
There are computable (non-decreasing) functions $f_{\sf cdl}, g_{\sf cdl}: \mathbb{N} \rightarrow \mathbb{N}$, and an algorithm, which given an $n$-vertex graph $G$, an integer $k \geq 1$ and a maximal clique $G'$ in $G$ such that: i) $\hhmdp{\C{H}_{\sf cdl}}(G) \leq k$, and ii) $|V(G')| \geq g_{\C{H}_{\sf cdl}}(k)$; in time bounded by $f_{\C{H}_{\sf cdl}}(k) \cdot n^{\C{O}(1)}$, outputs a set $X \subseteq V(G')$, such that $|X| \leq g_{\C{H}_{\sf cdl}}(k)$ and each vertex in $V(G') \setminus X$ is $(\C{H}_{\sf cdl},k)$-irrelevant in $G$. 
\end{observation} 

A result similar to the above can be obtained for the family of interval graphs, say, using Lemma 4.2 and Section 5 of~\cite{AgrawalM0Z19}.  

\begin{observation}\label{obs:interval-app}
There are computable (non-decreasing) functions $f_{\sf int}, g_{\sf int}: \mathbb{N} \rightarrow \mathbb{N}$, and an algorithm, which given an $n$-vertex graph $G$, an integer $k \geq 1$ and a maximal clique $G'$ in $G$ such that: i) $\hhmdp{\C{H}_{\sf int}}(G) \leq k$, and ii) $|V(G')| \geq g_{\C{H}_{\sf int}}(k)$; in time bounded by $f_{\C{H}_{\sf int}}(k) \cdot n^{\C{O}(1)}$, outputs a set $X \subseteq V(G')$, such that $|X| \leq g_{\C{H}_{\sf int}}(k)$ and each vertex in $V(G') \setminus X$ is $(\C{H}_{\sf int},k)$-irrelevant in $G$. 
\end{observation}

Next we state a result which will help us find an irrelevant vertex for the case of $\C{H}_{\sf bip}$, which can be obtained from Lemma 3.2 of~\cite{DBLP:journals/corr/LokshtanovRS17}. 

\begin{observation}\label{obs:bipartite-app}
There are computable (non-decreasing) functions $f_{\sf bip}, g_{\sf bip}: \mathbb{N} \rightarrow \mathbb{N}$, and an algorithm, which given an $n$-vertex graph $G$, an integer $k \geq 1$, and a well-linked set $Z$ in $G$ such that: i) $\hhmdp{\C{H}_{\sf bip}}(G) \leq k$, and ii) $|Z| \geq g_{\C{H}_{\sf bip}}(k)$;\footnote{For the definition of well-linked sets, please see~\cite{DBLP:journals/corr/LokshtanovRS17}} in time bounded by $f_{\C{H}_{\sf bip}}(k) \cdot n^{\C{O}(1)}$, outputs a set $X \subseteq Z$, such that $|X| \leq g_{\C{H}_{\sf bip}}(k)$ and each vertex in $V(G') \setminus X$ is $(\C{H}_{\sf bip},k)$-irrelevant in $G$.
\end{observation}

In the next lemma we show how we can turn Observation~\ref{obs:chordal-app} to~\ref{obs:bipartite-app} to find a boundaried $(\C{H},k)$-irrelevant vertex (together with the known algorithms for deletion to $\C{H} \in \mathbb{H}$~\cite{ReedSV04,Marx10,CaoM15,FominLPSZ20,CyganFKLMPPS15,DBLP:journals/corr/LokshtanovRS17}).   

\begin{lemma}\label{lem:convert-boundaried-irrelevant}
Consider $\C{H} \in \{\C{H}_{\sf chl}, \C{H}_{\sf int}, \C{H}_{\sf bip}\}$. There are computable functions $f^*_{\C{H}}, g^*_{\C{H}} : \mathbb{N} \times \mathbb{N} \rightarrow \mathbb{N}$, and an algorithm, which given a $t$-boundaried graph $G$, where $t$ is the largest number in $\Lambda(G)$, an integer $k \geq 0$, and an equivalence class $Q = (t^*, \Lambda \subseteq [t^*], (z_1,z_2, \cdots, z_{q[\C{H}_{\sf cdl}, t^*]}))$, such that if $t^* \neq 0$, then $t^* \in \Lambda$; in time bounded by $f^*_{\C{H}}(t,k) \cdot n^{\C{O}(1)}$ it does one of the following:
\begin{enumerate}
\item correctly conclude that $\tw(G) \leq g^*_{\C{H}}(t+k)$,\footnote{The treewidth of a boundaried graph is same as the treewidth of its unboundaried counterpart.},
\item find a boundaried $(\C{H},k)$-irrelevant vertex $v \in V(G) \setminus \delta(G)$, or
\item correctly conclude that there is no $S \subseteq V(G)$ of size at most $k$, such that $G-S$ belongs to an equivalence class $Q'$, where ${\sf ug}_{\C{H}_{\sf cdl}}(Q',Q) = 1$ (see Definition~\ref{def:ug-function-apps}). 
\end{enumerate}   
\end{lemma}
\begin{proof}
We explain the proof for the case when $\C{H} = \C{H}_{\sf cdl}$, as the proof for all the other cases can be obtained in a similar fashion. Recall that we have an ordering on ${\sf RPBO}(\C{H}_{\sf cdl},t)$, which has allowed us to assume that ${\sf RBPO}(\C{H},t) = \{P[\C{H},t,1], P[\C{H},t,2], \cdots,$ $P[\C{H},t,q[\C{H},t]]\}$. Let $q = |{\sf RPBO}(\C{H}_{\sf cdl},t)|$. The problem {\sc Chordal Vertex Deletion} admits an {\FPT} algorithm running in time $f(k) \cdot n^{\C{O}(1)}$, where $f$ is a computable function~\cite{Marx10}.

Let $g^*_{\C{H}}(t',k') = |{\sf RPBO}(\C{H}_{\sf cdl},t')| \cdot g_{\C{H}_{\sf cdl}}(k') + \sum_{i \in [q[\C{H}_{\sf cdl}, t']]} |V(P[\C{H}_{\sf cdl}, t'])| + k' + 1 $, for $t',k' \in \mathbb{N}$, where $g_{\C{H}_{\sf cdl}}$ is the function from Observation~\ref{obs:chordal-app}. For each $i \in [q]$, such that $z_i = 1$, we do the following:  
\begin{enumerate}
\item If $(G\oplus P[\C{H},t,i], k)$ is a no-instance of {\sc Chordal Vertex Deletion}, then report that there is no $S \subseteq V(G)$ of size at most $k$, such that $G-S$ belongs to an equivalence class $Q'$, where ${\sf ug}_{\C{H}_{\sf cdl}}(Q',Q) = 1$. The correctness of this step of the algorithm follows from its description. Moreover, it can be executed in time bounded by $f(k) \cdot n^{\C{O}(1)}$. 

\item We now assume that $(G\oplus P[\C{H},t,i], k)$ is a yes-instance of {\sc Chordal Vertex Deletion}, and we compute a tree decomposition $(T,\chi)$ for $G\oplus P[\C{H},t,i]$, where each bag is a clique along with at most $k+1$ more vertices (see~\cite{golumbic2004algorithmic}). If each bag in the tree decomposition has at most $g^*_{\C{H}}(t) + 1$ vertices, then the tree width of $G\oplus P[\C{H},t,i]$ (and thus, $G$) can be bounded by $g^*_{\C{H}}(t,k)$. Note that the above step can be done in time bounded by $f(k) \cdot n^{\C{O}(1)}$.  

\item Otherwise, we can obtain that there is a clique $G'$ in $G \oplus P[\C{H},t]$, such that $V(G') \subseteq V(G)$, and ii) $|V(G')| > |{\sf RPBO}(\C{H}_{\sf cdl},t)| \cdot g_{\C{H}_{\sf cdl}}(k)$ vertices. Now using Observation~\ref{obs:obst-t-boundary}, we obtain $X_i\subseteq V(G')$, such that: i) $|X_i| \leq g_{\C{H}_{\sf cdl}}(k)$, and ii) each $v \in V(G') \setminus X_i$ is $(\C{H}_{\sf cdl},k)$-relevant in $G$.  
\end{enumerate}

Let $X = \cup_{i \in q} X_i$ and $Y = V(G' \setminus )$. We note that we use the same $G'$ to compute $X_i$s, for each $i \in [q]$. Note that $|X|\leq |{\sf RPBO}(\C{H}_{\sf cdl},t)| \cdot g_{\C{H}_{\sf cdl}}(k)$ and $|V(G')| > |{\sf RPBO}(\C{H}_{\sf cdl},t)| \cdot g_{\C{H}_{\sf cdl}}(k)$, and thus $Y \neq \emptyset$. Moreover, notice that each $v \in Y$ is a boundaried $(\C{H}_{\sf cdl},k)$-irrelevant vertex. This concludes the proof. 
\end{proof}

Using Lemma 3.1, Theorem 6 and 7 from~\cite{DBLP:journals/corr/abs-1904-02944}, we can obtain the following result for $\C{H}_{\C{F}_{\sf top}}$. Roughly speaking, the notion of ``irrelevance'' with respect to extended-folio presented in~\cite{DBLP:journals/corr/abs-1904-02944} is a stronger notion than the one we give in Definition~\ref{def:irrelev-one} (which is equivalent to Definition~\ref{def:irrelev-two}, see Lemma~\ref{lemma:irrelevance}). This allows us to directly borrow their algorithmic detection of irrelevant vertices for our case.  

\begin{observation}\label{obs:top-app}
Consider a fixed finite family of graphs $\C{F}$. There are computable functions $f^*_{\C{F}_{\sf top}}, g^*_{\C{F}_{\sf top}}: \mathbb{N} \times \mathbb{N} \rightarrow \mathbb{N}$, and an algorithm, which given an $n$-vertex $t$-boundaried graph $G$ and an integer $k \geq 1$, in time bounded $f^*_{\C{F}_{\sf top}}(t,k) \cdot n^{\C{O}(1)}$, either correctly concludes that $\tw(G) \leq g^*_{\C{F}_{\sf min}}(t,k)$, or finds a boundaried $(\C{H}_{{\C{F}}_{\sf top}},k)$-irrelevant vertex $v \in V(G) \setminus \delta(G)$. 
\end{observation} 

We will next prove a lemma that will allow us to find irrelevant vertices for $\C{H}_{\C{F}_{\sf ind}}$, where $\C{F}$ is a finite family of graphs, using an application of the Sunflower Lemma~\cite{erdos1960intersection, CyganFKLMPPS15}.  

\begin{lemma}\label{lem:ind-app}
Consider a fixed finite family of graphs $\C{F}$. There are computable functions $f_{\C{F}_{\sf ind}}, g_{\C{F}_{\sf ind}}: \mathbb{N} \rightarrow \mathbb{N}$, and an algorithm, which given an $n$-vertex $t$-boundaried graph $G$ and an integer $k \geq 1$, where $n \geq g_{\C{F}_{\sf ind}}(k)$, in time bounded $f_{\C{F}_{\sf ind}}(k) \cdot n^{\C{O}(1)}$, finds a boundaried $(\C{H}_{\C{F}_ {\sf ind}},k)$-irrelevant vertex $v \in V(G) \setminus \delta(G)$. 
\end{lemma}
\begin{proof}
Let $\C{H}=\C{H}_{\C{F}_ {\sf ind}}$. Let $c=c_{\C{F}}$ be maximum among the size of the largest graph in $\C{F}$, denoted by $d$, and the size of $\C{F}$. So, $c$ is a fixed constant. Let ${\cal O}=\{O: O$ is a $t$-boundaried graph that is an induced subgraph of at least one graph in $\C{F}\}$. Notice that $|{\cal O}|\leq c\cdot 2^c\cdot \sum_{i=0}^{\min(t,c)}{c\choose i}\leq \OO(1)$, and for each $O\in{\cal O}$, $|V(O)|\leq c$. We define $f_{\C{F}_{\sf ind}}(x)=|{\cal O}|$, and $g_{\C{F}_{\sf ind}}(x)=d^2\cdot d!\cdot (x-1)^d\cdot |{\cal O}|+1$ for all $x\in\mathbb{N}$.

Now, we describe the algorithm. Let $(G,k)$ be its input. Then:
\begin{enumerate}
\item Initialize $M=\emptyset$.
\item For every $O\in{\cal O}$:
	\begin{enumerate}
	\item Let $G'=G\oplus O$.
	\item Let $U=V(G')$ and ${\cal Q}=\{A\subseteq U: G'[A]\in\C{F}\}$.
	\item Call the algorithm in Theorem 2.26~\cite{CyganFKLMPPS15} on $(U,{\cal Q},k)$ to obtain $(U',{\cal H}',k)$.
	\item Update $M\leftarrow M\cup (U'\cap V(G))$.
	\end{enumerate}
\item Return any vertex $v$ from $V(G)\setminus M$.
\end{enumerate}

We first argue that $V(G)\setminus M\neq\emptyset$, thus the algorithm returns a vertex. For this purpose, note that the number of iterations is $|{\cal O}|$. By Theorem 2.26~\cite{CyganFKLMPPS15}, in each one of them, at most $d^2\cdot d!\cdot (k-1)^d$ new vertices are inserted into $M$. Thus, at the end, $|M|\leq d^2\cdot d!\cdot (k-1)^d\cdot |{\cal O}|<g_{\C{F}_{\sf ind}}(k)$.
Second, we consider the time complexity of the algorithm. Notice that the algorithm in Theorem 2.26~\cite{CyganFKLMPPS15} runs in polynomial time, hence each iteration is performed in polynomial time, which means that the overall runtime is $|{\cal O}|\cdot n^{\OO(1)}\leq f_{\C{F}_{\sf ind}}(k) \cdot n^{\C{O}(1)}\leq n^{\OO(1)}$.

We now consider the correctness of the algorithm.
Trivially, ${\sf RBPO}({\C{H}},t)\subseteq{\cal O}$.
So, by Lemma \ref{lemma:irrelevance}, it suffice to show that $v$ is $k$-irrelevant in $G\oplus O$ for every $O\in{\cal O}$. For this purpose, fix some $O\in{\cal O}$, and consider the iteration corresponding to this $O$. We need to show that $(G',k)$ is a yes-instance if and only if $(G'-v,k)$ is a yes-instance. One direction is trivial: If $(G',k)$ is a yes-instance, then $(G'-v,k)$ is a yes-instance. For the other direction, suppose that $(G'-v,k)$ is a yes-instance. Let $U^\star=V(G'-v)$ and ${\cal Q}^\star=\{A\subseteq U^\star: G'[A]\in\C{F}\}$. Then, $(U^\star,{\cal Q}^\star,k)$ is a yes-instance of {\sc $d$-Hitting Set}. As $v\notin M$, we have that $U'\subseteq U^\star$ and ${\cal Q}'\subseteq {\cal Q}^\star$. However, this means that $(U',{\cal Q}',k)$ is also a yes-instance of {\sc $d$-Hitting Set}. In turn, because $(U,{\cal Q},k)$ and $(U',{\cal Q}',k)$ are equivalent (see, for example, Theorem 2.26~\cite{CyganFKLMPPS15}), we have that $(U,{\cal Q},k)$ is also a yes-instance of {\sc $d$-Hitting Set}. Hence, $(G',k)$ is a yes-instance. This completes the proof.
\end{proof}

We next handle the bounded treewidth case. 

\begin{lemma}\label{lemma:twbounded}
There is a function $f_{\C{H}}$, and an algorithm that given $\C{H} \in \mathbb{H}$, a $t$-boundaried graph $G$,  integers $k, \ell \in \mathbb{N}$, and an equivalence class $Q=(t^{\star}, \{t^{\star}\}\subseteq \Lambda \subseteq [t^{\star}], (z_1,\ldots,z_{q[\C{H},t]}))$ such that the $\tw(G) \leq \ell$, runs in time $f_{\C{H}}(k,\ell,t)n^{\OO(1)}$ and outputs a vertex subset $S$ of size at most $k$ 
such that $G-S$ belongs to an equivalence class which is at least as good as $Q$.  
\end{lemma}

\begin{proof}[Proof sketch]
Recall that ${\sf RBPO}(\C{H},t) = \{P[\C{H},t,1], P[\C{H},t,2], \cdots,$ $P[\C{H},t,q[\C{H},t]]\}$. For simplicity, let us denote $H_i=P[\C{H},t,i]$ for all $i\in [q[\C{H},t]]$. Recall that a graph $G'$ belongs to an equivalence class which is at least as good as $Q$ if the following holds.
\begin{itemize}
\item[(i)] $\max \Lambda(G') \leq t^{\star}$ and $\Lambda(G') \subseteq \Lambda$
\item[(ii)]  For each $i\in q[\C{H},t]]$, if $z_i=1$ then $G'\oplus H \in \C{H}$.  
\end{itemize}
One can write a {\sf CMSO} formula $\psi_1$ such that  a $t$-boundaried graph $G'$ 
satisfies the formula $\psi_1$ if and only if the property $(i)$ holds. Also, one can write a {\sf CMSO} formula $\psi_2$ such that $G'$ satisfies the formula $\psi_2$ if and only if $G'\in \C{H}$. 

Here our objective is to test whether there are $k$ vertices $x_1,\ldots,x_k$ such that 
for each $i\in q[\C{H},t]]$, $(G\oplus H_i)-S$ satisfies $\psi_1$ and for each $i\in [q[\C{H},t]]]$ such that $z_i=1$, $(G\oplus H_i) -S$ satisfies $\psi_2$, where $S=\{x_1,\ldots,x_k\}$. We can write one {\sf CMSO} formula for it. Towards that we construct a graph 
$$G^{\star}=((G\oplus H_1)\oplus H_2)\ldots ) \oplus H_{q[\C{H},t]]}.$$ Then we will have free variables for each vertex subsets $V(G)$, $V(H_1),\ldots V(H_{q[\C{H},t]]})$ and for each edge subsets $E(G)$, $E(H_1),\ldots E(H_{q[\C{H},t]]})$. Then, each graph $G\oplus H_i$ is specified by four free variables $V(G),V(H_i),E(G)$, and $E(H_i)$. Thus, our objective is to test the existence of  a vertex subset $S\subseteq V(G^{\star})$ of size $k$ such that for each $i$, the subgraph $G\oplus H_i$ satisfies $\psi_1$ and if $z_i=1$, then it satisfies $\psi_2$. Thus, it can be expressed using a {\sf CMSO} formula $\psi$. This formula size is upper bounded by a function of $k, t$ and $\C{H}$. 

We know that  $\tw(G) \leq \ell$ and $\{H_1,H_2,\ldots,H_{q[\C{H}}\}$ is a fixed finite set of graphs independent of the input graph, we have that the $\tw(G^{\star})$ is upper bounded by a function of $k,\ell,t$ and $\C{H}$. Therefore the lemma follows from Courcelle's Theorem~\cite{Courcelle90}.
\end{proof}

We are now ready to prove Theorem~\ref{thm:apps-uniform}.  

\begin{proof}[Proof of Theorem~\ref{thm:apps-uniform}]
Consider $\C{H} \in \mathbb{H}$. Let $f^*, g^*: \mathbb{N}\times \mathbb{N} \rightarrow \mathbb{N}$ be the appropriate functions based on $\C{H}$ returned by Lemma~\ref{lem:convert-boundaried-irrelevant}, Observation~\ref{obs:top-app} or Lemma~\ref{lem:ind-app}. To prove the theorem, we will use Theorem~\ref{thm:uni:elim}, together with the refinement of the equivalence class given in Definition~\ref{def:equiv-bhv}, the function ${\sf ug}_{\C{H}}$, given in Definition~\ref{def:ug-function-apps}, and the algorithm $\C{A}$ we describe next. Our algorithm $\C{A}$ will take as input a $t$-boundaried graph $G$, where $t \in \Gamma(G)$, an integer $k$, and an equivalence class $Q \in \C{Q}(\equiv_{u,\C{H}})$, in time bounded by $h(t,k) \cdot n^{\C{O}(1)}$, it will check if there is $S \subseteq V(G)$ of size at most $k$, such that $G-S$ is in the equivalence class $Q'$, where ${\sf ug}_{\C{H}}(Q',Q) = 1$. 

If Lemma~\ref{lem:convert-boundaried-irrelevant}, Observation~\ref{obs:top-app} or Lemma~\ref{lem:ind-app} concludes that the treewidth of $G$ is bounded by $g^*(t,k)$, then we resolve the instance using Lemma~\ref{lemma:twbounded}. Also, if using one of Lemma~\ref{lem:convert-boundaried-irrelevant}, Observation~\ref{obs:top-app} or Lemma~\ref{lem:ind-app} we are able to obtain the a set $S$ of desired type does not exists, the algorithm returns no. Otherwise, we have a boundaried $(\C{H},k)$-irrelevant vertex $v \in V(G)$, and we do the following: i) If $v \notin \delta(G)$, we (recursively) solve the instance $(G,k,Q)$; ii) Otherwise, let $Q'= t', \Lambda', \overline{z}'$ be the equivalence class obtained from $Q = (t^*, \Lambda, \overline{z})$, where $\Lambda' = \Lambda \setminus \lambda_G(v)$, $t'$ is the largest number in $\Lambda'$, and $\overline{z}'$ is the restriction of $\overline{z}$ to ${\sf RBPO}(\C{H},t') \subseteq {\sf RBPO}(\C{H},t^*)$. We (recursively) solve the instance $(G-\{v\},k,Q')$. The correctness of the algorithm follows from its description. Also, from Lemma~\ref{lem:convert-boundaried-irrelevant}, Observation~\ref{obs:top-app}, Lemma~\ref{lem:ind-app} and Lemma~\ref{lemma:twbounded}, we can obtain that the algorithm is uniform and it runs in time bounded by $h(t,k) \cdot n^{\C{O}(1)}$, for some function $h: \mathbb{N}\times \mathbb{N} \rightarrow \mathbb{N}$.  
\end{proof}

%% file: elimination-conclusion.tex
In this paper, we have shown that as far as the task of identifying the boundaries of tractability in parameterized complexity is concerned, two recently popular hybrid parameterizations that combine solution size and width measures ($\cH$-elimination distance and $\cH$-treewidth) are effectively only as powerful as the standard parameterization by the size of the modulator to $\cH$ for a host of commonly studied graph classes $\cH$. This is a surprising result, and one that unifies several recent results in the literature. Moreover, using our main result, we have resolved several open problems that were either stated explicitly or have naturally arisen in recent literature. We have also developed a framework for cross-parameterization and demonstrated how this can be applied to answer an open problem posed by Jansen et al.~\cite{JansenK021}.   

Hybrid parameterizations have been the subject of a flurry of interest in the last half-a-decade and it would be interesting to identify new, algorithmically useful, hybrid parameterizations that are provably stronger than existing ones. Further, as our characterization result gives non-uniform algorithms, we have developed a framework to design uniform {\FPT} algorithms to compute elimination distance to $\cH$ when $\cH$ has certain properties. We have demonstrated that these properties are fairly mild by showing a number of well-studied graph classes that possess these properties.

We leave the following questions arising from our work as interesting future research directions. 

\begin{itemize}
	\item Is {\probEDH} parameterized by $\hhtw(G)$ equivalent to the eight problems in Theorem~\ref{thm:mainEquiv} for hereditary, union-closed and CMSO definable $\cH$? We conjecture that the answer is yes. 
	\item Could one develop a framework to design uniform {\FPT} algorithms for {\probTDH} in the same spirit as our framework to design uniform {\FPT} algorithms for {\probEDH}? 
\end{itemize}

Finally, although we provide powerful classification tools, the fact still remains that $\hhtw(G)\leq \hhdepth(G)\leq \hhmd(G)$ and that each parameter could be arbitrarily smaller than the parameters to its right. Thus, it is still an interesting direction of research to study {\probVDH} for specific classes $\cH$ parameterized by $\hhdepth(G)$ and $\hhtw(G)$ and aim to optimize the running time, in the spirit of Jansen et al~\cite{JansenK021}. 

%% file: appendix.tex
\section{Problem Definitions}

\smallskip
\noindent{{\sc Chordal Vertex Deletion}}\\ 
\noindent{\em Input:}  An undirected graph $G$, and
a positive integer $k$.\\
\noindent{\em Output: } Does there exist a vertex subset $S$ of size at most $k$ such that
 $G-S$ is a chordal graph?\\

 \smallskip
\noindent{{\sc Feedback Vertex Set (FVS)}}\\ 
\noindent{\em Input: } An undirected graph $G$, and a positive integer $k$.\\
\noindent{\em Output: } Does there exist a vertex subset $S$ of size at most $k$ that intersects all cycles in $G$?\\

 \smallskip
\noindent{{\sc Hitting Set}}\\ 
\noindent{\em Input: } A universe $U$, a family $\cal A$ of sets over $U$,
and an integer $k$.\\
\noindent{\em Output: }Does there exist a set $X\subseteq U$ of size at most $k$ that has a nonempty intersection with every element of $\cal A$?\\

\smallskip
\noindent{{\sc Interval Vertex Deletion}}\\ 
\noindent{\em Input: } An undirected graph $G$, and
a positive integer $k$.\\
\noindent{\em Output: } Does there exist a vertex subset $S$ of size at most $k$ such that
$G-S$ is an interval graph?\\

\smallskip
\noindent{{\sc Odd Cycle Transversal (OCT)}}\\ 
\noindent{\em Input: } An undirected graph $G$, and
a positive integer $k$.\\
\noindent{\em Output: } Does there exist a vertex subset $S$ of size at most $k$ that intersects all odd cycles in $G$?\\

\smallskip
\noindent{{\sc Mutiway Cut}}\\ 
\noindent{\em Input: } An undirected graph $G$, a vertex subset $T$, and
a positive integer $k$.\\
\noindent{\em Question: } Does there exist a vertex subset $S$ of size at most $k$ such that
each connected component of $G-S$ has at most one vertex of $T$?\\

\smallskip
\noindent{{\sc Planar Vertex Deletion}}\\ 
\noindent{\em Input: } An undirected graph $G$, and
a positive integer $k$.\\
\noindent{\em Output: } Does there exist a vertex subset $S$ of size at most $k$ such that
$G-S$ is a planar graph?\\

\smallskip
\noindent{{\sc Subset Feedback Vertex Set (Subset FVS)}}\\ 
\noindent{\em Input: } An undirected graph $G$, a vertex subset $T$, and
a positive integer $k$.\\
\noindent{\em Output: } Does there exist a vertex subset $S$ of size at most $k$ that intersects all cycles containing a vertex of $T$?\\

\smallskip
\noindent{{\sc Subset Odd Cycle Transversal (Subset OCT)}}\\ 
\noindent{\em Input: } An undirected graph $G$, a vertex subset $T$, and
a positive integer $k$.\\
\noindent{\em Output: } Does there exist a vertex subset $S$ of size at most $k$ that intersects all odd cycles containing a vertex of $T$?\\

\smallskip
\noindent{{\sc Topological Minor Deletion}}\\ 
\noindent{\em Input: } An undirected graph $G$, a family of undirected graphs $\cal F$, and
a positive integer $k$.\\
\noindent{\em Output: } Does there exist a vertex subset $S$ of size at most $k$ such that
 $G-S$ contains no graph from $\cal H$ as a topological minor?\\

\smallskip
\noindent{{\sc Vertex Cover}}\\ 
\noindent{\em Input: } An undirected graph $G$, and
a positive integer $k$.\\
\noindent{\em Output: } Does there exist a vertex subset $S$ of size at most $k$ such that
 $G-S$ is edgeless?\\

%% file: main.bbl
\begin{thebibliography}{10}

\bibitem{abrahamson1993finite}
Karl Abrahamson and Michael Fellows.
\newblock Finite automata, bounded treewidth, and well-quasiordering.
\newblock {\em Contemporary Mathematics}, 147:539--539, 1993.

\bibitem{AgrawalKFR21}
Akanksha Agrawal, Lawqueen Kanesh, Fahad Panolan, M.~S. Ramanujan, and Saket
  Saurabh.
\newblock {An FPT Algorithm for Elimination Distance to Bounded Degree Graphs}.
\newblock In Markus Bl\"{a}ser and Benjamin Monmege, editors, {\em 38th
  International Symposium on Theoretical Aspects of Computer Science (STACS
  2021)}, volume 187 of {\em Leibniz International Proceedings in Informatics
  (LIPIcs)}, pages 5:1--5:11, Dagstuhl, Germany, 2021. Schloss Dagstuhl --
  Leibniz-Zentrum f{\"u}r Informatik.
\newblock \href {https://doi.org/10.4230/LIPIcs.STACS.2021.5}
  {\path{doi:10.4230/LIPIcs.STACS.2021.5}}.

\bibitem{DBLP:journals/talg/AgrawalLMSZ19}
Akanksha Agrawal, Daniel Lokshtanov, Pranabendu Misra, Saket Saurabh, and
  Meirav Zehavi.
\newblock Feedback vertex set inspired kernel for chordal vertex deletion.
\newblock {\em {ACM} Trans. Algorithms}, 15(1):11:1--11:28, 2019.
\newblock \href {https://doi.org/10.1145/3284356} {\path{doi:10.1145/3284356}}.

\bibitem{AgrawalM0Z19}
Akanksha Agrawal, Pranabendu Misra, Saket Saurabh, and Meirav Zehavi.
\newblock Interval vertex deletion admits a polynomial kernel.
\newblock In Timothy~M. Chan, editor, {\em Proceedings of the Thirtieth Annual
  {ACM-SIAM} Symposium on Discrete Algorithms, {SODA} 2019, San Diego,
  California, USA, January 6-9, 2019}, pages 1711--1730. {SIAM}, 2019.
\newblock \href {https://doi.org/10.1137/1.9781611975482.103}
  {\path{doi:10.1137/1.9781611975482.103}}.

\bibitem{Agrawal020}
Akanksha Agrawal and M.~S. Ramanujan.
\newblock On the parameterized complexity of clique elimination distance.
\newblock In Yixin Cao and Marcin Pilipczuk, editors, {\em 15th International
  Symposium on Parameterized and Exact Computation, {IPEC} 2020, December
  14-18, 2020, Hong Kong, China (Virtual Conference)}, volume 180 of {\em
  LIPIcs}, pages 1:1--1:13. Schloss Dagstuhl - Leibniz-Zentrum f{\"{u}}r
  Informatik, 2020.
\newblock \href {https://doi.org/10.4230/LIPIcs.IPEC.2020.1}
  {\path{doi:10.4230/LIPIcs.IPEC.2020.1}}.

\bibitem{ArnborgLS91}
Stefan Arnborg, Jens Lagergren, and Detlef Seese.
\newblock Easy problems for tree-decomposable graphs.
\newblock {\em Journal of Algorithms}, 12:308--340, 1991.

\bibitem{Bodlaender96}
Hans~L. Bodlaender.
\newblock A linear-time algorithm for finding tree-decompositions of small
  treewidth.
\newblock {\em {SIAM} J. Comput.}, 25(6):1305--1317, 1996.
\newblock \href {https://doi.org/10.1137/S0097539793251219}
  {\path{doi:10.1137/S0097539793251219}}.

\bibitem{BodlaenderFLPST16}
Hans~L. Bodlaender, Fedor~V. Fomin, Daniel Lokshtanov, Eelko Penninkx, Saket
  Saurabh, and Dimitrios~M. Thilikos.
\newblock (meta) kernelization.
\newblock {\em J. {ACM}}, 63(5):44:1--44:69, 2016.
\newblock \href {https://doi.org/10.1145/2973749} {\path{doi:10.1145/2973749}}.

\bibitem{BodlaenderF01}
Hans~L. Bodlaender and Babette van Antwerpen{-}de~Fluiter.
\newblock Reduction algorithms for graphs of small treewidth.
\newblock {\em Inf. Comput.}, 167(2):86--119, 2001.
\newblock \href {https://doi.org/10.1006/inco.2000.2958}
  {\path{doi:10.1006/inco.2000.2958}}.

\bibitem{BoriePT92}
Richard~B. Borie, R.~Gary Parker, and Craig~A. Tovey.
\newblock Automatic generation of linear-time algorithms from predicate
  calculus descriptions of problems on recursively constructed graph families.
\newblock {\em Algorithmica}, 7(5{\&}6):555--581, 1992.
\newblock \href {https://doi.org/10.1007/BF01758777}
  {\path{doi:10.1007/BF01758777}}.

\bibitem{BulianD16}
Jannis Bulian and Anuj Dawar.
\newblock Graph isomorphism parameterized by elimination distance to bounded
  degree.
\newblock {\em Algorithmica}, 75(2):363--382, 2016.
\newblock \href {https://doi.org/10.1007/s00453-015-0045-3}
  {\path{doi:10.1007/s00453-015-0045-3}}.

\bibitem{BulianD17}
Jannis Bulian and Anuj Dawar.
\newblock Fixed-parameter tractable distances to sparse graph classes.
\newblock {\em Algorithmica}, 79(1):139--158, 2017.
\newblock \href {https://doi.org/10.1007/s00453-016-0235-7}
  {\path{doi:10.1007/s00453-016-0235-7}}.

\bibitem{Cai96}
Leizhen Cai.
\newblock Fixed-parameter tractability of graph modification problems for
  hereditary properties.
\newblock {\em Inf. Process. Lett.}, 58(4):171--176, 1996.
\newblock \href {https://doi.org/10.1016/0020-0190(96)00050-6}
  {\path{doi:10.1016/0020-0190(96)00050-6}}.

\bibitem{CaoM15}
Yixin Cao and D{\'{a}}niel Marx.
\newblock Interval deletion is fixed-parameter tractable.
\newblock {\em {ACM} Trans. Algorithms}, 11(3):21:1--21:35, 2015.
\newblock \href {https://doi.org/10.1145/2629595} {\path{doi:10.1145/2629595}}.

\bibitem{CaoM16}
Yixin Cao and D{\'{a}}niel Marx.
\newblock Chordal editing is fixed-parameter tractable.
\newblock {\em Algorithmica}, 75(1):118--137, 2016.
\newblock \href {https://doi.org/10.1007/s00453-015-0014-x}
  {\path{doi:10.1007/s00453-015-0014-x}}.

\bibitem{cayley1889theorem}
Arthur Cayley.
\newblock A theorem on trees.
\newblock {\em Quarterly Journal of Pure and Applied Mathematics}, 23:376--378,
  1889.

\bibitem{DBLP:journals/siamcomp/ChitnisCHPP16}
Rajesh Chitnis, Marek Cygan, MohammadTaghi Hajiaghayi, Marcin Pilipczuk, and
  Michal Pilipczuk.
\newblock Designing {FPT} algorithms for cut problems using randomized
  contractions.
\newblock {\em {SIAM} J. Comput.}, 45(4):1171--1229, 2016.
\newblock URL: \url{http://dx.doi.org/10.1137/15M1032077}, \href
  {https://doi.org/10.1137/15M1032077} {\path{doi:10.1137/15M1032077}}.

\bibitem{Courcelle90}
Bruno Courcelle.
\newblock The monadic second-order logic of graphs. i. recognizable sets of
  finite graphs.
\newblock {\em Inf. Comput.}, 85(1):12--75, 1990.
\newblock \href {https://doi.org/10.1016/0890-5401(90)90043-H}
  {\path{doi:10.1016/0890-5401(90)90043-H}}.

\bibitem{Courcelle97}
Bruno Courcelle.
\newblock The expression of graph properties and graph transformations in
  monadic second-order logic.
\newblock {\em Handbook of Graph Grammars}, pages 313--400, 1997.

\bibitem{CourcelleMR00}
Bruno Courcelle, Johann~A. Makowsky, and Udi Rotics.
\newblock Linear time solvable optimization problems on graphs of bounded
  clique-width.
\newblock {\em Theory Comput. Syst.}, 33(2):125--150, 2000.
\newblock \href {https://doi.org/10.1007/s002249910009}
  {\path{doi:10.1007/s002249910009}}.

\bibitem{CourcelleO00}
Bruno Courcelle and Stephan Olariu.
\newblock Upper bounds to the clique width of graphs.
\newblock {\em Discret. Appl. Math.}, 101(1-3):77--114, 2000.
\newblock \href {https://doi.org/10.1016/S0166-218X(99)00184-5}
  {\path{doi:10.1016/S0166-218X(99)00184-5}}.

\bibitem{CourcelleO07}
Bruno Courcelle and Sang{-}il Oum.
\newblock Vertex-minors, monadic second-order logic, and a conjecture by seese.
\newblock {\em J. Comb. Theory, Ser. {B}}, 97(1):91--126, 2007.
\newblock \href {https://doi.org/10.1016/j.jctb.2006.04.003}
  {\path{doi:10.1016/j.jctb.2006.04.003}}.

\bibitem{CyganFKLMPPS15}
Marek Cygan, Fedor~V Fomin, {\L}ukasz Kowalik, Daniel Lokshtanov, D{\'a}niel
  Marx, Marcin Pilipczuk, Micha{\l} Pilipczuk, and Saket Saurabh.
\newblock {\em Parameterized algorithms}.
\newblock Springer, 2015.
\newblock \href {https://doi.org/10.1007/978-3-319-21275-3}
  {\path{doi:10.1007/978-3-319-21275-3}}.

\bibitem{DBLP:journals/siamdm/CyganPPW13}
Marek Cygan, Marcin Pilipczuk, Michal Pilipczuk, and Jakub~Onufry Wojtaszczyk.
\newblock Subset feedback vertex set is fixed-parameter tractable.
\newblock {\em {SIAM} J. Discret. Math.}, 27(1):290--309, 2013.

\bibitem{de1997algorithms}
Babette Lucie~Elisabeth de~Fluiter.
\newblock {\em Algorithms for graphs of small treewidth}.
\newblock PhD thesis, 1997.

\bibitem{DowneyF13}
Rodney~G. Downey and Michael~R. Fellows.
\newblock {\em Fundamentals of Parameterized Complexity}.
\newblock Texts in Computer Science. Springer, 2013.
\newblock \href {https://doi.org/10.1007/978-1-4471-5559-1}
  {\path{doi:10.1007/978-1-4471-5559-1}}.

\bibitem{EibenGHK19}
Eduard Eiben, Robert Ganian, Thekla Hamm, and O{-}joung Kwon.
\newblock Measuring what matters: {A} hybrid approach to dynamic programming
  with treewidth.
\newblock In Peter Rossmanith, Pinar Heggernes, and Joost{-}Pieter Katoen,
  editors, {\em 44th International Symposium on Mathematical Foundations of
  Computer Science, {MFCS} 2019, August 26-30, 2019, Aachen, Germany}, volume
  138 of {\em LIPIcs}, pages 42:1--42:15. Schloss Dagstuhl - Leibniz-Zentrum
  f{\"{u}}r Informatik, 2019.
\newblock \href {https://doi.org/10.4230/LIPIcs.MFCS.2019.42}
  {\path{doi:10.4230/LIPIcs.MFCS.2019.42}}.

\bibitem{EibenGK18}
Eduard Eiben, Robert Ganian, and O{-}joung Kwon.
\newblock A single-exponential fixed-parameter algorithm for
  distance-hereditary vertex deletion.
\newblock {\em J. Comput. Syst. Sci.}, 97:121--146, 2018.
\newblock \href {https://doi.org/10.1016/j.jcss.2018.05.005}
  {\path{doi:10.1016/j.jcss.2018.05.005}}.

\bibitem{erdos1960intersection}
Paul Erd{\"o}s and Richard Rado.
\newblock Intersection theorems for systems of sets.
\newblock {\em Journal of the London Mathematical Society}, 1(1):85--90, 1960.

\bibitem{corr/abs-2104-02998}
Fedor~V. Fomin, Petr~A. Golovach, and Dimitrios~M. Thilikos.
\newblock Parameterized complexity of elimination distance to first-order logic
  properties.
\newblock {\em CoRR (to appear in the proceedings fo LICS 2021)},
  abs/2104.02998, 2021.
\newblock URL: \url{https://arxiv.org/abs/2104.02998}, \href
  {http://arxiv.org/abs/2104.02998} {\path{arXiv:2104.02998}}.

\bibitem{FominLMS12}
Fedor~V. Fomin, Daniel Lokshtanov, Neeldhara Misra, and Saket Saurabh.
\newblock Planar {${\cal F}$}-deletion: Approximation, kernelization and
  optimal {FPT} algorithms.
\newblock In {\em 53rd Annual {IEEE} Symposium on Foundations of Computer
  Science, {FOCS} 2012, New Brunswick, NJ, USA, October 20-23, 2012}, pages
  470--479. {IEEE} Computer Society, 2012.
\newblock \href {https://doi.org/10.1109/FOCS.2012.62}
  {\path{doi:10.1109/FOCS.2012.62}}.

\bibitem{DBLP:journals/corr/abs-1904-02944}
Fedor~V. Fomin, Daniel Lokshtanov, Fahad Panolan, Saket Saurabh, and Meirav
  Zehavi.
\newblock Hitting topological minors is fpt.
\newblock {\em CoRR}, abs/1904.02944, 2019.
\newblock URL: \url{http://arxiv.org/abs/1904.02944}.

\bibitem{FominLPSZ20}
Fedor~V. Fomin, Daniel Lokshtanov, Fahad Panolan, Saket Saurabh, and Meirav
  Zehavi.
\newblock Hitting topological minors is {FPT}.
\newblock In Konstantin Makarychev, Yury Makarychev, Madhur Tulsiani, Gautam
  Kamath, and Julia Chuzhoy, editors, {\em Proccedings of the 52nd Annual {ACM}
  {SIGACT} Symposium on Theory of Computing, {STOC} 2020, Chicago, IL, USA,
  June 22-26, 2020}, pages 1317--1326. {ACM}, 2020.
\newblock \href {https://doi.org/10.1145/3357713.3384318}
  {\path{doi:10.1145/3357713.3384318}}.

\bibitem{FominLST18}
Fedor~V. Fomin, Daniel Lokshtanov, Saket Saurabh, and Dimitrios~M. Thilikos.
\newblock Kernels for (connected) dominating set on graphs with excluded
  topological minors.
\newblock {\em {ACM} Trans. Algorithms}, 14(1):6:1--6:31, 2018.
\newblock \href {https://doi.org/10.1145/3155298} {\path{doi:10.1145/3155298}}.

\bibitem{DBLP:journals/combinatorica/FominV12}
Fedor~V. Fomin and Yngve Villanger.
\newblock Treewidth computation and extremal combinatorics.
\newblock {\em Comb.}, 32(3):289--308, 2012.

\bibitem{GanianOR17}
Robert Ganian, Sebastian Ordyniak, and M.~S. Ramanujan.
\newblock Going beyond primal treewidth for {(M)ILP}.
\newblock In Satinder~P. Singh and Shaul Markovitch, editors, {\em Proceedings
  of the Thirty-First {AAAI} Conference on Artificial Intelligence, February
  4-9, 2017, San Francisco, California, {USA}}, pages 815--821. {AAAI} Press,
  2017.
\newblock URL:
  \url{http://aaai.org/ocs/index.php/AAAI/AAAI17/paper/view/14272}.

\bibitem{GanianRS17a}
Robert Ganian, M.~S. Ramanujan, and Stefan Szeider.
\newblock Backdoor treewidth for {SAT}.
\newblock In Serge Gaspers and Toby Walsh, editors, {\em Theory and
  Applications of Satisfiability Testing - {SAT} 2017 - 20th International
  Conference, Melbourne, VIC, Australia, August 28 - September 1, 2017,
  Proceedings}, volume 10491 of {\em Lecture Notes in Computer Science}, pages
  20--37. Springer, 2017.
\newblock \href {https://doi.org/10.1007/978-3-319-66263-3_2}
  {\path{doi:10.1007/978-3-319-66263-3_2}}.

\bibitem{GanianRS17d}
Robert Ganian, M.~S. Ramanujan, and Stefan Szeider.
\newblock Combining treewidth and backdoors for {CSP}.
\newblock In Heribert Vollmer and Brigitte Vall{\'{e}}e, editors, {\em 34th
  Symposium on Theoretical Aspects of Computer Science, {STACS} 2017, March
  8-11, 2017, Hannover, Germany}, volume~66 of {\em LIPIcs}, pages 36:1--36:17.
  Schloss Dagstuhl - Leibniz-Zentrum f{\"{u}}r Informatik, 2017.
\newblock \href {https://doi.org/10.4230/LIPIcs.STACS.2017.36}
  {\path{doi:10.4230/LIPIcs.STACS.2017.36}}.

\bibitem{GanianRS17}
Robert Ganian, M.~S. Ramanujan, and Stefan Szeider.
\newblock Combining treewidth and backdoors for {CSP}.
\newblock In Heribert Vollmer and Brigitte Vall{\'{e}}e, editors, {\em 34th
  Symposium on Theoretical Aspects of Computer Science, {STACS} 2017, March
  8-11, 2017, Hannover, Germany}, volume~66 of {\em LIPIcs}, pages 36:1--36:17.
  Schloss Dagstuhl - Leibniz-Zentrum f{\"{u}}r Informatik, 2017.
\newblock \href {https://doi.org/10.4230/LIPIcs.STACS.2017.36}
  {\path{doi:10.4230/LIPIcs.STACS.2017.36}}.

\bibitem{GanianRS17b}
Robert Ganian, M.~S. Ramanujan, and Stefan Szeider.
\newblock Discovering archipelagos of tractability for constraint satisfaction
  and counting.
\newblock {\em {ACM} Trans. Algorithms}, 13(2):29:1--29:32, 2017.
\newblock \href {https://doi.org/10.1145/3014587} {\path{doi:10.1145/3014587}}.

\bibitem{golumbic2004algorithmic}
Martin~Charles Golumbic.
\newblock {\em Algorithmic graph theory and perfect graphs}.
\newblock Elsevier, 2004.

\bibitem{DBLP:conf/stoc/GroheKMW11}
Martin Grohe, Ken{-}ichi Kawarabayashi, D{\'{a}}niel Marx, and Paul Wollan.
\newblock Finding topological subgraphs is fixed-parameter tractable.
\newblock In {\em Proceedings of the 43rd {ACM} Symposium on Theory of
  Computing, {STOC} 2011, San Jose, CA, USA, 6-8 June 2011}, pages 479--488,
  2011.
\newblock URL: \url{http://doi.acm.org/10.1145/1993636.1993700}, \href
  {https://doi.org/10.1145/1993636.1993700}
  {\path{doi:10.1145/1993636.1993700}}.

\bibitem{HlinenyOSG08}
Petr Hlinen{\'{y}}, Sang{-}il Oum, Detlef Seese, and Georg Gottlob.
\newblock Width parameters beyond tree-width and their applications.
\newblock {\em Comput. J.}, 51(3):326--362, 2008.
\newblock \href {https://doi.org/10.1093/comjnl/bxm052}
  {\path{doi:10.1093/comjnl/bxm052}}.

\bibitem{corr/abs-2105-04660}
Ashwin Jacob, Jari J.~H. de~Kroon, Diptapriyo Majumdar, and Venkatesh Raman.
\newblock Parameterized complexity of deletion to scattered graph classes.
\newblock {\em CoRR}, abs/2105.04660, 2021.
\newblock URL: \url{https://arxiv.org/abs/2105.04660}, \href
  {http://arxiv.org/abs/2105.04660} {\path{arXiv:2105.04660}}.

\bibitem{JacobM020}
Ashwin Jacob, Diptapriyo Majumdar, and Venkatesh Raman.
\newblock Parameterized complexity of deletion to scattered graph classes.
\newblock In Yixin Cao and Marcin Pilipczuk, editors, {\em 15th International
  Symposium on Parameterized and Exact Computation, {IPEC} 2020, December
  14-18, 2020, Hong Kong, China (Virtual Conference)}, volume 180 of {\em
  LIPIcs}, pages 18:1--18:17. Schloss Dagstuhl - Leibniz-Zentrum f{\"{u}}r
  Informatik, 2020.
\newblock \href {https://doi.org/10.4230/LIPIcs.IPEC.2020.18}
  {\path{doi:10.4230/LIPIcs.IPEC.2020.18}}.

\bibitem{corr/abs-2106-04191}
Bart M.~P. Jansen and Jari J.~H. de~Kroon.
\newblock {FPT} algorithms to compute the elimination distance to bipartite
  graphs and more.
\newblock {\em CoRR (to appear in the proceedings fo WG 2021)}, abs/2106.04191,
  2021.
\newblock URL: \url{https://arxiv.org/abs/2106.04191}, \href
  {http://arxiv.org/abs/2106.04191} {\path{arXiv:2106.04191}}.

\bibitem{JansenKW21}
Bart M.~P. Jansen, Jari J.~H. de~Kroon, and Michal Wlodarczyk.
\newblock Vertex deletion parameterized by elimination distance and even less.
\newblock {\em CoRR}, abs/2103.09715, 2021.
\newblock \href {http://arxiv.org/abs/2103.09715} {\path{arXiv:2103.09715}}.

\bibitem{JansenK021}
Bart M.~P. Jansen, Jari J.~H. de~Kroon, and Michal Wlodarczyk.
\newblock Vertex deletion parameterized by elimination distance and even less.
\newblock In Samir Khuller and Virginia~Vassilevska Williams, editors, {\em
  {STOC} '21: 53rd Annual {ACM} {SIGACT} Symposium on Theory of Computing,
  Virtual Event, Italy, June 21-25, 2021}, pages 1757--1769. {ACM}, 2021.
\newblock \href {https://doi.org/10.1145/3406325.3451068}
  {\path{doi:10.1145/3406325.3451068}}.

\bibitem{JansenLS14}
Bart M.~P. Jansen, Daniel Lokshtanov, and Saket Saurabh.
\newblock A near-optimal planarization algorithm.
\newblock In {\em Proceedings of the twenty-fifth annual ACM-SIAM symposium on
  Discrete algorithms}, pages 1802--1811. SIAM, {SIAM}, 2014.
\newblock \href {https://doi.org/10.1137/1.9781611973402.130}
  {\path{doi:10.1137/1.9781611973402.130}}.

\bibitem{JansenP18}
Bart M.~P. Jansen and Marcin Pilipczuk.
\newblock Approximation and kernelization for chordal vertex deletion.
\newblock {\em {SIAM} J. Discret. Math.}, 32(3):2258--2301, 2018.
\newblock \href {https://doi.org/10.1137/17M112035X}
  {\path{doi:10.1137/17M112035X}}.

\bibitem{JansenRV14}
Bart M.~P. Jansen, Venkatesh Raman, and Martin Vatshelle.
\newblock Parameter ecology for feedback vertex set.
\newblock {\em Tsinghua Science and Technology}, 19(4):387--409, 2014.
\newblock Special Issue dedicated to Jianer Chen.
\newblock \href {https://doi.org/10.1109/TST.2014.6867520}
  {\path{doi:10.1109/TST.2014.6867520}}.

\bibitem{DBLP:conf/soda/KakimuraKK12}
Naonori Kakimura, Ken{-}ichi Kawarabayashi, and Yusuke Kobayashi.
\newblock Erd{\"{o}}s-p{\'{o}}sa property and its algorithmic applications:
  parity constraints, subset feedback set, and subset packing.
\newblock In Yuval Rabani, editor, {\em Proceedings of the Twenty-Third Annual
  {ACM-SIAM} Symposium on Discrete Algorithms, {SODA} 2012, Kyoto, Japan,
  January 17-19, 2012}, pages 1726--1736. {SIAM}, 2012.

\bibitem{DBLP:conf/focs/KawarabayashiT11}
Ken{-}ichi Kawarabayashi and Mikkel Thorup.
\newblock The minimum k-way cut of bounded size is fixed-parameter tractable.
\newblock In {\em {IEEE} 52nd Annual Symposium on Foundations of Computer
  Science, {FOCS} 2011, Palm Springs, CA, USA, October 22-25, 2011}, pages
  160--169, 2011.
\newblock URL: \url{http://dx.doi.org/10.1109/FOCS.2011.53}, \href
  {https://doi.org/10.1109/FOCS.2011.53} {\path{doi:10.1109/FOCS.2011.53}}.

\bibitem{DBLP:journals/talg/0002LPRRSS16}
Eun~Jung Kim, Alexander Langer, Christophe Paul, Felix Reidl, Peter Rossmanith,
  Ignasi Sau, and Somnath Sikdar.
\newblock Linear kernels and single-exponential algorithms via protrusion
  decompositions.
\newblock {\em {ACM} Trans. Algorithms}, 12(2):21:1--21:41, 2016.
\newblock \href {https://doi.org/10.1145/2797140} {\path{doi:10.1145/2797140}}.

\bibitem{lekkeikerker1962representation}
C~Lekkeikerker and J~Boland.
\newblock Representation of a finite graph by a set of intervals on the real
  line.
\newblock {\em Fundamenta Mathematicae}, 51(1):45--64, 1962.

\bibitem{LindermayrSV20}
Alexander Lindermayr, Sebastian Siebertz, and Alexandre Vigny.
\newblock Elimination distance to bounded degree on planar graphs.
\newblock In Javier Esparza and Daniel Kr{\'{a}}l', editors, {\em 45th
  International Symposium on Mathematical Foundations of Computer Science,
  {MFCS} 2020, August 24-28, 2020, Prague, Czech Republic}, volume 170 of {\em
  LIPIcs}, pages 65:1--65:12. Schloss Dagstuhl - Leibniz-Zentrum f{\"{u}}r
  Informatik, 2020.
\newblock \href {https://doi.org/10.4230/LIPIcs.MFCS.2020.65}
  {\path{doi:10.4230/LIPIcs.MFCS.2020.65}}.

\bibitem{DBLP:journals/siamdm/LokshtanovMRS17}
Daniel Lokshtanov, Pranabendu Misra, M.~S. Ramanujan, and Saket Saurabh.
\newblock Hitting selected (odd) cycles.
\newblock {\em {SIAM} J. Discret. Math.}, 31(3):1581--1615, 2017.

\bibitem{LokshtanovNRRS14}
Daniel Lokshtanov, N.~S. Narayanaswamy, Venkatesh Raman, M.~S. Ramanujan, and
  Saket Saurabh.
\newblock Faster parameterized algorithms using linear programming.
\newblock {\em {ACM} Trans. Algorithms}, 11(2):15:1--15:31, 2014.
\newblock \href {https://doi.org/10.1145/2566616} {\path{doi:10.1145/2566616}}.

\bibitem{DBLP:journals/corr/LokshtanovRS17}
Daniel Lokshtanov, M.~S. Ramanujan, and Saket Saurabh.
\newblock The half-integral erd{\"{o}}s-p{\'{o}}sa property for non-null
  cycles.
\newblock {\em CoRR}, abs/1703.02866, 2017.
\newblock URL: \url{http://arxiv.org/abs/1703.02866}, \href
  {http://arxiv.org/abs/1703.02866} {\path{arXiv:1703.02866}}.

\bibitem{LokshtanovR0Z18}
Daniel Lokshtanov, M.~S. Ramanujan, Saket Saurabh, and Meirav Zehavi.
\newblock Reducing {CMSO} model checking to highly connected graphs.
\newblock In Ioannis Chatzigiannakis, Christos Kaklamanis, D{\'{a}}niel Marx,
  and Donald Sannella, editors, {\em 45th International Colloquium on Automata,
  Languages, and Programming, {ICALP} 2018, July 9-13, 2018, Prague, Czech
  Republic}, volume 107 of {\em LIPIcs}, pages 135:1--135:14. Schloss Dagstuhl
  - Leibniz-Zentrum f{\"{u}}r Informatik, 2018.
\newblock \href {https://doi.org/10.4230/LIPIcs.ICALP.2018.135}
  {\path{doi:10.4230/LIPIcs.ICALP.2018.135}}.

\bibitem{MajumdarR18}
Diptapriyo Majumdar and Venkatesh Raman.
\newblock Structural parameterizations of undirected feedback vertex set: {FPT}
  algorithms and kernelization.
\newblock {\em Algorithmica}, 80(9):2683--2724, 2018.
\newblock \href {https://doi.org/10.1007/s00453-018-0419-4}
  {\path{doi:10.1007/s00453-018-0419-4}}.

\bibitem{Marx06}
D{\'{a}}niel Marx.
\newblock Parameterized graph separation problems.
\newblock {\em Theor. Comput. Sci.}, 351(3):394--406, 2006.
\newblock \href {https://doi.org/10.1016/j.tcs.2005.10.007}
  {\path{doi:10.1016/j.tcs.2005.10.007}}.

\bibitem{Marx10}
D{\'{a}}niel Marx.
\newblock Chordal deletion is fixed-parameter tractable.
\newblock {\em Algorithmica}, 57(4):747--768, 2010.
\newblock \href {https://doi.org/10.1007/s00453-008-9233-8}
  {\path{doi:10.1007/s00453-008-9233-8}}.

\bibitem{NesetrilM06}
Jaroslav Nesetril and Patrice~Ossona de~Mendez.
\newblock Tree-depth, subgraph coloring and homomorphism bounds.
\newblock {\em Eur. J. Comb.}, 27(6):1022--1041, 2006.
\newblock \href {https://doi.org/10.1016/j.ejc.2005.01.010}
  {\path{doi:10.1016/j.ejc.2005.01.010}}.

\bibitem{Oum05}
Sang{-}il Oum.
\newblock Rank-width and vertex-minors.
\newblock {\em J. Comb. Theory, Ser. {B}}, 95(1):79--100, 2005.
\newblock \href {https://doi.org/10.1016/j.jctb.2005.03.003}
  {\path{doi:10.1016/j.jctb.2005.03.003}}.

\bibitem{Oum08}
Sang{-}il Oum.
\newblock Approximating rank-width and clique-width quickly.
\newblock {\em {ACM} Trans. Algorithms}, 5(1):10:1--10:20, 2008.
\newblock \href {https://doi.org/10.1145/1435375.1435385}
  {\path{doi:10.1145/1435375.1435385}}.

\bibitem{Oum08a}
Sang{-}il Oum.
\newblock Rank-width and well-quasi-ordering.
\newblock {\em {SIAM} J. Discret. Math.}, 22(2):666--682, 2008.
\newblock \href {https://doi.org/10.1137/050629616}
  {\path{doi:10.1137/050629616}}.

\bibitem{Oum17}
Sang{-}il Oum.
\newblock Rank-width: Algorithmic and structural results.
\newblock {\em Discret. Appl. Math.}, 231:15--24, 2017.
\newblock \href {https://doi.org/10.1016/j.dam.2016.08.006}
  {\path{doi:10.1016/j.dam.2016.08.006}}.

\bibitem{OumS06}
Sang{-}il Oum and Paul~D. Seymour.
\newblock Approximating clique-width and branch-width.
\newblock {\em J. Comb. Theory, Ser. {B}}, 96(4):514--528, 2006.
\newblock \href {https://doi.org/10.1016/j.jctb.2005.10.006}
  {\path{doi:10.1016/j.jctb.2005.10.006}}.

\bibitem{ReedSV04}
Bruce~A. Reed, Kaleigh Smith, and Adrian Vetta.
\newblock Finding odd cycle transversals.
\newblock {\em Oper. Res. Lett.}, 32(4):299--301, 2004.
\newblock \href {https://doi.org/10.1016/j.orl.2003.10.009}
  {\path{doi:10.1016/j.orl.2003.10.009}}.

\bibitem{RobertsonS95b}
Neil Robertson and Paul~D. Seymour.
\newblock Graph minors. {XIII.} {T}he disjoint paths problem.
\newblock {\em J. Comb. Theory, Ser. B}, 63(1):65--110, 1995.
\newblock \href {https://doi.org/10.1006/jctb.1995.1006}
  {\path{doi:10.1006/jctb.1995.1006}}.

\bibitem{SauST20}
Ignasi Sau, Giannos Stamoulis, and Dimitrios~M. Thilikos.
\newblock An fpt-algorithm for recognizing k-apices of minor-closed graph
  classes.
\newblock In Artur Czumaj, Anuj Dawar, and Emanuela Merelli, editors, {\em 47th
  International Colloquium on Automata, Languages, and Programming, {ICALP}
  2020, July 8-11, 2020, Saarbr{\"{u}}cken, Germany (Virtual Conference)},
  volume 168 of {\em LIPIcs}, pages 95:1--95:20. Schloss Dagstuhl -
  Leibniz-Zentrum f{\"{u}}r Informatik, 2020.
\newblock \href {https://doi.org/10.4230/LIPIcs.ICALP.2020.95}
  {\path{doi:10.4230/LIPIcs.ICALP.2020.95}}.

\bibitem{BevernKMN10}
Ren{\'{e}} van Bevern, Christian Komusiewicz, Hannes Moser, and Rolf
  Niedermeier.
\newblock Measuring indifference: Unit interval vertex deletion.
\newblock In Dimitrios~M. Thilikos, editor, {\em Graph Theoretic Concepts in
  Computer Science - 36th International Workshop, {WG} 2010, Zar{\'{o}}s,
  Crete, Greece, June 28-30, 2010 Revised Papers}, volume 6410 of {\em Lecture
  Notes in Computer Science}, pages 232--243, 2010.
\newblock \href {https://doi.org/10.1007/978-3-642-16926-7\_22}
  {\path{doi:10.1007/978-3-642-16926-7\_22}}.

\end{thebibliography}
